\documentclass[sigconf,nonacm]{acmart} 

\usepackage[utf8]{inputenc}

\usepackage{mathtools}
\usepackage{calc}
\usepackage{stmaryrd}
\usepackage{physics}

\usepackage{multicol}
\usepackage{array}
\usepackage{enumerate}
\usepackage[shortlabels]{enumitem}
\usepackage{adjustbox}
\usepackage{changepage}

\usepackage{crossreftools}
\usepackage{hyperref}
\usepackage{cleveref}

\usepackage{csquotes}

\usepackage{tikzit}
\usepackage{pgfplots}
\usepackage{tikz-cd}
\usepackage{nicematrix}

\usepackage{etoolbox}
\usepackage{tikz}  

\usepackage{todonotes}

\usepackage{lineno}

\usepackage{wrapfig}

\definecolor{zx_red}{RGB}{232, 165, 165}
\definecolor{zx_green}{RGB}{216, 248, 216}

\definecolor{nice_blue}{rgb}{0.36, 0.54, 0.66}
\definecolor{nice_yellow}{RGB}{253, 253, 150}

\usetikzlibrary{
    decorations.markings,
    decorations.pathmorphing,
    decorations.text,
    positioning,
    backgrounds,
    folding,
    shapes.geometric,
    shapes.misc,
    patterns,
    patterns.meta,
    arrows.meta
}

\pgfdeclarelayer{edgelayer}
\pgfdeclarelayer{nodelayer}
\pgfsetlayers{background,edgelayer,nodelayer,main}

\tikzstyle{every picture}=[baseline=-0.25em,scale=0.5]

\pgfkeys{/tikz/tikzit fill/.initial=0}
\pgfkeys{/tikz/tikzit draw/.initial=0}
\pgfkeys{/tikz/tikzit shape/.initial=0}
\pgfkeys{/tikz/tikzit category/.initial=0}

\tikzset{every path/.style={draw=black!80, line width=0.7pt}}
\tikzstyle{none}=[inner sep=0mm]

\tikzset{
    tickedge/.style={
        decoration={ markings,
            mark=at position .5 with {\draw (0,2pt) -- (0,-2pt);}
        },
        postaction={decorate}
    },
}

\tikzset{font=\normalsize}

\newif\ifsensitivetikz
\sensitivetikztrue

\newcommand{\tikzfigdir}{figures}
\newcommand{\tikzpdfdir}{figures/figures-pdf}

\renewcommand{\tikzfig}[1]{%
  \ifsensitivetikz
    \IfFileExists{\tikzfigdir/#1.tikz}{%
      {%
        \tikzstyle{every picture}=[tikzfig]%
        \input{\tikzfigdir/#1.tikz}%
      }%
    }{%
      \PackageError{tikzfig}%
        {Missing sensitive TikZ figure `#1'}%
        {Expected file: \tikzfigdir/#1.tikz}%
    }%
  \else
    \IfFileExists{\tikzpdfdir/#1.pdf}{%
      \includegraphics[scale=1,align=c]{\tikzpdfdir/#1.pdf}%
    }{%
    \IfFileExists{\tikzfigdir/#1.tikz}{%
      {%
        \tikzstyle{every picture}=[tikzfig]%
        \input{\tikzfigdir/#1.tikz}%
      }%
    }{%
      
      \PackageError{tikzfig}%
        {Missing precompiled figure `#1'}%
        {Expected file: \tikzpdfdir/#1.pdf}%
        }%
    }
  \fi
}

\newcommand{\scalefig}{0.7}

\usetikzlibrary{calc,fit}

\newcommand{\axiomboxscale}{0.78}

\newcommand{\axiomboxpic}[1]{%
    \scalebox{\axiomboxscale}{%
        \scalebox{\scalefig}{\input{./figures/#1.tikz}}%
    }%
}

\newcommand{\axiomboxrel}[1]{%
    \ensuremath{\eq\limits^{\scriptscriptstyle #1}}%
}

\newcommand{\axiomboxeq}[3]{%
    \begin{tabular}{@{}c@{\hspace{0.35em}}c@{\hspace{0.35em}}c@{}}
    \phantomsection\label{eq:#1}
        \axiomboxpic{#2}
        &
        \axiomboxrel{#1}
        &
        \axiomboxpic{#3}
    \end{tabular}%
}

\newcommand{\axiomboxeqthree}[4]{%
    \begin{tabular}{
        @{}c
        @{\hspace{0.25em}}c
        @{\hspace{0.25em}}c
        @{\hspace{0.25em}}c
        @{\hspace{0.25em}}c@{}
    }
        \phantomsection\label{eq:#1}
        \axiomboxpic{#2}
        &
        \axiomboxrel{#1}
        &
        \axiomboxpic{#3}
        &
        \axiomboxrel{#1}
        &
        \axiomboxpic{#4}
    \end{tabular}%
}

\newcommand{\generatorboxscale}{1}

\newcommand{\generatorboxpic}[1]{%
    \scalebox{\generatorboxscale}{%
        \scalebox{\scalefig}{\input{./figures/#1.tikz}}%
    }%
}

\newcommand{\generatorboxsem}[2]{%
    \begin{tabular}{@{}c@{\hspace{0.35em}}c@{\hspace{0.35em}}c@{}}
        $\interp{\generatorboxpic{#1}}_{\LFM}$
        &
        $\eq$
        &
        {\small $#2$}
    \end{tabular}%
}

\newcommand{\defaultboxscale}{1}

\newcommand{\boxpic}[1]{%
    \scalebox{\defaultboxscale}{%
        \scalebox{\scalefig}{\input{./figures/#1.tikz}}%
    }%
}

\newcommand{\boxrel}[1]{%
    \ensuremath{\eq\limits^{\scriptscriptstyle #1}}%
}

\newcommand{\boxeq}[3]{%
    \begin{tabular}{@{}c@{\hspace{0.35em}}c@{\hspace{0.35em}}c@{}}
        \boxpic{#2}
        &
        \boxrel{#1}
        &
        \boxpic{#3}
    \end{tabular}%
}

\newenvironment{derivationchain}{%
    \begin{center}%
    \setlength{\tabcolsep}{3pt}%
    \renewcommand{\arraystretch}{1.20}%
    \begin{adjustbox}{max width=\textwidth}%
    \begin{tabular}{ccccccc}%
}{%
    \end{tabular}%
    \end{adjustbox}%
    \end{center}%
}

\newcommand{\EaWGeneratorBlock}{%
\begin{minipage}[t]{7.8cm}
\centering
\setlength{\tabcolsep}{2pt}

\begin{tabular}{cc}

\generatorboxsem
    {EaW-gen/ket-0}
    {\ket{0}}
&
\generatorboxsem
    {EaW-gen/bra-0}
    {\bra{0}}
\\[1em]

\multicolumn{2}{c}{
    \generatorboxsem
        {EaW-gen/lw}
        {
            \ketbra{00}{0}
            +
            \ketbra{01}{1}
            +
            \ketbra{10}{1}
        }
}
\\[1em]

\multicolumn{2}{c}{
    \generatorboxsem
        {EaW-gen/rw}
        {
            \ketbra{0}{00}
            +
            \ketbra{1}{01}
            +
            \ketbra{1}{10}
        }
}
\\[1em]

\multicolumn{2}{c}{
    \generatorboxsem
        {EaW-gen/2-Z}
        {
            \ketbra{0}{0}
            +
            x\ketbra{1}{1}
        }
}

\end{tabular}
\end{minipage}%
}

\newcommand{\EpWGeneratorBlock}{%
\begin{minipage}[t]{7.8cm}
\centering
\setlength{\tabcolsep}{2pt}

\begin{tabular}{c}

\generatorboxsem
    {EpW-gen/cap}
    {
        \ket{00}+\ket{11}
    }
\\[1em]

\generatorboxsem
    {EpW-gen/cup}
    {
        \bra{00}+\bra{11}
    }

\end{tabular}
\end{minipage}%
}

\newcommand{\EZWGeneratorBlock}{%
\begin{minipage}[t]{8.8cm}
\centering
\setlength{\tabcolsep}{2pt}

\begin{tabular}{c}

\generatorboxsem
    {EZWbis-gen/4-Z}
    {
        \ketbra{00}{00}
        +
        x\ketbra{11}{11}
    }

\end{tabular}
\end{minipage}%
}

\newcommand{\aWGeneratorBlock}{%
\begin{minipage}[t]{7.8cm}
\centering
\setlength{\tabcolsep}{2pt}

\begin{tabular}{cc}

\generatorboxsem
    {aW-gen/1-w}
    {\ket{1}}
&
\generatorboxsem
    {aW-gen/co1-w}
    {\bra{1}}

\end{tabular}

\par\vspace*{0.7em}
{\small\itshape non-natural generators}

\end{minipage}%
}

\newcommand{\FpWGeneratorBlock}{%
\begin{minipage}[t]{7.8cm}
\centering
\setlength{\tabcolsep}{2pt}

\begin{tabular}{c}

\generatorboxsem
    {FpW-gen/2-w}
    {
        \ketbra{0}{1}
        +
        \ketbra{1}{0}
    }

\end{tabular}

\par\vspace*{0.7em}
{\small\itshape non-natural generator}

\end{minipage}%
}

\newcommand{\FZWGeneratorBlock}{%
\begin{minipage}[t]{8.8cm}
\centering

{\large\itshape No additional generator}

\end{minipage}%
}

\newcommand{\EaWAxiomBlock}{%
\begin{minipage}[t]{8.8cm}
\centering
\renewcommand{\arraystretch}{1.55}
\setlength{\tabcolsep}{2pt}

\begin{tabular}{ccc}

\axiomboxeq{A1}
    {EaW-rules/assoc_1}
    {EaW-rules/assoc_2}
&
\axiomboxeq{A2}
    {EaW-rules/bialg_1}
    {EaW-rules/bialg_2}
&
\axiomboxeq{A3}
    {EaW-rules/commut_1}
    {EaW-rules/commut_2}
\\[1em]

\axiomboxeq{A4}
    {EaW-rules/copy_1}
    {EaW-rules/copy_2}
&
\axiomboxeq{A5}
    {EaW-rules/erase_1}
    {EaW-rules/erase_2}
&
\axiomboxeq{A6}
    {EaW-rules/neutral_1}
    {EaW-rules/neutral_2}
\\[1em]

\axiomboxeq{A7}
    {EaW-rules/z-add_1}
    {EaW-rules/z-add_2}
&
\axiomboxeq{A8}
    {EaW-rules/z-copy_1}
    {EaW-rules/z-copy_2}
&
\axiomboxeq{A9}
    {EaW-rules/z-mult_1}
    {EaW-rules/z-mult_2}
\\[1em]

\axiomboxeq{A10}
    {EaW-rules/z-zero_1}
    {EaW-rules/z-zero_2}
&
\axiomboxeq{A11}
    {EaW-rules/empty_1}
    {EaW-rules/empty_2}
&
\axiomboxeq{A12}
    {EaW-rules/z-one_1}
    {EaW-rules/z-one_2}

\end{tabular}
\end{minipage}%
}

\newcommand{\EpWAxiomBlock}{%
\begin{minipage}[t]{10.8cm}
\centering
\renewcommand{\arraystretch}{1.55}
\setlength{\tabcolsep}{2pt}

\begin{tabular}{ccc}

\axiomboxeq{A13}
    {EpW-rules/loop_1}
    {EpW-rules/loop_2}
&
\axiomboxeq{A14}
    {EpW-rules/zero-float_1}
    {EpW-rules/zero-float_2}
&
\axiomboxeq{A15}
    {EpW-rules/z-cup-cap_1}
    {EpW-rules/z-cup-cap_2}
\\[1em]

\axiomboxeq{A16}
    {EpW-rules/cup-cap-triangle_1}
    {EpW-rules/cup-cap-triangle_2}
&
\axiomboxeq{A17}
    {EpW-rules/empty_1}
    {EpW-rules/empty_2}
&
\axiomboxeqthree{A18}
    {EpW-rules/snake_1}
    {EpW-rules/snake_2}
    {EpW-rules/snake_3}
\\[1em]

\axiomboxeq{A19}
    {EpW-rules/state-cup-cap_1}
    {EpW-rules/state-cup-cap_2}
&
\axiomboxeq{A20}
    {EpW-rules/state-triangle_1}
    {EpW-rules/state-triangle_2}
&
\axiomboxeq{A21}
    {EpW-rules/anticom_1}
    {EpW-rules/anticom_2}
\\[1em]

\axiomboxeq{A22}
    {EpW-rules/float_1}
    {EpW-rules/float_2}
&
\axiomboxeq{A23}
    {EpW-rules/op_1}
    {EpW-rules/op_2}
&
\axiomboxeq{A24}
    {EpW-rules/triangle-gadget_1}
    {EpW-rules/triangle-gadget_2}

\end{tabular}
\end{minipage}%
}

\newcommand{\aWAxiomBlock}{%
\begin{minipage}[t]{8.8cm}
\centering
\renewcommand{\arraystretch}{1.60}
\setlength{\tabcolsep}{2pt}

\begin{tabular}{ccc}

\axiomboxeq{A25}
    {aW-rules/dot-snake_1}
    {aW-rules/dot-snake_2}
&
\axiomboxeq{A26}
    {aW-rules/float_1}
    {aW-rules/float_2}
&
\axiomboxeq{A27}
    {aW-rules/zero-float_1}
    {aW-rules/zero-float_2}
\\[1em]

\axiomboxeq{A28}
    {aW-rules/op_1}
    {aW-rules/op_2}
&
\axiomboxeq{A29}
    {aW-rules/antinat_1}
    {aW-rules/antinat_2}
&
\axiomboxeq{A30}
    {aW-rules/zero-op_1}
    {aW-rules/zero-op_2}

\end{tabular}
\end{minipage}%
}

\newcommand{\FpWAxiomBlock}{%
\begin{minipage}[t]{10.8cm}
\centering
\renewcommand{\arraystretch}{1.60}
\setlength{\tabcolsep}{2pt}

\begin{tabular}{ccc}

\axiomboxeq{A31}
    {FpW-rules/loop-zero_1}
    {FpW-rules/loop-zero_2}
&
\axiomboxeq{A32}
    {FpW-rules/dot-cup-cap_1}
    {FpW-rules/dot-cup-cap_2}
&
\axiomboxeq{A33}
    {FpW-rules/float_1}
    {FpW-rules/float_2}
\\[1em]

\axiomboxeq{A34}
    {FpW-rules/flex_1}
    {FpW-rules/flex_2}
&
\axiomboxeq{A35}
    {FpW-rules/invol_1}
    {FpW-rules/invol_2}
&
\axiomboxeq{A36}
    {FpW-rules/1w-2w_1}
    {FpW-rules/1w-2w_2}
\\[1em]

\axiomboxeq{A37}
    {FpW-rules/state-cup_1}
    {FpW-rules/state-cup_2}

\end{tabular}
\end{minipage}%
}

\newcommand{\SEaWDerivedBlock}{%
\begin{minipage}[t]{8.8cm}
\centering
\renewcommand{\arraystretch}{1.55}
\setlength{\tabcolsep}{2pt}
\begin{tabular}{ccc}
\axiomboxeq{L1}
    {SEaW-lemma/add-mat_1}
    {SEaW-lemma/add-mat_2}
&
\axiomboxeq{L2}
    {SEaW-lemma/copy-mat_1}
    {SEaW-lemma/copy-mat_2}
&
\axiomboxeq{L3}
    {SEaW-lemma/mult-mat_1}
    {SEaW-lemma/mult-mat_2}
\\[1em]
\axiomboxeq{L4}
    {SEaW-lemma/erase_1}
    {SEaW-lemma/erase_2}
&
\axiomboxeq{L5}
    {SEaW-lemma/zero-mat_1}
    {SEaW-lemma/zero-mat_2}
&
\axiomboxeq{L6}
    {SEaW-lemma/direct-sum-mat_1}
    {SEaW-lemma/direct-sum-mat_2}
\\[1em]
\axiomboxeq{L7}
    {SEaW-lemma/lambda-mat_1}
    {SEaW-lemma/lambda-mat_2}
&
\axiomboxeq{L8}
    {SEaW-lemma/id-mat_1}
    {SEaW-lemma/id-mat_2}
&
\axiomboxeq{L9}
    {SEaW-lemma/rank-dec-mat_1}
    {SEaW-lemma/rank-dec-mat_2}
\end{tabular}
\end{minipage}%
}

\newcommand{\SEaWDerivedBlockNoLabel}{%
\begin{minipage}[t]{8.8cm}
\centering
\renewcommand{\arraystretch}{1.55}
\setlength{\tabcolsep}{2pt}
\begin{tabular}{ccc}
\boxeq{L1}
    {SEaW-lemma/add-mat_1}
    {SEaW-lemma/add-mat_2}
&
\boxeq{L2}
    {SEaW-lemma/copy-mat_1}
    {SEaW-lemma/copy-mat_2}
&
\boxeq{L3}
    {SEaW-lemma/mult-mat_1}
    {SEaW-lemma/mult-mat_2}
\\[1em]
\boxeq{L4}
    {SEaW-lemma/erase_1}
    {SEaW-lemma/erase_2}
&
\boxeq{L5}
    {SEaW-lemma/zero-mat_1}
    {SEaW-lemma/zero-mat_2}
&
\boxeq{L6}
    {SEaW-lemma/direct-sum-mat_1}
    {SEaW-lemma/direct-sum-mat_2}
\\[1em]
\boxeq{L7}
    {SEaW-lemma/lambda-mat_1}
    {SEaW-lemma/lambda-mat_2}
&
\boxeq{L8}
    {SEaW-lemma/id-mat_1}
    {SEaW-lemma/id-mat_2}
&
\boxeq{L9}
    {SEaW-lemma/rank-dec-mat_1}
    {SEaW-lemma/rank-dec-mat_2}
\end{tabular}
\end{minipage}%
}

\newcommand{\SEpWDerivedBlock}{%
\begin{minipage}[t]{10.8cm}
\centering
\renewcommand{\arraystretch}{1.55}
\setlength{\tabcolsep}{2pt}
\begin{tabular}{ccc}
\axiomboxeq{L10}
    {SEpW-lemma/gnode-add_1}
    {SEpW-lemma/gnode-add_2}
&
\axiomboxeq{L11}
    {SEpW-lemma/gnode-erase_1}
    {SEpW-lemma/gnode-erase_2}
&
\axiomboxeq{L12}
    {SEpW-lemma/mat-cup_1}
    {SEpW-lemma/mat-cup_2}
\\[1em]
\axiomboxeq{L13}
    {SEpW-lemma/gnode-dec_1}
    {SEpW-lemma/gnode-dec_2}
&
\axiomboxeq{L14}
    {SEpW-lemma/gnode-mat_1}
    {SEpW-lemma/gnode-mat_2}
&
\axiomboxeq{L15}
    {SEpW-lemma/mat-to-gnode_1}
    {SEpW-lemma/mat-to-gnode_2}
\\[1em]
\axiomboxeq{L16}
    {SEpW-lemma/mult-float_1}
    {SEpW-lemma/mult-float_2}
&
\axiomboxeq{L17}
    {SEpW-lemma/gen-anticom_2}
    {SEpW-lemma/gen-anticom_1}
&
\axiomboxeq{L18}
    {SEpW-lemma/float_1}
    {SEpW-lemma/float_2}
\\[1em]
\axiomboxeq{L19}
    {SEpW-lemma/op-gen_1}
    {SEpW-lemma/op-gen_2}
&
\axiomboxeq{L20}
    {SEpW-lemma/gnode-zero_1}
    {SEpW-lemma/gnode-zero_2}
&
\axiomboxeq{L21}
    {SEpW-lemma/invert-triangle_1}
    {SEpW-lemma/invert-triangle_2}
\\[1em]
\axiomboxeq{L22}
    {SEpW-lemma/triangle-cup_1}
    {SEpW-lemma/triangle-cup_2}
&
\axiomboxeq{L23}
    {SEpW-lemma/triangle-simp_1}
    {SEpW-lemma/triangle-simp_2}
&
{}
\end{tabular}
\end{minipage}%
}

\newcommand{\SaWDerivedBlock}{%
\begin{minipage}[t]{8.8cm}
\centering
\renewcommand{\arraystretch}{1.60}
\setlength{\tabcolsep}{2pt}
\begin{tabular}{ccc}
\axiomboxeq{L24}
    {SaW-lemma/det-mat_1}
    {SaW-lemma/det-mat_2}
&
\axiomboxeq{L25}
    {SaW-lemma/scaled-antinat_1}
    {SaW-lemma/scaled-antinat_2}
&
\axiomboxeq{L26}
    {SaW-lemma/dot-cup_1}
    {SaW-lemma/dot-cup_2}
\\[1em]
\axiomboxeq{L27}
    {SaW-lemma/op_1}
    {SaW-lemma/op_2}
&
\axiomboxeq{L28}
    {SaW-lemma/zero-float_1}
    {SaW-lemma/zero-float_2}
&
{}
\end{tabular}
\end{minipage}%
}

\newcommand{\SFpWDerivedBlock}{%
\begin{minipage}[t]{10.8cm}
\centering
\renewcommand{\arraystretch}{1.60}
\setlength{\tabcolsep}{2pt}
\begin{tabular}{ccc}
\axiomboxeq{L29}
    {SFpW-lemma/antiflex_1}
    {SFpW-lemma/antiflex_2}
&
\axiomboxeq{L30}
    {SFpW-lemma/unconnect_1}
    {SFpW-lemma/unconnect_2}
&
\axiomboxeq{L31}
    {SFpW-lemma/w-2-to-1_1}
    {SFpW-lemma/w-2-to-1_2}
\\[1em]
\axiomboxeq{L32}
    {SFpW-lemma/mat-x-interact_1}
    {SFpW-lemma/mat-x-interact_2}
&
\axiomboxeq{L33}
    {SFpW-lemma/scaled-antinat_1}
    {SFpW-lemma/scaled-antinat_2}
&
\axiomboxeq{L34}
    {SFpW-lemma/co-antiflex_1}
    {SFpW-lemma/co-antiflex_2}
\\[1em]
\axiomboxeq{L35}
    {SFpW-lemma/gnode-invert_1}
    {SFpW-lemma/gnode-invert_2}
&
\axiomboxeq{L36}
    {SFpW-lemma/dot-cup_1}
    {SFpW-lemma/dot-cup_2}
&
{}
\end{tabular}
\end{minipage}%
}

\DeclareRobustCommand{\svdots}{
  \vbox{%
    \baselineskip=0.33333\normalbaselineskip
    \lineskiplimit=0pt
    \hbox{.}\hbox{.}\hbox{.}%
    \kern-0.2\baselineskip
  }%
}

\newcommand{\eq}{=}
\renewcommand{\eq}{\mathop{=}}
\newcommand{\eqann}[2]{%
  \eq\limits_{#1}^{{\color{gray}\scriptstyle #2}}%
}

\newcommand{\ruleref}[2]{\hyperref[#1-rule:#2]{#2}}
\newcommand{\lemmaref}[1]{\hyperref[lem:#1]{#1}}
\newcommand{\eqhref}[1]{\hyperref[eq:#1]{#1}}

\newcommand{\N}{\ensuremath{\mathbb{N}}}
\newcommand{\Z}{\ensuremath{\mathbb{Z}}}

\newcommand{\C}{\ensuremath{\mathbb{C}}}
\newcommand{\A}{\ensuremath{\mathcal{A}}}
\newcommand{\G}{\ensuremath{\mathcal{G}}}
\newcommand{\U}{\ensuremath{\mathcal{U}}}
\newcommand{\Obj}{\ensuremath{\mathrm{Obj}}}

\newcommand{\T}{\ensuremath{\mathcal{T}}}
\renewcommand{\P}{\ensuremath{\mathcal{P}}}
\renewcommand{\S}{\ensuremath{\mathcal{S}}}

\newcommand{\Pf}{\mathrm{pf}}

\newcommand{\LFM}{\textbf{LFM}~}
\newcommand{\Qubit}{\textbf{Qubit}~}
\newcommand{\Qudit}{\textbf{Qudit}~}

\newcommand{\CAR}{\textbf{CAR}~}
\newcommand{\ECAR}{\textbf{EvenCAR}~}
\newcommand{\HCAR}{\textbf{HomCAR}~}
\newcommand{\ELFM}{\textbf{EvenLFM}~}

\newcommand{\EZW}{\textbf{EZW}~}

\newcommand{\FZW}{\textbf{FZW}}

\newcommand{\pW}{\textbf{pW}~}

\newcommand{\FpW}{\textbf{FW}}
\newcommand{\EpW}{\textbf{EW}~}
\newcommand{\aW}{\textbf{FoW}~}
\newcommand{\EaW}{\textbf{EoW}~}

\renewcommand{\to}{\mathop{\rightarrow}}
\newcommand{\rewrites}{\rightsquigarrow}

\DeclareMathSymbol{\shortminus}{\mathbin}{AMSa}{"39}

\makeatletter
\newcommand*{\bigcdot}{\mathpalette\bigcdot@{.5}}
\newcommand*\bigcdot@[2]{\mathbin{\vcenter{\hbox{\scalebox{#2}{$\m@th#1\bullet$}}}}}
\makeatother

\newcommand{\interp}[1]{\left\llbracket #1 \right\rrbracket}

\newcommand{\pmm}[1]{\begin{bsmallmatrix} #1 \end{bsmallmatrix}}

\newcommand{\interN}[2]{\ensuremath{\left\{ #1, \dots, #2\right\}}}

\newenvironment{smallarray}[1]
 {\null\,\vcenter\bgroup\scriptsize
  \renewcommand{\arraystretch}{0.7}%
  \arraycolsep=.13885em
  \hbox\bgroup$\array{@{}#1@{}}} 
 {\endarray$\egroup\egroup\,\null}

\renewcommand{\bra}[1]{\ensuremath{\left\langle #1 \right|}}
\renewcommand{\ket}[1]{\ensuremath{\left|  #1 \right\rangle}}
\renewcommand{\ketbra}[2]{\ket{#1}\!\!\bra{#2}}
\renewcommand{\braket}[2]{\ensuremath{\left\langle #1 | #2 \right\rangle}}

\crefname{notations}{Notations}{Notations}

\newtheorem{definition}{Definition}[subsection]
\crefname{definition}{Definition}{Definitions}

\crefname{history}{Historical note}{Historical notes}

\newtheorem{theorem}{Theorem}[subsection]
\crefname{theorem}{Theorem}{Theorems}

\newtheorem{proposition}{Proposition}[subsection]
\crefname{proposition}{Proposition}{Propositions}

\crefname{conj}{Conjecture}{conjecture}

\newtheorem{lemma}{Lemma}[subsection]
\crefname{lemma}{Lemma}{Lemmas}

\newtheorem{corollary}{Corollary}[subsection]
\crefname{corollary}{Corollary}{Corollaries}

\crefname{remark}{Remark}{Remarks}

\newtheorem*{example}{Example}
\crefname{example}{Example}{Examples}

\crefname{tn}{Illustration}{Illustrations}

\crefname{figure}{Figure}{Figures}

\newcolumntype{L}{>$l<$}

\input{preamble/href}

\tikzstyle{white phase dot}=[minimum size=4mm, font={\footnotesize\boldmath}, shape=rectangle, rounded corners=1.5mm, inner sep=0.8mm, outer sep=-2mm, scale=0.8, draw=black, fill=white]
\tikzstyle{white dot}=[inner sep=0.5mm, minimum width=0pt, minimum height=0pt, fill=white, draw=black, shape=circle]
\tikzstyle{dot}=[inner sep=0.5mm, minimum width=0pt, minimum height=0pt, fill=black, draw=black, shape=circle]
\tikzstyle{fswap}=[inner sep=0.7mm, minimum width=0pt, minimum height=0pt, draw=purple, shape=circle]
\tikzstyle{rmat}=[draw=purple, line width=0.8pt, signal, fill=white, signal to=east, signal from=west, inner sep=1pt, minimum height=6pt]
\tikzstyle{modenb}=[font={\tiny}]
\tikzstyle{presym}=[inner sep=0.7mm, minimum width=0pt, minimum height=0pt, draw=purple, shape=circle]
\tikzstyle{empty_box}=[fill=white, fill opacity=0, draw=black, dashed, shape=rectangle, minimum width=0.5cm, minimum height=0.5cm]
\tikzstyle{1_arrow_box}=[fill=white, draw=black, shape=rectangle, minimum width=0.5cm, minimum height=0.5cm]
\tikzstyle{2_arrow_box}=[fill=white, draw=black, shape=rectangle, minimum width=1cm, minimum height=1cm]
\tikzstyle{rw}=[shape=isosceles triangle, isosceles triangle stretches=true, fill=black, draw=black, minimum width=0.2cm, minimum height=1.5mm, inner sep=1pt, shape border rotate=0, outer sep=-0.1em]
\tikzstyle{lw}=[shape=isosceles triangle, isosceles triangle stretches=true, fill=black, draw=black, minimum width=0.2cm, minimum height=1.5mm, inner sep=1pt, shape border rotate=0, shape border rotate=180, outer sep=-0.1em]

\tikzstyle{thick}=[-, draw=black, line width=1.25pt, tikzit draw={rgb,255: red,122; green,122; blue,122}]
\tikzstyle{double}=[-, draw, double]
\tikzstyle{bg}=[-, draw=none,pattern={Dots[angle = 45, distance = 1.5pt]} , pattern color = gray!40]

\begin{document}

\title{Premonoidal Semantics and Scalable Diagrammatics of Fermionic Quantum Computing}
\author{Thomas Perez}
\orcid{0009-0006-7174-4535}
\affiliation{
  \institution{Inria}
  \institution{LIX, École Polytechnique}
  \institution{CPHT, CNRS}
  \institution{Institut Polytechnique de Paris}
  \city{Palaiseau}
  \country{France}
}
\email{thomas.perez@inria.fr}

\author{Titouan Carette}
\affiliation{
  \institution{LIX, École Polytechnique}
  \institution{CPHT, CNRS}
  \institution{Institut Polytechnique de Paris}
  \city{Palaiseau}
  \country{France}
}
\email{titouan.carette@polytechnique.edu}

\begin{abstract}
    Local fermionic mode (LFM) based quantum computation has pure state spaces that are isomorphic, via the Jordan–Wigner representation, to qubit state spaces, but its compositional structure is subtly different. Indeed, the algebraic formalism specifying how to embed fermionic systems underlies a notion of parallel composition, for which, in general, the usual interchange law of monoidal categories fails. We give a categorical account of this phenomenon. Starting from the CAR-algebraic semantics of fermionic gate application, we show that LFM processes form a symmetric premonoidal category whose centre is precisely the even, parity-preserving subcategory. The same processes can alternatively be organized with the ordinary tensor product and a fermionic presymmetry that is natural exactly on even maps. To reason diagrammatically in this latter presentation, we introduce pronaps, a relaxation of the notion of prop, in which permutations are not-necessarily-natural,  and use them to organize a hierarchy of fragments of the fermionic ZW calculus relevant to the study of  fermionic circuits on physically motivated gate-sets. We then extend scalable diagrammatic notation to pronaps. In this setting, syntactic sugars inspired by SZX and GSA are defined, with the specificity that matrix arrows encode minors and determinants, while graph triangles encode pfaffians of principal submatrices. These constructions yield elegant normal forms and completeness proofs for our presentations of the EoW, FoW, EW, and FW fragments. In particular, the FW presentation gives a new normal form for the matchgate fragment, distinct from the rWGS-X normal form of the planar-W (pW) presentation of the same category. The resulting framework connects fermionic circuit semantics, diagrammatic rewriting with scalable notations, and the algebra of determinants and pfaffians.
\end{abstract}

\maketitle
\thispagestyle{empty}

\section{Introduction}
        \paragraph{Motivation.}
        Quantum computation turns the physical principles of quantum mechanics into
        computational models on which algorithms can be specified, composed, and 
        compared to their classical counterparts. 
        The quantum circuit model achieves this by decomposing quantum computations
        into elementary gates acting on finite registers of qubits. Gates acting on
        subregisters can be assembled into larger circuits by sequential and parallel
        composition. Each gate is interpreted as a unitary map, and this compositional
        structure is naturally captured by symmetric monoidal categories. This symmetric 
        monoidal structure can be represented by string diagrams, with
        sequential and parallel composition encoded by graphical composition. Calculi
        such as the ZX-calculus~\cite{CoeckeDuncan2011} and the ZW-calculus~\cite{CoeckeKissinger2010, Hadzihasanovic2015} 
        axiomatize such categories via complete equational theories~\cite{JeandelPerdrixVilmart2020,Hadzihasanovic2015}, allowing semantics equalities to be derived by 
        diagrammatic rewriting.

        Fermionic quantum computing provides a particularly instructive test of this
        compositional picture. In the local-fermionic-mode (LFM) model of Bravyi and
        Kitaev \cite{BravyiKitaev2002}, information is encoded in the occupation of
        fermionic modes and transformations are naturally described in the CAR algebra
        generated by fermionic creation and annihilation operators. After choosing an ordering of
        the modes, the Jordan--Wigner representation identifies the state space of $m$
        fermionic modes with that of $m$ qubits. We argue that this vector-space coincidence, however,
        does not make the two compositional theories identical: fermionic exchange
        statistics remain visible when operations are applied to disjoint modes.
        The obstruction already appears for the one-mode flip gate $F\{i\}=f_i+f_i^\dagger$. For $i\neq j$, then:
        \[
        F\{i\}F\{j\}=-F\{j\}F\{i\}.
        \]
        More generally, the algebraic formalism specifying how fermionic systems are embedded
        into larger registers induces a notion of parallel composition for which the
        ordinary interchange law fails in general. This failure matches the
        super-interchange law of monoidal supercategories
        \cite{BrundanEllis2017}. We nevertheless describe it as a premonoidal structure,
        which is conceptually simpler and invites an operational interpretation of the gates 
        responsible for this failure through comparison with computational effects \cite{PowerRobinson1997}.

        Diagrammatic approaches to fermionic computation already exist. In particular, de
        Felice, Hadzihasanovic and Ng introduced a fermionic fragment of the ZW-calculus (\FZW) and proved it
        complete for homogeneous linear maps organized with the ordinary tensor product 
        \cite{DeFeliceHadzihasanovicNg2019}. Its even fragment (\EZW) was also given a complete axiomatization
        in \cite{Ng2018}. Our aim is complementary: we first identify
        the premonoidal structure induced directly by the semantics underlying Bravyi and
        Kitaev's LFM model~\cite{BravyiKitaev2002}, and then relate it precisely to the monoidal 
        presentation used for diagrammatic rewriting.
    \paragraph{Contributions.}
        Our first contribution is a categorical semantics for LFM circuits derived from
        their CAR-algebraic gate application. We show that CAR operators
        form a symmetric premonoidal category and that its centre is exactly the even
        subcategory. Via the Jordan--Wigner representation, this structure extends to
        homogeneous linear maps between registers of different sizes, including
        preparations and effects. On the same category of LFM maps, we identify two
        complementary presentations: with the ordinary tensor product the structure is
        monoidal but the fermionic permutation is natural only for even maps; with the
        CAR-induced premonoidal product the fermionic permutation is natural, 
        while two odd maps on disjoint systems fail monoidal interchange by a
        sign. This makes explicit the relation between the operator-algebraic and
        diagrammatic viewpoints on fermionic locality.

        Our second contribution is a diagrammatic framework adapted to the monoidal
        presentation with a non-natural fermionic permutation. We call the resulting categorical
        structures \emph{pronaps}, for PROducts and Non-nAtural Permutations, which are a relaxation of props. 
        Using the fermionic ZW-calculus as a starting point, we organize a hierarchy of fragments tailored to subsets of the
        Bravyi--Kitaev gate set. The even fragments \EaW~ and \EpW~ of \EZW~, are props; adding the
        relevant odd generators yields the pronap fragments \aW~ and \FpW~ of \FZW. The fragments \EpW~ and \aW~
        are new, while \FpW~ is a new pronap presentation of the category previously
        axiomatized by the planar W calculus (\pW)\cite{CaretteEtAl2023}. These fragments
        all are of interest to the study of fermionic circuits on a hierarchy of physically motivated gate-sets.

        Our third contribution is to extend scalable graphical notation \cite{
        CaretteHorsmanPerdrix2019} from props to pronaps. The extension makes explicit
        which equations can be scaled in the presence of non-natural generators and
        provides compact syntactic constructions for the algebraic structures occurring
        in the completeness proofs. In particular, matrix arrows represent weighted
        bipartite graphs and encode all minors of their label matrices, so that
        determinant identities such as Cauchy--Binet, the cofactor formula, Sylvester’s
        identity, and Schur-complement identities become diagrammatic rewrites. Additionally,
        graph triangles labelled by antisymmetric matrices encode principal pfaffians
        and satisfy the corresponding pfaffian identities.
        Finally, these scalable constructions give canonical normal forms and completeness
        proofs for the four fragments \EaW~, \aW~, \EpW~, and \FpW~. The \EaW~ normal form is a
        single matrix arrow. The \aW~ normal form is controlled by a full-rank matrix in
        reduced column-echelon form. The \EpW~ and \FpW~ normal forms combine such a matrix
        with a graph-triangle component, labelled by an antisymmetric matrix. For \FpW~, this yields a new
        normal form for the matchgate fragment, reminiscent of the AP normal form 
        for Clifford ZX-calculus~\cite{CaretteHorsmanPerdrix2019}, distinct from rWGS-X 
        normal-form derived in the planar-W (\pW) presentation~\cite{CaretteEtAl2023}. 
        Together, these results connect fermionic circuit semantics, complete scalable diagrammatic rewriting, and the
        determinant/pfaffian algebra underlying fermionic Gaussian and matchgate computations.
    \paragraph{Structure of the paper.}
        Section~\ref{sec:2} recalls the local-fermionic-mode model, introduces the CAR
        algebra and the Jordan--Wigner representation, and fixes the fermionic gate-set and the 
        fermionic swap used throughout the paper. 
        Section~\ref{sec:3} develops the categorical semantics: we derive the premonoidal structure from CAR
        gate application and establish the trade-off between monoidality and naturality
        of the fermionic permutation. Section~\ref{sec:4} introduces the diagrammatic
        languages, relates the existing fermionic ZW-calculus to our pronap presentation,
        and defines the \EaW~, \aW~, \EpW~, and \FpW~ fragments. Section~\ref{sec:5} extends
        scalable notation to pronaps, develops matrix arrows and graph triangles, and
        uses them to obtain the normal forms and completeness results. The appendices
        contain the proofs of the categorical statements, derived equations, and
        normal-form theorems.

\section{Local Fermionic Mode based quantum computing}\label{sec:2}
    Local Fermionic Mode (LFM) based quantum computing can be traced back, at least in explicit form, to a talk delivered by Seth Lloyd in 1996 at the Santa Barbara workshop on Physics of Information. The idea is perhaps best described in Lloyd's own words:
\textquote[S. Lloyd, Unconventional Quantum Computing Devices,~\cite{Lloyd2000}]{
Perhaps the most obvious basis for an unconventional quantum computer is the use
of particles with non-Boltzmann statistics in a regime where these statistics play a key role
in the dynamics of the device. For example, [I have] proposed the use of fermions
as the fundamental carriers of quantum information, so that a site or state occupied by a
fermion represents a 1 and an unoccupied site or state represents a 0.}

This idea was developed more systematically by Sergei Bravyi and Alexei Kitaev in their 2002
article \emph{Fermionic Quantum Computation}~\cite{BravyiKitaev2002}, which gives a precise
gate-based formulation of quantum computation with local fermionic modes, closely resembling
the usual quantum circuit model. Although their presentation is not phrased in these terms,
Bravyi and Kitaev essentially equip this computational model with a denotational semantics
expressed in the algebraic language of second quantization, using fermionic creation and
annihilation operators.
In \cref{subsec:2-1}, we first give a high-level conceptual discussion of what LFM-based quantum computation is,
and of possible experimental realizations, with the aim of building operational intuition for what the execution
 of an LFM circuit amounts to.
In \cref{subsec:2-2}, we then introduce the algebra of fermionic operators needed to describe
local fermionic modes mathematically. 
Finally, in \cref{subsec:2-3}, we use this formalism to give a precise definition of the Bravyi--Kitaev computational 
model through its denotational semantics.

\subsection{A Fermionic shell game}\label{subsec:2-1}

The computational degrees of freedom in Local Fermionic Mode (LFM) quantum
computation are not individual fermions, but the \emph{modes} they may occupy.
This distinction is important because identical fermions cannot be individually
labelled: they are fundamentally indistinguishable. Their fermionic statistics instead appear
through the antisymmetry of exchange of the many-particle state. Such exchange effects can
already be observed with a few particles, for instance with ultracold fermions
in optical tweezers \cite{BecherEtAl2020} or in electronic Hong--Ou--Mandel
experiments \cite{BocquillonEtAl2013}.
\begin{wrapfigure}{r}{0.55\columnwidth}
    \vspace{-0.8em}
    \centering
    \includegraphics[width=\linewidth]{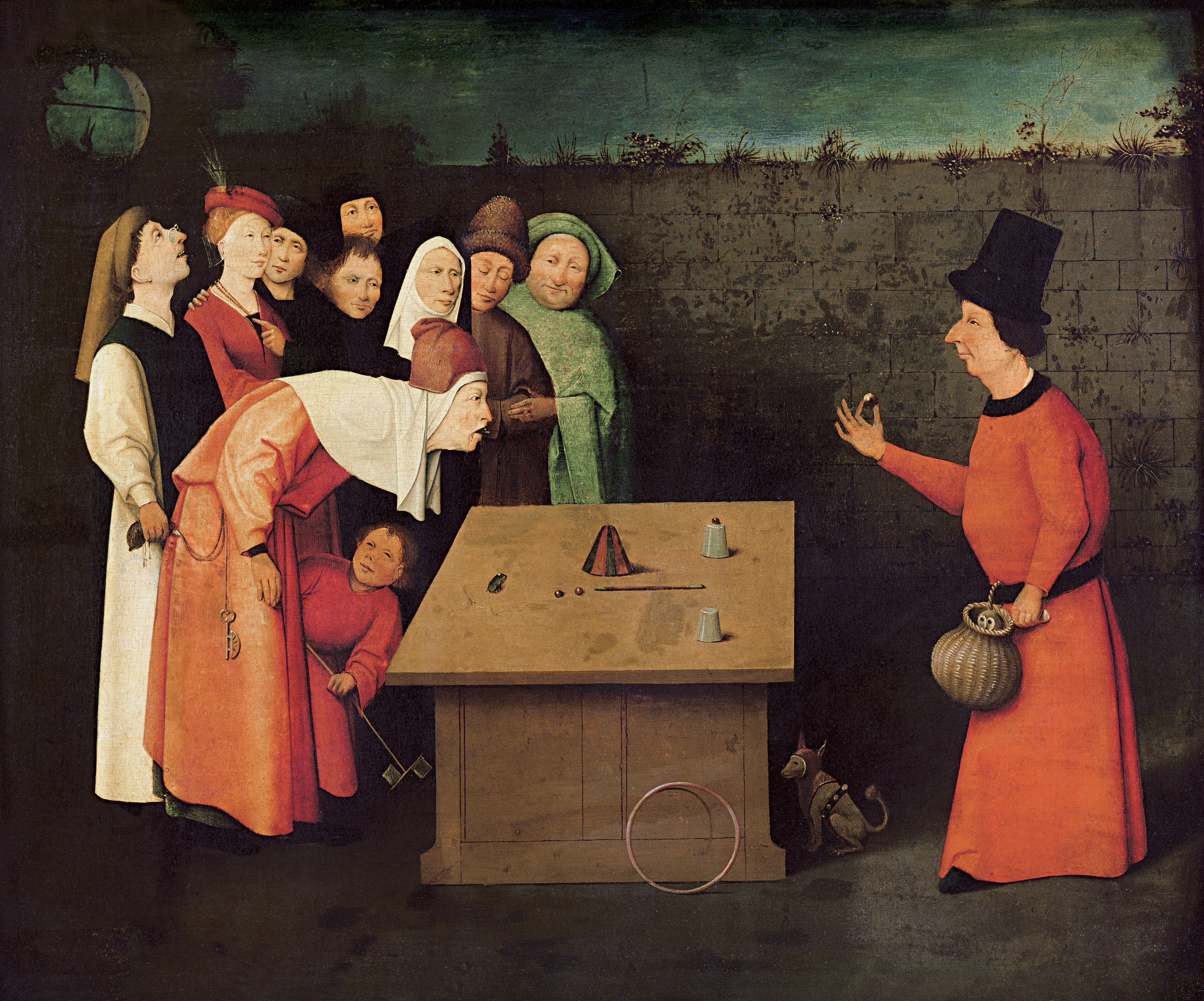}

    \vspace{0.1em}
    {\scriptsize\itshape
        An illustration of the shell-game analogy:
        \emph{The Conjurer}, workshop of Hieronymus Bosch,
        c.~1502. [Public domain; image source: The Yorck Project,
        Wikimedia Commons.]
    \par}
    \vspace{-4em}
\end{wrapfigure}
\\[-0.25em]
A useful operational picture for Local Fermionic Mode computation is that of
a shell game. The shells are the fermionic modes: they are the individually
addressable objects that can be moved, coupled, or measured. The balls
underneath them are indistinguishable fermions. What is recorded
computationally is therefore only whether each mode is occupied or empty,
with at most one fermion per mode by the Pauli exclusion principle.
An LFM register is thus described by a binary occupation pattern,
possibly in a quantum superposition of several such patterns : 
this is where the analogy with a classical shell game stops. 
Coherent transformations of the modes retain information about fermionic statistics:
in particular, exchanging occupied modes may introduce a sign in the
corresponding amplitude, which can become observable through interference.
Thus the modes provide an addressable computational interface, while the
fermionic nature of the particles remains visible in the transformations
acting on their occupation states.

Optical-tweezer platforms give a particularly concrete realization of this
picture~\cite{SerwaneEtAl2011,MurmannEtAl2015,SparEtAl2022,YanEtAl2022}. 
In this setting, a localized single-particle state of a tweezer acts as a fermionic mode: 
changing its phase, allowing tunnelling between modes, or coupling the occupations of several modes gives
an operational interpretation to the elementary gates of an LFM circuit
\cite{GonzalezCuadraEtAl2023,OttEtAl2025}. 
Such platforms normally conserve the total number of physical
atoms, whereas the general LFM model only imposes the weaker parity
superselection rule: fermionic systems do not admit coherent superpositions
between states of even and odd total fermion number
\cite{WickWightmanWigner1952,Friis2016,VidalEtAl2021}. Consequently, some
operations allowed by the abstract LFM formalism may require additional
physical resources on a number-conserving platform. We return to this point
when introducing our working fermionic gate-sets.

The following subsection introduces the algebraic formalism used to make these
ideas precise and to give a denotational semantics to LFM computation.

\subsection{Fermionic creation and annihilation operators}\label{subsec:2-2}

    A system of $m\in\N$ fermionic modes will be called a fermionic register of size $m$,
    where $m$ counts modes rather than particles.
    We now introduce the \emph{$*$-algebra}ic formalism used to describe such registers. \\[-1.5em]
        \begin{definition}[$*$-Algebra]
        A $*$-algebra $(\A,*)$ over $\C$ is an associative complex algebra
        $\A$ equipped with a conjugate-linear involution $*: \A \to \A$
        satisfying
            \[
                (ab)^* = b^*a^*.
            \]
        \end{definition}

    Fermionic registers are modeled using the abstract operator $*$-algebra with involution denoted $(-)^\dagger$
    generated by creation ($f^\dagger_{(-)}$) and annihilation ($f_{(-)}$) operators, which respectively create and annihilate particles in a given mode.
    These operators are subject to Canonical Anticommutation Relations (CAR), hence the name of the operator algebra.
    
    \begin{definition}[CAR algebra]
       The CAR algebra for $m$ modes $\A_m$, with involution denoted $(-)^\dagger$,
       is the unital $*$-algebra freely generated by operators $(f_1, \ldots, f_m)$, and subject to:
        \[
        f_i f_j + f_j f_i = 0 \qquad
        f_i^\dagger f_j^\dagger + f_j^\dagger f_i^\dagger = 0 \qquad
        f_i^\dagger f_j + f_j f_i^\dagger = \delta_{ij} 1.
        \]
    \end{definition}
       
    We denote elements of the algebra $a \in \A_m$ as non-commutative polynomials in the creation and annihilation operators. 
   This non-commutative polynomial is not unique. Indeed, using the CAR we can reorder the variables to rewrite uniquely 
   any element of $\A_m$ as:
    \[ 
    a = \sum_{\substack{I,J \subseteq \interN{1}{m}}} a_{I,J}f_{i_1}^\dagger \cdots f_{i_k}^\dagger f_{j_1} \cdots f_{j_l},
    \]
    where $I = \{i_1 < \cdots < i_k\}$ and $J = \{j_1 < \cdots < j_l\}$ are subsets of $\interN{1}{m}$, and $a_{I,J} \in \mathbb{C}$.

    The CAR preserve the parity of the degree of the monomials, which justifies the following definitions:

    \begin{definition}[Even operator, Odd operator, Homogeneous operator]
        An element $a \in \A_m$ is said:
        \begin{itemize}[label=\emph{|}]
            \item Even if it is a linear combination of monomials of even degrees,
            \item Odd if it is a linear combination of monomials of odd degrees,
            \item Homogeneous if it is either even or odd.
        \end{itemize}
    \end{definition}

    The preceding decomposition is precisely the canonical
    $\mathbb Z_2$-grading of the CAR algebra:
    \[
        \A_m = (\A_m)_{\bar 0} \oplus (\A_m)_{\bar 1},
        \qquad
        (\A_m)_{\bar p}(\A_m)_{\bar q}
        \subseteq
        (\A_m)_{\bar p+\bar q},
    \]
    where the degree is taken modulo $2$. In the usual terminology, 
    $\A_m$ is therefore a superalgebra \cite{Varadarajan2004}.
    It is important, however, that $\A_m$ is not supercommutative in
    general: homogeneous elements $a,b\in\A_m$ do not necessarily satisfy
    \[
        ab = (-1)^{|a||b|}ba.
    \]
    Such a graded commutation relation will hold below under additional
    conditions, in particular for homogeneous operators supported on
    disjoint sets of fermionic modes.

    Beware that the creation and annihilation operators do not correspond to physical transformations that are 
    expected to be implementable on a Fermionic computer. They should be understood as an intuitive way 
    to decompose the actual physical processes. One of the advantages of this decomposition is that the 
    homogeneous operators correspond to processes complying with the \emph{Parity Superselection Rule}. 
    Indeed, as an odd number of creation or annihilation of particles changes the parity of a state, while an even 
    number preserves the parity, non-homogeneous operators would introduce quantum superposition between even and odd
    states, which is forbidden.\\
    Operators of the CAR algebra can thus be thought of as acting on states, which in the case of 
    \emph{pure quantum mechanics} are described as vectors of a Hilbert space~\cite{NielsenChuang2010}. 
    To make this more formal, we will make use of the customary Dirac notations to denote vectors and matrices:


        \begin{definition}[Dirac notations]
            We denote the standard basis of the Hilbert space $\C^2$ as \textbf{kets}:
            \[
            \ket{0} =\pmm{1 \\ 0}, \quad \ket{1} = \pmm{0 \\ 1}.
            \]
                
            For binary words $x, y \in \{0,1\}^*$, the concatenated \textbf{ket} $\ket{xy}$ denotes the Kronecker product of the corresponding vectors:
            \[
            \ket{xy} = \ket{x} \otimes \ket{y}.
            \]

            The \textbf{bra} $\bra{x}$ is the conjugate transpose of the ket $\ket{x}$:
            \[
            \bra{x} = \ket{x}^\dagger.
            \]

            The inner product between a bra and a ket is given by:
            \[
            \braket{i}{j} = \bra{i}\ket{j} = \delta_{ij},
            \]
            where $\delta_{ij}$ is the Kronecker delta.
        \end{definition}

        The CAR algebra $\A_m$ being a finite dimensional $*$-algebra, of dimension $2^{2m} = 2^m \times 2^m$, there 
        is a straightforward way to provide a representation for it in the concrete $*$-algebra of endomorphisms of the 
        Hilbert space $\C^{2^m}$.

        \begin{definition}[Jordan-Wigner representation]
        Let $m \in \mathbb{N}$. The Jordan-Wigner representation 
        \[JW_m: \A_m \to \C^{2^m \times 2^m}\]
        is defined through its action on the computational basis of $\mathbb{C}^{2^m}$, where the binary string $(b_1 \cdots b_m)$ in
        \(
        \ket{b_1 \cdots b_m}
        \)
        is interpreted as the occupation-numbers of the state: mode $i$ is occupied if and only if $b_i=1$.
    
        $JW_m$ is the unique $*$-algebra morphism defined
        by:
        \[
        JW_m(f_k) \ket{x1y}
        =
        (-1)^{|x|}\ket{x0y},
        \qquad
        JW_m(f_k) \ket{x0y}
        =
        0,
        \]
        and
        \[
        JW_m(f_k^\dagger) \ket{x0y}
        =
        (-1)^{|x|}\ket{x1y},
        \qquad
        JW_m(f_k^\dagger) \ket{x1y}
        =
        0,
        \]
        where $x \in \{0,1\}^{k-1}$, $y \in \{0,1\}^{m-k}$, and $|x|$ denotes the Hamming weight of $x$.
        \end{definition}

    The Jordan-Wigner representation is often thought of as a way to encode fermionic systems into qubits. 
    In this work, we merely see it as a way to recover a Hilbert space formalism for pure fermionic systems, in which a happy 
    coincidence on dimension makes the state space of a register of $m$ fermionic modes isomorphic to the state 
    space of $m$ qubits.
    
    Above definitions of odd, even and homogeneous operators are easily extended to any linear maps, including state vectors.
    \begin{definition}[Even, Odd, Homogeneous linear map]
        Let $f = \sum\limits_{xy \in \{0,1\}^{n+m}} f_{xy} \ketbra{y}{x} \in \mathcal{L}(\mathbb{C}^{2^n}, \mathbb{C}^{2^m})$. We say that $f$ is:
        \begin{itemize}[label=\emph{|}]
            \item Even when $\forall xy \in \{0,1\}^{n+m}, f_{xy} \neq 0 \Rightarrow |x| + |y| \equiv 0 \mod 2$.
            \item Odd when $\forall xy \in \{0,1\}^{n+m}, f_{xy} \neq 0 \Rightarrow |x| + |y| \equiv 1 \mod 2$.
            \item Homogeneous when it is either even or odd.
        \end{itemize}
        
    \end{definition}

    This is precisely the usual notion of even and odd morphisms between
    super vector spaces~\cite{BrundanEllis2017}.
    One can check that this definition is consistent with the Jordan-Wigner representation,
    in the sense that $JW_m$ maps even operators to even linear maps, odd to odd, 
    and homogeneous to homogeneous. The representation of a physical operator should 
    therefore be an homogeneous linear map.

    Given its Jordan-Wigner representation, the CAR algebra formalism may seem redundant. Indeed, one could
    try to define a LFM circuit model through a denotational semantics in the more usual framework of Hilbert 
    spaces and linear maps.
    We argue that CAR algebras are central to Bravyi and Kitaev's proposal of such a model~\cite{BravyiKitaev2002}, 
    and that ignoring it would lead to missing important subtleties in the compositionality of Fermionic circuits,
    which we expand on in \cref{sec:3}. As a consequence, \cref{subsec:2-3} defines Fermionic circuits through an 
    operator algebraic denotational semantics. 


\subsection{Bravyi--Kitaev's computational model}\label{subsec:2-3}
As stated before, Bravyi and Kitaev's LFM computational model is analogous to the
quantum circuit model: computations are described by circuits whose gates are
applied to finite tuples of modes. Defining the model essentially amounts to
specifying what constitutes a fermionic gate and how such a gate is applied to
a tuple of modes.

In this work, we restrict ourselves to \emph{pure} Fermionic computations, whose
evolution is described by unitary transformations, as opposed to general mixed
quantum processes~\cite{NielsenChuang2010}. The Parity Superselection Rule
further constrains the closed evolution of a Fermionic system: physical
Hamiltonians, and therefore the unitaries they generate, preserve Fermionic
parity~\cite{BravyiKitaev2002}. We will nevertheless use a slightly broader notion of
Fermionic gate, allowing effective parity-reversing transformations whose
operational meaning will be discussed below.

\begin{definition}[Fermionic gates]
    An element $\U\in\A_m$ is \emph{unitary} when
    \[
        \U\U^\dagger=\U^\dagger\U=1_{\A_m}.
    \]
    A Fermionic gate is a homogeneous unitary $\G\in\A_m$. It is
    even, or parity-preserving, when $\G$ is even, and
    odd, or parity-reversing, when $\G$ is odd. 
\end{definition}

Originally, Bravyi and Kitaev only allow even gates. In particular, they provide a finite
set of local Fermionic gates which is approximately universal for quantum
computation~\cite{BravyiKitaev2002}.  Since $\A_m$ is
finite-dimensional, we use the usual exponential
$\exp(a)=\sum\limits_{k\geq 0}a^k/k!$.

\begin{definition}[Bravyi--Kitaev Fermionic gate set]
    We consider the following gates:
    \[
    \begin{aligned}
        \G_1
        &:= \exp\!\left(i\tfrac{\pi}{4}f_1^\dagger f_1\right),\\
        \G_2
        &:= \exp\!\left(
            i\tfrac{\pi}{4}
            (f_1^\dagger f_2+f_2^\dagger f_1)
        \right),\\
        \G_3
        &:= \exp\!\left(
            i\tfrac{\pi}{4}
            (f_2f_1+f_1^\dagger f_2^\dagger)
        \right),\\
        \G_4
        &:= \exp\!\left(i\pi f_1^\dagger f_1f_2^\dagger f_2\right).
    \end{aligned}
    \]
\end{definition}

All four gates are even, that is to say parity preserving. 
Under the Jordan--Wigner representation, the four Bravyi--Kitaev gates take the following matrix forms:
\[
JW(G_1)=
\begin{bmatrix}
1&0\\
0&e^{i\pi/4}
\end{bmatrix},
\qquad
JW(G_2)=
\begin{bmatrix}
1&0&0&0\\
0&\frac1{\sqrt2}&\frac{i}{\sqrt2}&0\\
0&\frac{i}{\sqrt2}&\frac1{\sqrt2}&0\\
0&0&0&1
\end{bmatrix},
\]
\[
JW(G_3)=
\begin{bmatrix}
\frac1{\sqrt2}&0&0&\frac{i}{\sqrt2}\\
0&1&0&0\\
0&0&1&0\\
\frac{i}{\sqrt2}&0&0&\frac1{\sqrt2}
\end{bmatrix},
\qquad
JW(G_4)=
\begin{bmatrix}
1&0&0&0\\
0&1&0&0\\
0&0&1&0\\
0&0&0&-1
\end{bmatrix}.
\]
\medskip

The gates $\G_1$, $\G_2$ and $\G_4$
also preserve the total particle number: they respectively describe an
occupation-dependent phase, coherent tunnelling between two modes, and an
interaction depending on their joint occupation. By contrast, $\G_3$ coherently
couples states whose particle numbers differ by two. This distinction has a
direct hardware interpretation. Number-conserving neutral-atom processors
naturally provide operations of the $\G_1$, $\G_2$ and $\G_4$ types, while
pairing operations of the $\G_3$ type require an additional resource. Bravyi
and Kitaev discuss such pairing terms in connection with an interaction with a
superconductor~\cite{BravyiKitaev2002}.

Among the gates considered in this work there is supplementary a simple odd operation,
which exchanges the empty and occupied states of a single fermionic mode.

\begin{definition}[Mode-flip gate]
    The \emph{mode-flip gate} is
    \[
        F:=f_1+f_1^\dagger\in\A_1.
    \]
\end{definition}

By the CAR,
\[
    F^\dagger=F,
    \qquad
    F^2
    =f_1f_1^\dagger+f_1^\dagger f_1
    =1,
\]
so $F$ is an odd Hermitian unitary. Under the Jordan--Wigner
representation,
\[
    JW_1(F)= X = \pmm{0 & 1 \\ 1 & 0} .
\]
It is therefore the fermionic analogue of a Pauli-$X$ gate on a single mode.

Although $F$ reverses fermionic parity, it does not create a coherent
superposition of the even and odd sectors. Indeed, $F$ is homogeneous: it
maps every even state to an odd state and every odd state to an even state.
In this sense, the mode-flip is compatible with the parity superselection
rule. This should be distinguished from \emph{parity conservation}, which
would forbid such a transformation on an isolated system.

We are not aware of a current fermionic platform in which $F$ is directly
available as a native operation. Its homogeneity nevertheless leaves open the
possibility of an effective implementation in a setting where the addressed
mode can exchange fermionic parity with additional physical degrees of
freedom. From the viewpoint of the mode alone, such an interaction could
realize $F$, while the complete physical system still preserves its total
fermionic parity.

On platforms whose native dynamics preserve parity, additional resources are
therefore required to realize the mode-flip. This is a fortiori the case for
platforms that conserve the total particle number, such as fermionic
neutral-atom tweezer arrays. A concrete example is provided by
Ott \emph{et al.}~\cite{OttEtAl2025}. They introduce an ancillary fermionic
reference and define a referenced mode for which
\[
    D(\theta)
    :=
    \exp\!\left(i\theta(f_1^\dagger+f_1)\right)
\]
can be implemented using number-conserving operations on the physical atoms.
At $\theta=\pi/2$,
\[
    D(\pi/2)=i(f_1^\dagger+f_1),
\]
which is, up to a global phase, precisely a mode-flip on the referenced mode.
The reference supplies the complementary change in particle number, so that
the effective referenced mode may be flipped while the total number of
physical atoms remains fixed.

Thus $F$ is a legitimate homogeneous gate of the LFM formalism, while
its physical realization depends on which conservation laws are imposed by
the underlying platform.

\paragraph{Application of a gate to Fermionic modes.}
The main difference with the usual quantum circuit model lies in the
application of a gate to a chosen tuple of modes. There, gate application is
defined using tensor products of Hilbert spaces. For Fermionic modes, it is
instead naturally expressed directly in the CAR algebras.

\begin{definition}[Tuple of modes and application morphism]
    Let
    \[
        t:\interN{1}{m}\hookrightarrow\interN{1}{m'}
    \]
    be an injection, which we equivalently regard as a tuple of $m$ pairwise
    distinct modes of an $m'$-mode register. We denote by
    \[
        \iota_t:\A_m\longrightarrow\A_{m'}
    \]
    the unique unital $*$-algebra morphism satisfying
    \[
        \iota_t(f_i)=f_{t(i)},
        \qquad
        \iota_t(f_i^\dagger)=f_{t(i)}^\dagger.
    \]
\end{definition}

\begin{lemma}\label{lem:application-morphism}
    The morphism $\iota_t$ is well defined and injective. Moreover, it preserves
    parity and maps unitary elements to unitary elements.
\end{lemma}

\begin{definition}[Fermionic gate application~\cite{BravyiKitaev2002}]
    Let $\G\in\A_m$ be a Fermionic gate and let
    $t:\interN{1}{m}\hookrightarrow\interN{1}{m'}$ be a tuple of modes.
    The application of $\G$ to $t$ is
    \[
        \G\{t\}:=\iota_t(\G)\in\A_{m'}.
    \]
\end{definition}

Thus gate application amounts algebraically to relabelling the creation and
annihilation operators occurring in the gate. This definition differs from
ordinary qubit gate application, where the ambient tensor-product structure
provides the embedding automatically.

\paragraph{Fermionic permutations.}
Bravyi and Kitaev introduce a Fermionic analogue of the swap gate, which
implements permutations of modes while respecting their exchange statistics.

\begin{definition}[Fermionic swap]
    The \emph{Fermionic swap} on two modes is
    \[
        \pi
        :=
        1-n_1-n_2+f_2^\dagger f_1+f_1^\dagger f_2
        \in\A_2.
    \]
    For distinct $i,j\in\interN{1}{m}$, we write
    $\pi\{i,j\}$ for its application to the tuple $(i,j)$.
\end{definition}

Under the Jordan--Wigner representation, the Fermionic swap is the usual
two-qubit swap with an additional phase on the doubly occupied state:
\[
    JW_2(\pi)
    =
    \ketbra{00}{00}
    +\ketbra{01}{10}
    +\ketbra{10}{01}
    -\ketbra{11}{11}.
\]
Thus $\pi$ exchanges the two occupation modes while accounting for the
minus sign arising from the exchange of two Fermions.

\begin{lemma}\label{lem:fermionic-swap-generators}
    The Fermionic swap is an even Hermitian unitary. Moreover, denoting 
    $\tau_{i,j} \in \mathfrak{S}_m$ the transposition of
    $i\neq j$, we have for every $k\in\interN{1}{m}$,
    \[
        \pi\{i,j\}f_k\pi\{i,j\}^{-1}
        =
        f_{\tau_{i,j}(k)},
        \qquad
        \pi\{i,j\}f_k^\dagger\pi\{i,j\}^{-1}
        =
        f_{\tau_{i,j}(k)}^\dagger.
    \]
\end{lemma}

The Fermionic swaps therefore implement ordinary permutations of the
\emph{mode labels} by conjugation.

\begin{lemma}[Representation of the symmetric group]\label{lem:representation-sym-group}
    For $r\in\interN{1}{m-1}$, let $\tau_r \in \mathfrak{S}_m$ be the transposition exchanging $r$ and $r+1$.
    Let us denote :
    \[
        \pi_r :=\pi\{r,r+1\}.
    \]
    The operators $\pi_r$ satisfy the Coxeter relations
    \[
        \pi_r^2=1,\qquad
        \pi_r\pi_s=\pi_s\pi_r\quad\text{if }|r-s|>1,
        \qquad
        \pi_r\pi_{r+1}\pi_r
        =
        \pi_{r+1}\pi_r\pi_{r+1}.
    \]
    Consequently there is a unique group homomorphism
    \[
        \mathfrak S_m\longrightarrow\A_m^\times,
        \qquad
        \sigma\longmapsto\pi_\sigma,
    \]
    sending $\tau_r$ to $\pi_r$.
\end{lemma}

For an arbitrary transposition $\tau_{i,j}$, the image $\pi_{\tau_{i,j}}$ obtained
from this representation agrees with the element $\pi\{i,j\}$ defined
above: both implement the same conjugation on all CAR generators and both
fix the vacuum.

\begin{definition}[Permutation action on the CAR algebra]
    The preceding representation induces an action
    \[
        \star:\mathfrak S_m\times\A_m\longrightarrow\A_m,
        \qquad
        \sigma\star a
        :=
         \pi_\sigma\,a\, \pi_\sigma^{-1}.
    \]
    Since $ \pi_\sigma^{-1}=\pi_{\sigma^{-1}}$, equivalently
    \[
        \sigma\star a
        =
         \pi_\sigma\,a\,\pi_{\sigma^{-1}}.
    \]
\end{definition}

\begin{lemma}\label{lem:G-equivariant}
    Let
    \[
        t:\interN{1}{m}\hookrightarrow\interN{1}{m'}
    \]
    be a tuple of modes and let $\sigma\in\mathfrak S_{m'}$. Then, for
    every $\G\in\A_m$,
    \[
        \sigma\star\G\{t\}
        =
        \G\{\sigma\circ t\}.
    \]
\end{lemma}

The distinction between even and odd Fermionic gates introduced above will be
crucial in the following section. Even operators supported on disjoint sets of
modes commute, as expected of independent physical operations. Odd operators
do not: for $i\neq j$,
\[
    F\{i\}F\{j\}
    =
    -F\{j\}F\{i\}.
\]
Consequently, once effective odd transformations are included, application on
disjoint modes does not satisfy the ordinary interchange law of a monoidal
category. The following section studies the categorical structure induced by
this obstruction.

\section{Categorical denotational semantics}\label{sec:3}
    In the previous section, we described Fermionic computations algebraically:
a computation on a fixed register is interpreted as an element of a CAR
algebra, while the application of a gate to a chosen collection of modes is
expressed through embeddings between CAR algebras. We now organize these
ingredients into a categorical semantics.

Two features will guide the construction. First, although gates supported on
disjoint sets of modes can be applied independently, odd Fermionic operations
anticommute rather than commute. Their coexistence therefore prevents the
usual tensorial interchange law and naturally leads to a premonoidal
structure. Second, CAR algebras describe endomorphisms of fixed-size
registers, but do not by themselves contain preparations, effects, or
processes between registers of different sizes. After identifying the
premonoidal structure already present in the CAR semantics, we will therefore
use the Jordan--Wigner representation to extend it to a category of general
pure Fermionic processes.

\subsection{The CAR category}
\label{subsec:3-1}

The CAR algebra $\A_m$ may be regarded as a one-object category whose
endomorphisms are the elements of $\A_m$, with composition given by the
algebra product. Collecting these categories for all register sizes gives a
first categorical semantics for Fermionic computations.

\begin{definition}[The $\CAR$ category]
    The category $\CAR$ has objects $\mathbb N$ and hom-sets
    \[
        \CAR[m,n]
        :=
        \begin{cases}
            \A_m & \text{if }m=n,\\
            \varnothing & \text{otherwise}.
        \end{cases}
    \]
    Composition in $\CAR[m,m]$ is multiplication in $\A_m$, and the
    identity on $m$ is $1_{\A_m}$. In particular,
    $\A_0\simeq\C$, so that $\CAR[0,0]\simeq\C$ contains the scalar
    morphisms.

    We denote by $\HCAR$ the wide subcategory whose morphisms are
    homogeneous operators, and by $\ECAR$ the wide subcategory whose
    morphisms are even operators.
\end{definition}

To compare this semantics with the usual Hilbert-space semantics, it is
convenient to use the usual semantics category of qubit registers.

\begin{definition}[The $\Qubit$ category]
    The category $\Qubit$ has objects $\mathbb N$ and
    \[
        \Qubit[m,n]
        :=
        \mathcal L\!\left(\C^{2^m},\C^{2^n}\right).
    \]
\end{definition}


The Jordan-Wigner $*$-algebra isomorphisms naturally extend to a functor 
from $\CAR$ to $\Qubit$.

\begin{definition}[Jordan--Wigner functor]
    The Jordan--Wigner representation defines a functor
    \[
        JW:\CAR\longrightarrow\Qubit
    \]
    by
    \[
        JW(m)=m,
        \qquad
        JW(P)=JW_m(P)
        \quad\text{for }P\in\A_m.
    \]
\end{definition}

Functoriality follows from the fact that each $JW_m$ is a unital
$*$-algebra morphism. Moreover, $JW$ is faithful and, for every
$m\in\mathbb N$, induces an isomorphism
\[
    \CAR[m,m]\simeq\Qubit[m,m].
\]
Thus every square linear map on an $m$-mode occupation space has a unique
description as an element of the CAR algebra. The functor is not full,
however, since $\CAR[m,n]=\varnothing$ whenever $m\neq n$. In particular,
$\CAR$ does not contain state preparations or effects. We return to this
limitation in \cref{subsec:3-3}.

\subsection{Premonoidality and Non-natural symmetry}
\label{subsec:3-2}

Bravyi and Kitaev's notion of applying a gate to a tuple of modes provides
more structure than the category $\CAR$ alone. In particular, a register can
be enlarged by idle modes on either side of an operation. These two
extensions are functorial separately, but, in the presence of odd Fermionic
operations, they do not jointly define a bifunctor. This is naturally captured
the notion of \emph{strict premonoidal category}.

\begin{definition}[Strict premonoidal category]
    A \emph{strict premonoidal category} $(\mathcal C,\owedge,I)$ is a
    category whose objects form a monoid
    $(\Obj(\mathcal C),\owedge,I)$, together with, for every object $A$,
    two functors
    \[
        A\owedge(-),
        \qquad
        (-)\owedge A,
    \]
    whose action on objects is given by $\owedge$, and satisfying
    \[
    \begin{aligned}
        I\owedge f &= f = f\owedge I,\\
        A\owedge(B\owedge f) &= (A\owedge B)\owedge f,\\
        A\owedge(f\owedge B) &= (A\owedge f)\owedge B,\\
        (f\owedge A)\owedge B &= f\owedge(A\owedge B).
    \end{aligned}
    \]
\end{definition}

    Strict premonoidal categories admit a string-diagram notation similar to
    monoidal categories. Morphisms are represented by boxes and objects by
    wires. Consider a morphism $f: B \to C$. Its string diagram presentation is
    $\scalebox{\scalefig}{\input{./figures/premonoidal-cat/f.tikz}}$.
    The composition $f\circ g$ for $g : A \to B$ si pictured as 
    $\scalebox{\scalefig}{\input{./figures/premonoidal-cat/gf.tikz}}$.
    The functors $A\owedge(-)$ and $(-)\owedge A$ correspond respectively
    to adding an idle wire above and below its diagram:
    \[
        \scalebox{\scalefig}{\input{./figures/premonoidal-cat/Awedgef_1.tikz}}
        \quad := \quad
        \scalebox{\scalefig}{\input{./figures/premonoidal-cat/Awedgef_2.tikz}}
    \]
    \[
        \scalebox{\scalefig}{\input{./figures/premonoidal-cat/fwedgeA_1.tikz}}
        \quad := \quad
        \scalebox{\scalefig}{\input{./figures/premonoidal-cat/fwedgeA_2.tikz}}.
    \]
    Contrary to the monoidal case, however, two boxes on disjoint wires
    cannot in general be moved past one another: their relative order may
    carry information.

We use the convention encoded by these diagrams that the modes of a register
are numbered from top to bottom. Adding modes underneath an operation therefore
does not change its mode indices, whereas adding modes above it shifts them.

\begin{definition}[Premonoidal product on $\CAR$]
    On objects, define
    \[
        m\owedge n:=m+n.
    \]
    For $a\in\A_m$ and $m'\in\mathbb N$, define
    \[
        a\owedge m'
        :=
        a\{1,\ldots,m\},
        \qquad
        m'\owedge a
        :=
        a\{m'+1,\ldots,m'+m\},
    \]
    as elements of $\A_{m+m'}$.
\end{definition}

\begin{proposition}\label{lem:car-strict-premonoidal}
    $(\CAR,\owedge,0)$ is a strict premonoidal category.
\end{proposition}

The failure of the interchange law
preventing $\owedge$ from being a monoidal cannot be witnessed
unless considering non-even morphisms, as even elements of the CAR algebras commute
with every other element so long as they are supported on disjoint sets of fermionic modes.

\begin{definition}[Commuting morphisms, central morphism, centre]
    Let $(\mathcal C,\owedge)$ be a strict premonoidal category.
    Two morphisms $f:A\to B$ and $g:C\to D$ \emph{commute} when
    \[
        \scalebox{\scalefig}{\input{./figures/premonoidal-cat/commute_1.tikz}}
        \quad = \quad
        \scalebox{\scalefig}{\input{./figures/premonoidal-cat/commute_2.tikz}}
    \]
    and
    \[
        \scalebox{\scalefig}{\input{./figures/premonoidal-cat/commute_3.tikz}}
        \quad = \quad
        \scalebox{\scalefig}{\input{./figures/premonoidal-cat/commute_4.tikz}}.
    \]
    A morphism is \emph{central} when it commutes with every morphism.
    The wide subcategory consisting of central morphisms is called the
    \emph{centre} of $\mathcal C$.
\end{definition}

\begin{proposition}
    \label{prop:center-CAR}
    The centre of $(\CAR,\owedge)$ is $\ECAR$.
\end{proposition}

In the homogeneous subcategory of \CAR, the obstruction to monoidality is precisely the Fermionic sign.

\begin{lemma}[Graded interchange]
    \label{lem:graded-interchange-CAR}
    Let $a\in\A_m$ and $a'\in\A_n$ be homogeneous, with parities
    $|a|,|a'|\in\mathbb Z_2$. Then
    \[
        (a\owedge n)(m\owedge a')
        =
        (-1)^{|a||a'|}
        (m\owedge a')(a\owedge n).
    \]
\end{lemma}The graded interchange law above is the usual super-interchange law
associated with the canonical $\mathbb Z_2$-grading of the CAR algebra.
In fact, our premonoidal structure satisfies the same signed interchange law as
homogeneous morphisms in a monoidal supercategory in the sense of
Brundan and Ellis~\cite{BrundanEllis2017}.

We nevertheless retain the premonoidal formulation for two reasons.
First, it arises directly from the notion of gate application introduced
in Section~\ref{sec:2}: an operation on a fermionic register can be
embedded into a larger register by adding idle modes either before or
after the modes on which it acts. Premonoidality therefore captures directly the
compositional structure induced by the CAR-algebraic definition of gate
application, without building the fermionic sign into a primitive
notion of parallel composition.

More importantly, premonoidality has a direct
operational interpretation. In a premonoidal category, two
operations acting on separate components cannot be exchanged, the sequential order 
in which they are applied matter. This was one of the
motivations for introducing premonoidal categories in the semantics of
programming languages, notably for computations with side effects and
in connection with models of concurrency~\cite{PowerRobinson1997}.
String-diagrammatic formalisms for premonoidal categories have also
been studied in this computational setting~\cite{RomanSobocinski2025}.

The same operational distinction appears concretely for fermionic
circuits. Two odd operations acting on disjoint sets of modes cannot be
exchanged freely: reversing their order introduces the sign of
\cref{lem:graded-interchange-CAR}. Even operations, on the other hand,
are central and may be exchanged with arbitrary operations. The
premonoidal viewpoint therefore makes explicit, at the level of the
circuit semantics, which parts of the computation seem to involve
``non-local'' effects on the system wavefunction.

\paragraph{Fermionic permutations.}

The Fermionic swaps introduced in \cref{subsec:2-3} provide a symmetry for
this premonoidal structure. Since we wish to distinguish the algebraic
permutation equations from naturality, we introduce the following
terminology.

\begin{definition}[Presymmetric premonoidal category]
    A strict premonoidal category $(\mathcal C,\owedge)$ is
    \emph{presymmetric}~ when it is equipped with central isomorphisms
    \[
        \pi_{A,B}:A\owedge B\longrightarrow B\owedge A,
    \]
    denoted
    $\scalebox{\scalefig}{\input{./figures/premonoidal-cat/presym.tikz}}$,
    satisfying the usual permutation coherence equations. In the
    one-generated strict setting considered here, these are generated by
    the involution and Yang--Baxter equations
    \[
        \scalebox{\scalefig}{\input{./figures/premonoidal-cat/invol_1.tikz}}
        \quad = \quad
        \scalebox{\scalefig}{\input{./figures/premonoidal-cat/invol_2.tikz}}, \quad \text{and}
    \]
    \[
        \scalebox{\scalefig}{\input{./figures/premonoidal-cat/yb_1.tikz}}
        \quad = \quad
        \scalebox{\scalefig}{\input{./figures/premonoidal-cat/yb_2.tikz}}.
    \]
\end{definition}

We call the non-natural symmetry involved in the above definition \emph{presymmetry}. A presymmetry
is a symmetry when it defines a natural transformation. In the following we call \emph{natural} with respect 
to a presymmetry the morphisms that are not obstacles to the naturality of the presymmetry. 

\begin{definition}[Natural morphism and Nature]
    Let $(\mathcal C,\owedge,\pi)$ be presymmetric. A morphism
    $f:A\to B$ is \emph{natural} with respect to $\pi$ when, for every
    object $C$,
    \[
        \scalebox{\scalefig}{\input{./figures/premonoidal-cat/nat_1.tikz}}
        \eq 
        \scalebox{\scalefig}{\input{./figures/premonoidal-cat/nat_2.tikz}}
        \quad \text{and} \quad
        \scalebox{\scalefig}{\input{./figures/premonoidal-cat/nat_3.tikz}}
        \eq
        \scalebox{\scalefig}{\input{./figures/premonoidal-cat/nat_4.tikz}}.
    \]
    The wide subcategory of natural morphisms is called the
    \emph{Nature} of the presymmetry. We call a presymmetric
    premonoidal category \emph{symmetric} when every morphism belongs to
    its Nature.
\end{definition}

For $m,n\in\mathbb N$, let
$\tau_{m,n}\in\mathfrak S_{m+n}$ be the block permutation exchanging the
first $m$ modes with the following $n$ modes while preserving the order
inside each block, and define
\[
    \pi_{m,n}:=\pi_{\tau_{m,n}}\in\A_{m+n}.
\]

\begin{proposition}
    \label{prop:CAR-symmetric-premonoidal}
    The family $(\pi_{m,n})$ makes
    $(\CAR,\owedge,\pi)$ into a strict symmetric premonoidal category.
\end{proposition}

The symmetric premonoidal category $(\CAR,\owedge,\pi)$ therefore captures
the endomorphism semantics of Fermionic circuits. It remains disconnected,
however: it contains no morphisms between registers of different sizes and
cannot directly express state preparations, effects, ancillae or
postselection. We now enlarge the semantics while preserving the structure
identified above, recovering a semantics in the category of Hilbert spaces 
and linear maps through the Jordan-Wigner functor.

\subsection{Categorical aspects of the Hilbert space semantics}
\label{subsec:3-3}

Let
\[
    \Pi_{m,n}
    :=
    JW_{m+n}(\pi_{m,n}).
\]

\begin{proposition}[Jordan--Wigner image of Fermionic permutations]
    \label{prop:graded-swap}
    For $x\in\{0,1\}^m$ and $y\in\{0,1\}^n$,
    \[
        \Pi_{m,n}
        \bigl(\ket{x}\otimes\ket{y}\bigr)
        =
        (-1)^{|x||y|}
        \ket{y}\otimes\ket{x}.
    \]
    In particular,
    \[
        \Pi_{1,1}
        =
        \ketbra{00}{00}
        +\ketbra{01}{10}
        +\ketbra{10}{01}
        -\ketbra{11}{11}.
    \]
\end{proposition}

Thus $\Pi$ is the graded swap associated with the parity
grading of the occupation basis.

\begin{definition}[$\LFM$ and $\ELFM$]
    The category $\LFM$ has objects $\mathbb N$ and homogeneous linear maps
    \[
        \LFM[m,n]
        \subseteq
        \mathcal L(\C^{2^m},\C^{2^n})
    \]
    as morphisms. Its wide subcategory of even linear maps is denoted
    $\ELFM$.
\end{definition}

By the results of \cref{subsec:2-2}, the Jordan--Wigner representation maps
homogeneous CAR operators to homogeneous linear maps and even CAR operators
to even linear maps. Thus it restricts to functors
\[
    JW:\HCAR\longrightarrow\LFM,
    \qquad
    JW:\ECAR\longrightarrow\ELFM.
\]
For each $m$, these restrictions are isomorphisms on the corresponding
endomorphism sets.

Unlike $\HCAR$, however, $\LFM$ also contains rectangular processes. In
particular, morphisms $0\to m$ are homogeneous state vectors and morphisms
$m\to0$ are homogeneous effects. It can therefore describe ancilla
preparation and postselection.

\paragraph{The monoidal/symmetric trade-off.}

The ordinary tensor product restricts to a monoidal product on $\LFM$.
However, the Fermionic permutation is not natural with respect to every
homogeneous linear map.

\begin{proposition}
    \label{prop:nature-LFM}
    The family~ $(\Pi_{m,n})$ defines a presymmetry on
    $(\Qubit,\otimes)$. Its Nature is exactly $\ELFM$.
\end{proposition}

Thus there is a trade-off: the ordinary tensor product is genuinely
monoidal, but the Fermionic symmetry is natural only for even processes.
The same data may instead be organized so as to make the Fermionic symmetry
natural at the price of weakening monoidality.

\begin{lemma}[Monoidal vs.\ symmetric trade-off]
    \label{lem:mon-sym-tradeoff}
    Let $(\mathcal C,\otimes,I,\pi)$ be a strict monoidal category equipped
    with a presymmetry. Define a product $\owedge$ on objects by
    \[
        A\owedge B:=A\otimes B,
    \]
    and define the functors $A\owedge (-)$ and $(-) \owedge A$ diagrammatically by
    \[
        \scalebox{\scalefig}{\input{./figures/premonoidal-cat/mon-sym-tradeoff_1.tikz}}
         := 
        \scalebox{\scalefig}{\input{./figures/premonoidal-cat/mon-sym-tradeoff_2.tikz}}
        \quad \text{and} \quad
        \scalebox{\scalefig}{\input{./figures/premonoidal-cat/mon-sym-tradeoff_3.tikz}}
         := 
        \scalebox{\scalefig}{\input{./figures/premonoidal-cat/mon-sym-tradeoff_4.tikz}}.
    \]
    Where the white background depicts the $\otimes$ structure while the gray background depicts 
    the $\owedge$ structure.\\
    Equivalently, for $f:B\to C$ and an object $A$,
    \[
        f\owedge A
        :=
        f\otimes A,
    \]
    and
    \[
        A\owedge f
        :=
        \pi_{C,A}
        \circ
        (f\otimes id_A)
        \circ
        \pi_{A,B}.
    \]
    Then $(\mathcal C,\owedge,\pi)$ is a strict symmetric premonoidal
    category. Moreover, the Nature of $(\mathcal C,\otimes,\pi)$ is
    contained in the centre of $(\mathcal C,\owedge)$.
\end{lemma}

Applying \cref{lem:mon-sym-tradeoff} to
$(\Qubit,\otimes,\Pi)$ gives a premonoidal product, also denoted
$\owedge$. Its action on homogeneous maps admits a particularly simple
description. Define the parity operator
\[
    \mathsf P_k
    :=
    \sum_{z\in\{0,1\}^k}
    (-1)^{|z|}\ketbra{z}{z}.
\]

\begin{lemma}
    \label{lem:LFM-whiskering}
    Let $f:m\to n$ be homogeneous of parity
    $|f|\in\mathbb Z_2$. Then
    \[
        f\owedge k
        =
        f\otimes id_k,
        \qquad
        k\owedge f
        =
        \mathsf P_k^{\,|f|}\otimes f.
    \]
\end{lemma}

In particular, $\owedge$ restricts to both $\LFM$ and $\ELFM$. The graphical
interchange equations introduced above now acquire their characteristic
Fermionic sign.

\begin{proposition}[Graded interchange in $\LFM$]
    \label{prop:graded-interchange-LFM}
    Let $f:m\to m'$ and $g:n\to n'$ be homogeneous. Then
    \[
        (f\owedge n')\circ(m\owedge g)
        =
        (-1)^{|f||g|}
        (m'\owedge g)\circ(f\owedge n).
    \]
    Consequently, the centre of $(\LFM,\owedge)$ is $\ELFM$.
\end{proposition}

We can now summarize the relation between the algebraic and Hilbert-space
semantics.

\begin{theorem}[Premonoidal Jordan--Wigner semantics]
    \label{thm:JW-premonoidal}
    The category $(\Qubit,\owedge,\Pi)$ is a strict symmetric
    premonoidal category, and its subcategories
    $(\LFM,\owedge,\Pi)$ and $(\ELFM,\owedge,\Pi)$ inherit this structure.
    Moreover,
    \[
        Z(\LFM,\owedge)=\ELFM,
    \]
    and the Jordan--Wigner functor
    \[
        JW:(\CAR,\owedge,\pi)
        \longrightarrow
        (\Qubit,\owedge,\Pi)
    \]
    is a strict symmetric premonoidal functor. It restricts to
    \[
        JW:\HCAR\longrightarrow\LFM,
        \qquad
        JW:\ECAR\longrightarrow\ELFM.
    \]
\end{theorem}

The same category $\LFM$ therefore carries two useful structures:
\[
    (\LFM,\otimes,\Pi)
    \qquad\text{and}\qquad
    (\LFM,\owedge,\Pi).
\]
With the ordinary tensor product, composition on independent systems is
strictly monoidal, but the Fermionic symmetry is natural only for even
maps. With the premonoidal product, the Fermionic symmetry becomes natural
for all homogeneous processes, while two odd processes on disjoint systems
fail interchange by a sign.

The second structure is the one directly induced by the CAR semantics and
therefore provides our denotational semantics for Fermionic circuits. The
first, however, is often more convenient for presenting and rewriting
linear maps. In the following section, we develop a local syntax for these
maps and study when its rewrites faithfully represent Fermionic circuits;
the distinction between $\otimes$ and $\owedge$ will allow us to keep track
of precisely where Fermionic signs come into play.

\section{Diagrammatic syntax of LFM-based quantum computing}\label{sec:4}
    In \cref{sec:2}, we presented the Bravyi and Kitaev LFM circuit model, and in \cref{sec:3}, we studied the categorical aspects 
of its two semantics, either in the symmetric premonoidal category $(\LFM, \owedge, \Pi)$ or in the presymmetric monoidal category 
$(\LFM,\otimes,\Pi)$.

In order to be executed on a fermionic machine, it is necessary to know how to rewrite a logical LFM circuit
$C$ in an equivalent circuit $C'$ whose gates are native on the target fermionic platform. Here, equivalent means 
that their denotations are the same.
\[
C\rewrites^* C' \qquad \text{with} \qquad \interp{C} = \interp{C'}
\]

To that extent, the monoidal semantics is better suited, for the rewriting can be encapsulated 
by a sequence of local rewrites which constitute an equational theory of the underlying monoidal
category. Such equational theories of circuits have a long history, from the
algebraic treatment of Boolean circuits by Lafont~\cite{Lafont2003}
to recent complete equational theories for quantum
circuits~\cite{ClementDelormePerdrix2024,ClementHeurtelMansfieldPerdrixValiron2023,
ClementDelormePerdrixVilmart2024}.
In categorical terms, they can be formalized as \emph{prop theories}.

\subsection{Props and their theories}

    When manipulating usual quantum circuits, like when manipulating fermionic circuits, the monoidal (or premonoidal) categories considered have more features: their monoid of objects is freely generated by a single object denoted $1$.

    \begin{example}
        This single object corresponds to:
        \begin{itemize}
            \item $1 = \mathbb{C}^2$ in $\Qubit$,
            \item $1 = \mathbb{C}^d$ in $\Qudit$    ,
            \item $1 = \mathbb{C}^2$ in $\LFM$,
            \item $1$ as the abstract object for one fermionic mode in \CAR.
        \end{itemize}
    \end{example}

    Given a set of usual quantum gates, the category of formal quantum circuits one can build with them is a symmetric 
    monoidal category whose monoid of objects is freely spanned by a single object: a one-qubit system, represented 
    with a wire. Such categories are called props.

    \begin{definition}[prop]\cite{Lack2004}
        A prop is a (strict) symmetric monoidal category whose monoid of objects is freely generated by a single object $1$.

        The $n$-fold monoidal product is denoted $1 + \dots + 1 = n$ for $n \in \mathbb{N}$. 
        The monoidal product on morphisms is denoted $\otimes$. 
        The generating symmetry is generally denoted $\sigma: 2 \to 2$.
    \end{definition}

        Fermionic circuits on Bravyi-Kitaev's gate set form a prop: 
        indeed, the semantics of all the gates is either in $(\ECAR, \owedge, \pi)$ or $(\ELFM, \otimes, \Pi)$, 
        which are props (arrows are all central and natural in both cases). 
        As a consequence, Bravyi and Kitaev's circuits can be manipulated as usual quantum circuits. This is the case 
        for any circuit on an even gate set.
    
    Props can be represented diagrammatically with a single type of wire, 
    corresponding to the core object $1$. 
    \[\scalebox{\scalefig}{\begin{tikzpicture}
	\begin{pgfonlayer}{nodelayer}
		\node [style=none] (0) at (-0.75, 0) {};
		\node [style=none] (1) at (0.75, 0) {};
	\end{pgfonlayer}
	\begin{pgfonlayer}{edgelayer}
		\draw (0.center) to (1.center);
	\end{pgfonlayer}
\end{tikzpicture}
}\]
    Composite objects $n = 1 + \cdots + 1$ are represented by $n$ parallel wires.
    \[\scalebox{\scalefig}{\begin{tikzpicture}
	\begin{pgfonlayer}{nodelayer}
		\node [style=none] (0) at (-0.75, -0.5) {};
		\node [style=none] (1) at (0.75, -0.5) {};
		\node [style=none] (2) at (-0.75, 0.5) {};
		\node [style=none] (3) at (0.75, 0.5) {};
		\node [style=none] (4) at (0, 0.25) {$\vdots$};
		\node [style=modenb] (5) at (0.5, 0) {$(n)$};
	\end{pgfonlayer}
	\begin{pgfonlayer}{edgelayer}
		\draw (0.center) to (1.center);
		\draw (2.center) to (3.center);
	\end{pgfonlayer}
\end{tikzpicture}
}\]
    
    \medskip

    In general, when considering the prop of quantum (resp. fermionic) circuits on a gate set (resp. even gate set), 
    its denotational semantics is equivalent to the category of formal circuits quotiented by all equations involving 
    the gates that are valid in the semantics. 
    This idea can be generalized to provide axiomatizations of props given by a formal category of diagrams over a 
    set of generators quotiented by a set of equations. 
    Such an axiomatization is called a prop theory.

    \begin{definition}[prop theory]
    A prop theory is a pair $\T = (\Sigma, E)$ consisting of a set of generators $\Sigma$ and a set of equations $E$:
    \begin{itemize}
        \item $\Sigma$ contains typed generators $g: n \to m$ with $n, m \in \mathbb{N}$,
        \item Formal diagrams $D, D', \dots$ are defined inductively as either:
        \begin{itemize}[label=\emph{|}]
            \item $g \in \Sigma$, a generator.
            \item $\mathrm{id}_1: 1 \to 1$, the identity on the core object.
            \item $\sigma: 2 \to 2$, the generating symmetry.
            \item $D \circ D' : n \to m$, the composition of two formal diagrams $D: k \to m$, $D': n \to k$,
            \item $D \otimes D' : n + n' \to m + m'$, the monoidal product of two formal diagrams 
            $D: n \to m$, $D': n' \to m'$.
        \end{itemize}
        \item $E$ contains pairs of formal diagrams $(D, D')$ of the same type $n \to m$.
    \end{itemize}
    \end{definition}  
        Given a prop theory $\T = (\Sigma, E)$, there is a corresponding prop generated by $\Sigma$ and quotiented 
        by $E$ and the axioms of symmetric monoidal categories. In practice, we denote both the theory and the corresponding 
        prop $\T$ and rely on the context to disambiguate.

    \begin{definition}[Semantic functor]
        Let $\T$ be a prop theory, $\P$ a prop. A symmetric monoidal functor $\interp{\cdot}^\T_\P: \T \to \P$ 
        is called a semantic functor of $\T$ in the prop $\P$. 
        We say that $\T$ is universal for $\P$ when $\interp{\cdot}^\T_\P$ is full, and that $\T$ is complete for $\P$ 
        when $\interp{\cdot}^\T_\P$ is faithful.

        When the context is clear, we simply write $\interp{\cdot}_\P$ for $\interp{\cdot}^\T_\P$.
    \end{definition}

        Let $\T_{\mathsf{BK}, \emptyset} = (\{\G_1, \G_2, \G_3, \G_4\}, \emptyset)$ be the theory of formal circuits on Bravyi and Kitaev's gate set. 
        We have defined in previous sections a semantic functor $\interp{\cdot}_{\LFM}: \T \to (\ELFM, \owedge, \Pi)$ through 
        $\interp{\cdot}_{\LFM} = J \circ \interp{\cdot}_{\CAR}$. This theory is not complete at all: 
        $E = \emptyset$ and thus different formal circuits implementing the same unitary are not equal in $\T_{\mathsf{BK},\emptyset}$.
        A complete theory for Bravyi-Kitaev's circuits $\T_{\mathsf{BK}} = (\{\G_1, \ldots, \G_4\}, E)$ would be such that for two circuits 
        $C_1, C_2: m \to m$ with $\interp{C_1}_{\LFM} = \interp{C_2}    _{\LFM}$, there is a way to transform $C_1$ into 
        $C_2$ by performing local rewrites $C \leftarrow C'$ for $(C, C') \in E$, where either $C$ or $C'$ is a subcircuit of either $C_1$ or $C_2$.\\

        Note that on \ELFM, $\owedge$ and $\otimes$ coincide.
        In order to deduce sound equations on circuits over Bravyi and Kitaev's gate set, it maybe useful to translate them
        to another prop theory interpreted in $(\ELFM, \otimes, \Pi)$. This method is similar to ZX-calulus for usual quantum circuits:
        translate quantum circuits to a more expressive prop theory (ZX-calculus), in which one can represent any linear maps, and which is 
        already proven complete for the target semantics. If a quantum circuit $C$ translates to a ZX diagram $D$, and that we know 
        how to rewrite $D \rewrites^* D'$ with $C'$ a circuit such that $\interp{C}=\interp{D}$ and $\interp{C'}=\interp{D'}$, then 
        $(C,C')$ is a sound equation.\\

        Based on the work of De Felice, Hadžihasanović, Ng\cite{DeFeliceHadzihasanovicNg2019, Ng2018}~, we propose to use the complete prop theory $\EZW$, whose
        generators are given together with their semantics in \cref{fig:fragment-generators}.

%
%
%


\newsavebox{\GenGridEaWSavebox}
\newsavebox{\GenGridEpWSavebox}
\newsavebox{\GenGridEZWSavebox}
\newsavebox{\GenGridaWSavebox}
\newsavebox{\GenGridFpWSavebox}
\newsavebox{\GenGridFZWSavebox}

\sbox{\GenGridEaWSavebox}{\EaWGeneratorBlock}
\sbox{\GenGridEpWSavebox}{\EpWGeneratorBlock}
\sbox{\GenGridEZWSavebox}{\EZWGeneratorBlock}
\sbox{\GenGridaWSavebox}{\aWGeneratorBlock}
\sbox{\GenGridFpWSavebox}{\FpWGeneratorBlock}
\sbox{\GenGridFZWSavebox}{\FZWGeneratorBlock}

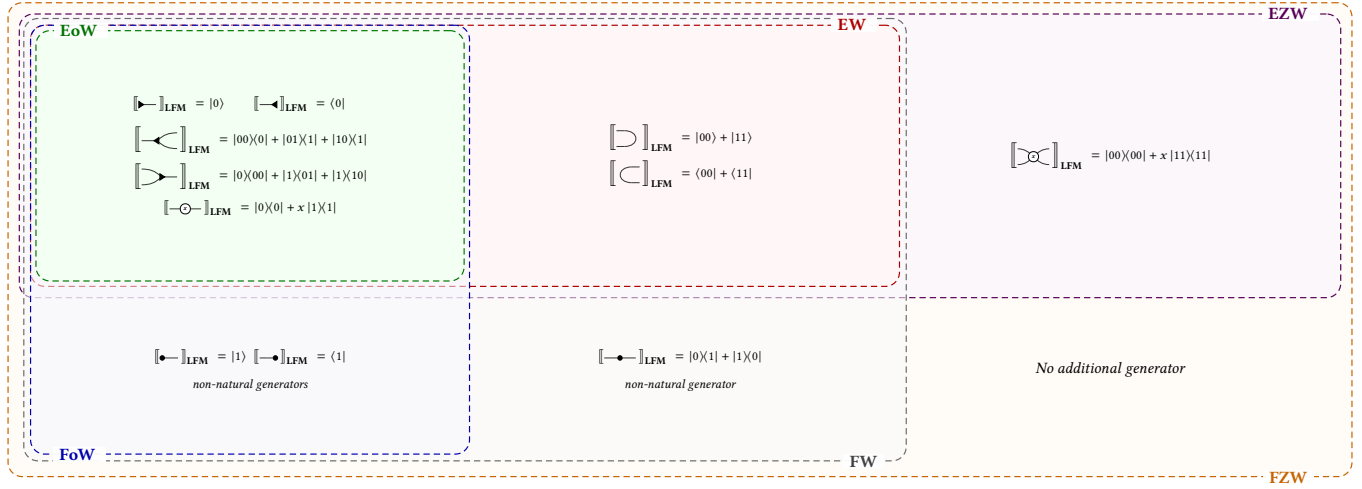
\begin{figure*}[t]
\centering

\begingroup
\sensitivetikztrue

\resizebox{\textwidth}{!}{%
\begin{tikzpicture}[
    scale=2,
    generator content/.style={inner sep=0pt,outer sep=0pt,draw=none},
    grid cell/.style={inner sep=0pt,outer sep=0pt,draw=none,minimum width=9.15cm},
    fragment region/.style={rounded corners=9pt,dash pattern=on 4pt off 2.5pt,line width=0.75pt},
    fragment title/.style={font=\Large\bfseries,fill=white,fill opacity=.96,text opacity=1,inner xsep=4pt,inner ysep=1pt}
]

%
%

\node[grid cell,minimum height=5.20cm] (cell-eaw) at (0,0) {};
\node[grid cell,minimum height=5.20cm] (cell-epw) at (9.55cm,0) {};
\node[grid cell,minimum height=5.20cm] (cell-ezw) at (19.10cm,0) {};

\node[grid cell,minimum height=3.15cm] (cell-aw)  at (0,-4.75cm) {};
\node[grid cell,minimum height=3.15cm] (cell-fpw) at (9.55cm,-4.75cm) {};
\node[grid cell,minimum height=3.15cm] (cell-fzw) at (19.10cm,-4.75cm) {};


\node[generator content] at (cell-eaw.center)
    {\usebox{\GenGridEaWSavebox}};

\node[generator content] at (cell-epw.center)
    {\usebox{\GenGridEpWSavebox}};

\node[generator content] at (cell-ezw.center)
    {\usebox{\GenGridEZWSavebox}};

\node[generator content] at (cell-aw.center)
    {\usebox{\GenGridaWSavebox}};

\node[generator content] at (cell-fpw.center)
    {\usebox{\GenGridFpWSavebox}};

\node[generator content] at (cell-fzw.center)
    {\usebox{\GenGridFZWSavebox}};

%
%

\begin{pgfonlayer}{background}

\node[
    fragment region,
     draw=orange!78!black,
    fill=orange!3,
    fit=(cell-eaw)(cell-epw)(cell-ezw)(cell-aw)(cell-fpw)(cell-fzw),
    inner sep=8mm
] (gen-grid-fzw-region) {};

\node[
    fragment region,
    draw=violet!68!black,
    fill=violet!3,
    fit=(cell-eaw)(cell-epw)(cell-ezw),
    inner sep=5.5mm
] (gen-grid-ezw-region) {};

\node[
    fragment region,
    draw=black!58,
    fill=black!2,
    fill opacity=.55,
    fit=(cell-eaw)(cell-epw)(cell-aw)(cell-fpw),
    inner sep=4.5mm
] (gen-grid-fpw-region) {};

\node[
    fragment region,
    draw=red!68!black,
    fill=red!3,
    fit=(cell-eaw)(cell-epw),
    inner sep=3.0mm
] (gen-grid-epw-region) {};

\node[
    fragment region,
    draw=blue!68!black,
    fill=blue!3,
    fill opacity=.55,
    fit=(cell-eaw)(cell-aw),
    inner sep=3.0mm
] (gen-grid-aw-region) {};

\node[
    fragment region,
    draw=green!48!black,
    fill=green!5,
    fit=(cell-eaw),
    inner sep=1.8mm
] (gen-grid-eaw-region) {};

\end{pgfonlayer}

%

\node[
    fragment title,
    anchor=west,
    text=orange!78!black
] at ([xshift=-20mm]gen-grid-fzw-region.south east)
    {$\FZW$};

\node[
    fragment title,
    anchor=east,
    text=violet!68!black
] at ([xshift=-5mm]gen-grid-ezw-region.north east)
    {$\EZW$};

\node[
    fragment title,
    anchor=east,
    text=black!75
] at ([xshift=-5mm]gen-grid-fpw-region.south east)
    {$\FpW$};

\node[
    fragment title,
    anchor=east,
    text=red!68!black
] at ([xshift=-5mm]gen-grid-epw-region.north east)
    {$\EpW$};

\node[
    fragment title,
    anchor=west,
    text=blue!68!black
] at ([xshift=5mm]gen-grid-aw-region.south west)
    {$\aW$};

\node[
    fragment title,
    anchor=west,
    text=green!45!black
] at ([xshift=4mm]gen-grid-eaw-region.north west)
    {$\EaW$};

\end{tikzpicture}%
}

\endgroup

\caption{
Generators of the six graphical theories and their semantics in $\LFM$.
Each generator is displayed in the smallest fragment in which it occurs.
The parameters $x$ range over $\mathbb{C}$.
The fermionic swap is the paradigmatic (pre)symmetry and is therefore
not displayed as a generator. 
}

\label{fig:fragment-generators}

\Description{
}

\end{figure*}

%
%

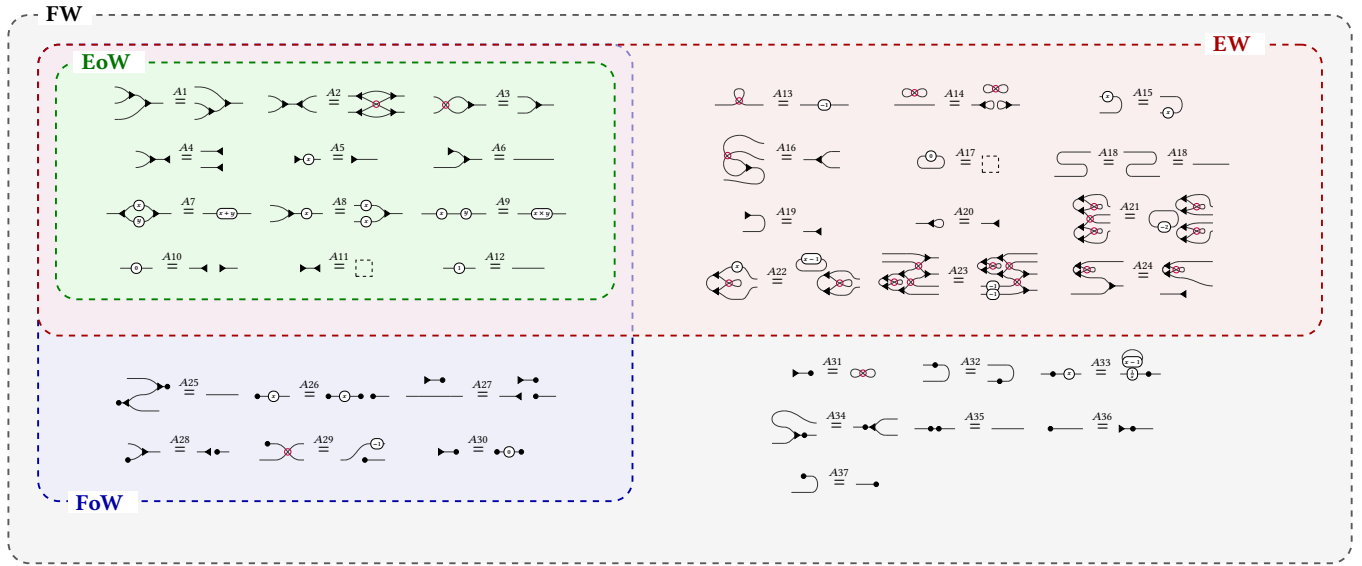
\begin{figure*}[t]
\centering

\begingroup
\sensitivetikztrue

\newsavebox{\EaWAxiomSavebox}
\newsavebox{\EpWAxiomSavebox}
\newsavebox{\aWAxiomSavebox}
\newsavebox{\FpWAxiomSavebox}

\sbox{\EaWAxiomSavebox}{\EaWAxiomBlock}
\sbox{\EpWAxiomSavebox}{\EpWAxiomBlock}
\sbox{\aWAxiomSavebox}{\aWAxiomBlock}
\sbox{\FpWAxiomSavebox}{\FpWAxiomBlock}


\resizebox{\textwidth}{!}{%
\begin{tikzpicture}[
    reset cm,
    axiom block/.style={
        inner sep=0pt,
        outer sep=0pt,
        draw=none
    },
    fragment region/.style={
        rounded corners=10pt,
        dashed,
        line width=0.9pt
    },
    fragment title/.style={
        font=\Large\bfseries,
        fill=white,
        inner xsep=4pt,
        inner ysep=1.5pt,
        draw=none
    }
]


\node[
    axiom block,
    anchor=north west
] (eaw) at (0,0) {\usebox{\EaWAxiomSavebox}};

\node[
    axiom block,
    anchor=north west
] (epw) at ([xshift=8mm]eaw.north east)
    {\usebox{\EpWAxiomSavebox}};

\node[
    axiom block,
    anchor=north west
] (aw) at ([yshift=-15mm]eaw.south west)
    {\usebox{\aWAxiomSavebox}};

\node[
    axiom block,
    anchor=north west
] (fpw) at ([yshift=-8mm]epw.south west)
    {\usebox{\FpWAxiomSavebox}};


\begin{pgfonlayer}{background}

\node[
    fragment region,
    draw=black!65,
    fill=black!4,
    fit=(eaw)(epw)(aw)(fpw),
    inner sep=11mm
] (fpw-region) {};

\node[
    fragment region,
    draw=blue!65!black,
    fill=blue!8,
    fill opacity=.55,
    fit=(eaw)(aw),
    inner sep=6mm
] (aw-region) {};

\node[
    fragment region,
    draw=red!65!black,
    fill=red!8,
    fill opacity=.55,
    fit=(eaw)(epw),
    inner sep=6mm
] (epw-region) {};

\node[
    fragment region,
    draw=green!50!black,
    fill=green!10,
    fill opacity=.75,
    fit=(eaw),
    inner sep=3mm
] (eaw-region) {};

\end{pgfonlayer}


\node[
    fragment title,
    anchor=west
] at ([xshift=5mm]fpw-region.north west)
    {$\FpW$};

\node[
    fragment title,
    anchor=west,
    text=blue!65!black
] at ([xshift=5mm]aw-region.south west)
    {$\aW$};

\node[
    fragment title,
    anchor=west,
    text=red!65!black
] at ([xshift=-20mm]epw-region.north east)
    {$\EpW$};

\node[
    fragment title,
    anchor=west,
    text=green!45!black
] at ([xshift=3mm]eaw-region.north west)
    {$\EaW$};

\end{tikzpicture}%
}

\endgroup

\caption{
Axioms of the four equational theories.
The overlapping regions represent inclusions between the corresponding
fragments; each axiom is displayed in the smallest fragment in which it
occurs. By convention, equations of \EaW~ and \aW~ are also valid when mirrored
with respect to the vertical axis.
}
\label{fig:equational-theories}

\Description{
The axioms are arranged in four overlapping regions corresponding to
EaW, aW, EpW, and FpW. EaW is contained in both aW and EpW, while FpW
contains all three regions.
}

\end{figure*}

    \begin{theorem}[De Felice, Hadžihasanović, Ng~\cite{Ng2018}]
            \raggedright
        There is a complete prop theory denoted $\EZW$ of $(\ELFM, \otimes, \Pi)$ with generators :\\[0.75em]
        {\centering\begin{tabular}{c c c c c c c c}
        $\scalebox{\scalefig}{{\begin{tikzpicture}
	\begin{pgfonlayer}{nodelayer}
		\node [style=none] (2) at (1, 0) {};
		\node [style=rw] (3) at (0, 0) {};
	\end{pgfonlayer}
	\begin{pgfonlayer}{edgelayer}
		\draw (2.center) to (3);
	\end{pgfonlayer}
\end{tikzpicture}
}}$
        & $\scalebox{\scalefig}{{\begin{tikzpicture}
	\begin{pgfonlayer}{nodelayer}
		\node [style=none] (2) at (-1.25, 0) {};
		\node [style=lw] (3) at (-0.25, 0) {};
	\end{pgfonlayer}
	\begin{pgfonlayer}{edgelayer}
		\draw (2.center) to (3);
	\end{pgfonlayer}
\end{tikzpicture}
}}$
        & $\scalebox{\scalefig}{{\begin{tikzpicture}
	\begin{pgfonlayer}{nodelayer}
		\node [style=none] (1) at (1, 0.5) {};
		\node [style=none] (2) at (-1.25, 0) {};
		\node [style=lw] (3) at (-0.25, 0) {};
		\node [style=none] (4) at (1, -0.5) {};
	\end{pgfonlayer}
	\begin{pgfonlayer}{edgelayer}
		\draw [in=-45, out=180] (4.center) to (3);
		\draw (2.center) to (3);
		\draw [in=45, out=180] (1.center) to (3);
	\end{pgfonlayer}
\end{tikzpicture}
}}$
        & $\scalebox{\scalefig}{{\begin{tikzpicture}
	\begin{pgfonlayer}{nodelayer}
		\node [style=none] (1) at (-1.25, 0.5) {};
		\node [style=none] (2) at (1, 0) {};
		\node [style=rw] (3) at (0, 0) {};
		\node [style=none] (4) at (-1.25, -0.5) {};
	\end{pgfonlayer}
	\begin{pgfonlayer}{edgelayer}
		\draw [in=-135, out=0] (4.center) to (3);
		\draw (2.center) to (3);
		\draw [in=135, out=0] (1.center) to (3);
	\end{pgfonlayer}
\end{tikzpicture}
}}$
        & $\scalebox{\scalefig}{{\begin{tikzpicture}
	\begin{pgfonlayer}{nodelayer}
		\node [style=none] (0) at (-0.25, 0.5) {};
		\node [style=none] (1) at (0.25, 0.5) {};
		\node [style=none] (2) at (0.25, -0.5) {};
		\node [style=none] (3) at (-0.25, -0.5) {};
	\end{pgfonlayer}
	\begin{pgfonlayer}{edgelayer}
		\draw (3.center) to (2.center);
		\draw (0.center) to (1.center);
		\draw [bend right=270, looseness=2.50] (1.center) to (2.center);
	\end{pgfonlayer}
\end{tikzpicture}
}}$
        &  $\scalebox{\scalefig}{{\begin{tikzpicture}
	\begin{pgfonlayer}{nodelayer}
		\node [style=none] (0) at (0.5, 0.5) {};
		\node [style=none] (1) at (0, 0.5) {};
		\node [style=none] (2) at (0, -0.5) {};
		\node [style=none] (3) at (0.5, -0.5) {};
	\end{pgfonlayer}
	\begin{pgfonlayer}{edgelayer}
		\draw (3.center) to (2.center);
		\draw (0.center) to (1.center);
		\draw [bend left=270, looseness=2.50] (1.center) to (2.center);
	\end{pgfonlayer}
\end{tikzpicture}
}}$ 
        &  $\scalebox{\scalefig}{{\begin{tikzpicture}
	\begin{pgfonlayer}{nodelayer}
		\node [style=none] (14) at (-1, 0.5) {};
		\node [style=none] (15) at (-1, -0.5) {};
		\node [style=none] (17) at (1, 0.5) {};
		\node [style=none] (18) at (1, -0.5) {};
		\node [style=white phase dot] (19) at (0, 0) {$x$};
	\end{pgfonlayer}
	\begin{pgfonlayer}{edgelayer}
		\draw [in=45, out=180] (17.center) to (19);
		\draw [in=-45, out=180] (18.center) to (19);
		\draw [in=135, out=0] (14.center) to (19);
		\draw [in=-135, out=0] (15.center) to (19);
	\end{pgfonlayer}
\end{tikzpicture}
}}$ 
        &  $\scalebox{\scalefig}{{\begin{tikzpicture}
	\begin{pgfonlayer}{nodelayer}
		\node [style=white phase dot] (0) at (0, 0) {$x$};
		\node [style=none] (1) at (-1, 0) {};
		\node [style=none] (2) at (1, 0) {};
	\end{pgfonlayer}
	\begin{pgfonlayer}{edgelayer}
		\draw (1.center) to (2.center);
	\end{pgfonlayer}
\end{tikzpicture}
}}$ \\
        \end{tabular}\\[0.75em]}
        \raggedright
    \end{theorem}

    Note that our set generators slightly differs from the one in \cite{Ng2018}. The reason is that we choose 
    to consider the fermionic swap
    \[\interp{\scalebox{\scalefig}{{\begin{tikzpicture}
	\begin{pgfonlayer}{nodelayer}
		\node [style=none] (1) at (-1, 0.5) {};
		\node [style=none] (2) at (1, 0.5) {};
		\node [style=none] (3) at (-1, -0.5) {};
		\node [style=none] (4) at (1, -0.5) {};
		\node [style=none] (9) at (0.75, -0.5) {};
		\node [style=none] (10) at (0.75, 0.5) {};
		\node [style=none] (11) at (-0.75, 0.5) {};
		\node [style=none] (12) at (-0.75, -0.5) {};
		\node [style=fswap] (13) at (0, 0) {};
	\end{pgfonlayer}
	\begin{pgfonlayer}{edgelayer}
		\draw [in=180, out=0] (11.center) to (9.center);
		\draw (1.center) to (11.center);
		\draw (9.center) to (4.center);
		\draw (2.center) to (10.center);
		\draw [in=0, out=-180] (10.center) to (12.center);
		\draw (12.center) to (3.center);
	\end{pgfonlayer}
\end{tikzpicture}
}}}_\LFM \quad = \quad \Pi_{1,1} \quad = \quad  \ketbra{00}{00} + \ketbra{10}{01} + \ketbra{01}{10} - \ketbra{11}{11}
    \]
    as paradigmatic symmetry instead of the usual swap 
    \[
    \interp{\scalebox{\scalefig}{{\begin{tikzpicture}
	\begin{pgfonlayer}{nodelayer}
		\node [style=none] (1) at (-1, 0.5) {};
		\node [style=none] (2) at (1, 0.5) {};
		\node [style=none] (3) at (-1, -0.5) {};
		\node [style=none] (4) at (1, -0.5) {};
		\node [style=none] (9) at (0.75, -0.5) {};
		\node [style=none] (10) at (0.75, 0.5) {};
		\node [style=none] (11) at (-0.75, 0.5) {};
		\node [style=none] (12) at (-0.75, -0.5) {};
	\end{pgfonlayer}
	\begin{pgfonlayer}{edgelayer}
		\draw [in=180, out=0] (11.center) to (9.center);
		\draw (1.center) to (11.center);
		\draw (9.center) to (4.center);
		\draw (2.center) to (10.center);
		\draw [in=0, out=-180] (10.center) to (12.center);
		\draw (12.center) to (3.center);
	\end{pgfonlayer}
\end{tikzpicture}
}}}_\LFM \quad = \quad \ketbra{00}{00} + \ketbra{10}{01} + \ketbra{01}{10} + \ketbra{11}{11}.
    \]
    Additionally, we choose to consider $\scalebox{\scalefig}{{\begin{tikzpicture}
	\begin{pgfonlayer}{nodelayer}
		\node [style=none] (14) at (-1, 0.5) {};
		\node [style=none] (15) at (-1, -0.5) {};
		\node [style=none] (17) at (1, 0.5) {};
		\node [style=none] (18) at (1, -0.5) {};
		\node [style=white phase dot] (19) at (0, 0) {$x$};
	\end{pgfonlayer}
	\begin{pgfonlayer}{edgelayer}
		\draw [in=45, out=180] (17.center) to (19);
		\draw [in=-45, out=180] (18.center) to (19);
		\draw [in=135, out=0] (14.center) to (19);
		\draw [in=-135, out=0] (15.center) to (19);
	\end{pgfonlayer}
\end{tikzpicture}
}}$ as a true generator instead of a mere syntactic sugar,
    so as to replace the classical swap  $\scalebox{\scalefig}{{\begin{tikzpicture}
	\begin{pgfonlayer}{nodelayer}
		\node [style=none] (1) at (-1, 0.5) {};
		\node [style=none] (2) at (1, 0.5) {};
		\node [style=none] (3) at (-1, -0.5) {};
		\node [style=none] (4) at (1, -0.5) {};
		\node [style=none] (9) at (0.75, -0.5) {};
		\node [style=none] (10) at (0.75, 0.5) {};
		\node [style=none] (11) at (-0.75, 0.5) {};
		\node [style=none] (12) at (-0.75, -0.5) {};
	\end{pgfonlayer}
	\begin{pgfonlayer}{edgelayer}
		\draw [in=180, out=0] (11.center) to (9.center);
		\draw (1.center) to (11.center);
		\draw (9.center) to (4.center);
		\draw (2.center) to (10.center);
		\draw [in=0, out=-180] (10.center) to (12.center);
		\draw (12.center) to (3.center);
	\end{pgfonlayer}
\end{tikzpicture}
}}$ as a generator by $\scalebox{\scalefig}{{\begin{tikzpicture}
	\begin{pgfonlayer}{nodelayer}
		\node [style=none] (14) at (-1, 0.5) {};
		\node [style=none] (15) at (-1, -0.5) {};
		\node [style=none] (17) at (1, 0.5) {};
		\node [style=none] (18) at (1, -0.5) {};
		\node [style=white phase dot] (19) at (0, 0) {$x$};
	\end{pgfonlayer}
	\begin{pgfonlayer}{edgelayer}
		\draw [in=45, out=180] (17.center) to (19);
		\draw [in=-45, out=180] (18.center) to (19);
		\draw [in=135, out=0] (14.center) to (19);
		\draw [in=-135, out=0] (15.center) to (19);
	\end{pgfonlayer}
\end{tikzpicture}
}}$,
     without losing or gaining expressivity as:
        \[
         \scalebox{\scalefig}{{\begin{tikzpicture}
	\begin{pgfonlayer}{nodelayer}
		\node [style=none] (1) at (-1, 0.5) {};
		\node [style=none] (2) at (1, 0.5) {};
		\node [style=none] (3) at (-1, -0.5) {};
		\node [style=none] (4) at (1, -0.5) {};
		\node [style=none] (9) at (0.75, -0.5) {};
		\node [style=none] (10) at (0.75, 0.5) {};
		\node [style=none] (11) at (-0.75, 0.5) {};
		\node [style=none] (12) at (-0.75, -0.5) {};
	\end{pgfonlayer}
	\begin{pgfonlayer}{edgelayer}
		\draw [in=180, out=0] (11.center) to (9.center);
		\draw (1.center) to (11.center);
		\draw (9.center) to (4.center);
		\draw (2.center) to (10.center);
		\draw [in=0, out=-180] (10.center) to (12.center);
		\draw (12.center) to (3.center);
	\end{pgfonlayer}
\end{tikzpicture}
}} = \scalebox{\scalefig}{{\begin{tikzpicture}
	\begin{pgfonlayer}{nodelayer}
		\node [style=none] (1) at (-1.5, 0.5) {};
		\node [style=none] (3) at (-1.5, -0.5) {};
		\node [style=none] (13) at (-0.75, 0.5) {};
		\node [style=none] (14) at (-0.75, -0.5) {};
		\node [style=none] (15) at (0.75, 0.5) {};
		\node [style=none] (16) at (0.75, -0.5) {};
		\node [style=white phase dot] (17) at (0, 0) {$-2$};
		\node [style=rw] (18) at (0.75, 0.5) {};
		\node [style=rw] (19) at (0.75, -0.5) {};
		\node [style=lw] (20) at (-0.75, 0.5) {};
		\node [style=lw] (21) at (-0.75, -0.5) {};
		\node [style=none] (23) at (3.25, 0.5) {};
		\node [style=none] (25) at (3.25, -0.5) {};
		\node [style=none] (26) at (2.75, -0.5) {};
		\node [style=none] (27) at (2.75, 0.5) {};
		\node [style=none] (28) at (1.25, 0.5) {};
		\node [style=none] (29) at (1.25, -0.5) {};
		\node [style=fswap] (30) at (2, 0) {};
	\end{pgfonlayer}
	\begin{pgfonlayer}{edgelayer}
		\draw (15.center) to (17);
		\draw (16.center) to (17);
		\draw (13.center) to (17);
		\draw (14.center) to (17);
		\draw (1.center) to (13.center);
		\draw (14.center) to (3.center);
		\draw [bend right] (18) to (20);
		\draw [bend left] (19) to (21);
		\draw [in=180, out=0] (28.center) to (26.center);
		\draw (26.center) to (25.center);
		\draw (23.center) to (27.center);
		\draw [in=0, out=-180] (27.center) to (29.center);
		\draw (16.center) to (29.center);
		\draw (28.center) to (15.center);
	\end{pgfonlayer}
\end{tikzpicture}
}}
        \quad \text{and} \quad
        \scalebox{\scalefig}{{\begin{tikzpicture}
	\begin{pgfonlayer}{nodelayer}
		\node [style=none] (14) at (-1, 0.5) {};
		\node [style=none] (15) at (-1, -0.5) {};
		\node [style=none] (17) at (1, 0.5) {};
		\node [style=none] (18) at (1, -0.5) {};
		\node [style=white phase dot] (19) at (0, 0) {$x$};
	\end{pgfonlayer}
	\begin{pgfonlayer}{edgelayer}
		\draw [in=45, out=180] (17.center) to (19);
		\draw [in=-45, out=180] (18.center) to (19);
		\draw [in=135, out=0] (14.center) to (19);
		\draw [in=-135, out=0] (15.center) to (19);
	\end{pgfonlayer}
\end{tikzpicture}
}} = \scalebox{\scalefig}{{\begin{tikzpicture}
	\begin{pgfonlayer}{nodelayer}
		\node [style=none] (55) at (-1.75, -2) {};
		\node [style=none] (108) at (-0.75, -1.25) {};
		\node [style=white phase dot] (110) at (2, 2) {$x$};
		\node [style=none] (111) at (1.75, -2) {};
		\node [style=none] (112) at (2.25, -2) {};
		\node [style=none] (113) at (0.75, -1.25) {};
		\node [style=none] (114) at (-1.75, 2) {};
		\node [style=none] (115) at (-3, 1) {};
		\node [style=none] (116) at (-0.75, 1.25) {};
		\node [style=none] (117) at (1.75, 2) {};
		\node [style=none] (118) at (2.25, 2) {};
		\node [style=none] (119) at (0.75, 1.25) {};
		\node [style=none] (122) at (-1.75, 1.25) {};
		\node [style=none] (123) at (-0.75, 2) {};
		\node [style=none] (127) at (1.75, 1.25) {};
		\node [style=none] (128) at (0.75, 2) {};
		\node [style=rw] (129) at (-1.25, 0) {};
		\node [style=none] (130) at (-1.75, -1.25) {};
		\node [style=none] (131) at (-0.75, -2) {};
		\node [style=lw] (134) at (1.25, 0) {};
		\node [style=none] (135) at (1.75, -1.25) {};
		\node [style=none] (136) at (0.75, -2) {};
		\node [style=fswap] (137) at (-1.25, 1.65) {};
		\node [style=fswap] (138) at (-1.25, -1.65) {};
		\node [style=white phase dot] (139) at (0.5, 0) {$\frac{1}{2}$};
		\node [style=none] (141) at (-0.25, 0.75) {};
		\node [style=none] (142) at (-0.75, -0.25) {};
		\node [style=none] (143) at (-0.25, -0.25) {};
		\node [style=none] (144) at (-0.75, 0.75) {};
		\node [style=none] (145) at (-0.75, -0.75) {};
		\node [style=none] (146) at (-0.25, -0.75) {};
		\node [style=fswap] (147) at (-0.5, -0.5) {};
		\node [style=fswap] (148) at (-0.75, 0) {};
		\node [style=none] (149) at (2.25, 2) {};
		\node [style=none] (150) at (3.25, 1) {};
		\node [style=none] (151) at (-1.75, -2) {};
		\node [style=none] (152) at (-3, -1) {};
		\node [style=none] (153) at (3.25, -1) {};
		\node [style=none] (154) at (2.25, -2) {};
	\end{pgfonlayer}
	\begin{pgfonlayer}{edgelayer}
		\draw [in=-180, out=0] (55.center) to (108.center);
		\draw (112.center) to (111.center);
		\draw [in=0, out=-180] (111.center) to (113.center);
		\draw (108.center) to (113.center);
		\draw [in=0, out=180] (114.center) to (115.center);
		\draw [in=180, out=0] (114.center) to (116.center);
		\draw (118.center) to (117.center);
		\draw [in=0, out=180] (117.center) to (119.center);
		\draw (116.center) to (119.center);
		\draw [in=180, out=15] (122.center) to (123.center);
		\draw [in=0, out=165] (127.center) to (128.center);
		\draw [in=165, out=-165, looseness=1.75] (129) to (130.center);
		\draw [in=-180, out=-15] (130.center) to (131.center);
		\draw [in=-165, out=165, looseness=1.75] (129) to (122.center);
		\draw [in=15, out=-15, looseness=1.75] (134) to (135.center);
		\draw [in=0, out=-165] (135.center) to (136.center);
		\draw (131.center) to (136.center);
		\draw (123.center) to (128.center);
		\draw [in=-15, out=15, looseness=1.75] (134) to (127.center);
		\draw (129) to (134);
		\draw [bend left=90] (144.center) to (141.center);
		\draw (144.center) to (142.center);
		\draw (141.center) to (143.center);
		\draw [in=90, out=-90, looseness=0.75] (142.center) to (146.center);
		\draw [in=75, out=-90] (143.center) to (145.center);
		\draw [bend right=90, looseness=1.50] (145.center) to (146.center);
		\draw [in=180, out=0] (149.center) to (150.center);
		\draw [in=0, out=-180] (151.center) to (152.center);
		\draw [in=0, out=180] (153.center) to (154.center);
	\end{pgfonlayer}
\end{tikzpicture}
}}
        \]
    The complete set of rewrite rules for the original theory (see~\cite{Ng2018}) directly translates to a complete 
    set of equations for this version.\\
    \raggedright
    In $\EZW$, the Bravyi-Kitaev gate set is represented as:\\[0.75em]
    
    {
        \centering
        \begin{tabular}{c c c}
         \scalebox{\scalefig}{{\input{./figures/gate-set/G1.tikz}}} & $\eq$ &  \scalebox{\scalefig}{{\begin{tikzpicture}
	\begin{pgfonlayer}{nodelayer}
		\node [style=white phase dot] (0) at (0, 0) {$e^{i\pi/4}$};
		\node [style=none] (1) at (-1.5, 0) {};
		\node [style=none] (2) at (1.5, 0) {};
	\end{pgfonlayer}
	\begin{pgfonlayer}{edgelayer}
		\draw (1.center) to (2.center);
	\end{pgfonlayer}
\end{tikzpicture}
}}\\[0.75em]
         \scalebox{\scalefig}{{\input{./figures/gate-set/G2.tikz}}} & $\eq$ &  \scalebox{\scalefig}{{\input{./figures/gate-set/G2-EZW.tikz}}}\\[0.75em]
         \scalebox{\scalefig}{{\input{./figures/gate-set/G3.tikz}}} & $\eq$ &  \scalebox{\scalefig}{{\input{./figures/gate-set/G3-EZW.tikz}}}\\[0.75em]
         \scalebox{\scalefig}{{\input{./figures/gate-set/G4.tikz}}} & $\eq$ &  \scalebox{\scalefig}{{\input{./figures/gate-set/G4-EZW.tikz}}}\\[0.75em]
        \end{tabular}\\[1em]}
        \raggedright

        The gray background stands for ambient $\owedge$ structure, while the $\otimes$ structure is pictured with white background.
        Recall that the $\owedge$ and $\otimes$ structure coincide on \ELFM.
        \raggedright
        \begin{example}[An even LFM circuit translated in $\EZW$]

        Consider the following LFM circuit, which can be translated in $\EZW$ as follows
        because the $\owedge$ and $\otimes$ structure coincide on \ELFM.
        \raggedright
        \medskip

        {\centering    
        \scalebox{\scalefig}{{\input{./figures/circuit-examples/EZW_1.tikz}}}\\[1em]
        $\downarrow$\\[1em]
        \scalebox{\scalefig}{{\input{./figures/circuit-examples/EZW_2.tikz}}}\\[1em]
        $\downarrow$\\[1em]
        \scalebox{\scalefig}{{\input{./figures/circuit-examples/EZW_3.tikz}}}\\
        }
        \raggedright
        \medskip
        \end{example}

        Note that $\scalebox{\scalefig}{{\begin{tikzpicture}
	\begin{pgfonlayer}{nodelayer}
		\node [style=none] (14) at (-1, 0.5) {};
		\node [style=none] (15) at (-1, -0.5) {};
		\node [style=none] (17) at (1, 0.5) {};
		\node [style=none] (18) at (1, -0.5) {};
		\node [style=white phase dot] (19) at (0, 0) {$x$};
	\end{pgfonlayer}
	\begin{pgfonlayer}{edgelayer}
		\draw [in=45, out=180] (17.center) to (19);
		\draw [in=-45, out=180] (18.center) to (19);
		\draw [in=135, out=0] (14.center) to (19);
		\draw [in=-135, out=0] (15.center) to (19);
	\end{pgfonlayer}
\end{tikzpicture}
}}$ is only involved in the expression for $G_4$. As a consequence, 
        studying the fragment generated by the other generators of $\EZW$ is of interest to manipulate circuits of 
        gate set $\{\G_1, \G_2, \G_3\}$. We propose the following axiomatization of this fragment, denoted $\EpW$.

        \begin{definition}[\EpW]
            \raggedright
            \EpW is the following prop theory with canonical symmetry $\pi$:\\[0.75em]
            {\centering
            Generators of \EpW \\[1.5em]
            \begin{tabular}{c c c c c c c}
                $\scalebox{\scalefig}{{\begin{tikzpicture}
	\begin{pgfonlayer}{nodelayer}
		\node [style=none] (2) at (1, 0) {};
		\node [style=rw] (3) at (0, 0) {};
	\end{pgfonlayer}
	\begin{pgfonlayer}{edgelayer}
		\draw (2.center) to (3);
	\end{pgfonlayer}
\end{tikzpicture}
}}$
                & $\scalebox{\scalefig}{{\begin{tikzpicture}
	\begin{pgfonlayer}{nodelayer}
		\node [style=none] (2) at (-1.25, 0) {};
		\node [style=lw] (3) at (-0.25, 0) {};
	\end{pgfonlayer}
	\begin{pgfonlayer}{edgelayer}
		\draw (2.center) to (3);
	\end{pgfonlayer}
\end{tikzpicture}
}}$
                & $\scalebox{\scalefig}{{\begin{tikzpicture}
	\begin{pgfonlayer}{nodelayer}
		\node [style=none] (1) at (1, 0.5) {};
		\node [style=none] (2) at (-1.25, 0) {};
		\node [style=lw] (3) at (-0.25, 0) {};
		\node [style=none] (4) at (1, -0.5) {};
	\end{pgfonlayer}
	\begin{pgfonlayer}{edgelayer}
		\draw [in=-45, out=180] (4.center) to (3);
		\draw (2.center) to (3);
		\draw [in=45, out=180] (1.center) to (3);
	\end{pgfonlayer}
\end{tikzpicture}
}}$
                & $\scalebox{\scalefig}{{\begin{tikzpicture}
	\begin{pgfonlayer}{nodelayer}
		\node [style=none] (1) at (-1.25, 0.5) {};
		\node [style=none] (2) at (1, 0) {};
		\node [style=rw] (3) at (0, 0) {};
		\node [style=none] (4) at (-1.25, -0.5) {};
	\end{pgfonlayer}
	\begin{pgfonlayer}{edgelayer}
		\draw [in=-135, out=0] (4.center) to (3);
		\draw (2.center) to (3);
		\draw [in=135, out=0] (1.center) to (3);
	\end{pgfonlayer}
\end{tikzpicture}
}}$
                & $\scalebox{\scalefig}{{\begin{tikzpicture}
	\begin{pgfonlayer}{nodelayer}
		\node [style=none] (0) at (-0.25, 0.5) {};
		\node [style=none] (1) at (0.25, 0.5) {};
		\node [style=none] (2) at (0.25, -0.5) {};
		\node [style=none] (3) at (-0.25, -0.5) {};
	\end{pgfonlayer}
	\begin{pgfonlayer}{edgelayer}
		\draw (3.center) to (2.center);
		\draw (0.center) to (1.center);
		\draw [bend right=270, looseness=2.50] (1.center) to (2.center);
	\end{pgfonlayer}
\end{tikzpicture}
}}$
                & $\scalebox{\scalefig}{{\begin{tikzpicture}
	\begin{pgfonlayer}{nodelayer}
		\node [style=none] (0) at (0.5, 0.5) {};
		\node [style=none] (1) at (0, 0.5) {};
		\node [style=none] (2) at (0, -0.5) {};
		\node [style=none] (3) at (0.5, -0.5) {};
	\end{pgfonlayer}
	\begin{pgfonlayer}{edgelayer}
		\draw (3.center) to (2.center);
		\draw (0.center) to (1.center);
		\draw [bend left=270, looseness=2.50] (1.center) to (2.center);
	\end{pgfonlayer}
\end{tikzpicture}
}}$
                & $\scalebox{\scalefig}{{\begin{tikzpicture}
	\begin{pgfonlayer}{nodelayer}
		\node [style=white phase dot] (0) at (0, 0) {$x$};
		\node [style=none] (1) at (-1, 0) {};
		\node [style=none] (2) at (1, 0) {};
	\end{pgfonlayer}
	\begin{pgfonlayer}{edgelayer}
		\draw (1.center) to (2.center);
	\end{pgfonlayer}
\end{tikzpicture}
}}$
            \end{tabular}\\[2em]
            \centering
            Equations of \EpW~are pictured on \cref{fig:equational-theories}}
            \centering
            \\
        \end{definition}

        Another fragment of interest is obtained by removing cups and caps, which are
        only needed to express $G_3$. We call it $\EaW$, for \emph{Even oriented W}:
        in the absence of cups and caps, every diagram admits an unambiguous global
        input-to-output orientation, with no wire bending from an output back towards an input.

    \begin{definition}[\EaW]
    \raggedright
    \EaW is the following prop theory with canonical symmetry $\Pi$:\\[0.75em]
        {\centering
        Generators of \EaW\\[1em]
        \begin{tabular}{c c c c c c}
            $\scalebox{\scalefig}{{\begin{tikzpicture}
	\begin{pgfonlayer}{nodelayer}
		\node [style=none] (2) at (1, 0) {};
		\node [style=rw] (3) at (0, 0) {};
	\end{pgfonlayer}
	\begin{pgfonlayer}{edgelayer}
		\draw (2.center) to (3);
	\end{pgfonlayer}
\end{tikzpicture}
}}$
            & $\scalebox{\scalefig}{{\begin{tikzpicture}
	\begin{pgfonlayer}{nodelayer}
		\node [style=none] (2) at (-1.25, 0) {};
		\node [style=lw] (3) at (-0.25, 0) {};
	\end{pgfonlayer}
	\begin{pgfonlayer}{edgelayer}
		\draw (2.center) to (3);
	\end{pgfonlayer}
\end{tikzpicture}
}}$
            & $\scalebox{\scalefig}{{\begin{tikzpicture}
	\begin{pgfonlayer}{nodelayer}
		\node [style=none] (1) at (1, 0.5) {};
		\node [style=none] (2) at (-1.25, 0) {};
		\node [style=lw] (3) at (-0.25, 0) {};
		\node [style=none] (4) at (1, -0.5) {};
	\end{pgfonlayer}
	\begin{pgfonlayer}{edgelayer}
		\draw [in=-45, out=180] (4.center) to (3);
		\draw (2.center) to (3);
		\draw [in=45, out=180] (1.center) to (3);
	\end{pgfonlayer}
\end{tikzpicture}
}}$
            & $\scalebox{\scalefig}{{\begin{tikzpicture}
	\begin{pgfonlayer}{nodelayer}
		\node [style=none] (1) at (-1.25, 0.5) {};
		\node [style=none] (2) at (1, 0) {};
		\node [style=rw] (3) at (0, 0) {};
		\node [style=none] (4) at (-1.25, -0.5) {};
	\end{pgfonlayer}
	\begin{pgfonlayer}{edgelayer}
		\draw [in=-135, out=0] (4.center) to (3);
		\draw (2.center) to (3);
		\draw [in=135, out=0] (1.center) to (3);
	\end{pgfonlayer}
\end{tikzpicture}
}}$
            & $\scalebox{\scalefig}{{\begin{tikzpicture}
	\begin{pgfonlayer}{nodelayer}
		\node [style=white phase dot] (0) at (0, 0) {$x$};
		\node [style=none] (1) at (-1, 0) {};
		\node [style=none] (2) at (1, 0) {};
	\end{pgfonlayer}
	\begin{pgfonlayer}{edgelayer}
		\draw (1.center) to (2.center);
	\end{pgfonlayer}
\end{tikzpicture}
}}$
        \end{tabular}\\[2em]
        Equations of \EaW~are pictured on \cref{fig:equational-theories}}
        \centering\\
        \raggedright
    \end{definition}

    Although the structure of $\EaW$ is very simple, being the smallest fragment carrying the Hopf algebra formed by $
    \left(\scalebox{\scalefig}{{\begin{tikzpicture}
	\begin{pgfonlayer}{nodelayer}
		\node [style=none] (1) at (-1.25, 0.5) {};
		\node [style=none] (2) at (1, 0) {};
		\node [style=rw] (3) at (0, 0) {};
		\node [style=none] (4) at (-1.25, -0.5) {};
	\end{pgfonlayer}
	\begin{pgfonlayer}{edgelayer}
		\draw [in=-135, out=0] (4.center) to (3);
		\draw (2.center) to (3);
		\draw [in=135, out=0] (1.center) to (3);
	\end{pgfonlayer}
\end{tikzpicture}
}},\scalebox{\scalefig}{{\begin{tikzpicture}
	\begin{pgfonlayer}{nodelayer}
		\node [style=none] (2) at (1, 0) {};
		\node [style=rw] (3) at (0, 0) {};
	\end{pgfonlayer}
	\begin{pgfonlayer}{edgelayer}
		\draw (2.center) to (3);
	\end{pgfonlayer}
\end{tikzpicture}
}}\right)$
    and 
    $\left(\scalebox{\scalefig}{{\begin{tikzpicture}
	\begin{pgfonlayer}{nodelayer}
		\node [style=none] (1) at (1, 0.5) {};
		\node [style=none] (2) at (-1.25, 0) {};
		\node [style=lw] (3) at (-0.25, 0) {};
		\node [style=none] (4) at (1, -0.5) {};
	\end{pgfonlayer}
	\begin{pgfonlayer}{edgelayer}
		\draw [in=-45, out=180] (4.center) to (3);
		\draw (2.center) to (3);
		\draw [in=45, out=180] (1.center) to (3);
	\end{pgfonlayer}
\end{tikzpicture}
}},\scalebox{\scalefig}{{\begin{tikzpicture}
	\begin{pgfonlayer}{nodelayer}
		\node [style=none] (2) at (-1.25, 0) {};
		\node [style=lw] (3) at (-0.25, 0) {};
	\end{pgfonlayer}
	\begin{pgfonlayer}{edgelayer}
		\draw (2.center) to (3);
	\end{pgfonlayer}
\end{tikzpicture}
}}\right)$
    with antipode
    $\scalebox{\scalefig}{{\begin{tikzpicture}
	\begin{pgfonlayer}{nodelayer}
		\node [style=white phase dot] (0) at (0, 0) {$-1$};
		\node [style=none] (1) at (-1, 0) {};
		\node [style=none] (2) at (1, 0) {};
	\end{pgfonlayer}
	\begin{pgfonlayer}{edgelayer}
		\draw (1.center) to (2.center);
	\end{pgfonlayer}
\end{tikzpicture}
}}$
    \cite{Hadzihasanovic2015}, it is sufficient to manipulate circuits of gate set $\{G_1, G_2\}$.\\
    \smallskip

    We have introduced a hierarchy of prop theories:
    \[  
        \begin{tikzcd}
            \EaW \arrow[rr, hook] &  & \EpW \arrow[rr, hook] &  & \EZW
        \end{tikzcd}
    \]
    which can be used to deduce rewrites on LFM circuits on increasingly smaller subsets of the gate set :
    \[
    \left\{ \G_1, \G_2\right\} \quad\subset\quad \left\{ \G_1, \G_2, \G_3\right\} \quad\subset\quad \left\{ \G_1, \G_2, \G_3, \G_4 \right\}
    \]
    What if we want the same for LFM circuits containing odd gates, such as
    $\scalebox{\scalefig}{{\input{./figures/gate-set/F.tikz}}}$, introduced in \cref{subsec:2-3} ?
    Then, the $\owedge$ and $\otimes$ structures do not coincide anymore, leading to a different way
    of translating LFM circuits into diagrams.
    \begin{example}[An LFM circuit with odd $F$ gates, with $\owedge$ and $\otimes$]

        Consider the following LFM circuit. Due to \cref{lem:mon-sym-tradeoff} 
        on the relation between the $\owedge$ and $\otimes$ structures, we have:
        \raggedright
        \medskip

        {\centering    
        \scalebox{\scalefig}{{\input{./figures/circuit-examples/FZW_1.tikz}}}\\[1em]
        $\downarrow$\\[1em]
        \scalebox{\scalefig}{{\input{./figures/circuit-examples/FZW_2.tikz}}}\\[1em]
        }
        \raggedright
        \medskip
        \end{example}

        Furthermore, with respect to $\otimes$, $\Pi$ is only presymmetric,
        and therefore $(\LFM, \otimes, \Pi)$ is not a prop, which calls for a relaxation of the prop
        notion : namely, \emph{pronaps} and their theories.

\subsection{Pronaps and their theories}

    The categories generated by fermionic gates do not have to be a prop: they may be only presymmetric monoidal. 
    This leads to the definition of a relaxed version of props: the \textbf{pronap}s.

    \begin{definition}[pronap]
    A \textbf{pronap} (for PROduct and Non-nAtural Permutations) is a presymmetric monoidal category whose monoid of objects is freely generated by a single object $1$.

    The $n$-fold monoidal product of $1$ is denoted $n$, and the monoidal product on arrows is denoted $\otimes$. The generating presymmetry is denoted $\pi$.
    \end{definition}

    \begin{proposition}
    $(\LFM, \otimes, \Pi)$ is a pronap.
    \end{proposition}

    As with props, we want to axiomatize pronaps via sets of generators and equations: a \textbf{pronap theory}. So as 
    to avoid overloading the set of equations with a naturality equation for each natural generator, we differentiate 
    natural and non-natural generators in the formal definition of these theories.

    \begin{definition}[pronap theory]
    A pronap theory is a pair $\T = (\Sigma, E)$ consisting of a set of generators $\Sigma$ and a set of equations $E$:
    \begin{itemize}
        \item $\Sigma$ contains typed generators $g: n \to m$ partitioned into two sets $\Sigma_{nat}$ and $\overline{\Sigma}_{nat}$, containing respectively natural and non-natural generators.
        \item Formal diagrams are defined inductively as follows:
        \begin{itemize}[label=\emph{|}]
            \item a generator $g \in \Sigma$,
            \item the identity $\mathrm{id}_1: 1 \to 1$,
            \item the generating presymmetry $\pi: 2 \to 2$,
            \item $D \circ D' : n \to m$, the composition of two formal diagrams $D: k \to m$, $D': n \to k$,
            \item the monoidal product $D \otimes D': n + n' \to m + m'$ for $D: n \to m$, $D': n' \to m'$ formal diagrams.
        \end{itemize}
        \item $E$ contains pairs of formal diagrams $(D, D')$ of the same type $n \to m$.
    \end{itemize}
    \end{definition}

    Given a pronap theory $(\Sigma_{nat} \cup \overline{\Sigma}_{nat}, E)$, there is a corresponding pronap generated by $\Sigma_{nat} \cup \overline{\Sigma}_{nat}$ and quotiented by $E$ and the axioms of presymmetric monoidal categories, including the naturality equations for all generators $g \in \Sigma_{nat}$. In practice, we denote both the theory and the corresponding pronap $\T$ and rely on the context to disambiguate.

    \begin{definition}[Semantics of a pronap theory]
    Let $\T$ be a pronap theory, $P$ a pronap. A semantics of $\T$ in $P$ is a monoidal functor *
    $\interp{\cdot}_P: \T \to P$ such that $\interp{\pi_\T}_P = \pi_P$.

    We say that $\T$ is \textbf{universal} for $P$ when $\interp{\cdot}_P$ is full, and \textbf{complete} for $P$ when $\interp{\cdot}_P$ is faithful.
    \end{definition}

    \begin{theorem}[De Felice, Hadžihasanović, Ng, \cite{DeFeliceHadzihasanovicNg2019}]
            \raggedright
        There is a complete pronap theory denoted $\FZW$ of $(\LFM, \otimes, \Pi)$ with generators :\\[0.75em]
        {\centering\begin{tabular}{ c c c c c c c c c}
        
        & $\scalebox{\scalefig}{{\begin{tikzpicture}
	\begin{pgfonlayer}{nodelayer}
		\node [style=none] (2) at (1, 0) {};
		\node [style=rw] (3) at (0, 0) {};
	\end{pgfonlayer}
	\begin{pgfonlayer}{edgelayer}
		\draw (2.center) to (3);
	\end{pgfonlayer}
\end{tikzpicture}
}}$
        & $\scalebox{\scalefig}{{\begin{tikzpicture}
	\begin{pgfonlayer}{nodelayer}
		\node [style=none] (2) at (-1.25, 0) {};
		\node [style=lw] (3) at (-0.25, 0) {};
	\end{pgfonlayer}
	\begin{pgfonlayer}{edgelayer}
		\draw (2.center) to (3);
	\end{pgfonlayer}
\end{tikzpicture}
}}$
        & $\scalebox{\scalefig}{{\begin{tikzpicture}
	\begin{pgfonlayer}{nodelayer}
		\node [style=none] (1) at (1, 0.5) {};
		\node [style=none] (2) at (-1.25, 0) {};
		\node [style=lw] (3) at (-0.25, 0) {};
		\node [style=none] (4) at (1, -0.5) {};
	\end{pgfonlayer}
	\begin{pgfonlayer}{edgelayer}
		\draw [in=-45, out=180] (4.center) to (3);
		\draw (2.center) to (3);
		\draw [in=45, out=180] (1.center) to (3);
	\end{pgfonlayer}
\end{tikzpicture}
}}$
        & $\scalebox{\scalefig}{{\begin{tikzpicture}
	\begin{pgfonlayer}{nodelayer}
		\node [style=none] (1) at (-1.25, 0.5) {};
		\node [style=none] (2) at (1, 0) {};
		\node [style=rw] (3) at (0, 0) {};
		\node [style=none] (4) at (-1.25, -0.5) {};
	\end{pgfonlayer}
	\begin{pgfonlayer}{edgelayer}
		\draw [in=-135, out=0] (4.center) to (3);
		\draw (2.center) to (3);
		\draw [in=135, out=0] (1.center) to (3);
	\end{pgfonlayer}
\end{tikzpicture}
}}$\\[1em]
        $\Sigma_{\text{nat}}$ \qquad
        & $\scalebox{\scalefig}{{\begin{tikzpicture}
	\begin{pgfonlayer}{nodelayer}
		\node [style=none] (0) at (-0.25, 0.5) {};
		\node [style=none] (1) at (0.25, 0.5) {};
		\node [style=none] (2) at (0.25, -0.5) {};
		\node [style=none] (3) at (-0.25, -0.5) {};
	\end{pgfonlayer}
	\begin{pgfonlayer}{edgelayer}
		\draw (3.center) to (2.center);
		\draw (0.center) to (1.center);
		\draw [bend right=270, looseness=2.50] (1.center) to (2.center);
	\end{pgfonlayer}
\end{tikzpicture}
}}$
        &  $\scalebox{\scalefig}{{\begin{tikzpicture}
	\begin{pgfonlayer}{nodelayer}
		\node [style=none] (0) at (0.5, 0.5) {};
		\node [style=none] (1) at (0, 0.5) {};
		\node [style=none] (2) at (0, -0.5) {};
		\node [style=none] (3) at (0.5, -0.5) {};
	\end{pgfonlayer}
	\begin{pgfonlayer}{edgelayer}
		\draw (3.center) to (2.center);
		\draw (0.center) to (1.center);
		\draw [bend left=270, looseness=2.50] (1.center) to (2.center);
	\end{pgfonlayer}
\end{tikzpicture}
}}$ 
        &  $\scalebox{\scalefig}{{\begin{tikzpicture}
	\begin{pgfonlayer}{nodelayer}
		\node [style=none] (14) at (-1, 0.5) {};
		\node [style=none] (15) at (-1, -0.5) {};
		\node [style=none] (17) at (1, 0.5) {};
		\node [style=none] (18) at (1, -0.5) {};
		\node [style=white phase dot] (19) at (0, 0) {$x$};
	\end{pgfonlayer}
	\begin{pgfonlayer}{edgelayer}
		\draw [in=45, out=180] (17.center) to (19);
		\draw [in=-45, out=180] (18.center) to (19);
		\draw [in=135, out=0] (14.center) to (19);
		\draw [in=-135, out=0] (15.center) to (19);
	\end{pgfonlayer}
\end{tikzpicture}
}}$ 
        &  $\scalebox{\scalefig}{{\begin{tikzpicture}
	\begin{pgfonlayer}{nodelayer}
		\node [style=white phase dot] (0) at (0, 0) {$x$};
		\node [style=none] (1) at (-1, 0) {};
		\node [style=none] (2) at (1, 0) {};
	\end{pgfonlayer}
	\begin{pgfonlayer}{edgelayer}
		\draw (1.center) to (2.center);
	\end{pgfonlayer}
\end{tikzpicture}
}}$ \\
        \end{tabular}\\[0.75em]
        \medskip
        \begin{tabular}{c c}
                $\overline{\Sigma}_{\text{nat}}$ \quad
                & $\scalebox{\scalefig}{{\begin{tikzpicture}
	\begin{pgfonlayer}{nodelayer}
		\node [style=none] (1) at (1, 0) {};
		\node [style=dot] (3) at (0, 0) {};
		\node [style=none] (4) at (-1, 0) {};
	\end{pgfonlayer}
	\begin{pgfonlayer}{edgelayer}
		\draw (3) to (4.center);
		\draw (1.center) to (3);
	\end{pgfonlayer}
\end{tikzpicture}
}}$
            \end{tabular}\\[2em]}
        \raggedright
    \end{theorem}

    Note that our set of generators is again different from the one in \cite{DeFeliceHadzihasanovicNg2019}. As for 
    $\EZW$, we get rid of the usual swap $\scalebox{\scalefig}{{\begin{tikzpicture}
	\begin{pgfonlayer}{nodelayer}
		\node [style=none] (1) at (-1, 0.5) {};
		\node [style=none] (2) at (1, 0.5) {};
		\node [style=none] (3) at (-1, -0.5) {};
		\node [style=none] (4) at (1, -0.5) {};
		\node [style=none] (9) at (0.75, -0.5) {};
		\node [style=none] (10) at (0.75, 0.5) {};
		\node [style=none] (11) at (-0.75, 0.5) {};
		\node [style=none] (12) at (-0.75, -0.5) {};
	\end{pgfonlayer}
	\begin{pgfonlayer}{edgelayer}
		\draw [in=180, out=0] (11.center) to (9.center);
		\draw (1.center) to (11.center);
		\draw (9.center) to (4.center);
		\draw (2.center) to (10.center);
		\draw [in=0, out=-180] (10.center) to (12.center);
		\draw (12.center) to (3.center);
	\end{pgfonlayer}
\end{tikzpicture}
}}$,
    both as a paradigmatic symmetry and as a generator, and replace it by:
    \[
    \scalebox{\scalefig}{{\begin{tikzpicture}
	\begin{pgfonlayer}{nodelayer}
		\node [style=none] (1) at (-1, 0.5) {};
		\node [style=none] (2) at (1, 0.5) {};
		\node [style=none] (3) at (-1, -0.5) {};
		\node [style=none] (4) at (1, -0.5) {};
		\node [style=none] (9) at (0.75, -0.5) {};
		\node [style=none] (10) at (0.75, 0.5) {};
		\node [style=none] (11) at (-0.75, 0.5) {};
		\node [style=none] (12) at (-0.75, -0.5) {};
		\node [style=fswap] (13) at (0, 0) {};
	\end{pgfonlayer}
	\begin{pgfonlayer}{edgelayer}
		\draw [in=180, out=0] (11.center) to (9.center);
		\draw (1.center) to (11.center);
		\draw (9.center) to (4.center);
		\draw (2.center) to (10.center);
		\draw [in=0, out=-180] (10.center) to (12.center);
		\draw (12.center) to (3.center);
	\end{pgfonlayer}
\end{tikzpicture}
}}
    \qquad
    \text{and}
    \qquad
    \scalebox{\scalefig}{{\begin{tikzpicture}
	\begin{pgfonlayer}{nodelayer}
		\node [style=none] (14) at (-1, 0.5) {};
		\node [style=none] (15) at (-1, -0.5) {};
		\node [style=none] (17) at (1, 0.5) {};
		\node [style=none] (18) at (1, -0.5) {};
		\node [style=white phase dot] (19) at (0, 0) {$x$};
	\end{pgfonlayer}
	\begin{pgfonlayer}{edgelayer}
		\draw [in=45, out=180] (17.center) to (19);
		\draw [in=-45, out=180] (18.center) to (19);
		\draw [in=135, out=0] (14.center) to (19);
		\draw [in=-135, out=0] (15.center) to (19);
	\end{pgfonlayer}
\end{tikzpicture}
}}
    \]
    The original presentation of the calculus uses the generator
    \[\interp{\scalebox{\scalefig}{{\begin{tikzpicture}
	\begin{pgfonlayer}{nodelayer}
		\node [style=none] (1) at (-1, 0.75) {};
		\node [style=none] (2) at (-1, 0.25) {};
		\node [style=none] (3) at (-1, -0.75) {};
		\node [style=none] (4) at (-0.75, 0) {\vdots};
		\node [style=none] (6) at (1, 0.75) {};
		\node [style=none] (7) at (1, 0.25) {};
		\node [style=none] (8) at (1, -0.75) {};
		\node [style=none] (9) at (0.75, 0) {\vdots};
		\node [style=dot] (10) at (0, 0) {};
	\end{pgfonlayer}
	\begin{pgfonlayer}{edgelayer}
		\draw [bend right] (6.center) to (10);
		\draw [bend right=15] (7.center) to (10);
		\draw [bend left] (8.center) to (10);
		\draw [bend left] (1.center) to (10);
		\draw [bend left=15] (2.center) to (10);
		\draw [bend right] (3.center) to (10);
	\end{pgfonlayer}
\end{tikzpicture}
}}}_\LFM = \sum_{\substack{xy~ \in \{0,1\}^{n+m} \\ |xy| =1}}\ketbra{x}{y}\]
    We choose to stick with the generators of $\EZW$, only adding one non-natural new generator $\scalebox{\scalefig}{{\begin{tikzpicture}
	\begin{pgfonlayer}{nodelayer}
		\node [style=none] (1) at (1, 0) {};
		\node [style=dot] (3) at (0, 0) {};
		\node [style=none] (4) at (-1, 0) {};
	\end{pgfonlayer}
	\begin{pgfonlayer}{edgelayer}
		\draw (3) to (4.center);
		\draw (1.center) to (3);
	\end{pgfonlayer}
\end{tikzpicture}
}}$, 
    which allows us to recover the same expressivity, as $\scalebox{\scalefig}{{\begin{tikzpicture}
	\begin{pgfonlayer}{nodelayer}
		\node [style=none] (1) at (-1, 0.75) {};
		\node [style=none] (2) at (-1, 0.25) {};
		\node [style=none] (3) at (-1, -0.75) {};
		\node [style=none] (4) at (-0.75, 0) {\vdots};
		\node [style=none] (6) at (1, 0.75) {};
		\node [style=none] (7) at (1, 0.25) {};
		\node [style=none] (8) at (1, -0.75) {};
		\node [style=none] (9) at (0.75, 0) {\vdots};
		\node [style=dot] (10) at (0, 0) {};
	\end{pgfonlayer}
	\begin{pgfonlayer}{edgelayer}
		\draw [bend right] (6.center) to (10);
		\draw [bend right=15] (7.center) to (10);
		\draw [bend left] (8.center) to (10);
		\draw [bend left] (1.center) to (10);
		\draw [bend left=15] (2.center) to (10);
		\draw [bend right] (3.center) to (10);
	\end{pgfonlayer}
\end{tikzpicture}
}}$ can be expressed as a syntactic sugar
    as follows:
    \[\scalebox{\scalefig}{{\begin{tikzpicture}
	\begin{pgfonlayer}{nodelayer}
		\node [style=none] (1) at (-1, 0.75) {};
		\node [style=none] (2) at (-1, 0.25) {};
		\node [style=none] (3) at (-1, -0.75) {};
		\node [style=none] (4) at (-0.75, 0) {\vdots};
		\node [style=none] (6) at (1, 0.75) {};
		\node [style=none] (7) at (1, 0.25) {};
		\node [style=none] (8) at (1, -0.75) {};
		\node [style=none] (9) at (0.75, 0) {\vdots};
		\node [style=dot] (10) at (0, 0) {};
	\end{pgfonlayer}
	\begin{pgfonlayer}{edgelayer}
		\draw [bend right] (6.center) to (10);
		\draw [bend right=15] (7.center) to (10);
		\draw [bend left] (8.center) to (10);
		\draw [bend left] (1.center) to (10);
		\draw [bend left=15] (2.center) to (10);
		\draw [bend right] (3.center) to (10);
	\end{pgfonlayer}
\end{tikzpicture}
}} \quad := \quad \scalebox{\scalefig}{{\begin{tikzpicture}
	\begin{pgfonlayer}{nodelayer}
		\node [style=dot] (3) at (0, 0) {};
		\node [style=none] (6) at (3.75, 1) {};
		\node [style=lw] (7) at (0.75, 0) {};
		\node [style=none] (10) at (3.75, 0.5) {};
		\node [style=lw] (11) at (1.5, -0.5) {};
		\node [style=none] (12) at (2.25, -1.25) {};
		\node [style=none] (14) at (4.25, -1.5) {};
		\node [style=none] (15) at (3.25, -1.25) {};
		\node [style=lw] (16) at (3.75, -1.25) {};
		\node [style=none] (17) at (5, -0.75) {};
		\node [style=none] (18) at (5, 1) {};
		\node [style=none] (19) at (5, 0.5) {};
		\node [style=none] (20) at (5, 0) {$\vdots$};
		\node [style=none] (21) at (2.75, -1.25) {$\dots$};
		\node [style=none] (23) at (-3.75, 1) {};
		\node [style=none] (27) at (-3.75, 0.5) {};
		\node [style=none] (29) at (-2.25, -1.25) {};
		\node [style=none] (31) at (-4.25, -1.5) {};
		\node [style=none] (32) at (-3.25, -1.25) {};
		\node [style=none] (34) at (-5, -0.75) {};
		\node [style=none] (35) at (-5, 1) {};
		\node [style=none] (36) at (-5, 0.5) {};
		\node [style=none] (37) at (-5, 0) {$\vdots$};
		\node [style=none] (38) at (-2.75, -1.25) {$\dots$};
		\node [style=rw] (39) at (-0.75, 0) {};
		\node [style=rw] (40) at (-1.5, -0.5) {};
		\node [style=rw] (41) at (-3.75, -1.25) {};
		\node [style=rw] (42) at (-4.5, -1.5) {};
		\node [style=lw] (43) at (4.5, -1.5) {};
	\end{pgfonlayer}
	\begin{pgfonlayer}{edgelayer}
		\draw [in=-45, out=180] (12.center) to (11);
		\draw [in=45, out=-180] (17.center) to (16);
		\draw (19.center) to (10.center);
		\draw (18.center) to (6.center);
		\draw (36.center) to (27.center);
		\draw (35.center) to (23.center);
		\draw (7) to (3);
		\draw [in=45, out=-180] (6.center) to (7);
		\draw [in=-45, out=180, looseness=1.25] (11) to (7);
		\draw [in=45, out=-180] (10.center) to (11);
		\draw [in=-45, out=180] (14.center) to (16);
		\draw (15.center) to (16);
		\draw [in=-135, out=0] (29.center) to (40);
		\draw [in=135, out=0] (34.center) to (41);
		\draw (3) to (39);
		\draw [in=135, out=0] (23.center) to (39);
		\draw [in=-135, out=0, looseness=1.25] (40) to (39);
		\draw [in=135, out=0] (27.center) to (40);
		\draw [in=-135, out=0] (31.center) to (41);
		\draw (32.center) to (41);
		\draw (43) to (14.center);
		\draw (42) to (31.center);
	\end{pgfonlayer}
\end{tikzpicture}
}}\]
    This presentation has the merits to reduce to one the number of non-natural generator of the pronap theory,
    which makes diagrammatic reasoning on $\FZW$ diagrams arguably easier. 
    Furthermore, the new generator $\scalebox{\scalefig}{{\begin{tikzpicture}
	\begin{pgfonlayer}{nodelayer}
		\node [style=none] (1) at (1, 0) {};
		\node [style=dot] (3) at (0, 0) {};
		\node [style=none] (4) at (-1, 0) {};
	\end{pgfonlayer}
	\begin{pgfonlayer}{edgelayer}
		\draw (3) to (4.center);
		\draw (1.center) to (3);
	\end{pgfonlayer}
\end{tikzpicture}
}}$ happens to be the translation
    in $\FZW$ of the $F$ gate:
    \[
    \scalebox{\scalefig}{{\input{./figures/gate-set/F.tikz}}}\quad \eq \quad \scalebox{\scalefig}{{\begin{tikzpicture}
	\begin{pgfonlayer}{nodelayer}
		\node [style=none] (1) at (1, 0) {};
		\node [style=dot] (3) at (0, 0) {};
		\node [style=none] (4) at (-1, 0) {};
	\end{pgfonlayer}
	\begin{pgfonlayer}{edgelayer}
		\draw (3) to (4.center);
		\draw (1.center) to (3);
	\end{pgfonlayer}
\end{tikzpicture}
}}
    \]
    As a result, it is possible to translate fermionic circuits on the gate-set $(\G_1, \G_2, \G_3, \G_4, F)$
    to $\FZW$ diagrams.

    \begin{example}[An LFM circuit with odd $F$ gates translated to $\FZW$]

        Consider the following LFM circuit, which can be translated in $\FZW$ as follows 
        due to \cref{lem:mon-sym-tradeoff} on the relation between the $\owedge$ and $\otimes$ structures:
        \raggedright
        \medskip

        {\centering    
        \scalebox{\scalefig}{{\input{./figures/circuit-examples/FZW_1.tikz}}}\\[1em]
        $\downarrow$\\[1em]
        \scalebox{\scalefig}{{\input{./figures/circuit-examples/FZW_2.tikz}}}\\[1em]
        $\downarrow$\\[1em]
        \scalebox{\scalefig}{{\input{./figures/circuit-examples/FZW_3.tikz}}}\\[1em]
        $\downarrow$\\[1em]
        \scalebox{\scalefig}{{\input{./figures/circuit-examples/FZW_4.tikz}}}\\[1em]
        }
        \raggedright
        \medskip
        \end{example}

    As before, the fragment obtained by removing the generator 
    $\scalebox{\scalefig}{{\begin{tikzpicture}
	\begin{pgfonlayer}{nodelayer}
		\node [style=none] (14) at (-1, 0.5) {};
		\node [style=none] (15) at (-1, -0.5) {};
		\node [style=none] (17) at (1, 0.5) {};
		\node [style=none] (18) at (1, -0.5) {};
		\node [style=white phase dot] (19) at (0, 0) {$x$};
	\end{pgfonlayer}
	\begin{pgfonlayer}{edgelayer}
		\draw [in=45, out=180] (17.center) to (19);
		\draw [in=-45, out=180] (18.center) to (19);
		\draw [in=135, out=0] (14.center) to (19);
		\draw [in=-135, out=0] (15.center) to (19);
	\end{pgfonlayer}
\end{tikzpicture}
}}$ is sufficient to 
    translate LFM circuits on the gate-set $(\G_1, \G_2, \G_3, F)$. This 
    fragment is called \FpW.

    \begin{definition}[$\FpW$]
        \raggedright
        $\FpW$ is the following pronap theory with canonical presymmetry $\scalebox{\scalefig}{\begin{tikzpicture}
	\begin{pgfonlayer}{nodelayer}
		\node [style=none] (1) at (-1, 0.5) {};
		\node [style=none] (2) at (1, 0.5) {};
		\node [style=none] (3) at (-1, -0.5) {};
		\node [style=none] (4) at (1, -0.5) {};
		\node [style=none] (9) at (0.75, -0.5) {};
		\node [style=none] (10) at (0.75, 0.5) {};
		\node [style=none] (11) at (-0.75, 0.5) {};
		\node [style=none] (12) at (-0.75, -0.5) {};
		\node [style=fswap] (13) at (0, 0) {};
	\end{pgfonlayer}
	\begin{pgfonlayer}{edgelayer}
		\draw [in=180, out=0] (11.center) to (9.center);
		\draw (1.center) to (11.center);
		\draw (9.center) to (4.center);
		\draw (2.center) to (10.center);
		\draw [in=0, out=-180] (10.center) to (12.center);
		\draw (12.center) to (3.center);
	\end{pgfonlayer}
\end{tikzpicture}
}$:\\[0.75em]
        \begin{adjustwidth}{-0cm}{0cm}
              {\centering\begin{tabular}{ c c c c c c c c c}
          $\Sigma_{\text{nat}}$ \qquad
        & $\scalebox{\scalefig}{{\begin{tikzpicture}
	\begin{pgfonlayer}{nodelayer}
		\node [style=none] (2) at (1, 0) {};
		\node [style=rw] (3) at (0, 0) {};
	\end{pgfonlayer}
	\begin{pgfonlayer}{edgelayer}
		\draw (2.center) to (3);
	\end{pgfonlayer}
\end{tikzpicture}
}}$
        & $\scalebox{\scalefig}{{\begin{tikzpicture}
	\begin{pgfonlayer}{nodelayer}
		\node [style=none] (2) at (-1.25, 0) {};
		\node [style=lw] (3) at (-0.25, 0) {};
	\end{pgfonlayer}
	\begin{pgfonlayer}{edgelayer}
		\draw (2.center) to (3);
	\end{pgfonlayer}
\end{tikzpicture}
}}$
        & $\scalebox{\scalefig}{{\begin{tikzpicture}
	\begin{pgfonlayer}{nodelayer}
		\node [style=none] (1) at (1, 0.5) {};
		\node [style=none] (2) at (-1.25, 0) {};
		\node [style=lw] (3) at (-0.25, 0) {};
		\node [style=none] (4) at (1, -0.5) {};
	\end{pgfonlayer}
	\begin{pgfonlayer}{edgelayer}
		\draw [in=-45, out=180] (4.center) to (3);
		\draw (2.center) to (3);
		\draw [in=45, out=180] (1.center) to (3);
	\end{pgfonlayer}
\end{tikzpicture}
}}$
        & $\scalebox{\scalefig}{{\begin{tikzpicture}
	\begin{pgfonlayer}{nodelayer}
		\node [style=none] (1) at (-1.25, 0.5) {};
		\node [style=none] (2) at (1, 0) {};
		\node [style=rw] (3) at (0, 0) {};
		\node [style=none] (4) at (-1.25, -0.5) {};
	\end{pgfonlayer}
	\begin{pgfonlayer}{edgelayer}
		\draw [in=-135, out=0] (4.center) to (3);
		\draw (2.center) to (3);
		\draw [in=135, out=0] (1.center) to (3);
	\end{pgfonlayer}
\end{tikzpicture}
}}$\\[1em]
        & $\scalebox{\scalefig}{{\begin{tikzpicture}
	\begin{pgfonlayer}{nodelayer}
		\node [style=none] (0) at (-0.25, 0.5) {};
		\node [style=none] (1) at (0.25, 0.5) {};
		\node [style=none] (2) at (0.25, -0.5) {};
		\node [style=none] (3) at (-0.25, -0.5) {};
	\end{pgfonlayer}
	\begin{pgfonlayer}{edgelayer}
		\draw (3.center) to (2.center);
		\draw (0.center) to (1.center);
		\draw [bend right=270, looseness=2.50] (1.center) to (2.center);
	\end{pgfonlayer}
\end{tikzpicture}
}}$
        &  $\scalebox{\scalefig}{{\begin{tikzpicture}
	\begin{pgfonlayer}{nodelayer}
		\node [style=none] (0) at (0.5, 0.5) {};
		\node [style=none] (1) at (0, 0.5) {};
		\node [style=none] (2) at (0, -0.5) {};
		\node [style=none] (3) at (0.5, -0.5) {};
	\end{pgfonlayer}
	\begin{pgfonlayer}{edgelayer}
		\draw (3.center) to (2.center);
		\draw (0.center) to (1.center);
		\draw [bend left=270, looseness=2.50] (1.center) to (2.center);
	\end{pgfonlayer}
\end{tikzpicture}
}}$ 
        &  $\scalebox{\scalefig}{{\begin{tikzpicture}
	\begin{pgfonlayer}{nodelayer}
		\node [style=white phase dot] (0) at (0, 0) {$x$};
		\node [style=none] (1) at (-1, 0) {};
		\node [style=none] (2) at (1, 0) {};
	\end{pgfonlayer}
	\begin{pgfonlayer}{edgelayer}
		\draw (1.center) to (2.center);
	\end{pgfonlayer}
\end{tikzpicture}
}}$ \\
        \end{tabular}\\[0.75em]
        \medskip
        \begin{tabular}{c c}
                $\overline{\Sigma}_{\text{nat}}$ \quad
                & $\scalebox{\scalefig}{{\begin{tikzpicture}
	\begin{pgfonlayer}{nodelayer}
		\node [style=none] (1) at (1, 0) {};
		\node [style=dot] (3) at (0, 0) {};
		\node [style=none] (4) at (-1, 0) {};
	\end{pgfonlayer}
	\begin{pgfonlayer}{edgelayer}
		\draw (3) to (4.center);
		\draw (1.center) to (3);
	\end{pgfonlayer}
\end{tikzpicture}
}}$
            \end{tabular}\\[2em]}
           \end{adjustwidth}
           \centering {
           Equations of \FpW~are pictured on \cref{fig:equational-theories}
           }\\

        \raggedright
    \end{definition}
    
    The pronap theory \FpW~ is an axiomatisation of the exact same category that was studied by Carette et al 
    in \cite{CaretteEtAl2023} through a \emph{pro} theory called \pW~, that is to say an equational theory with 
    no paradigmatic symmetry or presymmetry whatsoever.

    \begin{theorem}[Carette, Moutot, Perez, Vilmart]
    There is a complete pro theory $\pW$ and generators:
    
    {\centering
   \begin{tabular}{c c c c}
    $\scalebox{\scalefig}{{\begin{tikzpicture}
	\begin{pgfonlayer}{nodelayer}
		\node [style=white phase dot] (0) at (0, 0) {$x$};
		\node [style=none] (1) at (-1, 0) {};
		\node [style=none] (2) at (1, 0) {};
	\end{pgfonlayer}
	\begin{pgfonlayer}{edgelayer}
		\draw (1.center) to (2.center);
	\end{pgfonlayer}
\end{tikzpicture}
}}$
    & $\scalebox{\scalefig}{{\begin{tikzpicture}
	\begin{pgfonlayer}{nodelayer}
		\node [style=none] (0) at (-0.25, 0.5) {};
		\node [style=none] (1) at (0.25, 0.5) {};
		\node [style=none] (2) at (0.25, -0.5) {};
		\node [style=none] (3) at (-0.25, -0.5) {};
	\end{pgfonlayer}
	\begin{pgfonlayer}{edgelayer}
		\draw (3.center) to (2.center);
		\draw (0.center) to (1.center);
		\draw [bend right=270, looseness=2.50] (1.center) to (2.center);
	\end{pgfonlayer}
\end{tikzpicture}
}}$
    & $\scalebox{\scalefig}{{\begin{tikzpicture}
	\begin{pgfonlayer}{nodelayer}
		\node [style=none] (0) at (0.5, 0.5) {};
		\node [style=none] (1) at (0, 0.5) {};
		\node [style=none] (2) at (0, -0.5) {};
		\node [style=none] (3) at (0.5, -0.5) {};
	\end{pgfonlayer}
	\begin{pgfonlayer}{edgelayer}
		\draw (3.center) to (2.center);
		\draw (0.center) to (1.center);
		\draw [bend left=270, looseness=2.50] (1.center) to (2.center);
	\end{pgfonlayer}
\end{tikzpicture}
}}$ 
    & $\scalebox{\scalefig}{{\begin{tikzpicture}
	\begin{pgfonlayer}{nodelayer}
		\node [style=none] (1) at (-1, 0.75) {};
		\node [style=none] (2) at (-1, 0.25) {};
		\node [style=none] (3) at (-1, -0.75) {};
		\node [style=none] (4) at (-0.75, 0) {\vdots};
		\node [style=none] (6) at (1, 0.75) {};
		\node [style=none] (7) at (1, 0.25) {};
		\node [style=none] (8) at (1, -0.75) {};
		\node [style=none] (9) at (0.75, 0) {\vdots};
		\node [style=dot] (10) at (0, 0) {};
	\end{pgfonlayer}
	\begin{pgfonlayer}{edgelayer}
		\draw [bend right] (6.center) to (10);
		\draw [bend right=15] (7.center) to (10);
		\draw [bend left] (8.center) to (10);
		\draw [bend left] (1.center) to (10);
		\draw [bend left=15] (2.center) to (10);
		\draw [bend right] (3.center) to (10);
	\end{pgfonlayer}
\end{tikzpicture}
}}$
    \end{tabular}\\[0.75em]
   }
    \raggedright
    \end{theorem}

    In \pW~, the fermionic swap $\scalebox{\scalefig}{\begin{tikzpicture}
	\begin{pgfonlayer}{nodelayer}
		\node [style=none] (1) at (-1, 0.5) {};
		\node [style=none] (2) at (1, 0.5) {};
		\node [style=none] (3) at (-1, -0.5) {};
		\node [style=none] (4) at (1, -0.5) {};
		\node [style=none] (9) at (0.75, -0.5) {};
		\node [style=none] (10) at (0.75, 0.5) {};
		\node [style=none] (11) at (-0.75, 0.5) {};
		\node [style=none] (12) at (-0.75, -0.5) {};
		\node [style=fswap] (13) at (0, 0) {};
	\end{pgfonlayer}
	\begin{pgfonlayer}{edgelayer}
		\draw [in=180, out=0] (11.center) to (9.center);
		\draw (1.center) to (11.center);
		\draw (9.center) to (4.center);
		\draw (2.center) to (10.center);
		\draw [in=0, out=-180] (10.center) to (12.center);
		\draw (12.center) to (3.center);
	\end{pgfonlayer}
\end{tikzpicture}
}$ is only a syntactic sugar which is decomposed as: 
    \[\scalebox{\scalefig}{\begin{tikzpicture}
	\begin{pgfonlayer}{nodelayer}
		\node [style=none] (1) at (-1, 0.5) {};
		\node [style=none] (2) at (1, 0.5) {};
		\node [style=none] (3) at (-1, -0.5) {};
		\node [style=none] (4) at (1, -0.5) {};
		\node [style=none] (9) at (0.75, -0.5) {};
		\node [style=none] (10) at (0.75, 0.5) {};
		\node [style=none] (11) at (-0.75, 0.5) {};
		\node [style=none] (12) at (-0.75, -0.5) {};
		\node [style=fswap] (13) at (0, 0) {};
	\end{pgfonlayer}
	\begin{pgfonlayer}{edgelayer}
		\draw [in=180, out=0] (11.center) to (9.center);
		\draw (1.center) to (11.center);
		\draw (9.center) to (4.center);
		\draw (2.center) to (10.center);
		\draw [in=0, out=-180] (10.center) to (12.center);
		\draw (12.center) to (3.center);
	\end{pgfonlayer}
\end{tikzpicture}
} \quad := \quad \scalebox{\scalefig}{\begin{tikzpicture}
	\begin{pgfonlayer}{nodelayer}
		\node [style=none] (1) at (-2, 0.5) {};
		\node [style=none] (2) at (2, 0.5) {};
		\node [style=none] (3) at (-2, -0.5) {};
		\node [style=none] (4) at (2, -0.5) {};
		\node [style=none] (9) at (1, -0.5) {};
		\node [style=none] (10) at (1, 0.5) {};
		\node [style=none] (11) at (-1, 0.5) {};
		\node [style=none] (12) at (-1, -0.5) {};
		\node [style=dot] (13) at (-1, 0.5) {};
		\node [style=dot] (14) at (-1.5, 0.5) {};
		\node [style=dot] (15) at (-1.5, -0.5) {};
		\node [style=dot] (16) at (-1, -0.5) {};
		\node [style=dot] (17) at (-0.25, 0) {};
		\node [style=dot] (18) at (0.25, 0) {};
		\node [style=dot] (19) at (1, -0.5) {};
		\node [style=dot] (20) at (1, 0.5) {};
		\node [style=dot] (21) at (1.5, 0.5) {};
		\node [style=dot] (22) at (1.5, -0.5) {};
		\node [style=white phase dot] (23) at (0, 0.75) {$-1$};
		\node [style=white phase dot] (24) at (0, -0.75) {$-1$};
	\end{pgfonlayer}
	\begin{pgfonlayer}{edgelayer}
		\draw (1.center) to (11.center);
		\draw (9.center) to (4.center);
		\draw (2.center) to (10.center);
		\draw (12.center) to (3.center);
		\draw (17) to (18);
		\draw (13) to (17);
		\draw (17) to (16);
		\draw (18) to (19);
		\draw (18) to (20);
		\draw [bend left, looseness=0.75] (13) to (20);
		\draw [bend right, looseness=0.75] (16) to (19);
	\end{pgfonlayer}
\end{tikzpicture}
}\]
    As a consequence, \pW~ diagrams are essentially open planar graphs, hence the name \textbf{planar W} calculus.
    The \FpW~ presentation of the same category has, on the other hand, a paradigmatic presymmetry $\scalebox{\scalefig}{\begin{tikzpicture}
	\begin{pgfonlayer}{nodelayer}
		\node [style=none] (1) at (-1, 0.5) {};
		\node [style=none] (2) at (1, 0.5) {};
		\node [style=none] (3) at (-1, -0.5) {};
		\node [style=none] (4) at (1, -0.5) {};
		\node [style=none] (9) at (0.75, -0.5) {};
		\node [style=none] (10) at (0.75, 0.5) {};
		\node [style=none] (11) at (-0.75, 0.5) {};
		\node [style=none] (12) at (-0.75, -0.5) {};
		\node [style=fswap] (13) at (0, 0) {};
	\end{pgfonlayer}
	\begin{pgfonlayer}{edgelayer}
		\draw [in=180, out=0] (11.center) to (9.center);
		\draw (1.center) to (11.center);
		\draw (9.center) to (4.center);
		\draw (2.center) to (10.center);
		\draw [in=0, out=-180] (10.center) to (12.center);
		\draw (12.center) to (3.center);
	\end{pgfonlayer}
\end{tikzpicture}
}$
    and limits the number of non-natural generators with respect to this presymmetry to one : $\scalebox{\scalefig}{{\begin{tikzpicture}
	\begin{pgfonlayer}{nodelayer}
		\node [style=none] (1) at (1, 0) {};
		\node [style=dot] (3) at (0, 0) {};
		\node [style=none] (4) at (-1, 0) {};
	\end{pgfonlayer}
	\begin{pgfonlayer}{edgelayer}
		\draw (3) to (4.center);
		\draw (1.center) to (3);
	\end{pgfonlayer}
\end{tikzpicture}
}}$.\\

    \medskip

    Another interesting pronap can be obtained by adding non-natural generators $\scalebox{\scalefig}{{\begin{tikzpicture}
	\begin{pgfonlayer}{nodelayer}
		\node [style=none] (0) at (0.5, 0) {};
		\node [style=none] (1) at (0, 0) {};
		\node [style=dot] (4) at (-0.5, 0) {};
	\end{pgfonlayer}
	\begin{pgfonlayer}{edgelayer}
		\draw (0.center) to (1.center);
		\draw (1.center) to (4);
	\end{pgfonlayer}
\end{tikzpicture}
}}$ and $\scalebox{\scalefig}{{\begin{tikzpicture}
	\begin{pgfonlayer}{nodelayer}
		\node [style=none] (0) at (-0.5, 0) {};
		\node [style=none] (1) at (0, 0) {};
		\node [style=dot] (4) at (0.5, 0) {};
	\end{pgfonlayer}
	\begin{pgfonlayer}{edgelayer}
		\draw (0.center) to (1.center);
		\draw (1.center) to (4);
	\end{pgfonlayer}
\end{tikzpicture}
}}$ to $\EaW$,
    which can be understood as allowing preparation and postselection of states and effects $\ket{1}, \bra{1}$, that is of filled fermionic modes,
    to circuits on the gate set $(\G_1, \G_2)$.

    \begin{definition}[$\aW$]
        \raggedright
        $\aW$ is the following pronap theory with canonical presymmetry $\scalebox{\scalefig}{\begin{tikzpicture}
	\begin{pgfonlayer}{nodelayer}
		\node [style=none] (1) at (-1, 0.5) {};
		\node [style=none] (2) at (1, 0.5) {};
		\node [style=none] (3) at (-1, -0.5) {};
		\node [style=none] (4) at (1, -0.5) {};
		\node [style=none] (9) at (0.75, -0.5) {};
		\node [style=none] (10) at (0.75, 0.5) {};
		\node [style=none] (11) at (-0.75, 0.5) {};
		\node [style=none] (12) at (-0.75, -0.5) {};
		\node [style=fswap] (13) at (0, 0) {};
	\end{pgfonlayer}
	\begin{pgfonlayer}{edgelayer}
		\draw [in=180, out=0] (11.center) to (9.center);
		\draw (1.center) to (11.center);
		\draw (9.center) to (4.center);
		\draw (2.center) to (10.center);
		\draw [in=0, out=-180] (10.center) to (12.center);
		\draw (12.center) to (3.center);
	\end{pgfonlayer}
\end{tikzpicture}
}$:\\[0.75em]
            {\centering
            Generators of $\aW$ \\[1em]
            \begin{tabular}{c c c c c c c}
                $\Sigma_{\text{nat}}$ \quad
                & $\scalebox{\scalefig}{{\begin{tikzpicture}
	\begin{pgfonlayer}{nodelayer}
		\node [style=white phase dot] (0) at (0, 0) {$x$};
		\node [style=none] (1) at (-1, 0) {};
		\node [style=none] (2) at (1, 0) {};
	\end{pgfonlayer}
	\begin{pgfonlayer}{edgelayer}
		\draw (1.center) to (2.center);
	\end{pgfonlayer}
\end{tikzpicture}
}}$
                & $\scalebox{\scalefig}{{\begin{tikzpicture}
	\begin{pgfonlayer}{nodelayer}
		\node [style=none] (2) at (-1.25, 0) {};
		\node [style=lw] (3) at (-0.25, 0) {};
	\end{pgfonlayer}
	\begin{pgfonlayer}{edgelayer}
		\draw (2.center) to (3);
	\end{pgfonlayer}
\end{tikzpicture}
}}$
                & $\scalebox{\scalefig}{{\begin{tikzpicture}
	\begin{pgfonlayer}{nodelayer}
		\node [style=none] (2) at (1, 0) {};
		\node [style=rw] (3) at (0, 0) {};
	\end{pgfonlayer}
	\begin{pgfonlayer}{edgelayer}
		\draw (2.center) to (3);
	\end{pgfonlayer}
\end{tikzpicture}
}}$
                & $\scalebox{\scalefig}{{\begin{tikzpicture}
	\begin{pgfonlayer}{nodelayer}
		\node [style=none] (1) at (1, 0.5) {};
		\node [style=none] (2) at (-1.25, 0) {};
		\node [style=lw] (3) at (-0.25, 0) {};
		\node [style=none] (4) at (1, -0.5) {};
	\end{pgfonlayer}
	\begin{pgfonlayer}{edgelayer}
		\draw [in=-45, out=180] (4.center) to (3);
		\draw (2.center) to (3);
		\draw [in=45, out=180] (1.center) to (3);
	\end{pgfonlayer}
\end{tikzpicture}
}}$
                & $\scalebox{\scalefig}{{\begin{tikzpicture}
	\begin{pgfonlayer}{nodelayer}
		\node [style=none] (1) at (-1.25, 0.5) {};
		\node [style=none] (2) at (1, 0) {};
		\node [style=rw] (3) at (0, 0) {};
		\node [style=none] (4) at (-1.25, -0.5) {};
	\end{pgfonlayer}
	\begin{pgfonlayer}{edgelayer}
		\draw [in=-135, out=0] (4.center) to (3);
		\draw (2.center) to (3);
		\draw [in=135, out=0] (1.center) to (3);
	\end{pgfonlayer}
\end{tikzpicture}
}}$
            \end{tabular}\\[0.75em]
            \begin{tabular}{c c c}
                $\overline{\Sigma}_{\text{nat}}$ \quad
                & $\scalebox{\scalefig}{{\begin{tikzpicture}
	\begin{pgfonlayer}{nodelayer}
		\node [style=none] (0) at (0.5, 0) {};
		\node [style=none] (1) at (0, 0) {};
		\node [style=dot] (4) at (-0.5, 0) {};
	\end{pgfonlayer}
	\begin{pgfonlayer}{edgelayer}
		\draw (0.center) to (1.center);
		\draw (1.center) to (4);
	\end{pgfonlayer}
\end{tikzpicture}
}}$
                & $\scalebox{\scalefig}{{\begin{tikzpicture}
	\begin{pgfonlayer}{nodelayer}
		\node [style=none] (0) at (-0.5, 0) {};
		\node [style=none] (1) at (0, 0) {};
		\node [style=dot] (4) at (0.5, 0) {};
	\end{pgfonlayer}
	\begin{pgfonlayer}{edgelayer}
		\draw (0.center) to (1.center);
		\draw (1.center) to (4);
	\end{pgfonlayer}
\end{tikzpicture}
}}$
             \end{tabular}\\[2em]

            Equations of $\aW$ are pictured on \cref{fig:equational-theories}
            }
            \centering\\
        \raggedright
    \end{definition}

    To sum up, here is a picture of the hierarchy of fragments of \FZW ~:
    \[
        \begin{tikzcd}
    \aW \arrow[rr, hook]                   &  & \FpW \arrow[rr, hook]                  &  & \FZW                  \\
                                        &  &                                        &  &                       \\
    \EaW \arrow[rr, hook] \arrow[uu, hook] &  & \EpW \arrow[rr, hook] \arrow[uu, hook] &  & \EZW \arrow[uu, hook]
    \end{tikzcd}
    \]

    And there is the corresponding hierarchy of gate-sets from which each fragment can be used to write LFM circuits.
    \[
    \begin{tikzcd}
    {\substack{\{\G_1, \G_2\}\\ + \\ \ket{1}, \bra{1}}} \arrow[rr, "\subset" description]          &  & {\{\G_1, \G_2, \G_3, F\}} \arrow[rr, "\subset" description]                             &  & {\{\G_1, \G_2, \G_3, \G_4, F\}}                             \\
                                                                                    &  &                                                                                         &  &                                                             \\
    {\{\G_1, \G_2\}} \arrow[rr, "\subset" description] \arrow[uu, "\cup" description] &  & {\{\G_1, \G_2, \G_3\}} \arrow[rr, "\subset" description] \arrow[uu, "\cup" description] &  & {\{\G_1, \G_2, \G_3, \G_4\}} \arrow[uu, "\cup" description]
    \end{tikzcd}
    \]

    We already have completeness for $\FZW$, $\EZW$, and $\pW$, which is a theory for the same category as $\FpW$. 
    In \cref{sec:5}, we provide proofs of the completeness of $\EaW$, $\aW$, $\EpW$, and a new completeness proof for $\FpW$, 
    via the introduction of another formal tool: \emph{scalable notations for pronaps}.


\section{Scalable notations and completeness results}\label{sec:5}
        Scalable notations were introduced in the SZX-calculus to manipulate families of
    wires and generators compactly, and was later abstracted from ZX to arbitrary
    coloured props \cite{CaretteHorsmanPerdrix2019,CarettePerdrix2020}. This
    construction is closely related to graphical approaches to strictification and
    de-strictification, where the bookkeeping hidden by strict monoidal syntax is
    made explicit \cite{WilsonGhicaZanasi2024}. Scalable notations have since also
    appeared in Graphical Symplectic Algebra and, more recently, Graphical
    Algebraic Geometry \cite{BoothCaretteComfort2026,GaoShaikhKissinger2026},
    both of which connect to the broader Graphical Linear Algebra programme
    \cite{BonchiSobocinskiZanasi2017}. Here, we extend scalable notations from props
    to pronaps and use it to formulate the normal forms and completeness results
    of the fragments introduced above.

\subsection{Scalable notations}\label{subsec:scalable-notations}
    The pronap theories of subcategories of $(\LFM, \otimes, \pi)$ allow one to manipulate fermionic circuits through local rewrites. However, the diagrammatic reasoning involved may become tedious when considering large fermionic circuits whose underlying diagrams quickly become untractable.

    Scalable notations constitute a syntactical tool that solves this issue by giving a more compact description of certain large diagrams.

    The idea is to axiomatize a category $\S\P$ which is equivalent to $\P$, but whose monoid of objects is freely spanned by the set $\{\underline{m} \mid m \in \mathbb{N}\}$, where type $\underline{m}$ corresponds to a register of $m$ modes and is pictured as a thick wire $\underline{m}$.

    \begin{definition}[Scaled pronap]
    Given a pronap $\P$, $\S\P$ is the monoidal presymmetric category whose monoid of objects is freely spanned by $\{\underline{m} \mid m > 0\}$. The $k$-fold monoidal product of object $\underline{m}$ is denoted $\underline{m} + \dots + \underline{m} = k \cdot \underline{m}$. $\S\P$ is defined as:
    \[
    \S\P\left[\sum_{i=1}^k \underline{m_i}, \sum_{j=1}^{k'} \underline{m'_j}\right] := \P\left[\sum_{i=1}^k m_i, \sum_{j=1}^{k'} m'_j\right].
    \]
    \[
    \pi_{\underline{m_1} + \dots + \underline{m_k}, \hspace*{0.5em} \underline{m_1'} + \dots + \underline{m_{k'}'}} := \pi_{m_1 + \dots + m_k, \hspace*{0.5em} m_1' + \dots + m_{k'}'}.
    \]
    \end{definition}
    There is an immediate monoidal equivalence
    \[
    \begin{tikzcd}
    \P \arrow[rr, "\iota_\P", bend left] &  & \S\P \arrow[ll, "\text{unfold}_\P", bend left]
    \end{tikzcd}
    \]
    with
    \[
    \begin{cases}
    \iota_\P(m) = m \cdot \underline{1} \\
    \iota_\P(f) = f : m \cdot \underline{1} \to m' \cdot \underline{1}
    \end{cases}
    \]
    and
    \[
    \begin{cases}
    \text{unfold}_\P(\underline{m_1} + \dots + \underline{m_k}) = m_1 + \dots + m_k \\
    \text{unfold}_\P(f) = f
    \end{cases}
    \]
    Note that in $\S\P$, $\underline{m} + \underline{m'} \neq \underline{m + m'}$.

    These two distinct objects correspond to two different ways of grouping fermionic modes:

    \[\scalebox{\scalefig}{\input{./figures/scaled-gen/id_1.tikz}} \quad \neq \quad \scalebox{\scalefig}{\begin{tikzpicture}
	\begin{pgfonlayer}{nodelayer}
		\node [style=none] (2) at (-1, 0) {};
		\node [style=none] (3) at (1, 0) {};
		\node [style=modenb] (5) at (0, 0.25) {$n+m$};
	\end{pgfonlayer}
	\begin{pgfonlayer}{edgelayer}
		\draw [style=thick] (2.center) to (3.center);
	\end{pgfonlayer}
\end{tikzpicture}
} \]
    These two objects are however isomorphic, via
    \[
    \begin{cases}
    \mathrm{id}_{m + m'} \in \S\P[\underline{m} + \underline{m'}, \underline{m + m'}] \\
    \mathrm{id}_{m + m'} \in \S\P[\underline{m + m'}, \underline{m} + \underline{m'}]
    \end{cases}
    \]
    We denote these isomorphisms $\delta_{m,m'}$ and $\gamma_{m,m'}$, and call them the divider and the gatherer. Pictorially, they are represented as:

    \[\scalebox{\scalefig}{\input{./figures/scaled-gen/div.tikz}} \quad \text{and} \quad \scalebox{\scalefig}{\input{./figures/scaled-gen/gath.tikz}} \]
    
    and are subject to equations :

    \centering {
    \phantomsection\label{eq:Exp}\scalebox{\scalefig}{\input{./figures/scaled-rule/elim_1.tikz}} $\eq\limits_{Exp}$ \scalebox{\scalefig}{\input{./figures/scaled-rule/elim_2.tikz}} \quad 
    \phantomsection\label{eq:Elim}\scalebox{\scalefig}{\input{./figures/scaled-rule/exp_1.tikz}} $\eq\limits_{Elim}$ \scalebox{\scalefig}{\begin{tikzpicture}
	\begin{pgfonlayer}{nodelayer}
		\node [style=none] (0) at (0.75, 0) {};
		\node [style=none] (3) at (-0.75, 0) {};
		\node [style=modenb] (9) at (0, 0.25) {$n + n'$};
	\end{pgfonlayer}
	\begin{pgfonlayer}{edgelayer}
		\draw [style=thick] (3.center) to (0.center);
	\end{pgfonlayer}
\end{tikzpicture}
} \\
    \scalebox{\scalefig}{\input{./figures/scaled-rule/coherence_1.tikz}} $\eq\limits_{}$ \scalebox{\scalefig}{\input{./figures/scaled-rule/coherence_2.tikz}}\quad 
    \scalebox{\scalefig}{\input{./figures/scaled-rule/cocoherence_1.tikz}} $\eq\limits_{}$ \scalebox{\scalefig}{\input{./figures/scaled-rule/cocoherence_2.tikz}} 
    \centering
    }

    \raggedright
    Note that, modulo these equations, any two diagrams of the same type composed only of dividers and gatherers are equal. 
    It turns out that adding these families of generators and equations to a pronap theory for $\P$ results in an axiomatization of $\S\P$.

    \begin{definition}[Scaled pronap theory]
    Let $\T = (\Sigma_{nat} \cup \overline{\Sigma}_{nat}, E)$ be a pronap theory. The scaled theory $S\T$ is a pair $(S\Sigma, SE)$ consisting of a set of generators $S\Sigma$ and a set of equations $SE$.
    \begin{itemize}
        \item $S\Sigma$ contains typed generators:
        \begin{itemize}
            \item $g: m \cdot \underline{1} \to m' \cdot \underline{1}$ for all $g: m \to m'$ of $\Sigma$,
            \item together with $\begin{cases} \delta_{m,m'}: \underline{m} + \underline{m'} \to \underline{m + m'} \in S\Sigma_{nat} \\ \gamma_{m,m'}: \underline{m + m'} \to \underline{m} + \underline{m'} \in S\Sigma_{nat} \end{cases}$ for all $m, m' \in \mathbb{N}$.
        \end{itemize}
        \item A formal diagram is either:
        \begin{itemize}[label=\emph{|}]
            \item $g \in S\Sigma$,
            \item $\mathrm{id}_m: \underline{m} \to \underline{m}$,
            \item $\pi_{\underline{m},\underline{m'}}: \underline{m} + \underline{m'} \to \underline{m'} + \underline{m}$,
            \item $D \circ D': \sum\limits_{i=1}^k \underline{m_i} \to \sum\limits_{j=1}^{k'} \underline{m'_j}$ for $D': \sum\limits_{i=1}^k \underline{m_i} \to \sum\limits_{i=1}^{k''} \underline{m''_i}$ and $D: \sum\limits_{i=1}^{k''} \underline{m''_i} \to \sum\limits_{j=1}^{k'} \underline{m'_j}$,
            \item $D \otimes D': \sum\limits_{i=1}^l \underline{n_i} + \sum\limits_{j=1}^k \underline{m_j} \to \sum\limits_{i=1}^{l'} \underline{n'_i} + \sum\limits_{j=1}^{k'} \underline{m'_j}$\\ for $\begin{cases} D: \underline{n_1} + \dots + \underline{n_l} \to \underline{n'_1} + \dots + \underline{n'_{l'}}\\ D': \underline{m_1} + \dots + \underline{m_k} \to \underline{m'_1} + \dots + \underline{m'_{k'}}\end{cases}$.
        \end{itemize}
        \item $SE$ contains all pairs of formal diagrams $(D, D') \in E$ typed $m \cdot \underline{1} \to m' \cdot \underline{1}$ together with equations $\text{Exp}_{m,m'} = (\delta_{m,m'} \circ \gamma_{m,m'}, \mathrm{id}_{\underline{m}} \otimes \mathrm{id}_{\underline{m'}})$ and $\text{Elim}_{m,m'} = (\gamma_{m,m'} \circ \delta_{m,m'}, \mathrm{id}_{\underline{m + m'}})$ for $m, m' \in \mathbb{N}$.
    \end{itemize}
    \end{definition}
    Given a scaled pronap theory $S\T = (S\Sigma, SE)$ there is a corresponding presymmetric monoidal category generated by $S\Sigma$ and quotiented by $SE$ and the axioms of presymmetric monoidal categories, including the naturality equations for all generators $g \in \Sigma_{nat}$, which includes dividers and gatherers.

    Note that by construction, any semantics of $\T$ in pronap $\P$ $\interp{\cdot}_\P: \T \to \P$ induces a functor $\interp{\cdot}_{\S\P}: S\T \to \S\P$ related to $\interp{\cdot}_\P$ via 
    \[
        \begin{tikzcd}
        \T \arrow[dd, "\interp{\cdot}_{\S\P}"'] \arrow[rr, "\iota_\T", bend left] &  & \S\T \arrow[dd, "\interp{\cdot}_\P"] \arrow[ll, "\text{unfold}_\T", bend left] \\
                                                                                &  &                                                                                \\
        \P \arrow[rr, "\iota_\P", bend left]                                      &  & \S\P \arrow[ll, "\text{unfold}_\P", bend left]                                
        \end{tikzcd}
    \]
    Then, due to the category equivalences, for a pronap theory $\T$, we have:
    \[
    \begin{cases}
    \interp{\cdot}_\P: \T \to \P \text{ is faithful} \\
    \interp{\cdot}_{\S\P}: S\T \to \S\P \text{ is faithful}.
    \end{cases}
    \]
    In other words, proving completeness of $S\T$ for $\S\P$ is the same as proving completeness of $\T$ for $\P$.

    The interest of a scaled pronap theory lies in the possibility of providing compact syntactic sugars 
    for otherwise large diagrams. Such syntactic sugars are defined inductively. The most basic ones 
    are the scaled generators. 

    \begin{definition}[Scaled generator]
        Let $g \in \Sigma$ be a generator of a pronap theory $\T$. The $k$-scaled generator $\S_k(g)$ is the diagram of $S\T$ defined inductively on $k$ as:
        \begin{itemize}[label=\emph{|}]
                \item \scalebox{\scalefig}{\input{./figures/scaled-gen/S1g.tikz}} \quad := \quad \scalebox{\scalefig}{\input{./figures/scaled-gen/g.tikz}},
                \item \phantomsection\label{eq:Dist}\scalebox{\scalefig}{\input{./figures/scaled-gen/Skp1g_1.tikz}} \quad $\mathop{:=}\limits_{Dist}$ \quad \scalebox{\scalefig}{\input{./figures/scaled-gen/Skp1g_2.tikz}}
        \end{itemize}
    \end{definition}

    \begin{definition}[Scaled diagram]
    Let $\T = (\Sigma_{nat} \cup \overline{\Sigma}_{nat}, E)$ be a pronap theory, 
    $D$ a diagram of $\T$, and $k>0$.
    Define $\S_{k}(D)$ as the formal diagram of $S\T$ obtained by replacing generators $g$ in the expression of $D$ by $\S_{k}(g)$.
    \end{definition}
    When a diagram $D$ only involves natural generators in its expression, then \Dist~ also applies to $\S_{k+1}(D)$:

    \begin{lemma}\label{lemma:scaled-diag-nat}
    Let $\T$ be a pronap theory, $D$ a diagram of $\T$ which only involves natural generators $g \in \Sigma_{nat}$ in its expression.  Then, we have:
    \[\forall k > 0, \scalebox{\scalefig}{\begin{tikzpicture}
	\begin{pgfonlayer}{nodelayer}
		\node [style=none] (15) at (1.5, 0.25) {$\vdots$};
		\node [style=none] (16) at (2.25, 0.75) {};
		\node [style=none] (17) at (2.25, -0.75) {};
		\node [style={1_arrow_box}] (18) at (0, 0) {$\S_{k+1}D$};
		\node [style=none] (19) at (-1.5, 0.75) {};
		\node [style=none] (20) at (-1.5, -0.75) {};
		\node [style=none] (21) at (-1.5, 0.25) {$\vdots$};
		\node [style=none] (22) at (-2.25, 0.75) {};
		\node [style=none] (23) at (-2.25, -0.75) {};
		\node [style=modenb] (24) at (-1.75, 1) {$k+1$};
		\node [style=none] (29) at (1.5, 0.75) {};
		\node [style=none] (30) at (1.5, -0.75) {};
		\node [style=modenb] (31) at (1.75, 1) {$k+1$};
		\node [style=modenb] (32) at (1.75, -0.5) {$k+1$};
		\node [style=modenb] (33) at (-1.75, -0.5) {$k+1$};
	\end{pgfonlayer}
	\begin{pgfonlayer}{edgelayer}
		\draw [style=thick] (22.center) to (19.center);
		\draw [style=thick] (23.center) to (20.center);
		\draw [style=thick, in=165, out=0, looseness=1.25] (19.center) to (18);
		\draw [style=thick, in=-165, out=0, looseness=1.25] (20.center) to (18);
		\draw [style=thick] (16.center) to (29.center);
		\draw [style=thick] (17.center) to (30.center);
		\draw [style=thick, in=15, out=180, looseness=1.25] (29.center) to (18);
		\draw [style=thick, in=-15, out=180, looseness=1.25] (30.center) to (18);
	\end{pgfonlayer}
\end{tikzpicture}
} \eq\limits_{\S\T}^{Dist} \scalebox{\scalefig}{\begin{tikzpicture}
	\begin{pgfonlayer}{nodelayer}
		\node [style=none] (15) at (2, 0.25) {$\vdots$};
		\node [style=none] (16) at (2.5, 0.75) {};
		\node [style=none] (17) at (2.5, -0.75) {};
		\node [style={1_arrow_box}] (18) at (0, -0.5) {$\S_kD$};
		\node [style=none] (19) at (-2, 0.75) {};
		\node [style=none] (20) at (-2, -0.75) {};
		\node [style=none] (21) at (-2, 0.25) {$\vdots$};
		\node [style=none] (22) at (-2.5, 0.75) {};
		\node [style=none] (23) at (-2.5, -0.75) {};
		\node [style=modenb] (24) at (-2.25, 1) {$k$};
		\node [style=modenb] (25) at (2.25, 1) {$k$};
		\node [style=modenb] (26) at (2.25, -0.5) {$k$};
		\node [style=modenb] (27) at (-2.25, -0.5) {$k$};
		\node [style=none] (29) at (2, 0.75) {};
		\node [style=none] (30) at (2, -0.75) {};
		\node [style={1_arrow_box}] (31) at (0, 1) {$D$};
		\node [style=fswap] (32) at (1.25, 0.15) {};
		\node [style=fswap] (33) at (-1.25, 0.15) {};
	\end{pgfonlayer}
	\begin{pgfonlayer}{edgelayer}
		\draw [style=thick] (22.center) to (19.center);
		\draw [style=thick] (23.center) to (20.center);
		\draw [style=thick, in=165, out=0, looseness=1.25] (19.center) to (18);
		\draw [style=thick, in=-165, out=0, looseness=1.25] (20.center) to (18);
		\draw [in=0, out=180, looseness=0.75] (30.center) to (31);
		\draw [in=15, out=180] (29.center) to (31);
		\draw [in=-180, out=0, looseness=0.75] (20.center) to (31);
		\draw [in=165, out=0] (19.center) to (31);
		\draw [style=thick] (16.center) to (29.center);
		\draw [style=thick] (17.center) to (30.center);
		\draw [style=thick, in=15, out=180, looseness=1.25] (29.center) to (18);
		\draw [style=thick, in=-15, out=180, looseness=1.25] (30.center) to (18);
	\end{pgfonlayer}
\end{tikzpicture}
}\]
    \end{lemma}
    Note that this is not the case in general if $D$ involves non-natural generators. However, 
    when no parallel composition $\otimes$ is involved in $D$, then \Dist~ applies no matter 
    the naturality of the generators involved.

    \begin{lemma}{\label{lemma:scaled-diag-line}}
    Let $\T = (\Sigma, E)$ be a pronap theory, and $D$ a diagram such that $D$ does not involve monoidal product in its expression. Then,
    \[\forall k > 0, \scalebox{\scalefig}{\begin{tikzpicture}
	\begin{pgfonlayer}{nodelayer}
		\node [style=none] (15) at (1.5, 0.25) {$\vdots$};
		\node [style=none] (16) at (2.25, 0.75) {};
		\node [style=none] (17) at (2.25, -0.75) {};
		\node [style={1_arrow_box}] (18) at (0, 0) {$\S_{k+1}D$};
		\node [style=none] (19) at (-1.5, 0.75) {};
		\node [style=none] (20) at (-1.5, -0.75) {};
		\node [style=none] (21) at (-1.5, 0.25) {$\vdots$};
		\node [style=none] (22) at (-2.25, 0.75) {};
		\node [style=none] (23) at (-2.25, -0.75) {};
		\node [style=modenb] (24) at (-1.75, 1) {$k+1$};
		\node [style=none] (29) at (1.5, 0.75) {};
		\node [style=none] (30) at (1.5, -0.75) {};
		\node [style=modenb] (31) at (1.75, 1) {$k+1$};
		\node [style=modenb] (32) at (1.75, -0.5) {$k+1$};
		\node [style=modenb] (33) at (-1.75, -0.5) {$k+1$};
	\end{pgfonlayer}
	\begin{pgfonlayer}{edgelayer}
		\draw [style=thick] (22.center) to (19.center);
		\draw [style=thick] (23.center) to (20.center);
		\draw [style=thick, in=165, out=0, looseness=1.25] (19.center) to (18);
		\draw [style=thick, in=-165, out=0, looseness=1.25] (20.center) to (18);
		\draw [style=thick] (16.center) to (29.center);
		\draw [style=thick] (17.center) to (30.center);
		\draw [style=thick, in=15, out=180, looseness=1.25] (29.center) to (18);
		\draw [style=thick, in=-15, out=180, looseness=1.25] (30.center) to (18);
	\end{pgfonlayer}
\end{tikzpicture}
} \eq\limits_{\S\T}^{Dist} \scalebox{\scalefig}{\begin{tikzpicture}
	\begin{pgfonlayer}{nodelayer}
		\node [style=none] (15) at (2, 0.25) {$\vdots$};
		\node [style=none] (16) at (2.5, 0.75) {};
		\node [style=none] (17) at (2.5, -0.75) {};
		\node [style={1_arrow_box}] (18) at (0, -0.5) {$\S_kD$};
		\node [style=none] (19) at (-2, 0.75) {};
		\node [style=none] (20) at (-2, -0.75) {};
		\node [style=none] (21) at (-2, 0.25) {$\vdots$};
		\node [style=none] (22) at (-2.5, 0.75) {};
		\node [style=none] (23) at (-2.5, -0.75) {};
		\node [style=modenb] (24) at (-2.25, 1) {$k$};
		\node [style=modenb] (25) at (2.25, 1) {$k$};
		\node [style=modenb] (26) at (2.25, -0.5) {$k$};
		\node [style=modenb] (27) at (-2.25, -0.5) {$k$};
		\node [style=none] (29) at (2, 0.75) {};
		\node [style=none] (30) at (2, -0.75) {};
		\node [style={1_arrow_box}] (31) at (0, 1) {$D$};
		\node [style=fswap] (32) at (1.25, 0.15) {};
		\node [style=fswap] (33) at (-1.25, 0.15) {};
	\end{pgfonlayer}
	\begin{pgfonlayer}{edgelayer}
		\draw [style=thick] (22.center) to (19.center);
		\draw [style=thick] (23.center) to (20.center);
		\draw [style=thick, in=165, out=0, looseness=1.25] (19.center) to (18);
		\draw [style=thick, in=-165, out=0, looseness=1.25] (20.center) to (18);
		\draw [in=0, out=180, looseness=0.75] (30.center) to (31);
		\draw [in=15, out=180] (29.center) to (31);
		\draw [in=-180, out=0, looseness=0.75] (20.center) to (31);
		\draw [in=165, out=0] (19.center) to (31);
		\draw [style=thick] (16.center) to (29.center);
		\draw [style=thick] (17.center) to (30.center);
		\draw [style=thick, in=15, out=180, looseness=1.25] (29.center) to (18);
		\draw [style=thick, in=-15, out=180, looseness=1.25] (30.center) to (18);
	\end{pgfonlayer}
\end{tikzpicture}
}\]
    \end{lemma}
    This is crucial to understand which equations of the pronap theory $\T$ can be directly applied 
    to scaled diagrams, or in other words, which equation is \emph{scalable}.

    \begin{definition}[Scalable equation]
    Let $\T = (\Sigma, E)$ be a pronap theory, and $(D, D') \in E$. 
    We say that $(D, D')$ is scalable when
    \[
    \forall k > 0, \quad \S_{k}(D) \eq\limits_{\S\T} \S_{k}(D').
    \]
    \end{definition}

    \begin{example}
        For instance, equation \A28~ is scalable. 
        \[
         \scalebox{\scalefig}{\input{./figures/SaW-lemma/op_1.tikz}} \quad \eq \quad \scalebox{\scalefig}{\input{./figures/SaW-lemma/op_2.tikz}}
        \]
        However, equation \A25~ is not scalable: we prove 
        that a phase depending on the width of the wire has to appear on the right side of the scaled equation 
        for it to be sound.
        \[
        \scalebox{\scalefig}{\input{./figures/SaW-lemma/dot-cup_1.tikz}} \quad \eq \quad \scalebox{\scalefig}{\begin{tikzpicture}
	\begin{pgfonlayer}{nodelayer}
		\node [style=none] (23) at (2.5, 0) {};
		\node [style=none] (24) at (-2.5, 0) {};
		\node [style=modenb] (25) at (-2, 0.25) {$n$};
		\node [style=white phase dot] (26) at (0, 0) {$(-1)^{n+1}$};
	\end{pgfonlayer}
	\begin{pgfonlayer}{edgelayer}
		\draw [style=thick] (24.center) to (23.center);
	\end{pgfonlayer}
\end{tikzpicture}
}
        \] 
    \end{example}
    The following theorem uses \cref{lemma:scaled-diag-nat} and \cref{lemma:scaled-diag-line} to show that
    equations involving either only natural generators or no monoidal product in their expression are always scalable.

    \begin{theorem}\label{thm:scalable-equations}
    Let $\T = (\Sigma_{nat} \cup \overline{\Sigma}_{nat}, E)$ be a pronap theory. Let $(D, D') \in E$ be an equation. Then, 
    \begin{enumerate}[label=(\roman*)]
        \item\label{item:scalable-eq-1} If all generators involved in the expressions of $D$ and $D'$ are natural, then $(D, D')$ is scalable.
        \item\label{item:scalable-eq-2} If both $D$ and $D'$ involve no monoidal product in their expressions, then $(D, D')$ is scalable.
    \end{enumerate}
    In particular, in a prop theory, all equations are scalable.
    \end{theorem}
    In the following, whenever we write a diagrammatic derivation in which we use \cref{thm:scalable-equations} for some 
    equation A\#, we will write $\S A\#$ under the equality symbol. For scalable equations to which \cref{thm:scalable-equations}
    does not apply, we will need to prove it separately as an independent lemma. 

    \begin{example}
        Equations \A2, \A13, and \A24~ are all scalable due to \cref{thm:scalable-equations}.\ref{item:scalable-eq-1}.
        Equations \A26, \A31, and \A35~ are all scalable due to \cref{thm:scalable-equations}.\ref{item:scalable-eq-2}.
        Equation \A28~ is scalable, but involves both a non-natural generator and a monoidal product in its expression.
        To use its scaled version, we prove the lemma L27.
    \end{example}
    Scalable notations are not limited to scaled generators and scalable equations: more interesting syntactic sugars can be defined, with the 
    hope of depicting as compactly as possible the structure of a pronap theory. In \cref{subsec:matrix-arrows} and \cref{subsec:graph-triangles},
    we define and study two examples of such syntactic sugars, which we use in \cref{subsec:completeness} to prove completeness for all the fragments
    in our hierarchy.

\subsection{Matrix arrows}\label{subsec:matrix-arrows}
    The first use of scalable notations beyond scalable generators and equations was the introduction of matrix arrows in scalable ZX-calculus \cite{CaretteHorsmanPerdrix2019}. A similar construction can be carried in \EaW, and more generally in any prop theory carrying a bialgebra.

    \begin{definition}[Matrix arrow]
        In the scaled prop theory \S\EaW, we define matrix arrows inductively as follows:
        \begin{itemize}[label=\emph{|}]
            \item \phantomsection\label{eq:D1}$\scalebox{\scalefig}{{\input{./figures/mat-lemma/def_1.tikz}}} \mathop{:=}\limits^{D1} \scalebox{\scalefig}{{\input{./figures/mat-lemma/def_2.tikz}}}$
            \item \phantomsection\label{eq:D2}$\scalebox{\scalefig}{{\input{./figures/mat-lemma/def_3.tikz}}} \mathop{:=}\limits^{D2} \scalebox{\scalefig}{{\input{./figures/mat-lemma/def_4.tikz}}}$
            \item \phantomsection\label{eq:D3}$\scalebox{\scalefig}{{\input{./figures/mat-lemma/def_5.tikz}}} \mathop{:=}\limits^{D3} \scalebox{\scalefig}{{\input{./figures/mat-lemma/def_6.tikz}}}$
        \end{itemize}
    \end{definition}
    Intuitively, matrix arrows corresponds to weighted bipartite graphs, for example, the gate $\mathcal{G}_2 $ translates to the matrix arrow 
    corresponding to the matrix $\frac{1}{\sqrt{2}}\pmm{1 & i \\ i & 1}$:

    \[\scalebox{\scalefig}{{\input{./figures/gate-set/G2-EZW.tikz}}} \quad\eq\quad \scalebox{\scalefig}{{\input{./figures/gate-set/G2-scalable.tikz}}}\]
     \smallskip
    Using proofs by induction, on can derive the following equations:
    \[
    \SEaWDerivedBlockNoLabel
    \]
    These equations are analogue to the ZX case, with the difference that matrices and computation are over $\mathbb{C}$ and not $\mathbb{F}_2 $. However, if the interpretation in the ZX case only encodes the truth table of the associated linear map,
    the interpretation of our matrix arrows carries information about all the minors of the matrix:

    \[\interp{\scalebox{\scalefig}{{\input{./figures/SEaW-lemma/rank-dec-mat_1.tikz}}}}_{x,y} \eq \mathrm{det}_{x,y}(M)\]
    where $\det_{x,y}(M)$ is the determinant of the matrix obtained by only keeping the rows and columns of $M$ corresponding respectively to ones in $y$ and $x$, with the convention that the determinant of a non-square matrix is zero. This allows to read all the rewriting rules involving matrix arrows as statements on minors. For example, equations \L3~ is a direct encoding of the Cauchy-Binet formula: 

     \begin{derivationchain}
    $\det_{x,y}(MM')$ 
    & $\eq$ & 
    $\interp{\scalebox{\scalefig}{{\begin{tikzpicture}
	\begin{pgfonlayer}{nodelayer}
		\node [style=none] (0) at (-2, 0) {};
		\node [style=none] (1) at (2, 0) {};
		\node [style=rmat] (2) at (0, 0) {};
		\node [style=none] (3) at (0, 0.75) {$MM'$};
	\end{pgfonlayer}
	\begin{pgfonlayer}{edgelayer}
		\draw [style=thick] (0.center) to (1.center);
	\end{pgfonlayer}
\end{tikzpicture}
}}}_{x,y}$ 
    \\[1em]
    & $\eq$ & 
    $\interp{\scalebox{\scalefig}{{\input{./figures/mat-yoga/cauchy-binet_2.tikz}}}}_{x,y}$ 
    \\[1em]
    & $\eq$ & 
    $\left(\interp{\scalebox{\scalefig}{{\begin{tikzpicture}
	\begin{pgfonlayer}{nodelayer}
		\node [style=none] (0) at (-1, 0) {};
		\node [style=none] (1) at (1, 0) {};
		\node [style=rmat] (4) at (0, 0) {};
		\node [style=none] (5) at (0, 0.75) {$M$};
	\end{pgfonlayer}
	\begin{pgfonlayer}{edgelayer}
		\draw [style=thick] (0.center) to (1.center);
	\end{pgfonlayer}
\end{tikzpicture}
}}}\circ\interp{\scalebox{\scalefig}{{\begin{tikzpicture}
	\begin{pgfonlayer}{nodelayer}
		\node [style=none] (0) at (-1, 0) {};
		\node [style=none] (1) at (1, 0) {};
		\node [style=rmat] (2) at (0, 0) {};
		\node [style=none] (3) at (0, 0.75) {$M'$};
	\end{pgfonlayer}
	\begin{pgfonlayer}{edgelayer}
		\draw [style=thick] (0.center) to (1.center);
	\end{pgfonlayer}
\end{tikzpicture}
}}}\right)_{x,y}$ 
    \\[1em]
    & $\eq$ & 
    $\sum\limits_{z}\interp{\scalebox{\scalefig}{{\begin{tikzpicture}
	\begin{pgfonlayer}{nodelayer}
		\node [style=none] (0) at (-1, 0) {};
		\node [style=none] (1) at (1, 0) {};
		\node [style=rmat] (4) at (0, 0) {};
		\node [style=none] (5) at (0, 0.75) {$M$};
	\end{pgfonlayer}
	\begin{pgfonlayer}{edgelayer}
		\draw [style=thick] (0.center) to (1.center);
	\end{pgfonlayer}
\end{tikzpicture}
}}}_{x,z}\interp{\scalebox{\scalefig}{{\begin{tikzpicture}
	\begin{pgfonlayer}{nodelayer}
		\node [style=none] (0) at (-1, 0) {};
		\node [style=none] (1) at (1, 0) {};
		\node [style=rmat] (2) at (0, 0) {};
		\node [style=none] (3) at (0, 0.75) {$M'$};
	\end{pgfonlayer}
	\begin{pgfonlayer}{edgelayer}
		\draw [style=thick] (0.center) to (1.center);
	\end{pgfonlayer}
\end{tikzpicture}
}}}_{z,y}$ 
    \\[1em]
    & $\eq$ & 
    $\sum\limits_{z}\det_{x,z}(M)\det_{z,y}(M')$
    \end{derivationchain}
    If matrix arrows are defined in \S\EaW they can be straightforwardly imported in its extensions \S\EpW, \S\aW and \S\FpW, notice that matrix arrows are even. For example in \S\EpW and \S\FpW, matrix arrows interact with cups and caps as follows:
    \[
    \scalebox{\scalefig}{{\input{./figures/SEpW-lemma/mat-cup_1.tikz}}}\quad\eq\quad\scalebox{\scalefig}{{\input{./figures/SEpW-lemma/mat-cup_2.tikz}}}
    \]
    This is to be compared with the ZX case where flipping a matrix arrows doesn't gives a matrix arrows, here we see how ZW matrix arrows are in as sense more symmetric. \\
    In \S\aW and \S\FpW, the interaction with odd scaled generators is particularly interesting, in particular we can derive the cornerstone rule of the interaction between matrix arrows and determinants:\\[1em]  
    \[
    \scalebox{\scalefig}{{\input{./figures/SaW-lemma/det-mat_1.tikz}}}\quad\eq\quad\scalebox{\scalefig}{{\input{./figures/SaW-lemma/det-mat_2.tikz}}}
    \]
    As a direct consequence, we have that $\interp{\scalebox{\scalefig}{{\input{./figures/mat-yoga/float-det.tikz}}}}=\det M$. 

    This should be linked to the combinatorial interpretation of planar ZW diagrams in \cite{CaretteEtAl2023}:
    instead of planar graphs, we rather have bipartite graphs whose crossing are born by the fermionic swap, and the diagram exactly counts 
    the weighted perfect matchings of this bipartite graphs, although it multiplies each perfect matching by $(-1)^s$ 
    where $s$ is the number of crossings involved in the particular matching. This is the reason why we get the determinant instead 
    of the permanent.
    This allows for diagrammatic analogue of well-known identities involving determinants. The derivable equation \L32~ encodes the cofactor matrix formula for the inverse.

    Diagrammatically:

    \begin{tabular}{ccc}
    \scalebox{\scalefig}{\input{./figures/mat-yoga/invert-comat_2.tikz}}
    &
    $\eq\limits_{\L8, \L3}$
    &
    \scalebox{\scalefig}{\input{./figures/mat-yoga/invert-comat_3.tikz}}
    \\[1em]
    &
    $\eq\limits_{\substack{\D2,\D3,\D1,\\\S\A5, \S\A35}}$
    &
    \scalebox{\scalefig}{\input{./figures/mat-yoga/invert-comat_4.tikz}}
    \\[1em]
    &
    $\eq\limits_{\substack{\L32, \A26}}$
    &
    \scalebox{\scalefig}{\input{./figures/mat-yoga/invert-comat_5.tikz}}
    \\[1em]
    &
    $\eq\limits_{\substack{\A26, \A9}}$
    &
    \scalebox{\scalefig}{\input{./figures/mat-yoga/invert-comat_6.tikz}}
    \end{tabular}\\
    Therefore:\\
    {\centering
    \begin{tabular}{c c l}
    $P^{-1}_{i,j}$ 
    &     
    $\eq\limits_{\substack{}}$
    & 
    $\interp{\scalebox{\scalefig}{{\begin{tikzpicture}
	\begin{pgfonlayer}{nodelayer}
		\node [style=none] (0) at (-1, 0) {};
		\node [style=none] (1) at (1, 0) {};
		\node [style=rmat] (4) at (0, 0) {};
		\node [style=none] (5) at (0.25, 0.75) {$P^{-1}$};
	\end{pgfonlayer}
	\begin{pgfonlayer}{edgelayer}
		\draw [style=thick] (0.center) to (1.center);
	\end{pgfonlayer}
\end{tikzpicture}
}}}_{2^i , 2^j }$ 
    \\[1.2em]
    &     
    $\eq\limits_{\substack{}}$
    & 
    $\frac{(-1)^{i+j}}{detP}\interp{\scalebox{\scalefig}{{\input{./figures/mat-yoga/invert-comat_7.tikz}}}}_{2^{n+1}-2^i , 2^{n+1}-2^j }$ 
    \\[1.2em]
     & 
    $\eq\limits_{\substack{}}$
    & 
    $\frac{(-1)^{i+j}}{\det P}\det P^{(j,i)}$
    \\[1.2em]
    \end{tabular}\\}
    \raggedright
    where \(P^{(j,i)}\) is obtained from \(P\) by deleting row \(j\) and column \(i\).\\
    The well-known diagrammatic version of trace cyclicity here corresponds to Sylvester's identity:

    \begin{tabular}{ccc}
    \scalebox{\scalefig}{\begin{tikzpicture}
	\begin{pgfonlayer}{nodelayer}
		\node [style=dot] (1) at (2, 0) {};
		\node [style=rmat] (7) at (0, 0) {};
		\node [style=none] (9) at (0, 0.75) {$I_n + MM'$};
		\node [style=modenb] (11) at (1.5, 0.25) {$n$};
		\node [style=dot] (12) at (-2, 0) {};
	\end{pgfonlayer}
	\begin{pgfonlayer}{edgelayer}
		\draw [style=thick] (12) to (1);
	\end{pgfonlayer}
\end{tikzpicture}
}
    &
    $\eq\limits_{\substack{\L1,\L3}}$
    &
    \scalebox{\scalefig}{\input{./figures/mat-yoga/sylvester_2.tikz}}
    \\[1em]
    &
    $\eq\limits_{\substack{\L8, \S\A36}}$
    &
    \scalebox{\scalefig}{\input{./figures/mat-yoga/sylvester_3.tikz}}
    \\[1em]
    &
    $\eq\limits_{\substack{\L34}}$
    &
    \scalebox{\scalefig}{\input{./figures/mat-yoga/sylvester_4.tikz}}
    \\[1em]
    &
    $\eq\limits_{\substack{\S\A6}}$
    &
    \scalebox{\scalefig}{\input{./figures/mat-yoga/sylvester_5.tikz}}
    \\[1em]
    &
    $\eq\limits_{\substack{\S\A35, \L12, \L3}}$
    &
    \scalebox{\scalefig}{\input{./figures/mat-yoga/sylvester_6.tikz}}
    \\[1em]
    &
    $\eq\limits_{\substack{\L12}}$
    &
    \scalebox{\scalefig}{\input{./figures/mat-yoga/sylvester_7.tikz}}
    \\[1em]
    &
    $\eq\limits_{\substack{\L12}}$
    &
    \scalebox{\scalefig}{\input{./figures/mat-yoga/sylvester_8.tikz}}
    \\[2em]
    &
    $\eq\limits_{\substack{\L3, \L12, \S\A35\\ \S\A6, \L34,\L8\\ \S\A36, \L3, \L1}}$
    &
    \scalebox{\scalefig}{\begin{tikzpicture}
	\begin{pgfonlayer}{nodelayer}
		\node [style=dot] (1) at (2, 0) {};
		\node [style=rmat] (7) at (0, 0) {};
		\node [style=none] (9) at (0, 0.75) {$I_m + M'M$};
		\node [style=modenb] (11) at (1.5, 0.25) {$m$};
		\node [style=dot] (12) at (-2, 0) {};
	\end{pgfonlayer}
	\begin{pgfonlayer}{edgelayer}
		\draw [style=thick] (12) to (1);
	\end{pgfonlayer}
\end{tikzpicture}
}
    \end{tabular}\\
    Which yields 
    \[\det(I_n + MM') = \det(I_m + M'M)\]
    Finally, one can quickly derive the diagrammatic version of the Schur complement formula while computing a block determinant.

    Let $M_1 \in \C^{k\times k}$, $M_2 \in \C^{k\times n}$, $M_3 \in \C^{n\times k}$, $M_4 \in \C^{n\times n}$ with $M_4$ invertible.

    Then:

    \begin{tabular}{ccc}
    \scalebox{\scalefig}{\input{./figures/mat-yoga/shur_1.tikz}}
    &
    $\eq\limits_{\substack{\L8,\L3,\L2}}$
    &
    \scalebox{\scalefig}{\input{./figures/mat-yoga/shur_2.tikz}}
    \\[1em]
    &
    $\eq\limits_{\substack{\L24}}$
    &
    \scalebox{\scalefig}{\input{./figures/mat-yoga/shur_3.tikz}}
    \\[1em]
    &
    $\eq\limits_{\substack{\L3,\S\A3}}$
    &
    \scalebox{\scalefig}{\input{./figures/mat-yoga/shur_4.tikz}}
    \\[1em]
    \end{tabular}

    \begin{tabular}{c c c}
    \hspace*{5em}
    &
    $\eq\limits_{\substack{\L25}}$
    &
    \scalebox{\scalefig}{\input{./figures/mat-yoga/shur_5.tikz}}
    \\[1em]
    &
    $\eq\limits_{\substack{\L26}}$
    &
    \scalebox{\scalefig}{\input{./figures/mat-yoga/shur_6.tikz}}
    \\[1em]
    &
    $\eq\limits_{\substack{L7, \L3}}$
    &
    \scalebox{\scalefig}{\input{./figures/mat-yoga/shur_7.tikz}}
    \\[1em]
    &
    $\eq\limits_{\substack{\L1}}$
    &
    \scalebox{\scalefig}{\input{./figures/mat-yoga/shur_8.tikz}}
    \end{tabular}
    \medskip

    This yields $\det \pmm{M_1 & M_2 \\ M_3 & M_4}  = \det M_4 \det(M_1 - M_2M_4^{-1}M_3) $.

    More derived equations involving matrix arrows can be found in \cref{fig:derived-equations}, and are used extensively in the completeness proofs of \cref{subsec:completeness}.

    \subsection{Graph triangles}\label{subsec:graph-triangles}
    The last scalable syntactic sugar we introduce already appeared in analogue form in ZX-calculus to represent graph states \cite{CaretteHorsmanPerdrix2019}, but mainly in symplectic graphical algebra \cite{BoothCaretteComfort2026}. 
    As with matrix arrows we take inspiration from the syntactic similarities while obtaining slightly different behaviors in the end.

    \begin{definition}[Graph triangle]
        Let $Z \in \C^{n \times n}$ be an antisymmetric matrix.
        We define the graph triangle as follows:
        \[\phantomsection\label{eq:D4}\scalebox{\scalefig}{\begin{tikzpicture}
	\begin{pgfonlayer}{nodelayer}
		\node [style=rw] (0) at (-0.5, 0) {};
		\node [style=none] (1) at (1, 0) {};
		\node [style=modenb] (2) at (0.75, 0.25) {$n$};
		\node [style=none] (3) at (-0.5, 0.75) {$Z$};
	\end{pgfonlayer}
	\begin{pgfonlayer}{edgelayer}
		\draw [style=thick] (1.center) to (0);
	\end{pgfonlayer}
\end{tikzpicture}
} \mathop{:=}\limits^{D4} \scalebox{\scalefig}{\begin{tikzpicture}
	\begin{pgfonlayer}{nodelayer}
		\node [style=none] (1) at (1, 0) {};
		\node [style=modenb] (2) at (0.75, 0.25) {$n$};
		\node [style=none] (3) at (-1.25, 1.25) {$\frac{Z}{2}$};
		\node [style=rw] (4) at (-0.5, 0) {};
		\node [style=none] (5) at (-1.25, 0.5) {};
		\node [style=none] (6) at (-1.25, -0.5) {};
		\node [style=none] (7) at (-1.75, 0.5) {};
		\node [style=none] (8) at (-1.75, -0.5) {};
		\node [style=rmat] (9) at (-1.5, 0.5) {};
	\end{pgfonlayer}
	\begin{pgfonlayer}{edgelayer}
		\draw [style=thick, in=135, out=0] (5.center) to (4);
		\draw [style=thick, in=-135, out=0] (6.center) to (4);
		\draw [style=thick] (1.center) to (4);
		\draw [style=thick] (6.center) to (8.center);
		\draw [style=thick] (5.center) to (7.center);
		\draw [style=thick, bend right=270, looseness=1.50] (8.center) to (7.center);
	\end{pgfonlayer}
\end{tikzpicture}
}\]
         \[\phantomsection\label{eq:D5}\scalebox{\scalefig}{\begin{tikzpicture}
	\begin{pgfonlayer}{nodelayer}
		\node [style=none] (1) at (-0.5, 0) {};
		\node [style=modenb] (2) at (-0.25, 0.25) {$n$};
		\node [style=none] (3) at (1, 0.75) {$Z$};
		\node [style=lw] (4) at (1, 0) {};
	\end{pgfonlayer}
	\begin{pgfonlayer}{edgelayer}
		\draw [style=thick] (1.center) to (4);
	\end{pgfonlayer}
\end{tikzpicture}
} \mathop{:=}\limits^{D5} \scalebox{\scalefig}{\begin{tikzpicture}
	\begin{pgfonlayer}{nodelayer}
		\node [style=none] (1) at (-1.75, 0) {};
		\node [style=modenb] (2) at (-1.5, 0.25) {$n$};
		\node [style=none] (3) at (0.75, 1.25) {$\frac{Z}{2}$};
		\node [style=none] (5) at (0.5, 0.5) {};
		\node [style=none] (6) at (0.5, -0.5) {};
		\node [style=none] (7) at (1, 0.5) {};
		\node [style=none] (8) at (1, -0.5) {};
		\node [style=rmat] (9) at (0.75, 0.5) {};
		\node [style=lw] (10) at (-0.25, 0) {};
	\end{pgfonlayer}
	\begin{pgfonlayer}{edgelayer}
		\draw [style=thick] (6.center) to (8.center);
		\draw [style=thick] (5.center) to (7.center);
		\draw [style=thick, bend left=270, looseness=1.50] (8.center) to (7.center);
		\draw [style=thick, in=45, out=180] (5.center) to (10);
		\draw [style=thick, in=-45, out=180] (6.center) to (10);
		\draw [style=thick] (1.center) to (10);
	\end{pgfonlayer}
\end{tikzpicture}
}\]
    \end{definition}
    From this definition we can obtain an alternative inductive one, revealing the underlying stucture:
    \[
    \scalebox{\scalefig}{{\begin{tikzpicture}
	\begin{pgfonlayer}{nodelayer}
		\node [style=rw] (0) at (-0.5, 0) {};
		\node [style=none] (2) at (1, 0) {};
		\node [style=none] (3) at (-0.5, 1) {$\pmm{Z & C \\ -C^T & Z'}$};
		\node [style=modenb] (4) at (0.5, 0.25) {$n+n'$};
	\end{pgfonlayer}
	\begin{pgfonlayer}{edgelayer}
		\draw [style=thick] (2.center) to (0);
	\end{pgfonlayer}
\end{tikzpicture}
}}\quad\eq\quad\scalebox{\scalefig}{{\input{./figures/SEpW-lemma/gnode-dec_2.tikz}}}
    \]
    We see from this definition that graph triangles can be seen as weighted graphs as in the following example:
    \[
        \scalebox{\scalefig}{{\input{./figures/graph-yoga/example_1.tikz}}}\quad\eq\quad\scalebox{\scalefig}{{\input{./figures/graph-yoga/example_2.tikz}}}
    \]
    Notice that all crossings involve fermionic swaps, this provides a striking connection not to the determinant but to the Pfaffian this time.
    Indeed:

    \[\interp{\scalebox{\scalefig}{\begin{tikzpicture}
	\begin{pgfonlayer}{nodelayer}
		\node [style=rw] (0) at (-0.5, 0) {};
		\node [style=none] (1) at (1, 0) {};
		\node [style=modenb] (2) at (0.75, 0.25) {$n$};
		\node [style=none] (3) at (-0.5, 0.75) {$Z$};
	\end{pgfonlayer}
	\begin{pgfonlayer}{edgelayer}
		\draw [style=thick] (1.center) to (0);
	\end{pgfonlayer}
\end{tikzpicture}
}}_{x}= \Pf_{x} (Z)\]
    where $\Pf_{x}(Z)$ is the pfaffian of the antisymmetric matrix obtained by only keeping the rows and columns of $Z$ corresponding  to ones in  $x$, with the convention that the pfaffian of a matrix of odd size is zero. 
    We can derive the following properties:

        \begin{tabular}{ccc}
        \boxeq{L10}
            {SEpW-lemma/gnode-add_1}
            {SEpW-lemma/gnode-add_2}
        &
        \boxeq{L20}    
        {SEpW-lemma/gnode-zero_1}
        {SEpW-lemma/gnode-zero_2}
        &
        \boxeq{L11}
            {SEpW-lemma/gnode-erase_1}
            {SEpW-lemma/gnode-erase_2}
        \\[1em]
        \end{tabular}
    When interacting with cups and caps, we get:
    \[
    \boxeq{L22}
    {SEpW-lemma/triangle-cup_1}
    {SEpW-lemma/triangle-cup_2}
    \]
    Graph triangles interact with matrix arrows as:
    \[
    \boxeq{L14}
    {SEpW-lemma/gnode-mat_1}
    {SEpW-lemma/gnode-mat_2}
    \] 
    Which, for $M \leftarrow P\in \C^{n\times n}$ invertible yields the pfaffian identity: 

    \[\Pf(PZP^T)=\det (P) \Pf (Z)\]
    Again, graph triangles are defined in \S\EpW but also available in its extension \S\FpW. Graph triangles are even, and therefore natural.

    Surprisingly, for $A \in \C^{2n\times 2n}$ an invertible antisymmetric matrix, the interaction of the associated graph triangle 
    with the non-natural black node is pictured as:

    \[
    \boxeq{L35}
    {SFpW-lemma/gnode-invert_1}
    {SFpW-lemma/gnode-invert_2}
    \]

    Which encodes the formula for the pfaffian of the inverse. Finally, we can again obtain a Schur complement formula, but for the Pfaffian this time:

    \begin{tabular}{ccc}
    \scalebox{\scalefig}{\begin{tikzpicture}
	\begin{pgfonlayer}{nodelayer}
		\node [style=none] (0) at (2, 0) {};
		\node [style=rw] (1) at (-1, 0) {};
		\node [style=none] (2) at (-1.5, 1.5) {$\pmm{A & C \\ -C^T & Z }$};
		\node [style=dot] (3) at (2, 0) {};
	\end{pgfonlayer}
	\begin{pgfonlayer}{edgelayer}
		\draw [style=thick] (1) to (0.center);
	\end{pgfonlayer}
\end{tikzpicture}
}
    &
    $\eq\limits_{\substack{\L13}}$
    &
    \scalebox{\scalefig}{\input{./figures/graph-yoga/shur_2.tikz}}
    \\[1em]
    &
    $\eq\limits_{\substack{\S\A35}}$
    &
    \scalebox{\scalefig}{\input{./figures/graph-yoga/shur_3.tikz}}
    \\[1em]
    &
    $\eq\limits_{\substack{\L35}}$
    &
    \scalebox{\scalefig}{\input{./figures/graph-yoga/shur_4.tikz}}
    \\[1em]
    &
    $\eq\limits_{\substack{Dist, \S\A36, \L34}}$
    &
    \scalebox{\scalefig}{\input{./figures/graph-yoga/shur_5.tikz}}
    \\[1em]
    &
    $\eq\limits_{\substack{\S\A35, \S\A6, \S\A18, \L3}}$
    &
    \scalebox{\scalefig}{\input{./figures/graph-yoga/shur_6.tikz}}
    \\[1em]
    &
    $\eq\limits_{\substack{\L14, \L10}}$
    &
    \scalebox{\scalefig}{\input{./figures/graph-yoga/shur_7.tikz}}
    \\[1em]
    \end{tabular}\\
    which yields
    $\Pf\pmm{A & C \\ -C^T & Z} = \Pf(A)\Pf(C^TA^{-1}C + Z)$.
    More derived equations involving graph triangles can be found in \cref{fig:derived-equations}.
%
%

\begin{figure*}[t]
\centering

\begingroup
\sensitivetikztrue

\newsavebox{\SEaWDerivedSavebox}
\newsavebox{\SEpWDerivedSavebox}
\newsavebox{\SaWDerivedSavebox}
\newsavebox{\SFpWDerivedSavebox}

\sbox{\SEaWDerivedSavebox}{\SEaWDerivedBlock}
\sbox{\SEpWDerivedSavebox}{\SEpWDerivedBlock}
\sbox{\SaWDerivedSavebox}{\SaWDerivedBlock}
\sbox{\SFpWDerivedSavebox}{\SFpWDerivedBlock}


\resizebox{\textwidth}{!}{%
\begin{tikzpicture}[
    reset cm,
    axiom block/.style={
        inner sep=0pt,
        outer sep=0pt,
        draw=none
    },
    fragment region/.style={
        rounded corners=10pt,
        dashed,
        line width=0.9pt
    },
    fragment title/.style={
        font=\Large\bfseries,
        fill=white,
        inner xsep=4pt,
        inner ysep=1.5pt,
        draw=none
    }
]


\node[
    axiom block,
    anchor=north west
] (eaw) at (0,0) {\usebox{\SEaWDerivedSavebox}};

\node[
    axiom block,
    anchor=north west
] (epw) at ([xshift=8mm]eaw.north east)
    {\usebox{\SEpWDerivedSavebox}};

\node[
    axiom block,
    anchor=north west
] (aw) at ([yshift=-45mm]eaw.south west)
    {\usebox{\SaWDerivedSavebox}};

\node[
    axiom block,
    anchor=north west
] (fpw) at ([yshift=-9mm]epw.south west)
    {\usebox{\SFpWDerivedSavebox}};


\begin{pgfonlayer}{background}

\node[
    fragment region,
    draw=black!65,
    fill=black!4,
    fit=(eaw)(epw)(aw)(fpw),
    inner sep=11mm
] (fpw-region) {};

\node[
    fragment region,
    draw=blue!65!black,
    fill=blue!8,
    fill opacity=.55,
    fit=(eaw)(aw),
    inner sep=6mm
] (aw-region) {};

\node[
    fragment region,
    draw=red!65!black,
    fill=red!8,
    fill opacity=.55,
    fit=(eaw)(epw),
    inner sep=6mm
] (epw-region) {};

\node[
    fragment region,
    draw=green!50!black,
    fill=green!10,
    fill opacity=.75,
    fit=(eaw),
    inner sep=3mm
] (eaw-region) {};

\end{pgfonlayer}


\node[
    fragment title,
    anchor=west
] at ([xshift=5mm]fpw-region.north west)
    {$\FpW$};

\node[
    fragment title,
    anchor=west,
    text=blue!65!black
] at ([xshift=5mm]aw-region.south west)
    {$\aW$};

\node[
    fragment title,
    anchor=west,
    text=red!65!black
] at ([xshift=5mm]epw-region.north west)
    {$\EpW$};

\node[
    fragment title,
    anchor=west,
    text=green!45!black
] at ([xshift=3mm]eaw-region.north west)
    {$\EaW$};

\end{tikzpicture}%
}

\endgroup

\caption{
Derived equations used throughout the completeness proofs.
Each equation $L\#$ is displayed in the smallest scaled fragment
in which it is derived. By convention, equations of \S\EaW~ and \S\aW~ are also valid when mirrored
with respect to the vertical axis.
For readability, we use the convention that $M$ denotes an arbitrary matrix, 
$P$ and $Q$ invertible matrices, $Z$ a square antisymmetric matrix, 
and $A$ an invertible antisymmetric matrix.
}
\label{fig:derived-equations}

\Description{
The derived equations are partitioned into four regions:
those holding already in scaled EaW, those introduced in scaled EpW,
those introduced in scaled aW, and those introduced in scaled FpW.
}

\end{figure*}
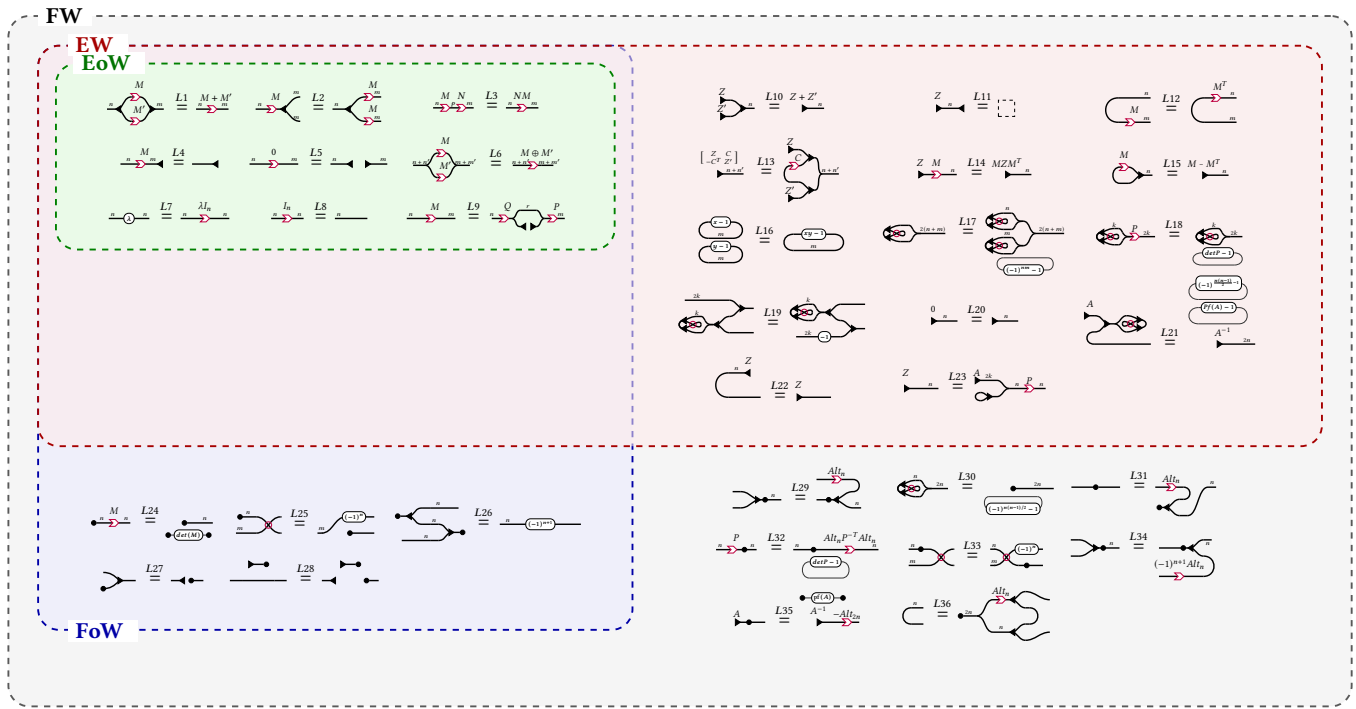

    We now have all the necessary syntactic sugar and rules to properly state our completeness proofs.

\subsection{Completeness}\label{subsec:completeness}
The completeness results for the different fragments all follow the same general strategy: we identify \emph{normal-form diagrams} and show that every diagram can be rewritten into a unique such form. More precisely, for any two diagrams $D$ and $D'$ with the same type, we construct rewriting procedures
\[
D \rewrites^{*} N(D)
\qquad\text{and}\qquad
D' \rewrites^{*} N(D'),
\]

where $N(D)$ and $N(D')$ belong to a fixed syntactic class of normal forms. The crucial point is then that the interpretation of a normal form uniquely determines its syntax. Hence, whenever
\[
\interp{D}=\interp{D'},
\]

we have $N(D)=N(D')$, and therefore
\[
D = D'
\]

in the corresponding equational theory. 

Reduction to normal form thus gives both a canonical representative for each denotation and, as a consequence, completeness of the theory.

For $\S\EaW$, it is convenient to first describe normal forms for morphisms
\[
\underline{n}\longrightarrow\underline{m}.
\]

Using dividers and gatherers, any diagram
\[
D:\underline{n_1}+\cdots+\underline{n_k}\longrightarrow
\underline{m_1}+\cdots+\underline{m_l}
\]

can be rewritten as
\[
\delta_{\underline{m_1+\cdots+m_l},\underline{m_1}+\cdots+\underline{m_l}}\circ D' \circ \gamma_{\underline{n_1}+\cdots+\underline{n_k},\underline{n_1+\cdots+n_k}}
\]

where
\[
D':\underline{n_1+\cdots+n_k}\longrightarrow
\underline{m_1+\cdots+m_l}.
\]

Thus, a normal form for morphisms $\underline{n}\to\underline{m}$ transfers directly to arbitrary diagrams by first gathering the inputs, applying the normal form to the resulting diagram $D'$, and finally dividing the outputs.

\begin{theorem}[Completeness of $\S\EaW$]
For any diagram $D:\underline{n}\to\underline{m}$ of $\S\EaW$, there exists a unique matrix $M \in \C^{m \times n}$ such that $D$ can 
be rewritten as :
\[
\scalebox{\scalefig}{\input{./figures/nf-theorem/EaW.tikz}}
\]
Furthermore, two diagrams reduce to the same normal form if and only if they have the same semantics.
\end{theorem}

For the remaining fragments, it is more convenient to reduce the normal-form problem for arbitrary morphisms to the corresponding problem for states. 

This reduction is based on a form of \emph{map--state duality}, although the precise mechanism differs between the fragments. 

In the compact-closed setting, inputs can be bent into outputs, reducing the study of morphisms to that of states.

In $\S\aW$, this mechanism is not expressed by the usual compact-closed structure of cups $\scalebox{\scalefig}{\begin{tikzpicture}
	\begin{pgfonlayer}{nodelayer}
		\node [style=none] (0) at (0.5, 0.5) {};
		\node [style=none] (1) at (0, 0.5) {};
		\node [style=none] (2) at (0, -0.5) {};
		\node [style=none] (3) at (0.5, -0.5) {};
	\end{pgfonlayer}
	\begin{pgfonlayer}{edgelayer}
		\draw (3.center) to (2.center);
		\draw (0.center) to (1.center);
		\draw [bend left=270, looseness=2.50] (1.center) to (2.center);
	\end{pgfonlayer}
\end{tikzpicture}
}$ and 
caps $\scalebox{\scalefig}{\begin{tikzpicture}
	\begin{pgfonlayer}{nodelayer}
		\node [style=none] (0) at (-0.25, 0.5) {};
		\node [style=none] (1) at (0.25, 0.5) {};
		\node [style=none] (2) at (0.25, -0.5) {};
		\node [style=none] (3) at (-0.25, -0.5) {};
	\end{pgfonlayer}
	\begin{pgfonlayer}{edgelayer}
		\draw (3.center) to (2.center);
		\draw (0.center) to (1.center);
		\draw [bend right=270, looseness=2.50] (1.center) to (2.center);
	\end{pgfonlayer}
\end{tikzpicture}
}$, which are out of the scope of the fragment. 

Their role is instead played by: 
\[
\scalebox{\scalefig}{\begin{tikzpicture}
	\begin{pgfonlayer}{nodelayer}
		\node [style=dot] (0) at (-2, 0) {};
		\node [style=lw] (1) at (-1, 0) {};
		\node [style=none] (2) at (0.5, 1) {};
		\node [style=lw] (3) at (-1, 0) {};
		\node [style=none] (4) at (0.5, -1) {};
		\node [style=modenb] (5) at (-1.5, 0.25) {$n$};
	\end{pgfonlayer}
	\begin{pgfonlayer}{edgelayer}
		\draw [style=thick] (0) to (1);
		\draw [style=thick, in=45, out=-180] (2.center) to (1);
		\draw [style=thick, in=-45, out=180] (4.center) to (3);
	\end{pgfonlayer}
\end{tikzpicture}
} \qquad \scalebox{\scalefig}{\begin{tikzpicture}
	\begin{pgfonlayer}{nodelayer}
		\node [style=dot] (0) at (2.75, 0) {};
		\node [style=none] (2) at (0.25, 1) {};
		\node [style=none] (4) at (0.25, -1) {};
		\node [style=rw] (5) at (1.75, 0) {};
		\node [style=none] (6) at (-3, 1) {};
		\node [style=none] (7) at (-3, -1) {};
		\node [style=modenb] (8) at (2.25, 0.25) {$n$};
		\node [style=white phase dot] (9) at (-1.25, 1) {$(-1)^{n+1}$};
	\end{pgfonlayer}
	\begin{pgfonlayer}{edgelayer}
		\draw [style=thick] (0) to (5);
		\draw [style=thick, in=135, out=0] (2.center) to (5);
		\draw [style=thick, in=-135, out=0] (4.center) to (5);
		\draw [style=thick] (6.center) to (2.center);
		\draw [style=thick] (4.center) to (7.center);
	\end{pgfonlayer}
\end{tikzpicture}
} 
\]

Then, using \L26, we have:
\[
\scalebox{\scalefig}{\begin{tikzpicture}
	\begin{pgfonlayer}{nodelayer}
		\node [style=none] (0) at (-1.5, 0) {};
		\node [style=none] (1) at (1.5, 0) {};
		\node [style={1_arrow_box}] (2) at (0, 0) {$D$};
		\node [style=modenb] (3) at (-1, 0.25) {$n$};
		\node [style=modenb] (4) at (1, 0.25) {$m$};
	\end{pgfonlayer}
	\begin{pgfonlayer}{edgelayer}
		\draw [style=thick] (0.center) to (1.center);
	\end{pgfonlayer}
\end{tikzpicture}
} 
\eq\limits_{\L26, \Exp}
\scalebox{\scalefig}{\begin{tikzpicture}
	\begin{pgfonlayer}{nodelayer}
		\node [style=none] (0) at (-1.5, -2.75) {};
		\node [style=none] (1) at (1.5, -2.75) {};
		\node [style={1_arrow_box}] (2) at (0, -2.75) {$D$};
		\node [style=modenb] (3) at (-1, -2.5) {$n$};
		\node [style=modenb] (4) at (1, -2.5) {$m$};
		\node [style=dot] (5) at (-3.25, -1.75) {};
		\node [style=lw] (6) at (-2.25, -1.75) {};
		\node [style=none] (7) at (-1.5, -0.75) {};
		\node [style=lw] (8) at (-2.25, -1.75) {};
		\node [style=none] (9) at (-1.5, -2.75) {};
		\node [style=modenb] (10) at (-2.75, -1.5) {$n$};
		\node [style=none] (11) at (2.5, -1.75) {};
		\node [style=none] (12) at (1.5, -0.75) {};
		\node [style=none] (13) at (3, -1.75) {};
		\node [style=none] (14) at (4.25, -2.75) {};
		\node [style=none] (15) at (3.5, -1.75) {};
		\node [style=none] (17) at (3, -1.75) {};
		\node [style=dot] (18) at (6, 0.25) {};
		\node [style=none] (19) at (4.25, 1.25) {};
		\node [style=none] (20) at (4.25, -0.75) {};
		\node [style=rw] (21) at (5, 0.25) {};
		\node [style=none] (22) at (-0.5, 1.25) {};
		\node [style=none] (23) at (4.25, -0.75) {};
		\node [style=modenb] (24) at (5.5, 0.5) {$n$};
		\node [style=white phase dot] (25) at (2.75, 1.25) {$(-1)^{n+1}$};
		\node [style=none] (26) at (-2, 0) {};
		\node [style=none] (27) at (-3.25, 0) {};
		\node [style=modenb] (28) at (3, -1.5) {$n+m$};
		\node [style=none] (29) at (4.25, -2.75) {};
		\node [style=none] (30) at (6, -2.75) {};
		\node [style=modenb] (31) at (5.25, -2.5) {$m$};
	\end{pgfonlayer}
	\begin{pgfonlayer}{edgelayer}
		\draw [style=thick] (0.center) to (1.center);
		\draw [style=thick] (5) to (6);
		\draw [style=thick, in=45, out=-180] (7.center) to (6);
		\draw [style=thick, in=-45, out=180] (9.center) to (8);
		\draw [style=thick, in=-180, out=0] (1.center) to (11.center);
		\draw [style=thick] (12.center) to (7.center);
		\draw [style=thick, in=180, out=0] (12.center) to (11.center);
		\draw [style=thick] (13.center) to (11.center);
		\draw [style=thick, in=0, out=180] (14.center) to (15.center);
		\draw [style=thick] (17.center) to (15.center);
		\draw [style=thick] (18) to (21);
		\draw [style=thick, in=135, out=0] (19.center) to (21);
		\draw [style=thick, in=-135, out=0] (20.center) to (21);
		\draw [style=thick] (22.center) to (19.center);
		\draw [style=thick] (20.center) to (23.center);
		\draw [style=thick, in=0, out=180] (23.center) to (15.center);
		\draw [style=thick, in=-180, out=0] (26.center) to (22.center);
		\draw [style=thick] (27.center) to (26.center);
		\draw [style=thick] (30.center) to (29.center);
	\end{pgfonlayer}
\end{tikzpicture}
}
\eq\limits{}
\scalebox{\scalefig}{\begin{tikzpicture}
	\begin{pgfonlayer}{nodelayer}
		\node [style=none] (11) at (-3.5, -1.75) {};
		\node [style=none] (13) at (-1.5, -1.75) {};
		\node [style=none] (14) at (0, -2.75) {};
		\node [style=none] (15) at (-1, -1.75) {};
		\node [style=none] (17) at (-1.5, -1.75) {};
		\node [style=dot] (18) at (2.25, 0.25) {};
		\node [style=none] (19) at (0.5, 1.25) {};
		\node [style=none] (20) at (0.5, -0.75) {};
		\node [style=rw] (21) at (1.25, 0.25) {};
		\node [style=none] (22) at (-1.5, 1.25) {};
		\node [style=none] (23) at (0, -0.75) {};
		\node [style=modenb] (24) at (1.75, 0.5) {$n$};
		\node [style=white phase dot] (25) at (-0.5, 1.25) {$(-1)^{n+1}$};
		\node [style=none] (26) at (-3, 0) {};
		\node [style=none] (27) at (-4.25, 0) {};
		\node [style=modenb] (28) at (-1.75, -1.5) {$n+m$};
		\node [style={2_arrow_box}] (29) at (-3.5, -1.75) {$D'$};
		\node [style=none] (30) at (2.25, -2.75) {};
		\node [style=modenb] (31) at (1.5, -2.5) {$m$};
	\end{pgfonlayer}
	\begin{pgfonlayer}{edgelayer}
		\draw [style=thick] (13.center) to (11.center);
		\draw [style=thick, in=0, out=180] (14.center) to (15.center);
		\draw [style=thick] (17.center) to (15.center);
		\draw [style=thick] (18) to (21);
		\draw [style=thick, in=135, out=0] (19.center) to (21);
		\draw [style=thick, in=-135, out=0] (20.center) to (21);
		\draw [style=thick] (22.center) to (19.center);
		\draw [style=thick] (20.center) to (23.center);
		\draw [style=thick, in=0, out=180] (23.center) to (15.center);
		\draw [style=thick, in=-180, out=0] (26.center) to (22.center);
		\draw [style=thick] (27.center) to (26.center);
		\draw [style=thick] (30.center) to (14.center);
	\end{pgfonlayer}
\end{tikzpicture}
}
\]

a map can be transformed into a state, reduced to normal form there, and subsequently reconstructed as a map. 

This allows the normal-form problem for arbitrary morphisms to be reduced to the normal-form problem for diagrams of type
\[
0\longrightarrow\underline{n}.
\]

Recall the notion of reduced column echelon form matrix:

\begin{definition}[Reduced column echelon form]
    Let $M\in\C^{m\times n}$. We say that $M$ is in
    \emph{reduced column echelon form} (RCEF) if $M^T$ is in
    reduced row echelon form.
    Equivalently:
    \begin{enumerate}
        \item all zero columns of $M$, if any, occur after all non-zero
        columns;
        \item the first non-zero entry of each non-zero column is equal
        to $1$. We call this entry the \emph{pivot} of the column, and
        its row the corresponding \emph{pivot row};
        \item the pivot of each non-zero column lies strictly below the
        pivot of the preceding non-zero column;
        \item each pivot is the only non-zero entry in its row.
    \end{enumerate}
    Thus, if $M$ has rank $r$, its $r$ non-zero columns have pivot rows
    \[
        1\leq j_1<\cdots<j_r\leq m.
    \]
\end{definition}

    For example,
    \[
        M=
        \begin{bsmallmatrix}
            1 & 0 & 0 \\
            a & 0 & 0 \\
            0 & 1 & 0 \\
            b & c & 0 \\
            0 & 0 & 1
        \end{bsmallmatrix}
    \]
    is in RCEF. Its pivot rows are
    \[
        (j_1,j_2,j_3)=(1,3,5).
    \]
    Indeed,
    \[
        M^T=
        \begin{bsmallmatrix}
            1 & a & 0 & b & 0 \\
            0 & 0 & 1 & c & 0 \\
            0 & 0 & 0 & 0 & 1
        \end{bsmallmatrix}
    \]
    is in reduced row echelon form.\\
RCEF matrices are crucial to the normal for of state diagrams of $\S\aW$:

\begin{theorem}[Completeness of $\S\aW$]
For any diagram $D:0\to\underline{n}$ of $\S\aW$, there exists a unique triple $(r,x,M)$,
where $0\leq r \leq n$, $x\in\C$, and $M\in\C^{n\times r}$ is a full rank matrix in RCEF, such that $D$ can 
be rewritten as:
\[
\scalebox{\scalefig}{\input{./figures/nf-theorem/aW_1.tikz}}
\]
With convention that $(r,x,M) = (n,0,I_n)$ whenever $\interp{D} = 0$,
that for $r = 0$, $M$ is the empty matrix of type $\C^{n\times0}$,
and that in this case the empty matrix arrow is merely a syntactic sugar for:
\[
\scalebox{\scalefig}{\input{./figures/nf-theorem/aW-spec-case_1.tikz}} \quad := \quad \scalebox{\scalefig}{\begin{tikzpicture}
	\begin{pgfonlayer}{nodelayer}
		\node [style=none] (4) at (2, 0) {};
		\node [style=modenb] (8) at (1.75, 0.25) {$n$};
		\node [style=rw] (9) at (0.25, 0) {};
	\end{pgfonlayer}
	\begin{pgfonlayer}{edgelayer}
		\draw [style=thick] (9) to (4.center);
	\end{pgfonlayer}
\end{tikzpicture}
}
\]
Furthermore, two diagrams reduce to the same normal form if and only if they have the same semantics.
\end{theorem}

For the $\FpW$ and $\EpW$ fragments, the presence of cups and caps gives an ordinary map--state duality: inputs can be bent into outputs 
and conversely using \S\A18, 
so that a normal form for states immediately induces one for morphisms of arbitrary arity. 

The completeness proofs for $\S\FpW$ and $\S\EpW$ can therefore be organised by first reducing to normal form 
the associated state and then recovering the original input--output type by means of the compact structure.

\begin{theorem}[Completeness of $\S\FpW$]
For any  diagram $D:0\to\underline{n}$ of $\S\FpW$, there exists a unique quadruple $(r,x,Z,M)$, 
where $0\leq r \leq n$, $x\in\C$, $M\in\C^{n\times r}$ is a full rank matrix in RCEF, and $Z\in \C^{n\times n}$ 
is an antisymmetric matrix whose columns and rows corresponding to pivots of $M$ are zero, such that $D$ can 
be rewritten as:
\[
\scalebox{\scalefig}{\input{./figures/nf-theorem/FpW_1.tikz}}
\]
With the convention that $(r,x,Z,M) = (n,0,0,I_n)$ whenever $\interp{D} = 0$,
that for $r = 0$, $M$ is the empty matrix of type $\C^{n\times0}$,
and that in this case the empty matrix arrow is merely a syntactic sugar for:
\[
\scalebox{\scalefig}{\input{./figures/nf-theorem/aW-spec-case_1.tikz}} \quad := \quad \scalebox{\scalefig}{\begin{tikzpicture}
	\begin{pgfonlayer}{nodelayer}
		\node [style=none] (4) at (2, 0) {};
		\node [style=modenb] (8) at (1.75, 0.25) {$n$};
		\node [style=rw] (9) at (0.25, 0) {};
	\end{pgfonlayer}
	\begin{pgfonlayer}{edgelayer}
		\draw [style=thick] (9) to (4.center);
	\end{pgfonlayer}
\end{tikzpicture}
}
\]
Furthermore, two diagrams reduce to the same normal form if and only if they have the same semantics.
\end{theorem}

The normal form for $\S\EpW$ is very similar to the $\S\FpW$ normal form, the main difference being that the integer
$r =2k$ has to be even, and that using: 
\[
\boxeq{L30}
    {SFpW-lemma/unconnect_1}
    {SFpW-lemma/unconnect_2}
\]

the bottom-left corner of the normal form can be imported in $\S\EpW$.

The two completeness results are then very similar up to this trick.

\begin{theorem}[Completeness of $\S\EpW$]
    For any  diagram $D:0\to\underline{n}$ of $\S\EpW$, there exists a unique quadruple $(k,x,Z,M)$, 
    where $0\leq 2k \leq n$, $x\in\C$, $M\in\C^{n\times 2k}$ is a full rank matrix in RCEF, and $Z\in \C^{n\times n}$ 
    is an antisymmetric matrix whose columns and rows corresponding to pivots of $M$ are zero, such that $D$ can 
    be rewritten as:
    \[
    \scalebox{\scalefig}{\input{./figures/nf-theorem/EpW.tikz}}
    \]
    With the convention that $(k,x,Z,M) = (n/2,0,0,I_{n})$ if $n$ even, $(k,x,Z,M) = ((n-1)/2,0,0,\pmm{I_{n-1}\\0 \cdots 0 }\in \C^{n\times (n-1)})$ if it is odd  whenever $\interp{D} = 0$.
    For $k = 0$, $M$ is the empty matrix of type $\C^{n\times0}$,
    and in this case the empty matrix arrow is merely a syntactic sugar for:
    \[
    \scalebox{\scalefig}{\input{./figures/nf-theorem/EpW-spec-case_1.tikz}} \quad := \quad \scalebox{\scalefig}{\begin{tikzpicture}
	\begin{pgfonlayer}{nodelayer}
		\node [style=none] (4) at (2, 0) {};
		\node [style=modenb] (8) at (1.75, 0.25) {$n$};
		\node [style=rw] (9) at (0.25, 0) {};
	\end{pgfonlayer}
	\begin{pgfonlayer}{edgelayer}
		\draw [style=thick] (9) to (4.center);
	\end{pgfonlayer}
\end{tikzpicture}
}
    \]
    Furthermore, two diagrams reduce to the same normal form if and only if they have the same semantics.
\end{theorem}

As explained in \cref{subsec:scalable-notations}, the completeness  of a scaled pronap theory implies the 
completeness  of its unscaled counterpart, yielding the following corollary.

\begin{corollary}[Completeness of the unscaled fragments]
    \label{cor:unscaled-completeness}
    For each fragment
    \[
        \mathcal P \in \{\EaW,\aW,\EpW,\FpW\},
    \]
    and for any two diagrams $D_1,D_2:m\to n$ in $\mathcal P$,
    \[
        \llbracket D_1\rrbracket
        =
        \llbracket D_2\rrbracket
        \quad\Longrightarrow\quad
        D_1\eq\limits_{\mathcal{P}}D_2.
    \]
    Equivalently, the semantic interpretation functor of each of
    $\EaW$, $\aW$, $\EpW$, and $\FpW$ is faithful, and therefore 
    these pronap theories are complete.
\end{corollary}

\section*{Conclusion}
%

In this work, we have developed a categorical and diagrammatic account of the circuits underlying the Bravyi--Kitaev model of fermionic computation. We first provided a categorical semantics for these circuits and observed that, in full generality, their compositional structure is not monoidal but only \emph{premonoidal}. This distinction reflects the genuinely fermionic character of the model and clarifies the algebraic structure required to reason compositionally about its circuits. Since premonoidal structure is a natural tool for modelling effectful programming languages, this observation would be coherent with  an operational interpretation in which the gates responsible for the failure of monoidality carry a computational effect.

We then introduced equational theories for four fragments of the FZW calculus: $\EaW$, $\EpW$, $\aW$, and $\FpW$. Although $\EaW$ is algebraically elementary, it plays a crucial structural role in the development of the other fragments. The calculi $\EpW$ and $\aW$ are new: $\EpW$ is expressive enough to represent all Gaussian gates appearing in the Bravyi--Kitaev gate set, while $\aW$ is the smallest theory in which the paradigmatic swap is non-natural, and is therefore formulated as a pronap. Finally, $\FpW$ captures the previously studied $\pW$ calculus, but with a different presentation, again formulated in terms of pronaps.

To establish completeness for these fragments, we extended the formalism of scalable notations to the setting of pronaps. This extension allows families of wires and generators to be manipulated directly at the diagrammatic level and provides uniform normal forms for all four fragments. For FW, this yields a new canonical normal-form presentation of the matchgate fragment, distinct from the reduced WGS-X normal form previously obtained for the planar-W calculus~\cite{CaretteEtAl2023}, and expressed uniformly in terms of the matrix-arrow and graph-triangle constructions used throughout our hierarchy. Beyond their role in the completeness arguments, the scalable syntactic constructions introduced along the way have an independent algebraic significance: they provide a diagrammatic representation of operations governing determinants and Pfaffians. As a consequence, several classical identities involving these polynomials admit short and purely diagrammatic proofs.

Altogether, these results provide a common categorical and rewriting framework connecting fermionic computation, scalable diagrammatic reasoning, and the algebra of determinants and Pfaffians. They also isolate a hierarchy of fragments in which increasingly expressive classes of fermionic circuits can be studied through complete graphical languages, providing a foundation for the further investigation of fermionic circuit compilation, optimisation, and classical simulation.

\newpage

\bibliographystyle{acm}
\bibliography{lics}

\newpage

\appendix
\clearpage
\onecolumn
\section{Proofs of \cref{sec:2}}

\begin{proof}[Proof of \cref{lem:application-morphism}]
    Since $t$ is injective, its images satisfy the same CAR:
    \[
        \{f_{t(i)},f_{t(j)}\}=0,\qquad
        \{f_{t(i)}^\dagger,f_{t(j)}^\dagger\}=0,\qquad
        \{f_{t(i)}^\dagger,f_{t(j)}\}=\delta_{ij}1.
    \]
    Hence the assignment on generators extends uniquely to a unital
    $*$-algebra morphism.

    For injectivity, recall that the normally ordered monomials
    \[
        f_{i_1}^\dagger\cdots f_{i_k}^\dagger
        f_{j_1}\cdots f_{j_\ell}
    \]
    form a basis of $\A_m$. Since $t$ is injective, their images are, up to
    the signs required to reorder the indices, distinct normally ordered
    monomials of $\A_{m'}$. They are therefore linearly independent, and
    $\iota_t$ is injective.

    Finally, $\iota_t$ preserves the degree of every monomial and hence its
    parity. If $\U$ is unitary, then
    \[
        \iota_t(\U)\iota_t(\U)^\dagger
        =
        \iota_t(\U\U^\dagger)
        =1,
    \]
    and similarly in the opposite order.
\end{proof}

\begin{proof}[Proof of \cref{lem:fermionic-swap-generators}]
    It suffices first to work in $\A_2$. Set
    \[
        a:=1-n_1-n_2,
        \qquad
        h:=f_2^\dagger f_1+f_1^\dagger f_2,
    \]
    so that $\pi=a+h$. Clearly $\pi^\dagger=\pi$, and the CAR give
    \[
        a^2=1-n_1-n_2+2n_1n_2,
        \qquad
        h^2=n_1+n_2-2n_1n_2,
        \qquad
        ah=ha=0.
    \]
    Hence $\pi^2=1$, so $\pi$ is a Hermitian unitary. Every term of $\pi$
    has even degree, so $\pi$ is even.

    Again using the CAR,
    \[
        \pi f_1
        =
        f_1-n_2f_1-n_1f_2
        =
        f_2\pi,
    \]
    and symmetrically $\pi f_2=f_1\pi$. Since $\pi^2=1$,
    \[
        \pi f_1\pi=f_2,
        \qquad
        \pi f_2\pi=f_1.
    \]
    Applying the morphism associated with $\tau_{i,j}$ gives the result for
    $k=i,j$. If $k\notin\{i,j\}$, $\pi\{i,j\}$ is even and supported only
    on modes $i,j$, hence it commutes with $f_k$ by the CAR. Taking adjoints
    gives the corresponding identities for the creation operators.
\end{proof}

\begin{proof}[Proof of \cref{lem:representation-sym-group}]
    The relation $\pi_r^2=1$ follows from
    Lemma~\ref{lem:fermionic-swap-generators}. If $|r-s|>1$, the two
    Fermionic swaps are even and have disjoint supports, so the CAR imply
    $\pi_r\pi_s=\pi_s\pi_r$.

    It remains to prove the braid relation. Let
    \[
        U:=\pi_r\pi_{r+1}\pi_r,
        \qquad
        V:=\pi_{r+1}\pi_r\pi_{r+1}.
    \]
    By Lemma~\ref{lem:fermionic-swap-generators}, conjugation by $U$ and
    $V$ induces the same permutation
    \[
        \tau_r\tau_{r+1}\tau_r=\tau_{r+1}\tau_r\tau_{r+1}
    \]
    on every generator $f_k$ and $f_k^\dagger$. Hence $V^{-1}U$ commutes
    with every generator of $\A_m$, and therefore belongs to the centre of
    $\A_m$. By the Jordan--Wigner isomorphism,
    $Z(\A_m)=\C 1$, so $U=\lambda V$ for some $\lambda\in\C$.

    Each Fermionic swap fixes the vacuum $\ket{0^m}$, and therefore so do
    $U$ and $V$. It follows that $\lambda=1$, proving the braid relation.
    The existence and uniqueness of the homomorphism then follow from the
    Coxeter presentation of $\mathfrak S_m$.
\end{proof}

\begin{proof}[Proof of \cref{lem:G-equivariant}]
    By Lemma~\ref{lem:fermionic-swap-generators} and the preceding
    representation,
    \[
        \sigma\star f_k=f_{\sigma(k)},
        \qquad
        \sigma\star f_k^\dagger=f_{\sigma(k)}^\dagger.
    \]
    Hence the two $*$-algebra morphisms
    \[
        (\sigma\star -)\circ\iota_t
        \qquad\text{and}\qquad
        \iota_{\sigma\circ t}
    \]
    agree on every generator of $\A_m$, and are therefore equal. Applying
    them to $\G$ gives the result.
\end{proof}
\section{Proofs of \cref{sec:3}}
\begin{proof}[Proof of \cref{lem:car-strict-premonoidal}]
    For every fixed $k$, the maps
    \[
        k\owedge(-)
        \qquad\text{and}\qquad
        (-)\owedge k
    \]
    are induced by the mode embeddings introduced in
    \cref{subsec:2-3}. Since these embeddings are unital algebra
    morphisms, they preserve identities and composition, and hence
    define functors.

    The unit equations follow from the identity embedding
    $\A_m\to\A_m$. The associativity equations follow from composition
    of the corresponding index embeddings. Explicitly, for
    $a\in\A_m$,
    \[
        k\owedge(\ell\owedge a)
        =
        (k+\ell)\owedge a,
        \qquad
        (a\owedge\ell)\owedge k
        =
        a\owedge(\ell+k),
    \]
    and
    \[
        (k\owedge a)\owedge\ell
        =
        k\owedge(a\owedge\ell).
    \]
    Thus $(\CAR,\owedge,0)$ is strict premonoidal.
\end{proof}

\begin{proof}[Proof of \cref{lem:graded-interchange-CAR}]
    Let $a\in\A_m$ and $a'\in\A_n$ be homogeneous, with parities
    $|a|,|a'|\in\mathbb Z_2$.
    The generators appearing in $a\owedge n$ have indices in
    $\{1,\ldots,m\}$, whereas those appearing in $m\owedge a'$
    have indices in
    $\{m+1,\ldots,m+n\}$.

    Since the two sets of modes are disjoint, the CAR imply that
    exchanging one generator from the first monomial with one generator
    from the second contributes a factor $-1$. The result follows by linearity: 
    \[
        (a\owedge n)(m\owedge a')
        =
        (-1)^{|a||a'|}
        (m\owedge a')(a\owedge n).
    \]
\end{proof}

\begin{proof}[Proof of \cref{prop:center-CAR}]
    By \cref{lem:graded-interchange-CAR}, every even operator commutes
    with every homogeneous operator supported on a disjoint block of
    modes, and therefore, by linearity, with every operator supported
    on such a block. Hence every even operator is central.

    Conversely, let $a\in\A_m$ and write uniquely
    \[
        a=a_{\mathrm{even}}+a_{\mathrm{odd}}.
    \]
    Suppose that $a_{\mathrm{odd}}\neq0$. Let
    \[
        F:=f_1+f_1^\dagger\in\A_1
    \]
    be the mode-flip introduced in \cref{subsec:2-3}. Since $F$ is odd,
    \cref{lem:graded-interchange-CAR} gives
    \[
    \begin{aligned}
        &(a\owedge1)(m\owedge F)
        -(m\owedge F)(a\owedge1)\\
        &\qquad=
        2(a_{\mathrm{odd}}\owedge1)(m\owedge F).
    \end{aligned}
    \]
    The right-hand side is non-zero: the application embedding is
    injective, so $a_{\mathrm{odd}}\owedge1\neq0$, while
    $m\owedge F$ is invertible since $F$ is unitary.
    Hence $a$ is not central.

    Therefore the central morphisms are exactly the even operators.
\end{proof}

\begin{proof}[Proof of \cref{prop:CAR-symmetric-premonoidal}]
    Every $\pi_{m,n}$ is even and therefore central by
    \cref{prop:center-CAR}. The involution and coherence equations
    follow from the corresponding identities between block
    permutations, since $\sigma\mapsto\pi_\sigma$ is a representation
    of the symmetric group.

    It remains to prove naturality. Let $a\in\A_m$. By the equivariance
    lemma of \cref{subsec:2-3}, conjugation by the block permutation
    transports the first block of $m$ modes to the second block:
    \[
        \pi_{m,n}(a\owedge n)\pi_{n,m}
        =
        n\owedge a.
    \]
    Equivalently,
    \[
        \pi_{m,n}(a\owedge n)
        =
        (n\owedge a)\pi_{m,n}.
    \]
    Similarly, for $a'\in\A_n$,
    \[
        \pi_{m,n}(m\owedge a')\pi_{n,m}
        =
        a'\owedge m,
    \]
    or equivalently
    \[
        \pi_{m,n}(m\owedge a')
        =
        (a'\owedge m)\pi_{m,n}.
    \]
    These are precisely the two naturality equations. Hence every
    morphism of $\CAR$ is natural with respect to $\pi$.
\end{proof}

\begin{proof}[Proof of \cref{prop:graded-swap}]
    Decompose the block permutation exchanging the first $m$ modes with
    the following $n$ modes into adjacent Fermionic swaps. Each mode of
    the first block crosses each mode of the second block exactly once.
    A Fermionic swap contributes a minus sign precisely when both
    crossed modes are occupied.

    Hence, on $\ket{x}\otimes\ket{y}$, the number of sign-producing
    crossings is $|x||y|$, and therefore
    \[
        \Pi_{m,n}
        (\ket{x}\otimes\ket{y})
        =
        (-1)^{|x||y|}
        \ket{y}\otimes\ket{x}.
    \]
\end{proof}

\begin{proof}[Proof of \cref{prop:nature-LFM}]
    The involutivity and coherence equations follow from
    \cref{prop:graded-swap}, equivalently from the corresponding block
    permutations.

    Let $f:m\to n$ be even and write
    \[
        f=\sum_{x,y}f_{y,x}\ketbra{y}{x}.
    \]
    Whenever $f_{y,x}\neq0$,
    \[
        |x|\equiv|y|\pmod2.
    \]
    Hence, for every $z\in\{0,1\}^k$,
    \[
    \begin{aligned}
        \Pi_{n,k}(f\otimes id_k)\ket{xz}
        &=
        \sum_y
        f_{y,x}(-1)^{|y||z|}\ket{zy}\\
        &=
        \sum_y
        f_{y,x}(-1)^{|x||z|}\ket{zy}\\
        &=
        (id_k\otimes f)\Pi_{m,k}\ket{xz}.
    \end{aligned}
    \]
    The second naturality equation is analogous. Thus every even map
    belongs to the Nature of $\Pi$.

    Conversely, suppose that $f$ is natural. Taking $k=1$ and applying
    the preceding naturality equation to
    $\ket{x}\otimes\ket{1}$ yields
    \[
        \sum_y f_{y,x}(-1)^{|y|}\ket{1y}
        =
        \sum_y f_{y,x}(-1)^{|x|}\ket{1y}.
    \]
    By linear independence of the computational basis, every non-zero
    coefficient $f_{y,x}$ therefore satisfies
    \[
        (-1)^{|y|}=(-1)^{|x|}.
    \]
    Hence $|x|+|y|$ is even whenever $f_{y,x}\neq0$, so $f$ is even.
    The Nature of $\Pi$ is therefore exactly $\ELFM$.
\end{proof}

\begin{proof}[Proof of \cref{lem:mon-sym-tradeoff}]
    For every object $A$, the functoriality of $(-)\owedge A$ follows
    directly from the bifunctoriality of $\otimes$, since
    \[
        f\owedge A=f\otimes id_A.
    \]

    For $A\owedge(-)$, let $f:B\to C$ and $g:C\to D$. Then
    \[
    \begin{aligned}
        (A\owedge g)\circ(A\owedge f)
        &=
        \pi_{D,A}(g\otimes id_A)
        \pi_{A,C}
        \pi_{C,A}(f\otimes id_A)
        \pi_{A,B}\\
        &=
        \pi_{D,A}
        ((g\circ f)\otimes id_A)
        \pi_{A,B}\\
        &=
        A\owedge(g\circ f),
    \end{aligned}
    \]
    using the involutivity
    $\pi_{A,C}\pi_{C,A}=id_{C\otimes A}$.
    Moreover,
    \[
        A\owedge id_B
        =
        \pi_{B,A}\pi_{A,B}
        =
        id_{A\otimes B}.
    \]
    Thus both left and right extension by an object are functorial.

    The unit and strict associativity equations follow by expanding the
    definitions and using the coherence equations of the presymmetry.
    Likewise, the same coherence equations imply that the morphisms
    $\pi_{A,B}$ are central for $\owedge$.

    By construction, $\pi$ is natural with respect to every morphism for
    the product $\owedge$. Indeed, for $f:B\to C$,
    \[
    \begin{aligned}
        (A\owedge f)\circ\pi_{B,A}
        &=
        \pi_{C,A}(f\otimes id_A)
        \pi_{A,B}\pi_{B,A}\\
        &=
        \pi_{C,A}(f\otimes id_A)\\
        &=
        \pi_{C,A}\circ(f\owedge A),
    \end{aligned}
    \]
    and similarly
    \[
    \begin{aligned}
        \pi_{A,C}\circ(A\owedge f)
        &=
        \pi_{A,C}\pi_{C,A}
        (f\otimes id_A)\pi_{A,B}\\
        &=
        (f\otimes id_A)\pi_{A,B}\\
        &=
        (f\owedge A)\circ\pi_{A,B}.
    \end{aligned}
    \]
    Hence $(\mathcal C,\owedge,\pi)$ is strict symmetric
    premonoidal.

    Finally, suppose that $f:B\to C$ belongs to the Nature of
    $(\mathcal C,\otimes,\pi)$. Naturality for $\otimes$ gives
    \[
        \pi_{C,A}(f\otimes id_A)
        =
        (id_A\otimes f)\pi_{B,A},
    \]
    and therefore
    \[
    \begin{aligned}
        A\owedge f
        &=
        \pi_{C,A}(f\otimes id_A)\pi_{A,B}\\
        &=
        id_A\otimes f.
    \end{aligned}
    \]
    On the other hand, by definition,
    \[
        f\owedge A=f\otimes id_A.
    \]
    The ordinary interchange law for $\otimes$ then implies that $f$
    commutes, with respect to $\owedge$, with every morphism. Thus $f$
    is central for $\owedge$.
\end{proof}

\begin{proof}[Proof of \cref{lem:LFM-whiskering}]
    The equality
    \[
        f\owedge k=f\otimes id_k
    \]
    is the definition.

    For the second equality, let
    $z\in\{0,1\}^k$ and $x\in\{0,1\}^m$. By definition,
    \[
        k\owedge f
        =
        \Pi_{n,k}(f\otimes id_k)\Pi_{k,m}.
    \]
    Hence
    \[
    \begin{aligned}
        (k\owedge f)\ket{zx}
        &=
        (-1)^{|z||x|}
        \Pi_{n,k}
        (f\otimes id_k)\ket{xz}\\
        &=
        \sum_y
        f_{y,x}
        (-1)^{|z||x|+|y||z|}
        \ket{zy}.
    \end{aligned}
    \]
    Whenever $f_{y,x}\neq0$, homogeneity gives
    \[
        |y|\equiv|x|+|f|\pmod2.
    \]
    Therefore
    \[
        |z||x|+|y||z|
        \equiv
        |f||z|
        \pmod2,
    \]
    and consequently
    \[
        (k\owedge f)\ket{zx}
        =
        (\mathsf P_k^{\,|f|}\otimes f)\ket{zx}.
    \]
    Thus
    \[
        k\owedge f
        =
        \mathsf P_k^{\,|f|}\otimes f.
    \]
\end{proof}

\begin{proof}[Proof of \cref{prop:graded-interchange-LFM}]
    Let $f:m\to m'$ and $g:n\to n'$ be homogeneous.
    By \cref{lem:LFM-whiskering},
    \[
    \begin{aligned}
        (f\owedge n')\circ(m\owedge g)
        &=
        (f\otimes id_{n'})
        \circ
        (\mathsf P_m^{\,|g|}\otimes g)\\
        &=
        (f\mathsf P_m^{\,|g|})\otimes g,
    \end{aligned}
    \]
    whereas
    \[
    \begin{aligned}
        (m'\owedge g)\circ(f\owedge n)
        &=
        (\mathsf P_{m'}^{\,|g|}\otimes g)
        \circ
        (f\otimes id_n)\\
        &=
        (\mathsf P_{m'}^{\,|g|}f)\otimes g.
    \end{aligned}
    \]
    Since $f$ has parity $|f|$,
    \[
        \mathsf P_{m'}f
        =
        (-1)^{|f|}
        f\mathsf P_m.
    \]
    Hence
    \[
        \mathsf P_{m'}^{\,|g|}f
        =
        (-1)^{|f||g|}
        f\mathsf P_m^{\,|g|},
    \]
    and therefore
    \[
        (f\owedge n')\circ(m\owedge g)
        =
        (-1)^{|f||g|}
        (m'\owedge g)\circ(f\owedge n).
    \]

    In particular, every even map is central. Conversely, if $f$ is a
    non-zero odd map, choose any non-zero odd one-mode map, for instance
    the mode-flip $F$. The two corresponding interchange composites are
    non-zero and differ by a minus sign, so $f$ is not central.
    Therefore the centre is exactly $\ELFM$.
\end{proof}

\begin{proof}[Proof of \cref{thm:JW-premonoidal}]
    The first statement follows from
    \cref{lem:mon-sym-tradeoff}, and the description of the centre of
    $\LFM$ follows from \cref{prop:graded-interchange-LFM}.

    It remains to show that $JW$ preserves the two extensions by an
    object. Let $a\in\A_m$.

    Since $a\owedge n$ acts on the first $m$ modes and the additional
    $n$ modes are appended afterwards,
    \[
        JW(a\owedge n)
        =
        JW(a)\otimes id_n
        =
        JW(a)\owedge n.
    \]

    For the other extension, equivariance of the CAR action under the
    Fermionic block permutation gives
    \[
        n\owedge a
        =
        \pi_{m,n}(a\owedge n)\pi_{n,m}.
    \]
    Applying $JW$ yields
    \[
    \begin{aligned}
        JW(n\owedge a)
        &=
        \Pi_{m,n}
        (JW(a)\otimes id_n)
        \Pi_{n,m}\\
        &=
        n\owedge JW(a),
    \end{aligned}
    \]
    by the definition of the premonoidal product on the Hilbert-space
    semantics.

    Finally,
    \[
        JW(\pi_{m,n})=\Pi_{m,n},
    \]
    so $JW$ preserves the presymmetry as well.
\end{proof}
\section{Proofs of \cref{sec:5}}
\begin{proof}[Proof of \cref{lemma:scaled-diag-nat}]
    By induction on $D$:
    \begin{itemize}[label=|]
    \item Base cases:
            \[\begin{cases}
                 g \to \text{definition of }\S_{k_1}(D)\text{, immediate.}\\
                 \mathrm{id}_1 \to \text{ follows from $\text{Elim}_{1,k}$.}\\
                 \pi \to\text{ follows from naturality of dividers and gatherers.}
            \end{cases}\]

    \item Case $D = D' \circ D''$:\smallskip
    \[
    \scalebox{\scalefig}{\input{./figures/proof-sec-4/circ_1.tikz}} \eq
    \scalebox{\scalefig}{\input{./figures/proof-sec-4/circ_2.tikz}} \eq
    \scalebox{\scalefig}{\input{./figures/proof-sec-4/circ_3.tikz}} 
    \]
    \item Case $D = D' \otimes D''$\smallskip
   \[
    \scalebox{\scalefig}{\input{./figures/proof-sec-4/otime_1.tikz}} \eq
    \scalebox{\scalefig}{\input{./figures/proof-sec-4/otime_2.tikz}} \eq
    \scalebox{\scalefig}{\input{./figures/proof-sec-4/otime_3.tikz}} 
    \]
    Note that this step is only true so long as all the generators in the expression of $D$ are natural. The other inductive steps are however valid for any diagram.
    \end{itemize}
    \end{proof}

    \begin{proof}[Proof of \cref{lemma:scaled-diag-line}]
    Same as above without the step for $\otimes$ in the induction.
    \end{proof}

    \begin{proof}[Proof of \cref{thm:scalable-equations}]
    Let $(D, D')$ be an equation such that $D$ and $D'$ only involve natural generators. Using \cref{lemma:scaled-diag-nat}, we get :
    \[
    \scalebox{\scalefig}{\input{./figures/proof-sec-4/thm_1.tikz}} \eq
    \scalebox{\scalefig}{\input{./figures/proof-sec-4/thm_2.tikz}} \eq
    \scalebox{\scalefig}{\input{./figures/proof-sec-4/thm_3.tikz}} \eq
    \scalebox{\scalefig}{\input{./figures/proof-sec-4/thm_4.tikz}}
    \]
    Similarly, using \cref{lemma:scaled-diag-line}, we get the second point.
    \end{proof}
\clearpage
\onecolumn
\sensitivetikzfalse

\subsection{Proofs of the derived equations}
\subsubsection{Derived equations in $\S\EaW$}
\begin{proof}[Proof of \L5]
            The base case comes from \EaW equation \A10. 
            \phantomsection\label{proof:SEaW-zero-mat}
                \begin{derivationchain}
    \scalebox{\scalefig}{\input{./figures/aW-proof/zero-mat_1.tikz}} &
    $\eq$ &
    \scalebox{\scalefig}{\input{./figures/aW-proof/zero-mat_2.tikz}} &
    $\eq\limits_{IH}$ &
    \scalebox{\scalefig}{\input{./figures/aW-proof/zero-mat_3.tikz}} &
    $\eq\limits_{\S \A6}$ &
    \scalebox{\scalefig}{\input{./figures/aW-proof/zero-mat_0.tikz}} \\[1em]
\end{derivationchain}
            \end{proof}
\begin{proof}[Proof of \L6]
            \phantomsection\label{proof:SEaW-direct-sum-mat}
                \begin{derivationchain}
    \scalebox{\scalefig}{\input{./figures/aW-proof/direct-sum-mat_1.tikz}} &
    $\eq\limits_{\S \A6}$ &
    \scalebox{\scalefig}{\input{./figures/aW-proof/direct-sum-mat_2.tikz}} &
    $\eq\limits_{\L5}$ &
    \scalebox{\scalefig}{\input{./figures/aW-proof/direct-sum-mat_3.tikz}} &
    $\eq$ &
    \scalebox{\scalefig}{\input{./figures/aW-proof/direct-sum-mat_4.tikz}} \\[1em]
    \scalebox{\scalefig}{\input{./figures/aW-proof/direct-sum-mat_4.tikz}} &
    $\eq\limits_{\L5}$ &
    \scalebox{\scalefig}{\input{./figures/aW-proof/direct-sum-mat_5.tikz}} &
    $\eq\limits_{D_2,D_3}$ &
    \scalebox{\scalefig}{\input{./figures/aW-proof/direct-sum-mat_0.tikz}} &
    {} &
    {} \\[1em]
\end{derivationchain}
            \end{proof}
\begin{proof}[Proof of \L7]
                \phantomsection\label{proof:SEaW-lambda-mat}
                    \begin{derivationchain}
    \scalebox{\scalefig}{\input{./figures/SEaW-lemma-proof/lambda-mat_1.tikz}} &
    $\eq\limits_{\L6}$ &
    \scalebox{\scalefig}{\input{./figures/SEaW-lemma-proof/lambda-mat_2.tikz}} &
    $\eq\limits_{\D1}$ &
    \scalebox{\scalefig}{\input{./figures/SEaW-lemma-proof/lambda-mat_3.tikz}} &
    $\eq\limits_{IH}$ &
    \scalebox{\scalefig}{\input{./figures/SEaW-lemma-proof/lambda-mat_4.tikz}} \\[1em]
    \scalebox{\scalefig}{\input{./figures/SEaW-lemma-proof/lambda-mat_4.tikz}} &
    $\eq\limits_{\Dist}$ &
    \scalebox{\scalefig}{\input{./figures/SEaW-lemma-proof/lambda-mat_5.tikz}} &
    {} &
    {} &
    {} &
    {} \\[1em]
\end{derivationchain}
            \end{proof}
\begin{proof}[Proof of \L8]
            \phantomsection\label{proof:SEaW-id-mat}
                \begin{derivationchain}
    \scalebox{\scalefig}{\input{./figures/SEaW-lemma-proof/id-mat_1.tikz}} &
    $\eq\limits_{\L7}$ &
    \scalebox{\scalefig}{\input{./figures/SEaW-lemma-proof/id-mat_2.tikz}} &
    $\eq\limits_{\S \A12}$ &
    \scalebox{\scalefig}{\begin{tikzpicture}
	\begin{pgfonlayer}{nodelayer}
		\node [style=modenb] (2) at (-1.25, 0.25) {$n$};
		\node [style=modenb] (3) at (1.25, 0.25) {$n$};
		\node [style=none] (7) at (1.5, 0) {};
		\node [style=none] (12) at (-1.5, 0) {};
	\end{pgfonlayer}
	\begin{pgfonlayer}{edgelayer}
		\draw [style=thick] (12.center) to (7.center);
	\end{pgfonlayer}
\end{tikzpicture}
} &
    {} &
    {} \\[1em]
\end{derivationchain}
            \end{proof}
\begin{proof}[Proof of \L1]
            \phantomsection\label{proof:SEaW-add-mat}
    After unfolding the scalar case of the matrix-arrow definition, the base case is exactly \A7.
                \begin{derivationchain}
    \scalebox{\scalefig}{\input{./figures/aW-proof/add-mat_1.tikz}} &
    $\eq\limits_{\D2}$ &
    \scalebox{\scalefig}{\input{./figures/aW-proof/add-mat_2.tikz}} &
    $\eq\limits_{\S \A1, \S \A3}$ &
    \scalebox{\scalefig}{\input{./figures/aW-proof/add-mat_3.tikz}} &
    $\eq\limits_{\Dist}$ &
    \scalebox{\scalefig}{\input{./figures/aW-proof/add-mat_4.tikz}} \\[1em]
    \scalebox{\scalefig}{\input{./figures/aW-proof/add-mat_4.tikz}} &
    $\eq$ &
    \scalebox{\scalefig}{\input{./figures/aW-proof/add-mat_5.tikz}} &
    $\eq\limits_{IH}$ &
    \scalebox{\scalefig}{\input{./figures/aW-proof/add-mat_6.tikz}} &
    $\eq\limits_{\D2}$ &
    \scalebox{\scalefig}{\input{./figures/aW-proof/add-mat_0.tikz}} \\[1em]
\end{derivationchain}
            \end{proof}
\begin{proof}[Proof of \L3]
            \phantomsection\label{proof:SEaW-mult-mat}
    After unfolding the scalar case of the matrix-arrow definition, the base case is exactly \A9.
                \begin{derivationchain}
    \scalebox{\scalefig}{\input{./figures/aW-proof/mult-mat_1.tikz}} &
    $\eq\limits_{\D2,\D3}$ &
    \scalebox{\scalefig}{\input{./figures/aW-proof/mult-mat_2.tikz}} &
    $\eq\limits_{\Exp}$ &
    \scalebox{\scalefig}{\input{./figures/aW-proof/mult-mat_3.tikz}} &
    $\eq\limits_{IH}$ &
    \scalebox{\scalefig}{\input{./figures/aW-proof/mult-mat_4.tikz}} \\[1em]
    \scalebox{\scalefig}{\input{./figures/aW-proof/mult-mat_4.tikz}} &
    $\eq\limits_{\L1}$ &
    \scalebox{\scalefig}{\input{./figures/aW-proof/mult-mat_0.tikz}} &
    {} &
    {} &
    {} &
    {} \\[1em]
\end{derivationchain}
            \end{proof}
\begin{proof}[Proof of \L4]
            \phantomsection\label{proof:SEaW-erase-row-col}
            The base case follows from \D1 and \A5.\\
            \[\scalebox{\scalefig}{\input{./figures/aW-proof/erase-row-col-aux_1.tikz}} = \scalebox{\scalefig}{\begin{tikzpicture}
	\begin{pgfonlayer}{nodelayer}
		\node [style=none] (2) at (-0.75, 0) {};
		\node [style=modenb] (4) at (-0.25, 0.25) {$n$};
		\node [style=lw] (5) at (1.25, 0) {};
	\end{pgfonlayer}
	\begin{pgfonlayer}{edgelayer}
		\draw [style=thick] (5) to (2.center);
	\end{pgfonlayer}
\end{tikzpicture}
}\]
                \begin{derivationchain}
\multicolumn{7}{l}{Case $\begin{bsmallmatrix}M & M'
                        \end{bsmallmatrix}$ :}\\[0.5em]
    \scalebox{\scalefig}{\input{./figures/aW-proof/erase-row-col-aux_3.tikz}} &
    $\eq\limits_{\D2}$ &
    \scalebox{\scalefig}{\input{./figures/aW-proof/erase-row-col-aux_4.tikz}} &
    $\eq\limits_{\S \A4}$ &
    \scalebox{\scalefig}{\input{./figures/aW-proof/erase-row-col-aux_5.tikz}} &
    $\eq\limits_{IH}$ &
    \scalebox{\scalefig}{\input{./figures/aW-proof/erase-row-col-aux_6.tikz}} \\[1em]
    \scalebox{\scalefig}{\input{./figures/aW-proof/erase-row-col-aux_6.tikz}} &
    $\eq\limits_{\Dist}$ &
    \scalebox{\scalefig}{\input{./figures/aW-proof/erase-row-col-aux_7.tikz}} &
    {} &
    {} &
    {} &
    {} \\[1em]
\end{derivationchain}
\begin{derivationchain}
\multicolumn{7}{l}{Case $\begin{bsmallmatrix} M \\ M'
                        \end{bsmallmatrix}$ :}\\[0.5em]
    \scalebox{\scalefig}{\input{./figures/aW-proof/erase-row-col-aux_8.tikz}} &
    $\eq\limits_{\D3}$ &
    \scalebox{\scalefig}{\input{./figures/aW-proof/erase-row-col-aux_9.tikz}} &
    $\eq\limits_{\Dist}$ &
    \scalebox{\scalefig}{\input{./figures/aW-proof/erase-row-col-aux_10.tikz}} &
    $\eq\limits_{IH}$ &
    \scalebox{\scalefig}{\input{./figures/aW-proof/erase-row-col-aux_11.tikz}} \\[1em]
    \scalebox{\scalefig}{\input{./figures/aW-proof/erase-row-col-aux_11.tikz}} &
    $\eq\limits_{\S \A6}$ &
    \scalebox{\scalefig}{\begin{tikzpicture}
	\begin{pgfonlayer}{nodelayer}
		\node [style=none] (8) at (-0.75, 0) {};
		\node [style=modenb] (10) at (-0.5, 0.25) {$n$};
		\node [style=lw] (14) at (0.25, 0) {};
	\end{pgfonlayer}
	\begin{pgfonlayer}{edgelayer}
		\draw [style=thick] (14) to (8.center);
	\end{pgfonlayer}
\end{tikzpicture}
} &
    {} &
    {} &
    {} &
    {} \\[1em]
\end{derivationchain}
            \end{proof}
\begin{proof}[Proof of \L2]
            \phantomsection\label{proof:SEaW-copy-mat}
    After unfolding the scalar case of the matrix-arrow definition, the base case is exactly \A8.
                We proceed by induction on matrices: \\
                \begin{derivationchain}
\multicolumn{7}{l}{Case $\begin{bsmallmatrix}M & M'
                        \end{bsmallmatrix}$ :}\\[0.5em]
    \scalebox{\scalefig}{\input{./figures/aW-proof/copy-mat_1.tikz}} &
    $\eq\limits_{\D2}$ &
    \scalebox{\scalefig}{\input{./figures/aW-proof/copy-mat_2.tikz}} &
    $\eq\limits_{\S \A2}$ &
    \scalebox{\scalefig}{\input{./figures/aW-proof/copy-mat_3.tikz}} &
    $\eq\limits_{IH}$ &
    \scalebox{\scalefig}{\input{./figures/aW-proof/copy-mat_4.tikz}} \\[1em]
    \scalebox{\scalefig}{\input{./figures/aW-proof/copy-mat_4.tikz}} &
    $\eq$ &
    \scalebox{\scalefig}{\input{./figures/aW-proof/copy-mat_5.tikz}} &
    $\eq\limits_{\Exp}$ &
    \scalebox{\scalefig}{\input{./figures/aW-proof/copy-mat_6.tikz}} &
    $\eq\limits_{\Dist}$ &
    \scalebox{\scalefig}{\input{./figures/aW-proof/copy-mat_7.tikz}} \\[1em]
    \scalebox{\scalefig}{\input{./figures/aW-proof/copy-mat_7.tikz}} &
    $\eq\limits_{\D2}$ &
    \scalebox{\scalefig}{\input{./figures/aW-proof/copy-mat_8.tikz}} &
    {} &
    {} &
    {} &
    {} \\[1em]
\end{derivationchain}
                \begin{derivationchain}
\multicolumn{7}{l}{Case $\begin{bsmallmatrix} M \\ M'
                        \end{bsmallmatrix}$ :}\\[0.5em]
    \scalebox{\scalefig}{\input{./figures/aW-proof/copy-mat_9.tikz}} &
    $\eq\limits_{\Dist}$ &
    \scalebox{\scalefig}{\input{./figures/aW-proof/copy-mat_10.tikz}} &
    $\eq\limits_{\D3}$ &
    \scalebox{\scalefig}{\input{./figures/aW-proof/copy-mat_11.tikz}} &
    $\eq\limits_{IH}$ &
    \scalebox{\scalefig}{\input{./figures/aW-proof/copy-mat_12.tikz}} \\[1em]
    \scalebox{\scalefig}{\input{./figures/aW-proof/copy-mat_12.tikz}} &
    $\eq\limits_{\S \A1, \S \A3}$ &
    \scalebox{\scalefig}{\input{./figures/aW-proof/copy-mat_13.tikz}} &
    $\eq\limits_{\D3}$ &
    \scalebox{\scalefig}{\input{./figures/aW-proof/copy-mat_14.tikz}} &
    {} &
    {} \\[1em]
\end{derivationchain}
            \end{proof}
\begin{proof}[Proof of \L9]
            \phantomsection\label{proof:SEaW-rank-dec-mat}
                Note that any matrix $M \in \C^{m\times n}$ of rank $r$ can be decomposed as $M = P'Q'$
                    with $P' \in \C^{m\times r}$ and $Q' \in \C^{r\times n}$ both of rank $r$. Then, one can complete $P'$ and $Q'$ 
                    into invertible matrices $P \in \mathcal{GL}_m(\C)$ and $Q \in \mathcal{GL}_n(\C)$ by adding columns and rows
                \begin{derivationchain}
    \scalebox{\scalefig}{\input{./figures/aW-proof/rank-dec-mat_1.tikz}} &
    $\eq\limits_{\L3}$ &
    \scalebox{\scalefig}{\input{./figures/aW-proof/rank-dec-mat_2.tikz}} &
    $\eq\limits_{\D2,\D3, \L4}$ &
    \scalebox{\scalefig}{\input{./figures/aW-proof/rank-dec-mat_0.tikz}} &
    {} &
    {} \\[1em]
\end{derivationchain}
            \end{proof}
\subsubsection{Derived equations in $\S\EpW$}
\begin{proof}[Proof of \L10]
                    \phantomsection\label{proof:SEpW-triangle-add}
                    \begin{derivationchain}
    \scalebox{\scalefig}{\input{./figures/SEpW-lemma-proof/triangle-add_1.tikz}} &
    $\eq$ &
    \scalebox{\scalefig}{\input{./figures/SEpW-lemma-proof/triangle-add_2.tikz}} &
    $\eq\limits_{\D4}$ &
    \scalebox{\scalefig}{\input{./figures/SEpW-lemma-proof/triangle-add_3.tikz}} &
    $\eq\limits_{\S \A1, \S \A3, \S \A16}$ &
    \scalebox{\scalefig}{\input{./figures/SEpW-lemma-proof/triangle-add_4.tikz}} \\[1em]
    \scalebox{\scalefig}{\input{./figures/SEpW-lemma-proof/triangle-add_4.tikz}} &
    $\eq\limits_{\L1}$ &
    \scalebox{\scalefig}{\input{./figures/SEpW-lemma-proof/triangle-add_5.tikz}} &
    $\eq\limits_{\D4}$ &
    \scalebox{\scalefig}{\input{./figures/SEpW-lemma-proof/triangle-add_6.tikz}} &
    {} &
    {} \\[1em]
\end{derivationchain}
                \end{proof}
\begin{proof}[Proof of \L11]
                    \phantomsection\label{proof:SEpW-triangle-erase}
                    \begin{derivationchain}
    \scalebox{\scalefig}{\input{./figures/SEpW-lemma-proof/triangle-erase_1.tikz}} &
    $\eq\limits_{\D4}$ &
    \scalebox{\scalefig}{\input{./figures/SEpW-lemma-proof/triangle-erase_2.tikz}} &
    $\eq\limits_{\S \A4}$ &
    \scalebox{\scalefig}{\input{./figures/SEpW-lemma-proof/triangle-erase_3.tikz}} &
    $\eq\limits_{\L4, \S \A16}$ &
    \scalebox{\scalefig}{\input{./figures/SEpW-lemma-proof/triangle-erase_4.tikz}} \\[1em]
    \scalebox{\scalefig}{\input{./figures/SEpW-lemma-proof/triangle-erase_4.tikz}} &
    $\eq\limits_{\S \A17}$ &
    \scalebox{\scalefig}{\begin{tikzpicture}
	\begin{pgfonlayer}{nodelayer}
		\node [style={empty_box}] (0) at (0, 0) {};
	\end{pgfonlayer}
\end{tikzpicture}
} &
    {} &
    {} &
    {} &
    {} \\[1em]
\end{derivationchain}
                \end{proof}
\begin{proof}[Proof of \L12]
                \phantomsection\label{proof:SEpW-mat-cup}
    After unfolding the scalar case of the matrix-arrow definition, the base case is exactly \A15.
                    \begin{derivationchain}
    \scalebox{\scalefig}{\input{./figures/SEpW-lemma-proof/mat-cup_1.tikz}} &
    $\eq\limits_{\D2}$ &
    \scalebox{\scalefig}{\input{./figures/SEpW-lemma-proof/mat-cup_2.tikz}} &
    $\eq\limits_{\Dist}$ &
    \scalebox{\scalefig}{\input{./figures/SEpW-lemma-proof/mat-cup_3.tikz}} &
    $\eq\limits_{IH}$ &
    \scalebox{\scalefig}{\input{./figures/SEpW-lemma-proof/mat-cup_4.tikz}} \\[1em]
    \scalebox{\scalefig}{\input{./figures/SEpW-lemma-proof/mat-cup_4.tikz}} &
    $\eq\limits_{\S \A16, \S \A18}$ &
    \scalebox{\scalefig}{\input{./figures/SEpW-lemma-proof/mat-cup_5.tikz}} &
    $\eq\limits_{\D3}$ &
    \scalebox{\scalefig}{\input{./figures/SEpW-lemma-proof/mat-cup_6.tikz}} &
    {} &
    {} \\[1em]
\end{derivationchain}
                    \begin{derivationchain}
    \scalebox{\scalefig}{\input{./figures/SEpW-lemma-proof/mat-cup_7.tikz}} &
    $\eq\limits_{\D3}$ &
    \scalebox{\scalefig}{\input{./figures/SEpW-lemma-proof/mat-cup_8.tikz}} &
    $\eq\limits_{\S \A16, \S \A18}$ &
    \scalebox{\scalefig}{\input{./figures/SEpW-lemma-proof/mat-cup_9.tikz}} &
    $\eq\limits_{IH}$ &
    \scalebox{\scalefig}{\input{./figures/SEpW-lemma-proof/mat-cup_10.tikz}} \\[1em]
    \scalebox{\scalefig}{\input{./figures/SEpW-lemma-proof/mat-cup_10.tikz}} &
    $\eq\limits_{\Dist}$ &
    \scalebox{\scalefig}{\input{./figures/SEpW-lemma-proof/mat-cup_11.tikz}} &
    $\eq\limits_{\D2}$ &
    \scalebox{\scalefig}{\input{./figures/SEpW-lemma-proof/mat-cup_12.tikz}} &
    {} &
    {} \\[1em]
\end{derivationchain}
                \end{proof}
\begin{proof}[Proof of \L13]
                    \phantomsection\label{proof:SEpW-triangle-dec}
                    \begin{derivationchain}
    \scalebox{\scalefig}{\begin{tikzpicture}
	\begin{pgfonlayer}{nodelayer}
		\node [style=rw] (0) at (-0.5, 0) {};
		\node [style=none] (2) at (1, 0) {};
		\node [style=none] (3) at (-0.5, 1) {$\pmm{Z & C \\ -C^T & Z'}$};
		\node [style=modenb] (4) at (0.5, 0.25) {$n+n'$};
	\end{pgfonlayer}
	\begin{pgfonlayer}{edgelayer}
		\draw [style=thick] (2.center) to (0);
	\end{pgfonlayer}
\end{tikzpicture}
} &
    $\eq\limits_{\D4}$ &
    \scalebox{\scalefig}{\input{./figures/SEpW-lemma-proof/triangle-dec_2.tikz}} &
    $\eq\limits_{\D2,\D3}$ &
    \scalebox{\scalefig}{\input{./figures/SEpW-lemma-proof/triangle-dec_3.tikz}} &
    $\eq\limits_{\Dist}$ &
    \scalebox{\scalefig}{\input{./figures/SEpW-lemma-proof/triangle-dec_4.tikz}} \\[1em]
    \scalebox{\scalefig}{\input{./figures/SEpW-lemma-proof/triangle-dec_4.tikz}} &
    $\eq\limits_{\Exp}$ &
    \scalebox{\scalefig}{\input{./figures/SEpW-lemma-proof/triangle-dec_5.tikz}} &
    $\eq\limits_{\S \A1}$ &
    \scalebox{\scalefig}{\input{./figures/SEpW-lemma-proof/triangle-dec_6.tikz}} &
    $\eq\limits_{\S \A13}$ &
    \scalebox{\scalefig}{\input{./figures/SEpW-lemma-proof/triangle-dec_7.tikz}} \\[1em]
    \scalebox{\scalefig}{\input{./figures/SEpW-lemma-proof/triangle-dec_7.tikz}} &
    $\eq\limits_{\S \A16, \S \A1}$ &
    \scalebox{\scalefig}{\input{./figures/SEpW-lemma-proof/triangle-dec_8.tikz}} &
    $\eq\limits_{\S \A8}$ &
    \scalebox{\scalefig}{\input{./figures/SEpW-lemma-proof/triangle-dec_9.tikz}} &
    $\eq\limits_{\L12, \L3}$ &
    \scalebox{\scalefig}{\input{./figures/SEpW-lemma-proof/triangle-dec_10.tikz}} \\[1em]
\end{derivationchain}
                    \begin{derivationchain}
    \scalebox{\scalefig}{\input{./figures/SEpW-lemma-proof/triangle-dec_10.tikz}} &
    $\eq\limits_{\S \A1}$ &
    \scalebox{\scalefig}{\input{./figures/SEpW-lemma-proof/triangle-dec_11.tikz}} &
    $\eq\limits_{\S \A3}$ &
    \scalebox{\scalefig}{\input{./figures/SEpW-lemma-proof/triangle-dec_12.tikz}} &
    $\eq\limits_{\S \A16, \L1}$ &
    \scalebox{\scalefig}{\input{./figures/SEpW-lemma-proof/triangle-dec_13.tikz}} \\[1em]
\end{derivationchain}
                \end{proof}
\begin{proof}[Proof of \L14]
                    \phantomsection\label{proof:SEpW-triangle-mat}
                    \begin{derivationchain}
    \scalebox{\scalefig}{\input{./figures/SEpW-lemma-proof/triangle-mat_1.tikz}} &
    $\eq\limits_{\D4}$ &
    \scalebox{\scalefig}{\input{./figures/SEpW-lemma-proof/triangle-mat_2.tikz}} &
    $\eq\limits_{\L2}$ &
    \scalebox{\scalefig}{\input{./figures/SEpW-lemma-proof/triangle-mat_3.tikz}} &
    $\eq\limits_{\L12}$ &
    \scalebox{\scalefig}{\input{./figures/SEpW-lemma-proof/triangle-mat_4.tikz}} \\[1em]
    \scalebox{\scalefig}{\input{./figures/SEpW-lemma-proof/triangle-mat_4.tikz}} &
    $\eq\limits_{\L3}$ &
    \scalebox{\scalefig}{\input{./figures/SEpW-lemma-proof/triangle-mat_5.tikz}} &
    $\eq\limits_{\D4}$ &
    \scalebox{\scalefig}{\begin{tikzpicture}
	\begin{pgfonlayer}{nodelayer}
		\node [style=rw] (13) at (-0.25, 0) {};
		\node [style=none] (21) at (1.25, 0) {};
		\node [style=none] (22) at (-0.25, 0.75) {$MZM^T$};
		\node [style=modenb] (23) at (1, 0.25) {$n$};
	\end{pgfonlayer}
	\begin{pgfonlayer}{edgelayer}
		\draw [style=thick] (13) to (21.center);
	\end{pgfonlayer}
\end{tikzpicture}
} &
    {} &
    {} \\[1em]
\end{derivationchain}
                \end{proof}
\begin{proof}[Proof of \L15]
                \phantomsection\label{proof:SEpW-mat-to-triangle}
                    \begin{derivationchain}
    \scalebox{\scalefig}{\begin{tikzpicture}
	\begin{pgfonlayer}{nodelayer}
		\node [style=none] (1) at (1, 0) {};
		\node [style=modenb] (2) at (0.75, 0.25) {$n$};
		\node [style=none] (3) at (-0.5, 0.75) {$M\shortminus M^T$};
		\node [style=rw] (4) at (-0.5, 0) {};
	\end{pgfonlayer}
	\begin{pgfonlayer}{edgelayer}
		\draw [style=thick] (1.center) to (4);
	\end{pgfonlayer}
\end{tikzpicture}
} &
    $\eq\limits_{\D4}$ &
    \scalebox{\scalefig}{\input{./figures/SEpW-lemma-proof/mat-to-triangle_2.tikz}} &
    $\eq\limits_{\L1}$ &
    \scalebox{\scalefig}{\input{./figures/SEpW-lemma-proof/mat-to-triangle_3.tikz}} &
    $\eq\limits_{\S \A16, \S \A18}$ &
    \scalebox{\scalefig}{\input{./figures/SEpW-lemma-proof/mat-to-triangle_4.tikz}} \\[1em]
    \scalebox{\scalefig}{\input{./figures/SEpW-lemma-proof/mat-to-triangle_4.tikz}} &
    $\eq\limits_{\S \A1, \S \A3}$ &
    \scalebox{\scalefig}{\input{./figures/SEpW-lemma-proof/mat-to-triangle_5.tikz}} &
    $\eq\limits_{\L12}$ &
    \scalebox{\scalefig}{\input{./figures/SEpW-lemma-proof/mat-to-triangle_6.tikz}} &
    $\eq\limits_{\S \A13}$ &
    \scalebox{\scalefig}{\input{./figures/SEpW-lemma-proof/mat-to-triangle_7.tikz}} \\[1em]
    \scalebox{\scalefig}{\input{./figures/SEpW-lemma-proof/mat-to-triangle_7.tikz}} &
    $\eq$ &
    \scalebox{\scalefig}{\input{./figures/SEpW-lemma-proof/mat-to-triangle_8.tikz}} &
    $\eq$ &
    \scalebox{\scalefig}{\input{./figures/SEpW-lemma-proof/mat-to-triangle_9.tikz}} &
    $\eq$ &
    \scalebox{\scalefig}{\input{./figures/SEpW-lemma-proof/mat-to-triangle_10.tikz}} \\[1em]
\end{derivationchain}
                \end{proof}
\begin{proof}[Proof of \L16]
                    \phantomsection\label{proof:SEpW-mult-float}
                        \begin{derivationchain}
    \scalebox{\scalefig}{\input{./figures/SEpW-lemma-proof/mult-float_1.tikz}} &
    $\eq\limits_{\S \A7}$ &
    \scalebox{\scalefig}{\input{./figures/SEpW-lemma-proof/mult-float_2.tikz}} &
    $\eq\limits_{\S \A16, \S \A3}$ &
    \scalebox{\scalefig}{\input{./figures/SEpW-lemma-proof/mult-float_3.tikz}} &
    $\eq\limits_{\S \A13}$ &
    \scalebox{\scalefig}{\input{./figures/SEpW-lemma-proof/mult-float_4.tikz}} \\[1em]
    \scalebox{\scalefig}{\input{./figures/SEpW-lemma-proof/mult-float_4.tikz}} &
    $\eq\limits_{\S \A22}$ &
    \scalebox{\scalefig}{\input{./figures/SEpW-lemma-proof/mult-float_5.tikz}} &
    $\eq\limits_{\S \A9}$ &
    \scalebox{\scalefig}{\input{./figures/SEpW-lemma-proof/mult-float_6.tikz}} &
    $\eq\limits_{\S \A13, \S \A3, \S \A16, \S \A7}$ &
    \scalebox{\scalefig}{\input{./figures/SEpW-lemma-proof/mult-float_7.tikz}} \\[1em]
\end{derivationchain}
                \end{proof}
\begin{proof}[Proof of \L17]
                \phantomsection\label{proof:SEpW-gen-anticom}
                To show \L17, we start by showing that for all $n \in \N$,
                \[ \scalebox{\scalefig}{\input{./figures/SEpW-lemma-proof/gen-anticom_-1.tikz}} \eq\limits_{AUX1}  \scalebox{\scalefig}{\input{./figures/SEpW-lemma-proof/gen-anticom_0.tikz}} \eq\limits_{AUX2}  \scalebox{\scalefig}{\input{./figures/SEpW-lemma-proof/gen-anticom_0bis.tikz}}\]
                We show that by induction on $n$, the base case being exactly equation \A21.
                    \begin{derivationchain}
    \scalebox{\scalefig}{\input{./figures/SEpW-lemma-proof/gen-anticom_1.tikz}} &
    $\eq\limits_{\Dist}$ &
    \scalebox{\scalefig}{\input{./figures/SEpW-lemma-proof/gen-anticom_2.tikz}} &
    $\eq$ &
    \scalebox{\scalefig}{\input{./figures/SEpW-lemma-proof/gen-anticom_3.tikz}} &
    $\eq\limits_{\S \A21}$ &
    \scalebox{\scalefig}{\input{./figures/SEpW-lemma-proof/gen-anticom_4.tikz}} \\[1em]
    \scalebox{\scalefig}{\input{./figures/SEpW-lemma-proof/gen-anticom_4.tikz}} &
    $\eq$ &
    \scalebox{\scalefig}{\input{./figures/SEpW-lemma-proof/gen-anticom_5.tikz}} &
    $\eq\limits_{\Dist}$ &
    \scalebox{\scalefig}{\input{./figures/SEpW-lemma-proof/gen-anticom_6.tikz}} &
    $\eq\limits_{IH}$ &
    \scalebox{\scalefig}{\input{./figures/SEpW-lemma-proof/gen-anticom_7.tikz}} \\[1em]
    \scalebox{\scalefig}{\input{./figures/SEpW-lemma-proof/gen-anticom_7.tikz}} &
    $\eq\limits_{\Exp, \Elim}$ &
    \scalebox{\scalefig}{\input{./figures/SEpW-lemma-proof/gen-anticom_8.tikz}} &
    $\eq\limits_{\Dist}$ &
    \scalebox{\scalefig}{\input{./figures/SEpW-lemma-proof/gen-anticom_9.tikz}} &
    $\eq\limits_{\L16}$ &
    \scalebox{\scalefig}{\input{./figures/SEpW-lemma-proof/gen-anticom_10.tikz}} \\[1em]
\end{derivationchain}
                The second equality is proven similarly, applying rewrites symmetrically with respect to the horizontal direction.\\
                Now that this auxiliary result is proven, let $n,m \in \N$, and $k \leq m$.
                \begin{derivationchain}
    \scalebox{\scalefig}{\input{./figures/SEpW-lemma-proof/gen-anticom_11.tikz}} &
    $\eq\limits_{AUX1}$ &
    \scalebox{\scalefig}{\input{./figures/SEpW-lemma-proof/gen-anticom_12.tikz}} &
    $\eq\limits_{AUX2}$ &
    \scalebox{\scalefig}{\input{./figures/SEpW-lemma-proof/gen-anticom_13.tikz}} &
    $\eq\limits_{\L16}$ &
    \scalebox{\scalefig}{\input{./figures/SEpW-lemma-proof/gen-anticom_14.tikz}} \\[1em]
\end{derivationchain}
                Repeating this operation $m$ times let go from the case $(n+m, 0)$ to the case $(n,m)$. Using \L16,
                the global phase obtained is 
                \[(-1)^{\left(\sum\limits_{k=1}^m n+m-1\right)} = (-1)^{m(n+m+1)} = (-1)^{nm}\]
                Which suffices to prove \L17.
            \end{proof}
\textbf{Auxiliary linear algebra for \L18 and \L24.}
We use the following standard facts about elementary matrices.
\begin{definition}[Elementary matrix]
A \emph{transvection matrix} is a matrix
\[
    T_{ij}(\lambda)=I_n+\lambda E_{ij},
    \qquad i\neq j,\quad \lambda\in\mathbb{C}.
\]
A \emph{dilation matrix} is a matrix
\[
    D_i(\alpha)=I_n+(\alpha-1)E_{ii},
    \qquad \alpha\in\mathbb{C}^{\ast}.
\]
An elementary matrix is either a transvection matrix or a dilation matrix.
\end{definition}
\begin{lemma}\label{lemma:elementary-small-blocks}
Every elementary matrix can be written using products, direct sums, and
permutation conjugations of the matrices
\[
    \begin{bmatrix}0&1\\[2pt]1&0\end{bmatrix},
    \qquad
    \begin{bmatrix}1&0\\[2pt]\lambda&1\end{bmatrix}
    \quad(\lambda\in\mathbb{C}),
    \qquad
    \begin{bmatrix}x\end{bmatrix}
    \quad(x\in\mathbb{C}^{\ast}).
\]
If the ambient dimension is even, the last family may equivalently be replaced by
\[
    \begin{bmatrix}x&0\\0&1\end{bmatrix}
    \qquad(x\in\mathbb{C}^{\ast}),
\]
so that all direct-sum factors may be chosen of even size.
\end{lemma}
\begin{proof}
A dilation
\[
    D_i(\alpha)=\operatorname{diag}(1,\ldots,1,\alpha,1,\ldots,1)
\]
is obtained, up to a permutation of the coordinates, by taking the direct sum
of the one-dimensional blocks $\begin{bsmallmatrix}\alpha\end{bsmallmatrix}$
and $\begin{bsmallmatrix}1\end{bsmallmatrix}$.
Likewise, every transvection is permutation-conjugate to a direct sum of
\[
    \begin{bmatrix}1&0\\[2pt]\lambda&1\end{bmatrix}
\]
with identity blocks. Permutation matrices are generated by adjacent
transpositions, themselves obtained by taking direct sums of
$\begin{bsmallmatrix}0&1\\1&0\end{bsmallmatrix}$ with identities.
In even ambient dimension, a dilation is permutation-conjugate to a direct sum
containing one block
$\begin{bsmallmatrix}\alpha&0\\0&1\end{bsmallmatrix}$ and otherwise
$2\times2$ identity blocks. The same pairing of the remaining identity
coordinates gives the claimed even-size decomposition for transvections.
\end{proof}
\begin{lemma}\label{lemma:GL-generated-elementary}
The group $\mathcal{GL}_n(\mathbb{C})$ is generated by elementary matrices.
\end{lemma}
\begin{proof}
Gaussian elimination reduces any $M\in\mathcal{GL}_n(\mathbb{C})$ to the
identity using row additions, non-zero row scalings, and row swaps.
Row additions and non-zero row scalings are left multiplications by
transvections and dilations. A swap of two adjacent coordinates is also a
product of such matrices; on the corresponding two-dimensional block,
\[
    \begin{bmatrix}0&1\\1&0\end{bmatrix}
    =
    T_{12}(1)\,T_{21}(-1)\,T_{12}(1)\,D_1(-1).
\]
Hence there exist elementary matrices $E_1,\ldots,E_r$ such that
$E_r\cdots E_1M=I_n$. Since
$T_{ij}(\lambda)^{-1}=T_{ij}(-\lambda)$ and
$D_i(\alpha)^{-1}=D_i(\alpha^{-1})$, the inverses are again elementary,
and therefore $M$ is a product of elementary matrices.
\end{proof}
By \cref{lemma:elementary-small-blocks,lemma:GL-generated-elementary},
together with \L6 and \L3, it is enough in \L18 and \L24 to establish the
displayed small-matrix cases and to check stability under direct sums and
products.
\begin{proof}[Proof of \L18]
    \phantomsection\label{proof:SEpW-det-mat}
    \textbf{Base case:} $\pmm{0&1\\1&0}$\\
    \begin{derivationchain}
    \scalebox{\scalefig}{\input{./figures/SEpW-lemma-proof/det-mat_1.tikz}} &
    $\eq\limits_{\D2,\D3,\D1,\S \A12}$ &
    \scalebox{\scalefig}{\input{./figures/SEpW-lemma-proof/det-mat_2.tikz}} &
    $\eq\limits_{\S \A10}$ &
    \scalebox{\scalefig}{\input{./figures/SEpW-lemma-proof/det-mat_3.tikz}} &
    $\eq\limits_{\Exp,\S \A6}$ &
    \scalebox{\scalefig}{\input{./figures/SEpW-lemma-proof/det-mat_4.tikz}} \\[1em]
    \scalebox{\scalefig}{\input{./figures/SEpW-lemma-proof/det-mat_4.tikz}} &
    $\eq\limits_{\S \A13}$ &
    \scalebox{\scalefig}{\input{./figures/SEpW-lemma-proof/det-mat_5.tikz}} &
    $\eq\limits_{\S \A8}$ &
    \scalebox{\scalefig}{\input{./figures/SEpW-lemma-proof/det-mat_6.tikz}} &
    $\eq\limits_{\S \A13}$ &
    \scalebox{\scalefig}{\input{./figures/SEpW-lemma-proof/det-mat_7.tikz}} \\[1em]
    \scalebox{\scalefig}{\input{./figures/SEpW-lemma-proof/det-mat_7.tikz}} &
    $\eq\limits_{\S \A22   }$ &
    \scalebox{\scalefig}{\input{./figures/SEpW-lemma-proof/det-mat_8.tikz}} &
    {} &
    {} &
    {} &
    {} \\[1em]
\end{derivationchain}
    \textbf{Base case:} $\pmm{1&0\\\lambda&1}$\\
    \begin{derivationchain}
    \scalebox{\scalefig}{\input{./figures/SEpW-lemma-proof/det-mat_9.tikz}} &
    $\eq\limits_{\D2,\D3,\D1,\S \A12, \S \A10}$ &
    \scalebox{\scalefig}{\input{./figures/SEpW-lemma-proof/det-mat_10.tikz}} &
    $\eq\limits_{\S \A6, \Exp}$ &
    \scalebox{\scalefig}{\input{./figures/SEpW-lemma-proof/det-mat_11.tikz}} &
    $\eq\limits_{\S \A18}$ &
    \scalebox{\scalefig}{\input{./figures/SEpW-lemma-proof/det-mat_12.tikz}} \\[1em]
    \scalebox{\scalefig}{\input{./figures/SEpW-lemma-proof/det-mat_12.tikz}} &
    $\eq\limits_{\S \A3}$ &
    \scalebox{\scalefig}{\input{./figures/SEpW-lemma-proof/det-mat_13.tikz}} &
    $\eq\limits_{\S \A24}$ &
    \scalebox{\scalefig}{\input{./figures/SEpW-lemma-proof/det-mat_14.tikz}} &
    $\eq\limits_{\S \A5, \S \A18}$ &
    \scalebox{\scalefig}{\input{./figures/SEpW-lemma-proof/det-mat_15.tikz}} \\[1em]
    \scalebox{\scalefig}{\input{./figures/SEpW-lemma-proof/det-mat_15.tikz}} &
    $\eq\limits_{\S \A6, \S \A17}$ &
    \scalebox{\scalefig}{\input{./figures/SEpW-lemma-proof/det-mat_16.tikz}} &
    {} &
    {} &
    {} &
    {} \\[1em]
\end{derivationchain}
    \textbf{Base case:} $\pmm{x & 0\\0 & 1}$\\
    \begin{derivationchain}
    \scalebox{\scalefig}{\input{./figures/SEpW-lemma-proof/det-mat_17.tikz}} &
    $\eq\limits_{\D2, \D3, \D1, \S \A10}$ &
    \scalebox{\scalefig}{\input{./figures/SEpW-lemma-proof/det-mat_18.tikz}} &
    $\eq\limits_{\S \A6, \S \A12}$ &
    \scalebox{\scalefig}{\input{./figures/SEpW-lemma-proof/det-mat_19.tikz}} &
    $\eq\limits_{\S \A22}$ &
    \scalebox{\scalefig}{\input{./figures/SEpW-lemma-proof/det-mat_19-2.tikz}} \\[1em]
\end{derivationchain}
     \textbf{Induction step:} $M\oplus M'$\\
    \begin{derivationchain}
    \scalebox{\scalefig}{\input{./figures/SEpW-lemma-proof/det-mat_20.tikz}} &
    $\eq\limits_{\L6}$ &
    \scalebox{\scalefig}{\input{./figures/SEpW-lemma-proof/det-mat_21.tikz}} &
    $\eq\limits_{\L17, \Exp}$ &
    \scalebox{\scalefig}{\input{./figures/SEpW-lemma-proof/det-mat_22.tikz}} &
    $\eq\limits_{IH}$ &
    \scalebox{\scalefig}{\input{./figures/SEpW-lemma-proof/det-mat_23.tikz}} \\[1em]
    \scalebox{\scalefig}{\input{./figures/SEpW-lemma-proof/det-mat_23.tikz}} &
    $\eq\limits_{\L17,\L16}$ &
    \scalebox{\scalefig}{\input{./figures/SEpW-lemma-proof/det-mat_24.tikz}} &
    {} &
    {} &
    {} &
    {} \\[1em]
\end{derivationchain}
     \textbf{Induction step:} $M' M $\\
    \begin{derivationchain}
    \scalebox{\scalefig}{\input{./figures/SEpW-lemma-proof/det-mat_25.tikz}} &
    $\eq\limits_{\L3}$ &
    \scalebox{\scalefig}{\input{./figures/SEpW-lemma-proof/det-mat_26.tikz}} &
    $\eq\limits_{IH}$ &
    \scalebox{\scalefig}{\input{./figures/SEpW-lemma-proof/det-mat_27.tikz}} &
    $\eq\limits_{\L16}$ &
    \scalebox{\scalefig}{\input{./figures/SEpW-lemma-proof/det-mat_28.tikz}} \\[1em]
\end{derivationchain}
    \end{proof}
\begin{proof}[Proof of \L20]
        \phantomsection\label{proof:SEpW-triangle-zero}
            \begin{derivationchain}
    \scalebox{\scalefig}{\begin{tikzpicture}
	\begin{pgfonlayer}{nodelayer}
		\node [style=rw] (13) at (-0.75, 0) {};
		\node [style=none] (22) at (-0.75, 0.75) {$0$};
		\node [style=modenb] (24) at (0, 0.25) {$n$};
		\node [style=none] (25) at (0.75, 0) {};
	\end{pgfonlayer}
	\begin{pgfonlayer}{edgelayer}
		\draw [style=thick] (25.center) to (13);
	\end{pgfonlayer}
\end{tikzpicture}
} &
    $\eq\limits_{\D4}$ &
    \scalebox{\scalefig}{\input{./figures/SEpW-lemma-proof/triangle-zero_2.tikz}} &
    $\eq\limits_{\L5}$ &
    \scalebox{\scalefig}{\input{./figures/SEpW-lemma-proof/triangle-zero_3.tikz}} &
    $\eq\limits_{\S \A19}$ &
    \scalebox{\scalefig}{\input{./figures/SEpW-lemma-proof/triangle-zero_4.tikz}} \\[1em]
    \scalebox{\scalefig}{\input{./figures/SEpW-lemma-proof/triangle-zero_4.tikz}} &
    $\eq\limits_{\S \A6}$ &
    \scalebox{\scalefig}{\begin{tikzpicture}
	\begin{pgfonlayer}{nodelayer}
		\node [style=modenb] (24) at (0, 0.25) {$n$};
		\node [style=none] (25) at (0.75, 0) {};
		\node [style=rw] (27) at (-0.75, 0) {};
	\end{pgfonlayer}
	\begin{pgfonlayer}{edgelayer}
		\draw [style=thick] (25.center) to (27);
	\end{pgfonlayer}
\end{tikzpicture}
} &
    {} &
    {} &
    {} &
    {} \\[1em]
\end{derivationchain}
    \end{proof}
\textbf{Auxiliary Darboux decomposition for \L21.}
\begin{lemma}[Darboux]\label{lemma:Darboux}
Let $A \in \mathbb{C}^{2n \times 2n}$ be an invertible antisymmetric matrix,
i.e.\ $A^T=-A$. Then there exists an invertible matrix
$P \in \mathbb{C}^{2n \times 2n}$ such that
\[
    A = P^T J P,
    \qquad
    J =
    \begin{pmatrix}
        0 & I_n \\
        -I_n & 0
    \end{pmatrix}.
\]
Moreover,
\[
    \det(P)=(-1)^{\frac{n(n-1)}{2}}Pf(A),
\]
and
\[
    A^{-1}=P^{-1}(-J)P^{-T}.
\]
\end{lemma}
\begin{proof}
Consider the non-degenerate antisymmetric bilinear form
\[
    \omega(x,y)=x^TAy
\]
on $\mathbb{C}^{2n}$. Choose $e_1\neq0$. By non-degeneracy, there exists
$f_1$ such that $\omega(e_1,f_1)=1$. The restriction of $\omega$ to
$\operatorname{span}(e_1,f_1)$ is non-degenerate, hence
\[
    \mathbb{C}^{2n}
    =
    \operatorname{span}(e_1,f_1)
    \oplus
    \operatorname{span}(e_1,f_1)^{\perp_\omega},
\]
and the restriction of $\omega$ to the orthogonal complement is again
non-degenerate. Iterating yields a basis
\[
    (e_1,\ldots,e_n,f_1,\ldots,f_n)
\]
such that
\[
    \omega(e_i,e_j)=\omega(f_i,f_j)=0,
    \qquad
    \omega(e_i,f_j)=\delta_{ij}.
\]
Therefore, for a suitable $P\in GL_{2n}(\mathbb{C})$,
\[
    A=P^TJP.
\]
Using
\[
    Pf(Q^TBQ)=\det(Q)Pf(B)
\]
(which can be deduced from \L14), 
and
\[
    Pf(J)=(-1)^{\frac{n(n-1)}{2}},
\]
we obtain
\[
    Pf(A)
    =
    \det(P)Pf(J)
    =
    (-1)^{\frac{n(n-1)}{2}}\det(P),
\]
hence
\[
    \det(P)
    =
    (-1)^{\frac{n(n-1)}{2}}Pf(A).
\]
Finally, since $J^2=-I_{2n}$, we have $J^{-1}=-J$, and therefore
\[
    A^{-1}
    =(P^TJP)^{-1}
    =P^{-1}J^{-1}P^{-T}
    =P^{-1}(-J)P^{-T}.
\]
\end{proof}
\begin{proof}[Proof of \L21] \phantomsection\label{proof:SEpW-invert-triangle}
        Let $A \in \mathbb{C}^{2n \times 2n}$ be an invertible antisymmetric matrix.
        Using \cref{lemma:Darboux}, we get
        \[
            A = P^T J P,
            \qquad
            J =
            \begin{pmatrix}
                0 & I_n \\
                -I_n & 0
            \end{pmatrix}.
        \]
        For some $P \in \mathbb{C}^{2n \times 2n}$ invertible. Then:\\
        \smallskip
            \begin{derivationchain}
    \scalebox{\scalefig}{\input{./figures/SEpW-lemma-proof/invert-triangle_1.tikz}} &
    $\eq\limits_{\L14}$ &
    \scalebox{\scalefig}{\input{./figures/SEpW-lemma-proof/invert-triangle_2.tikz}} &
    $\eq\limits_{\L13}$ &
    \scalebox{\scalefig}{\input{./figures/SEpW-lemma-proof/invert-triangle_3.tikz}} &
    $\eq\limits_{\L8, \L20}$ &
    \scalebox{\scalefig}{\input{./figures/SEpW-lemma-proof/invert-triangle_4.tikz}} \\[1em]
    \scalebox{\scalefig}{\input{./figures/SEpW-lemma-proof/invert-triangle_4.tikz}} &
    $\eq\limits_{\S \A6}$ &
    \scalebox{\scalefig}{\input{./figures/SEpW-lemma-proof/invert-triangle_5.tikz}} &
    $\eq\limits_{\L8, \L3}$ &
    \scalebox{\scalefig}{\input{./figures/SEpW-lemma-proof/invert-triangle_6.tikz}} &
    $\eq\limits_{\L2, \L3, \L8}$ &
    \scalebox{\scalefig}{\input{./figures/SEpW-lemma-proof/invert-triangle_7.tikz}} \\[1em]
    \scalebox{\scalefig}{\input{./figures/SEpW-lemma-proof/invert-triangle_7.tikz}} &
    $\eq\limits_{\L18, \L12}$ &
    \scalebox{\scalefig}{\input{./figures/SEpW-lemma-proof/invert-triangle_8.tikz}} &
    $\eq\limits_{\Dist}$ &
    \scalebox{\scalefig}{\input{./figures/SEpW-lemma-proof/invert-triangle_9.tikz}} &
    $\eq\limits_{\S \A3}$ &
    \scalebox{\scalefig}{\input{./figures/SEpW-lemma-proof/invert-triangle_10.tikz}} \\[1em]
    \scalebox{\scalefig}{\input{./figures/SEpW-lemma-proof/invert-triangle_10.tikz}} &
    $\eq\limits_{\S \A1}$ &
    \scalebox{\scalefig}{\input{./figures/SEpW-lemma-proof/invert-triangle_11.tikz}} &
    {} &
    {} &
    {} &
    {} \\[1em]
\end{derivationchain}
            \begin{derivationchain}
    \scalebox{\scalefig}{\input{./figures/SEpW-lemma-proof/invert-triangle_11.tikz}} &
    $\eq\limits_{\S \A13}$ &
    \scalebox{\scalefig}{\input{./figures/SEpW-lemma-proof/invert-triangle_12.tikz}} &
    $\eq\limits_{\S \A7, \S \A10}$ &
    \scalebox{\scalefig}{\input{./figures/SEpW-lemma-proof/invert-triangle_13.tikz}} &
    $\eq\limits_{\S \A6    }$ &
    \scalebox{\scalefig}{\input{./figures/SEpW-lemma-proof/invert-triangle_14.tikz}} \\[1em]
    \scalebox{\scalefig}{\input{./figures/SEpW-lemma-proof/invert-triangle_14.tikz}} &
    $\eq\limits_{\Dist}$ &
    \scalebox{\scalefig}{\input{./figures/SEpW-lemma-proof/invert-triangle_15.tikz}} &
    $\eq\limits_{\S \A13}$ &
    \scalebox{\scalefig}{\input{./figures/SEpW-lemma-proof/invert-triangle_16.tikz}} &
    $\eq\limits_{\S \A6, \L20, \L13}$ &
    \scalebox{\scalefig}{\input{./figures/SEpW-lemma-proof/invert-triangle_17.tikz}} \\[1em]
    \scalebox{\scalefig}{\input{./figures/SEpW-lemma-proof/invert-triangle_17.tikz}} &
    $\eq\limits_{\L14}$ &
    \scalebox{\scalefig}{\input{./figures/SEpW-lemma-proof/invert-triangle_18.tikz}} &
    {} &
    {} &
    {} &
    {} \\[1em]
\end{derivationchain}
    \end{proof}
\begin{proof}[Proof of \L19]
    \phantomsection\label{proof:SEpW-op-gen}
        \begin{derivationchain}
    \scalebox{\scalefig}{\input{./figures/SEpW-lemma-proof/op-gen_1.tikz}} &
    $\eq\limits_{\Dist}$ &
    \scalebox{\scalefig}{\input{./figures/SEpW-lemma-proof/op-gen_2.tikz}} &
    $\eq$ &
    \scalebox{\scalefig}{\input{./figures/SEpW-lemma-proof/op-gen_3.tikz}} &
    $\eq\limits_{\S \A23}$ &
    \scalebox{\scalefig}{\input{./figures/SEpW-lemma-proof/op-gen_4.tikz}} \\[1em]
    \scalebox{\scalefig}{\input{./figures/SEpW-lemma-proof/op-gen_4.tikz}} &
    $\eq\limits_{IH}$ &
    \scalebox{\scalefig}{\input{./figures/SEpW-lemma-proof/op-gen_5.tikz}} &
    $\eq$ &
    \scalebox{\scalefig}{\input{./figures/SEpW-lemma-proof/op-gen_6.tikz}} &
    $\eq\limits_{\Dist}$ &
    \scalebox{\scalefig}{\input{./figures/SEpW-lemma-proof/op-gen_7.tikz}} \\[1em]
\end{derivationchain}
    \end{proof}
\begin{proof}[Proof of \L22]
        \phantomsection\label{proof:SEpW-triangle-cup}
            \begin{derivationchain}
    \scalebox{\scalefig}{\input{./figures/SEpW-lemma-proof/triangle-cup_1.tikz}} &
    $\eq\limits_{\D4}$ &
    \scalebox{\scalefig}{\input{./figures/SEpW-lemma-proof/triangle-cup_2.tikz}} &
    $\eq\limits_{\S \A16, \S \A18}$ &
    \scalebox{\scalefig}{\input{./figures/SEpW-lemma-proof/triangle-cup_3.tikz}} &
    $\eq\limits_{\S \A3}$ &
    \scalebox{\scalefig}{\input{./figures/SEpW-lemma-proof/triangle-cup_4.tikz}} \\[1em]
    \scalebox{\scalefig}{\input{./figures/SEpW-lemma-proof/triangle-cup_4.tikz}} &
    $\eq\limits_{\S \A18}$ &
    \scalebox{\scalefig}{\input{./figures/SEpW-lemma-proof/triangle-cup_5.tikz}} &
    $\eq\limits_{\D4}$ &
    \scalebox{\scalefig}{\begin{tikzpicture}
	\begin{pgfonlayer}{nodelayer}
		\node [style=none] (3) at (1, 0) {};
		\node [style=modenb] (6) at (1, 0.25) {$n$};
		\node [style=none] (8) at (-0.25, 0.75) {$Z$};
		\node [style=none] (17) at (1.5, 0) {};
		\node [style=rw] (21) at (-0.25, 0) {};
	\end{pgfonlayer}
	\begin{pgfonlayer}{edgelayer}
		\draw [style=thick] (17.center) to (21);
	\end{pgfonlayer}
\end{tikzpicture}
} &
    {} &
    {} \\[1em]
\end{derivationchain}
    \end{proof}
\textbf{Auxiliary decomposition for \L23.}
\begin{lemma}\label{lemma:antisymmetric-kernel-decomposition}
Let $Z\in\mathbb{C}^{n\times n}$ be antisymmetric. Then there exist
$k\leq n$, an invertible antisymmetric matrix
$A\in\mathbb{C}^{(n-k)\times(n-k)}$, and an invertible matrix
$P\in GL_n(\mathbb{C})$ such that
\[
    Z=P^T(A\oplus0_k)P,
\]
where $0_k$ denotes the $k\times k$ zero matrix. In particular, $n-k$ is even.
\end{lemma}
\begin{proof}
Consider the antisymmetric bilinear form
\[
    \omega(x,y)=x^TZy
\]
on $\mathbb{C}^n$, and let
\[
    K=\ker Z
     =\{x\in\mathbb{C}^n\mid
       \omega(x,y)=0\text{ for every }y\in\mathbb{C}^n\}.
\]
Set $k=\dim K$, and choose a vector-space complement $V$ of $K$, so that
\[
    \mathbb{C}^n=V\oplus K.
\]
The restriction $\omega|_V$ is non-degenerate. Indeed, if $x\in V$ satisfies
$\omega(x,v)=0$ for every $v\in V$, then also $\omega(x,k)=0$ for every
$k\in K$. Hence $\omega(x,y)=0$ for every $y\in\mathbb{C}^n$, so $x\in K$.
Since $V\cap K=\{0\}$, we obtain $x=0$.
Therefore the matrix $A$ of $\omega|_V$ in any basis of $V$ is antisymmetric
and invertible. Extending this basis by a basis of $K$, the matrix of $\omega$
is $A\oplus0_k$. Hence, after the corresponding change of basis, there exists
$P\in GL_n(\mathbb{C})$ such that
\[
    Z=P^T(A\oplus0_k)P.
\]
Since an invertible antisymmetric matrix has even size, $n-k$ is even.
\end{proof}
\begin{proof}[Proof of \L23]
    \phantomsection\label{proof:SEpW-triangle-simp}
        \begin{derivationchain}
    \scalebox{\scalefig}{\begin{tikzpicture}
	\begin{pgfonlayer}{nodelayer}
		\node [style=rw] (1) at (-1, 0) {};
		\node [style=none] (2) at (1, 0) {};
		\node [style=none] (3) at (-1, 1) {$Z$};
		\node [style=modenb] (4) at (0.5, 0.25) {$n$};
	\end{pgfonlayer}
	\begin{pgfonlayer}{edgelayer}
		\draw [style=thick] (2.center) to (1);
	\end{pgfonlayer}
\end{tikzpicture}
} &
    $\eq\limits_{\L13}$ &
    \scalebox{\scalefig}{\input{./figures/SEpW-lemma-proof/triangle-simp_2.tikz}} &
    $\eq\limits_{\L20,\L5}$ &
    \scalebox{\scalefig}{\input{./figures/SEpW-lemma-proof/triangle-simp_3.tikz}} &
    $\eq\limits_{\S \A6}$ &
    \scalebox{\scalefig}{\input{./figures/SEpW-lemma-proof/triangle-simp_4.tikz}} \\[1em]
\end{derivationchain}
    \end{proof}
\subsubsection{Derived equations in $\S\aW$}
\begin{proof}[Proof of \L25]
                \phantomsection\label{proof:SaW-scaled-antinat}
    The base case $n=1$ is exactly \A29.
                We proceed by induction on $n\in\N^*$.\\[1em]
                \begin{derivationchain}
    \scalebox{\scalefig}{\input{./figures/SaW-lemma-proof/scaled-antinat_1.tikz}} &
    $\eq\limits_{\Dist}$ &
    \scalebox{\scalefig}{\input{./figures/SaW-lemma-proof/scaled-antinat_2.tikz}} &
    $\eq\limits_{IH}$ &
    \scalebox{\scalefig}{\input{./figures/SaW-lemma-proof/scaled-antinat_3.tikz}} &
    $\eq\limits_{\S \A9, \Dist}$ &
    \scalebox{\scalefig}{\input{./figures/SaW-lemma-proof/scaled-antinat_4.tikz}} \\[1em]
\end{derivationchain}
            \end{proof}
\begin{proof}[Proof of \L26]
                \phantomsection\label{proof:SaW-dot-cup}
                We proceed by induction on $n\in\N^*$.\\[1em]
                    \begin{derivationchain}
    \scalebox{\scalefig}{\input{./figures/SaW-lemma-proof/dot-cup_1.tikz}} &
    $\eq\limits_{\Dist}$ &
    \scalebox{\scalefig}{\input{./figures/SaW-lemma-proof/dot-cup_2.tikz}} &
    $\eq$ &
    \scalebox{\scalefig}{\input{./figures/SaW-lemma-proof/dot-cup_2bis.tikz}} &
    $\eq\limits_{\L25}$ &
    \scalebox{\scalefig}{\input{./figures/SaW-lemma-proof/dot-cup_3.tikz}} \\[1em]
    \scalebox{\scalefig}{\input{./figures/SaW-lemma-proof/dot-cup_3.tikz}} &
    $\eq\limits_{\L25}$ &
    \scalebox{\scalefig}{\input{./figures/SaW-lemma-proof/dot-cup_4.tikz}} &
    $\eq\limits_{\S \A25}$ &
    \scalebox{\scalefig}{\input{./figures/SaW-lemma-proof/dot-cup_5.tikz}} &
    $\eq\limits_{IH}$ &
    \scalebox{\scalefig}{\input{./figures/SaW-lemma-proof/dot-cup_6.tikz}} \\[1em]
    \scalebox{\scalefig}{\input{./figures/SaW-lemma-proof/dot-cup_6.tikz}} &
    $\eq\limits_{\S \A9}$ &
    \scalebox{\scalefig}{\input{./figures/SaW-lemma-proof/dot-cup_7.tikz}} &
    $\eq\limits_{\Dist}$ &
    \scalebox{\scalefig}{\input{./figures/SaW-lemma-proof/dot-cup_8.tikz}} &
    {} &
    {} \\[1em]
\end{derivationchain}
            \end{proof}
\begin{proof}[Proof of \L24]
                \phantomsection\label{proof:SaW-det-mat}
                    \textbf{Base case:} $\pmm{0&1\\1&0}$\\
                    \begin{derivationchain}
    \scalebox{\scalefig}{\input{./figures/SaW-lemma-proof/det-mat_1.tikz}} &
    $\eq\limits_{\D2,\D3}$ &
    \scalebox{\scalefig}{\input{./figures/SaW-lemma-proof/det-mat_2.tikz}} &
    $\eq\limits_{\A10}$ &
    \scalebox{\scalefig}{\input{./figures/SaW-lemma-proof/det-mat_3.tikz}} &
    $\eq\limits_{\A6}$ &
    \scalebox{\scalefig}{\input{./figures/SaW-lemma-proof/det-mat_4.tikz}} \\[1em]
    \scalebox{\scalefig}{\input{./figures/SaW-lemma-proof/det-mat_4.tikz}} &
    $\eq\limits_{\Dist}$ &
    \scalebox{\scalefig}{\input{./figures/SaW-lemma-proof/det-mat_5.tikz}} &
    $\eq\limits_{\A29}$ &
    \scalebox{\scalefig}{\input{./figures/SaW-lemma-proof/det-mat_6.tikz}} &
    $\eq\limits_{\A26}$ &
    \scalebox{\scalefig}{\input{./figures/SaW-lemma-proof/det-mat_7.tikz}} \\[1em]
    \scalebox{\scalefig}{\input{./figures/SaW-lemma-proof/det-mat_7.tikz}} &
    $\eq\limits_{\Dist}$ &
    \scalebox{\scalefig}{\input{./figures/SaW-lemma-proof/det-mat_8.tikz}} &
    {} &
    {} &
    {} &
    {} \\[1em]
\end{derivationchain}
                    \textbf{Base case:} $\pmm{1&0\\\lambda&1}$\\
                    \begin{derivationchain}
    \scalebox{\scalefig}{\input{./figures/SaW-lemma-proof/det-mat_9.tikz}} &
    $\eq\limits_{\D2,\D3,\A10}$ &
    \scalebox{\scalefig}{\input{./figures/SaW-lemma-proof/det-mat_10.tikz}} &
    $\eq\limits_{\A6}$ &
    \scalebox{\scalefig}{\input{./figures/SaW-lemma-proof/det-mat_11.tikz}} &
    $\eq\limits_{\Dist}$ &
    \scalebox{\scalefig}{\input{./figures/SaW-lemma-proof/det-mat_12.tikz}} \\[1em]
    \scalebox{\scalefig}{\input{./figures/SaW-lemma-proof/det-mat_12.tikz}} &
    $\eq\limits_{\A28}$ &
    \scalebox{\scalefig}{\input{./figures/SaW-lemma-proof/det-mat_13.tikz}} &
    $\eq\limits_{\A5, \A6}$ &
    \scalebox{\scalefig}{\input{./figures/SaW-lemma-proof/det-mat_14.tikz}} &
    $\eq\limits_{\Dist}$ &
    \scalebox{\scalefig}{\begin{tikzpicture}
	\begin{pgfonlayer}{nodelayer}
		\node [style=none] (10) at (0.75, 0) {};
		\node [style=none] (12) at (2.75, 0) {};
		\node [style=modenb] (21) at (2.5, 0.25) {$2$};
		\node [style=dot] (24) at (0.75, 0) {};
	\end{pgfonlayer}
	\begin{pgfonlayer}{edgelayer}
		\draw [style=thick] (10.center) to (12.center);
	\end{pgfonlayer}
\end{tikzpicture}
} \\[1em]
    \scalebox{\scalefig}{\begin{tikzpicture}
	\begin{pgfonlayer}{nodelayer}
		\node [style=none] (10) at (0.75, 0) {};
		\node [style=none] (12) at (2.75, 0) {};
		\node [style=modenb] (21) at (2.5, 0.25) {$2$};
		\node [style=dot] (24) at (0.75, 0) {};
	\end{pgfonlayer}
	\begin{pgfonlayer}{edgelayer}
		\draw [style=thick] (10.center) to (12.center);
	\end{pgfonlayer}
\end{tikzpicture}
} &
    $\eq\limits_{\S\A12, \S\A26, \Dist, \A26}$ &
    \scalebox{\scalefig}{\input{./figures/SaW-lemma-proof/det-mat_16.tikz}} &
    {} &
    {} &
    {} &
    {} \\[1em]
\end{derivationchain}
                    \textbf{Base case:} $\pmm{x}$\\
                    \begin{derivationchain}
    \scalebox{\scalefig}{\input{./figures/SaW-lemma-proof/det-mat_17.tikz}} &
    $\eq\limits_{\D1}$ &
    \scalebox{\scalefig}{\begin{tikzpicture}
	\begin{pgfonlayer}{nodelayer}
		\node [style=none] (10) at (-1, 0) {};
		\node [style=none] (12) at (1.5, 0) {};
		\node [style=modenb] (21) at (1.25, 0.25) {$1$};
		\node [style=dot] (24) at (-1, 0) {};
		\node [style=white phase dot] (25) at (0.25, 0) {$x$};
	\end{pgfonlayer}
	\begin{pgfonlayer}{edgelayer}
		\draw (24) to (12.center);
	\end{pgfonlayer}
\end{tikzpicture}
} &
    $\eq\limits_{\A26}$ &
    \scalebox{\scalefig}{\input{./figures/SaW-lemma-proof/det-mat_19.tikz}} &
    {} &
    {} \\[1em]
\end{derivationchain}
                    \textbf{Induction step:} $M\oplus M'$\\
                    \begin{derivationchain}
    \scalebox{\scalefig}{\input{./figures/SaW-lemma-proof/det-mat_20.tikz}} &
    $\eq\limits_{\L6}$ &
    \scalebox{\scalefig}{\input{./figures/SaW-lemma-proof/det-mat_21.tikz}} &
    $\eq\limits_{\Dist}$ &
    \scalebox{\scalefig}{\input{./figures/SaW-lemma-proof/det-mat_22.tikz}} &
    $\eq\limits_{IH}$ &
    \scalebox{\scalefig}{\input{./figures/SaW-lemma-proof/det-mat_23.tikz}} \\[1em]
    \scalebox{\scalefig}{\input{./figures/SaW-lemma-proof/det-mat_23.tikz}} &
    $\eq\limits_{\A26, \A9, \Dist}$ &
    \scalebox{\scalefig}{\input{./figures/SaW-lemma-proof/det-mat_24.tikz}} &
    {} &
    {} &
    {} &
    {} \\[1em]
\end{derivationchain}
                    \textbf{Induction step:} $M' M $\\
                    \begin{derivationchain}
    \scalebox{\scalefig}{\input{./figures/SaW-lemma-proof/det-mat_25.tikz}} &
    $\eq\limits_{\L3}$ &
    \scalebox{\scalefig}{\input{./figures/SaW-lemma-proof/det-mat_26.tikz}} &
    $\eq\limits_{IH}$ &
    \scalebox{\scalefig}{\input{./figures/SaW-lemma-proof/det-mat_27.tikz}} &
    $\eq\limits_{IH, \A26, \A9}$ &
    \scalebox{\scalefig}{\input{./figures/SaW-lemma-proof/det-mat_28.tikz}} \\[1em]
\end{derivationchain}
            \end{proof}
\begin{proof}[Proof of \L27]
                \phantomsection\label{proof:SaW-op}
                    \begin{derivationchain}
    \scalebox{\scalefig}{\input{./figures/SaW-lemma-proof/op_1.tikz}} &
    $\eq\limits_{\Dist}$ &
    \scalebox{\scalefig}{\input{./figures/SaW-lemma-proof/op_2.tikz}} &
    $\eq\limits_{\L25}$ &
    \scalebox{\scalefig}{\input{./figures/SaW-lemma-proof/op_3.tikz}} &
    $\eq\limits_{\A28,IH}$ &
    \scalebox{\scalefig}{\input{./figures/SaW-lemma-proof/op_4.tikz}} \\[1em]
    \scalebox{\scalefig}{\input{./figures/SaW-lemma-proof/op_4.tikz}} &
    $\eq\limits_{\S \A5}$ &
    \scalebox{\scalefig}{\input{./figures/SaW-lemma-proof/op_5.tikz}} &
    {} &
    {} &
    {} &
    {} \\[1em]
\end{derivationchain}
            \end{proof}
\begin{proof}[Proof of \L28]
                \phantomsection\label{proof:SaW-zero-erase}
                    \begin{derivationchain}
    \scalebox{\scalefig}{\input{./figures/SaW-lemma-proof/zero-erase_1.tikz}} &
    $\eq\limits_{\Dist}$ &
    \scalebox{\scalefig}{\input{./figures/SaW-lemma-proof/zero-erase_2.tikz}} &
    $\eq\limits_{\A27}$ &
    \scalebox{\scalefig}{\input{./figures/SaW-lemma-proof/zero-erase_3.tikz}} &
    $\eq$ &
    \scalebox{\scalefig}{\input{./figures/SaW-lemma-proof/zero-erase_4.tikz}} \\[1em]
    \scalebox{\scalefig}{\input{./figures/SaW-lemma-proof/zero-erase_4.tikz}} &
    $\eq\limits_{IH}$ &
    \scalebox{\scalefig}{\input{./figures/SaW-lemma-proof/zero-erase_5.tikz}} &
    $\eq\limits_{\Dist}$ &
    \scalebox{\scalefig}{\input{./figures/SaW-lemma-proof/zero-erase_6.tikz}} &
    {} &
    {} \\[1em]
\end{derivationchain}
            \end{proof}
\subsubsection{Derived equations in $\S\FpW$}
\begin{proof}[Proof of \L29]
                \phantomsection\label{proof:SFpW-antiflex}
                    \begin{derivationchain}
    \scalebox{\scalefig}{\input{./figures/SFpW-lemma-proof/antiflex_1.tikz}} &
    $\eq\limits_{\Dist}$ &
    \scalebox{\scalefig}{\input{./figures/SFpW-lemma-proof/antiflex_2.tikz}} &
    $\eq\limits_{IH}$ &
    \scalebox{\scalefig}{\input{./figures/SFpW-lemma-proof/antiflex_3.tikz}} &
    $\eq\limits_{\A34}$ &
    \scalebox{\scalefig}{\input{./figures/SFpW-lemma-proof/antiflex_4.tikz}} \\[1em]
    \scalebox{\scalefig}{\input{./figures/SFpW-lemma-proof/antiflex_4.tikz}} &
    $\eq\limits_{\A6}$ &
    \scalebox{\scalefig}{\input{./figures/SFpW-lemma-proof/antiflex_4-1.tikz}} &
    $\eq\limits_{\A34}$ &
    \scalebox{\scalefig}{\input{./figures/SFpW-lemma-proof/antiflex_4-2.tikz}} &
    $\eq\limits_{\A36}$ &
    \scalebox{\scalefig}{\input{./figures/SFpW-lemma-proof/antiflex_4-3.tikz}} \\[1em]
    \scalebox{\scalefig}{\input{./figures/SFpW-lemma-proof/antiflex_4-3.tikz}} &
    $\eq\limits_{\L25}$ &
    \scalebox{\scalefig}{\input{./figures/SFpW-lemma-proof/antiflex_4-4.tikz}} &
    $\eq\limits_{\A36,\A34,\A6}$ &
    \scalebox{\scalefig}{\input{./figures/SFpW-lemma-proof/antiflex_5.tikz}} &
    $\eq\limits_{\Dist}$ &
    \scalebox{\scalefig}{\input{./figures/SFpW-lemma-proof/antiflex_6.tikz}} \\[1em]
    \scalebox{\scalefig}{\input{./figures/SFpW-lemma-proof/antiflex_6.tikz}} &
    $\eq\limits_{\L6}$ &
    \scalebox{\scalefig}{\input{./figures/SFpW-lemma-proof/antiflex_7.tikz}} &
    $\eq$ &
    \scalebox{\scalefig}{\input{./figures/SFpW-lemma-proof/antiflex_8.tikz}} &
    {} &
    {} \\[1em]
\end{derivationchain}
            \end{proof}
\begin{proof}[Proof of \L30]
                \phantomsection\label{proof:SFpW-unconnect}
                    \begin{derivationchain}
    \scalebox{\scalefig}{\input{./figures/SFpW-lemma-proof/unconnect_1.tikz}} &
    $\eq\limits_{\L17}$ &
    \scalebox{\scalefig}{\input{./figures/SFpW-lemma-proof/unconnect_2.tikz}} &
    $\eq\limits_{\A35}$ &
    \scalebox{\scalefig}{\input{./figures/SFpW-lemma-proof/unconnect_2-1.tikz}} &
    $\eq\limits_{\A32}$ &
    \scalebox{\scalefig}{\input{./figures/SFpW-lemma-proof/unconnect_2-2.tikz}} \\[1em]
    \scalebox{\scalefig}{\input{./figures/SFpW-lemma-proof/unconnect_2-2.tikz}} &
    $\eq\limits_{\A34}$ &
    \scalebox{\scalefig}{\input{./figures/SFpW-lemma-proof/unconnect_2-3.tikz}} &
    $\eq\limits_{\A13,\A18}$ &
    \scalebox{\scalefig}{\input{./figures/SFpW-lemma-proof/unconnect_2-4.tikz}} &
    $\eq\limits_{\A7,\A10,\A19}$ &
    \scalebox{\scalefig}{\input{./figures/SFpW-lemma-proof/unconnect_2-5.tikz}} \\[1em]
    \scalebox{\scalefig}{\input{./figures/SFpW-lemma-proof/unconnect_2-5.tikz}} &
    $\eq\limits_{\A36}$ &
    \scalebox{\scalefig}{\input{./figures/SFpW-lemma-proof/unconnect_3.tikz}} &
    $\eq\limits_{IH}$ &
    \scalebox{\scalefig}{\input{./figures/SFpW-lemma-proof/unconnect_4.tikz}} &
    $\eq\limits_{\L16}$ &
    \scalebox{\scalefig}{\input{./figures/SFpW-lemma-proof/unconnect_5.tikz}} \\[1em]
    \scalebox{\scalefig}{\input{./figures/SFpW-lemma-proof/unconnect_5.tikz}} &
    $\eq\limits_{\Dist}$ &
    \scalebox{\scalefig}{\input{./figures/SFpW-lemma-proof/unconnect_6.tikz}} &
    {} &
    {} &
    {} &
    {} \\[1em]
\end{derivationchain}
            \end{proof}
\begin{proof}[Proof of \L31]
                \phantomsection\label{proof:SFpW-w-2-to-1}
                    \begin{derivationchain}
    \scalebox{\scalefig}{\input{./figures/SFpW-lemma-proof/w-2-to-1_1.tikz}} &
    $\eq\limits_{\S \A36}$ &
    \scalebox{\scalefig}{\input{./figures/SFpW-lemma-proof/w-2-to-1_2.tikz}} &
    $\eq\limits_{\L29}$ &
    \scalebox{\scalefig}{\input{./figures/SFpW-lemma-proof/w-2-to-1_3.tikz}} &
    $\eq\limits_{\S \A6}$ &
    \scalebox{\scalefig}{\begin{tikzpicture}
	\begin{pgfonlayer}{nodelayer}
		\node [style=none] (1) at (0, 0) {};
		\node [style=none] (14) at (3, 0) {};
		\node [style=modenb] (17) at (2.75, 0.25) {$n$};
		\node [style=dot] (22) at (1.5, 0) {};
	\end{pgfonlayer}
	\begin{pgfonlayer}{edgelayer}
		\draw [style=thick] (14.center) to (1.center);
	\end{pgfonlayer}
\end{tikzpicture}
} \\[1em]
\end{derivationchain}
            \end{proof}
\begin{proof}[Proof of \L32]
                \phantomsection\label{proof:SFpW-mat-w-interact}
                    \begin{derivationchain}
    \scalebox{\scalefig}{\input{./figures/SFpW-lemma-proof/mat-w-interact_1.tikz}} &
    $\eq\limits_{\L31}$ &
    \scalebox{\scalefig}{\input{./figures/SFpW-lemma-proof/mat-w-interact_2.tikz}} &
    $\eq\limits_{\L12}$ &
    \scalebox{\scalefig}{\input{./figures/SFpW-lemma-proof/mat-w-interact_3.tikz}} &
    $\eq\limits_{\L8, \L3}$ &
    \scalebox{\scalefig}{\input{./figures/SFpW-lemma-proof/mat-w-interact_4.tikz}} \\[1em]
    \scalebox{\scalefig}{\input{./figures/SFpW-lemma-proof/mat-w-interact_4.tikz}} &
    $\eq\limits_{\L2}$ &
    \scalebox{\scalefig}{\input{./figures/SFpW-lemma-proof/mat-w-interact_5.tikz}} &
    $\eq\limits_{\L24}$ &
    \scalebox{\scalefig}{\input{./figures/SFpW-lemma-proof/mat-w-interact_6.tikz}} &
    $\eq\limits_{\L8,\L3}$ &
    \scalebox{\scalefig}{\input{./figures/SFpW-lemma-proof/mat-w-interact_7.tikz}} \\[1em]
    \scalebox{\scalefig}{\input{./figures/SFpW-lemma-proof/mat-w-interact_7.tikz}} &
    $\eq\limits_{\L2,\L12, \L24}$ &
    \scalebox{\scalefig}{\input{./figures/SFpW-lemma-proof/mat-w-interact_8.tikz}} &
    $\eq\limits_{\L24}$ &
    \scalebox{\scalefig}{\input{./figures/SFpW-lemma-proof/mat-w-interact_9.tikz}} &
    $\eq\limits_{\L2}$ &
    \scalebox{\scalefig}{\input{./figures/SFpW-lemma-proof/mat-w-interact_10.tikz}} \\[1em]
    \scalebox{\scalefig}{\input{./figures/SFpW-lemma-proof/mat-w-interact_10.tikz}} &
    $\eq\limits_{\L31}$ &
    \scalebox{\scalefig}{\input{./figures/SFpW-lemma-proof/mat-w-interact_11.tikz}} &
    {} &
    {} &
    {} &
    {} \\[1em]
\end{derivationchain}
            \end{proof}
\begin{proof}[Proof of \L33]
                \phantomsection\label{proof:SFpW-scaled-antinat}
                    \begin{derivationchain}
    \scalebox{\scalefig}{\input{./figures/SFpW-lemma-proof/scaled-antinat_1.tikz}} &
    $\eq\limits_{\L31}$ &
    \scalebox{\scalefig}{\input{./figures/SFpW-lemma-proof/scaled-antinat_2.tikz}} &
    $\eq$ &
    \scalebox{\scalefig}{\input{./figures/SFpW-lemma-proof/scaled-antinat_3.tikz}} &
    $\eq\limits_{\L25}$ &
    \scalebox{\scalefig}{\input{./figures/SFpW-lemma-proof/scaled-antinat_4.tikz}} \\[1em]
    \scalebox{\scalefig}{\input{./figures/SFpW-lemma-proof/scaled-antinat_4.tikz}} &
    $\eq\limits_{\L31}$ &
    \scalebox{\scalefig}{\input{./figures/SFpW-lemma-proof/scaled-antinat_5.tikz}} &
    {} &
    {} &
    {} &
    {} \\[1em]
\end{derivationchain}
            \end{proof}
\begin{proof}[Proof of \L34]
                \phantomsection\label{proof:SFpW-co-antiflex}
                    \begin{derivationchain}
    \scalebox{\scalefig}{\input{./figures/SFpW-lemma-proof/co-antiflex_1.tikz}} &
    $\eq\limits_{\S \A3}$ &
    \scalebox{\scalefig}{\input{./figures/SFpW-lemma-proof/co-antiflex_1-1.tikz}} &
    $\eq\limits_{\L29}$ &
    \scalebox{\scalefig}{\input{./figures/SFpW-lemma-proof/co-antiflex_2.tikz}} &
    $\eq\limits_{\L33}$ &
    \scalebox{\scalefig}{\input{./figures/SFpW-lemma-proof/co-antiflex_3.tikz}} \\[1em]
    \scalebox{\scalefig}{\input{./figures/SFpW-lemma-proof/co-antiflex_3.tikz}} &
    $\eq$ &
    \scalebox{\scalefig}{\input{./figures/SFpW-lemma-proof/co-antiflex_4.tikz}} &
    $\eq\limits_{\S \A13, \S \A3}$ &
    \scalebox{\scalefig}{\input{./figures/SFpW-lemma-proof/co-antiflex_5.tikz}} &
    $\eq\limits_{\L7, \L3}$ &
    \scalebox{\scalefig}{\input{./figures/SFpW-lemma-proof/co-antiflex_6.tikz}} \\[1em]
\end{derivationchain}
            \end{proof}
\begin{proof}[Proof of \L35]
                \phantomsection\label{proof:SFpW-triangle-invert}
                    \begin{derivationchain}
    \scalebox{\scalefig}{\input{./figures/SFpW-lemma-proof/triangle-invert_1.tikz}} &
    $\eq\limits_{\S \A17, \L16}$ &
    \scalebox{\scalefig}{\input{./figures/SFpW-lemma-proof/triangle-invert_2.tikz}} &
    $\eq\limits_{\L21}$ &
    \scalebox{\scalefig}{\input{./figures/SFpW-lemma-proof/triangle-invert_3.tikz}} &
    $\eq\limits_{\L30}$ &
    \scalebox{\scalefig}{\input{./figures/SFpW-lemma-proof/triangle-invert_4.tikz}} \\[1em]
    \scalebox{\scalefig}{\input{./figures/SFpW-lemma-proof/triangle-invert_4.tikz}} &
    $\eq\limits_{\S \A36}$ &
    \scalebox{\scalefig}{\input{./figures/SFpW-lemma-proof/triangle-invert_5.tikz}} &
    $\eq\limits_{\L34}$ &
    \scalebox{\scalefig}{\input{./figures/SFpW-lemma-proof/triangle-invert_6.tikz}} &
    $\eq\limits_{\L3,\L8,\S \A6, \S \A18}$ &
    \scalebox{\scalefig}{\input{./figures/SFpW-lemma-proof/triangle-invert_7.tikz}} \\[1em]
\end{derivationchain}
            \end{proof}
\begin{proof}[Proof of \L36]
                \phantomsection\label{proof:SFpW-dot-cup}
                    \begin{derivationchain}
    \scalebox{\scalefig}{\input{./figures/SFpW-lemma-proof/dot-cup_1.tikz}} &
    $\eq\limits_{\S\A35}$ &
    \scalebox{\scalefig}{\input{./figures/SFpW-lemma-proof/dot-cup_2.tikz}} &
    $\eq\limits_{\L31}$ &
    \scalebox{\scalefig}{\input{./figures/SFpW-lemma-proof/dot-cup_3.tikz}} &
    $\eq\limits_{\L12, \S\A18}$ &
    \scalebox{\scalefig}{\input{./figures/SFpW-lemma-proof/dot-cup_4.tikz}} \\[1em]
    \scalebox{\scalefig}{\input{./figures/SFpW-lemma-proof/dot-cup_4.tikz}} &
    $\eq\limits_{\L31}$ &
    \scalebox{\scalefig}{\input{./figures/SFpW-lemma-proof/dot-cup_5.tikz}} &
    $\eq\limits_{\L2}$ &
    \scalebox{\scalefig}{\input{./figures/SFpW-lemma-proof/dot-cup_6.tikz}} &
    $\eq\limits_{\Dist}$ &
    \scalebox{\scalefig}{\input{./figures/SFpW-lemma-proof/dot-cup_7.tikz}} \\[1em]
\end{derivationchain}
            \end{proof}

\clearpage
\onecolumn

\section{Normal-form rewrites}
\begin{definition}[Reduced column echelon form]
    Let $M\in\C^{m\times n}$. We say that $M$ is in
    \emph{reduced column echelon form} (RCEF) if $M^T$ is in
    reduced row echelon form.

    Equivalently:
    \begin{enumerate}
        \item all zero columns of $M$, if any, occur after all non-zero
        columns;
        \item the first non-zero entry of each non-zero column is equal
        to $1$. We call this entry the \emph{pivot} of the column, and
        its row the corresponding \emph{pivot row};
        \item the pivot of each non-zero column lies strictly below the
        pivot of the preceding non-zero column;
        \item each pivot is the only non-zero entry in its row.
    \end{enumerate}

    Thus, if $M$ has rank $r$, its $r$ non-zero columns have pivot rows
    \[
        1\leq j_1<\cdots<j_r\leq m.
    \]

    For example,
    \[
        M=
        \begin{bsmallmatrix}
            1 & 0 & 0 \\
            a & 0 & 0 \\
            0 & 1 & 0 \\
            b & c & 0 \\
            0 & 0 & 1
        \end{bsmallmatrix}
    \]
    is in RCEF. Its pivot rows are
    \[
        (j_1,j_2,j_3)=(1,3,5).
    \]
    Indeed,
    \[
        M^T=
        \begin{bsmallmatrix}
            1 & a & 0 & b & 0 \\
            0 & 0 & 1 & c & 0 \\
            0 & 0 & 0 & 0 & 1
        \end{bsmallmatrix}
    \]
    is in reduced row echelon form.
\end{definition}

\begin{lemma}\label{lemma:dec-rcef}
    Let $M\in\C^{n\times r}$ be a matrix in RCEF with
    \[
        \operatorname{rank}(M)=r.
    \]
    Let
    \[
        1\leq j_1<\cdots<j_r\leq n
    \]
    be its pivot rows.

    Then there exist a permutation matrix
    $P\in\C^{n\times n}$ and a matrix
    $B\in\C^{(n-r)\times r}$ such that
    \[
        M
        =
        P
        \begin{bmatrix}
            I_r\\
            B
        \end{bmatrix}.
    \]
\end{lemma}

\begin{proof}
    Since $M$ has rank $r$ and has exactly $r$ columns, every column
    of $M$ is non-zero and hence contains a pivot.

    Let
    \[
        j_1<\cdots<j_r
    \]
    be the pivot rows, and let
    \[
        k_1<\cdots<k_{n-r}
    \]
    be the increasing enumeration of
    \[
        \{1,\ldots,n\}\setminus\{j_1,\ldots,j_r\}.
    \]

    By definition of RCEF, the pivot in the $i$-th column is equal to
    $1$ and is the only non-zero entry of the $j_i$-th row. Therefore
    the submatrix formed by the pivot rows is
    \[
        M_{[j_1,\ldots,j_r]}=I_r.
    \]

    Let
    \[
        B:=M_{[k_1,\ldots,k_{n-r}]}
        \in\C^{(n-r)\times r}
    \]
    be the submatrix formed by the non-pivot rows.

    Let $Q$ be the permutation matrix which moves the pivot rows
    $j_1,\ldots,j_r$ to the first $r$ positions, while preserving their
    order, and moves the remaining rows
    $k_1,\ldots,k_{n-r}$ underneath them, again preserving their order.
    Then
    \[
        QM
        =
        \begin{bmatrix}
            M_{[j_1,\ldots,j_r]}\\
            M_{[k_1,\ldots,k_{n-r}]}
        \end{bmatrix}
        =
        \begin{bmatrix}
            I_r\\
            B
        \end{bmatrix}.
    \]

    Since $Q$ is a permutation matrix,
    \[
        Q^{-1}=Q^T
    \]
    is also a permutation matrix. Setting $P:=Q^{-1}$ gives
    \[
        M
        =
        P
        \begin{bmatrix}
            I_r\\
            B
        \end{bmatrix},
    \]
    as required.
\end{proof}

\subsection{Completeness of $\EaW$}

The generators used below can all be expressed as matrix arrows.  In
particular, the following rewrites reduce the non-structural generators to
matrix-arrow form.

\textbf{To normal form.}

\begin{derivationchain}
    \scalebox{\scalefig}{\input{./figures/aW-scaled-gen/coW-arrow_1.tikz}} &
    $\eq\limits_{\L8,D3}$ &
    \scalebox{\scalefig}{\input{./figures/aW-scaled-gen/coW-arrow_2.tikz}} &
    {} &
    {} &
    {} &
    {} \\[1em]
    \scalebox{\scalefig}{\input{./figures/aW-scaled-gen/W-arrow_1.tikz}} &
    $\eq\limits_{\L8,D2}$ &
    \scalebox{\scalefig}{\input{./figures/aW-scaled-gen/W-arrow_2.tikz}} &
    {} &
    {} &
    {} &
    {} \\[1em]
    \scalebox{\scalefig}{\input{./figures/aW-scaled-gen/fswap-arrow_1.tikz}} &
    $\eq\limits_{\L8,\S\A6,\L5,D2,D3}$ &
    \scalebox{\scalefig}{\input{./figures/aW-scaled-gen/fswap-arrow_2.tikz}} &
    {} &
    {} &
    {} &
    {} \\[1em]
    \scalebox{\scalefig}{\input{./figures/aW-scaled-gen/bin-Z-arrow_1.tikz}} &
    $\eq\limits_{L7}$ &
    \scalebox{\scalefig}{\input{./figures/aW-scaled-gen/bin-Z-arrow_2.tikz}} &
    {} &
    {} &
    {} &
    {} \\[1em]
\end{derivationchain}

Then \L6~ and \L3~ enable to write any diagram as a matrix arrow.

\textbf{Uniqueness of normal form.}
Let $M, M' \in \C^{m\times n}$. Denote $N(M)$ (resp. $N(M')$) the matrix arrow labelled by $M$ (resp. $M'$).
Suppose
\[
\interp{N(M)} = \interp{N(M')}
\]
then, note that
\[
\interp{N(M)}_{2^i, 2^j} = M_{i,j}
\qquad
\text{and}
\qquad
\interp{N(M')}_{2^i, 2^j} = M'_{i,j}
\]
Therefore $M=M'$.

\subsection{Completeness of $\aW$}

\textbf{To normal form.}
Due to \L26, we can rewrite any diagram $D:0\to \underline{n}$:
\[\scalebox{\scalefig}{\begin{tikzpicture}
	\begin{pgfonlayer}{nodelayer}
		\node [style={1_arrow_box}] (0) at (-1, 0) {$D$};
		\node [style=none] (1) at (0.75, 0) {};
		\node [style=modenb] (2) at (0.25, 0.25) {$n$};
	\end{pgfonlayer}
	\begin{pgfonlayer}{edgelayer}
		\draw [style=thick] (0) to (1.center);
	\end{pgfonlayer}
\end{tikzpicture}
} \quad\eq\limits_{\L26}\quad \scalebox{\scalefig}{\input{./figures/aW-proof/pre-pre-nf_2.tikz}}\]
Therefore, we can assume without loss of generality that $D$ contains at least $n$ generators 
$\scalebox{\scalefig}{\begin{tikzpicture}
	\begin{pgfonlayer}{nodelayer}
		\node [style=none] (0) at (0.5, 0) {};
		\node [style=none] (1) at (0, 0) {};
		\node [style=dot] (4) at (-0.5, 0) {};
	\end{pgfonlayer}
	\begin{pgfonlayer}{edgelayer}
		\draw (0.center) to (1.center);
		\draw (1.center) to (4);
	\end{pgfonlayer}
\end{tikzpicture}
}$ and $\scalebox{\scalefig}{\begin{tikzpicture}
	\begin{pgfonlayer}{nodelayer}
		\node [style=none] (0) at (-0.5, 0) {};
		\node [style=none] (1) at (0, 0) {};
		\node [style=dot] (4) at (0.5, 0) {};
	\end{pgfonlayer}
	\begin{pgfonlayer}{edgelayer}
		\draw (0.center) to (1.center);
		\draw (1.center) to (4);
	\end{pgfonlayer}
\end{tikzpicture}
}$, which we can isolate
and gather respectively to the left and right sides of the diagram using \A29 and \Dist, 
rewriting the remaining $\EaW$ part as a matrix arrow, which gives us as starting point the following diagram:

\[\scalebox{\scalefig}{\begin{tikzpicture}
	\begin{pgfonlayer}{nodelayer}
		\node [style={1_arrow_box}] (0) at (-1, 0) {$D$};
		\node [style=none] (1) at (0.75, 0) {};
		\node [style=modenb] (2) at (0.25, 0.25) {$n$};
	\end{pgfonlayer}
	\begin{pgfonlayer}{edgelayer}
		\draw [style=thick] (0) to (1.center);
	\end{pgfonlayer}
\end{tikzpicture}
} \quad\eq\limits_{}\quad \scalebox{\scalefig}{\input{./figures/aW-proof/pre-nf_2.tikz}}\]

with $l \geq n$ and $k \geq n$. Then we derive:

\[\scalebox{\scalefig}{\input{./figures/aW-proof/pre-nf_2.tikz}} \quad\eqann{\D2}{M_1 = \pmm{M_2\\M_3}}\quad \scalebox{\scalefig}{\input{./figures/aW-proof/pre-nf_4.tikz}}\]

\textbf{Case} $r = \mathrm{rank}(M_3) < l$:\\
\begin{derivationchain}
    \scalebox{\scalefig}{\input{./figures/aW-proof/pre-nf_2.tikz}} &
    $\eqann{\D2}{M_1 = \pmm{M_2\\M_3}}$ &
    \scalebox{\scalefig}{\input{./figures/aW-proof/pre-nf_4.tikz}} &
    $\eq\limits_{\L9}$ &
    \scalebox{\scalefig}{\input{./figures/aW-proof/pre-nf_5-case-0.tikz}} &
    $\eq\limits_{\L24, \S\A30}$ &
    \scalebox{\scalefig}{\input{./figures/aW-proof/pre-nf_6-case-0.tikz}} \\[1em]
    \scalebox{\scalefig}{\input{./figures/aW-proof/pre-nf_6-case-0.tikz}} &
    $\eq\limits_{\Dist}$ &
    \scalebox{\scalefig}{\input{./figures/aW-proof/pre-nf_7-case-0.tikz}} &
    $\eq\limits_{\S\A26, \S\A10, \S\A30}$ &
    \scalebox{\scalefig}{\input{./figures/aW-proof/pre-nf_8-case-0.tikz}} &
    $\eq\limits_{\L28, \S\A30, \S\A26}$ &
    \scalebox{\scalefig}{\input{./figures/aW-proof/pre-nf_9-case-0.tikz}} \\[1em]
    \scalebox{\scalefig}{\input{./figures/aW-proof/pre-nf_9-case-0.tikz}} &
    $\eq\limits_{\Dist, \L4, \S\A6}$ &
    \scalebox{\scalefig}{\input{./figures/aW-proof/pre-nf_10-case-0.tikz}} &
    $\eq\limits_{\L28}$ &
    \scalebox{\scalefig}{\input{./figures/aW-proof/pre-nf_11-case-0.tikz}} &
    $\eq\limits_{\S\A30, \S\A26}$ &
    \scalebox{\scalefig}{\input{./figures/aW-proof/pre-nf_12-case-0.tikz}} \\[1em]
\end{derivationchain}
Note that this is the normal form for the zero diagram of type $0 \to n$.\\

\textbf{Case} $\mathrm{rank}(M_3) \geq l$:\\
When $\mathrm{rank}(M_3) \geq l$, we automatically have $\mathrm{rank}M_3 = l$ as 
\[\mathrm{rank}(M_3) \leq \mathrm{min}(l,k)\]
We derive :\\
\begin{derivationchain}
    \scalebox{\scalefig}{\input{./figures/aW-proof/pre-nf_2.tikz}} &
    $\eqann{\D2}{M_1 = \pmm{M_2\\M_3}}$ &
    \scalebox{\scalefig}{\input{./figures/aW-proof/pre-nf_4.tikz}} &
    $\eq\limits_{\L9}$ &
    \scalebox{\scalefig}{\input{./figures/aW-proof/pre-nf_5.tikz}} &
    $\eq\limits_{\L8,\L3,\L2,\L24}$ &
    \scalebox{\scalefig}{\input{./figures/aW-proof/pre-nf_6.tikz}} \\[1em]
    \scalebox{\scalefig}{\input{./figures/aW-proof/pre-nf_6.tikz}} &
    $\eq\limits_{Dist,\L24}$ &
    \scalebox{\scalefig}{\input{./figures/aW-proof/pre-nf_7.tikz}} &
    $\eqann{\L3}{M_4 = M_2Q_1^{-1}}$ &
    \scalebox{\scalefig}{\input{./figures/aW-proof/pre-nf_8.tikz}} &
    $\eqann{Dist,\S\A6}{x_1 = \det P_1 \det Q_1}$ &
    \scalebox{\scalefig}{\input{./figures/aW-proof/pre-nf_11.tikz}} \\[1em]
    \scalebox{\scalefig}{\input{./figures/aW-proof/pre-nf_11.tikz}} &
    $\eq\limits_{\L25}$ &
    \scalebox{\scalefig}{\input{./figures/aW-proof/pre-nf_12.tikz}} &
    $\eq\limits_{L27}$ &
    \scalebox{\scalefig}{\input{./figures/aW-proof/pre-nf_13.tikz}} &
    $\eqann{\S\A26}{M_4 = \pmm{B_1 & M5}}$ &
    \scalebox{\scalefig}{\input{./figures/aW-proof/pre-nf_14.tikz}} \\[1em]
    \scalebox{\scalefig}{\input{./figures/aW-proof/pre-nf_14.tikz}} &
    $\eqann{\S\A30,\S\A26,\S\A9}{x_2 = (-1)^{l(k-l)}x_1}$ &
    \scalebox{\scalefig}{\input{./figures/aW-proof/pre-nf_15.tikz}} &
    {} &
    {} &
    {} &
    {} \\[1em]
\end{derivationchain}

\textbf{Case $k=l$.}
In the case $k=l$, we have $r_2 = 0$. By abuse of notation $M_5$ is then the 
empty matrix $ \pmm{~} \in \C^{n\times 0}$ and we denote:
\[
\scalebox{\scalefig}{\input{./figures/nf-theorem/aW-spec-case_1.tikz}} \quad := \quad \scalebox{\scalefig}{\begin{tikzpicture}
	\begin{pgfonlayer}{nodelayer}
		\node [style=none] (4) at (2, 0) {};
		\node [style=modenb] (8) at (1.75, 0.25) {$n$};
		\node [style=rw] (9) at (0.25, 0) {};
	\end{pgfonlayer}
	\begin{pgfonlayer}{edgelayer}
		\draw [style=thick] (9) to (4.center);
	\end{pgfonlayer}
\end{tikzpicture}
}
\]
By convention, this empty matrix is considered full rank and in RCEF. 
This yields the normal form of states $\lambda \ket{0^n}$ for $\lambda \in \C$.

\textbf{Case $k > l$.}
When $k > l$, we have $r_2 > 0$ and we obtain:
\[
    \scalebox{\scalefig}{\input{./figures/aW-proof/pre-nf_15.tikz}}.
\]
As before, if $\mathrm{rank}(M_5) \neq r_2$, we can apply \L9, \L24, \L28, and \S\A30~
to get back to the normal form for the zero diagram of type $0 \to n$. Otherwise,
we can assume that $\mathrm{rank}(M_5) = r_2 \leq n$. Then
Using \L24, $\S\A26$, and $\S\A9$, we can multiply $M_5$ on the right by a
sequence of elementary matrices so as to put it in Reduced Column Echelon
Form (RCEF), while multiplying $x_2$ by the determinant of these matrices.
This gives
\[
    \scalebox{\scalefig}{\input{./figures/aW-proof/pre-nf_17.tikz}},
\]
with $M_6$ full rank and in RCEF.

\medskip

\textbf{Uniqueness of normal form.}

Let
\[
    N=N(r,x,M)
    \qquad\text{and}\qquad
    N'=N(r',x',M')
\]
be two $\S\aW$ normal forms of the same type $0\to n$, where
\(
    M\in\C^{n\times r},
    M'\in\C^{n\times r'}
\) 
are in RCEF and have ranks $r$ and $r'$ respectively. Let
\(
    1\leq j_1<\cdots<j_r\leq n,
    1\leq j'_1<\cdots<j'_{r'}\leq n
\) 
be their respective pivot rows.

Assume that
\[
    \interp{N}=\interp{N'},
\quad \text{i.e.} \quad
            \interp{\scalebox{\scalefig}{\input{./figures/aW-proof/unique-nf_1.tikz}}}
            \eq
            \interp{\scalebox{\scalefig}{\input{./figures/aW-proof/unique-nf_2.tikz}}}.
\]

Due to \cref{lemma:dec-rcef}, there exists permutation matrices 
$P, P' \in \C^{n\times n}$ such that

\[
    P M
    =
    \begin{bmatrix}
        I_r\\
        B
    \end{bmatrix}
    \qquad 
    \text{and} 
    \qquad
    P'M'
    =
    \begin{bmatrix}
        I_{r'}\\
        B'
    \end{bmatrix}.
\]

For $s\leq n$, we write
\[
    E_s:=
    \begin{bmatrix}
        I_s & 0
    \end{bmatrix}
    \in\C^{s\times n}.
\]
Hence
\[
    E_r P M = I_r,
    \qquad
    E_{r'} P' M' = I_{r'}.
\]

We prove equality of the normal-form data in five steps. The first
step separates the zero case; in the remaining steps we may therefore
assume that the common interpretation is non-zero.

\medskip

\noindent
\(\boxed{\text{\itshape Step 1}:\ 
\interp{N}=\interp{N'}=0
\ \Longrightarrow\ x=x'=0}\)

First, let us show that a normal form has zero interpretation if and
only if its global scalar vanishes.

If $x=0$, then $L28$ rewrites $N$ to
\scalebox{\scalefig}{\input{./figures/aW-proof/zero-nf.tikz}},
the distinguished normal form of the zero $0\to n$ state. Therefore
\[
    x=0
    \quad\Longrightarrow\quad
    \interp{N}=0.
\]

Conversely, suppose that $x\neq0$. Diagrammatically, we get 

\begin{derivationchain}
    \scalebox{\scalefig}{\input{./figures/aW-nf-uniqueness-proof/step1_1.tikz}} &
    $\eqann{\L3, \D2}{M = P^-1\pmm{I_r\\B}}$ &
    \scalebox{\scalefig}{\input{./figures/aW-nf-uniqueness-proof/step1_2.tikz}} &
    $\eqann{\L8}{}$ &
    \scalebox{\scalefig}{\input{./figures/aW-nf-uniqueness-proof/step1_3.tikz}} &
\end{derivationchain}

By plugging the gadget
\begin{derivationchain}
    \scalebox{\scalefig}{\input{./figures/aW-nf-uniqueness-proof/gadget_1.tikz}} 
\end{derivationchain}
in N, we get:

\begin{derivationchain}
    \scalebox{\scalefig}{\input{./figures/aW-nf-uniqueness-proof/step1_4.tikz}} &
    $\eqann{\L3,\L8, \Dist}{}$ &
    \scalebox{\scalefig}{\input{./figures/aW-nf-uniqueness-proof/step1_5.tikz}} &
    $\eqann{\L4}{}$ &
    \scalebox{\scalefig}{\input{./figures/aW-nf-uniqueness-proof/step1_6.tikz}} &
    $\eqann{\S\A6}{}$ &
    \scalebox{\scalefig}{\input{./figures/aW-nf-uniqueness-proof/step1_8.tikz}} \\[1em]

    \scalebox{\scalefig}{\input{./figures/aW-nf-uniqueness-proof/step1_8.tikz}} &
    $\eqann{\S\A6, \S\A5}{}$ &
    \scalebox{\scalefig}{\input{./figures/aW-nf-uniqueness-proof/step1_9.tikz}} &
    $\eqann{\L26}{}$ &
    \scalebox{\scalefig}{\input{./figures/aW-nf-uniqueness-proof/step1_10.tikz}} &
    $\eqann{\S\A11}{}$ &
    \scalebox{\scalefig}{\begin{tikzpicture}
	\begin{pgfonlayer}{nodelayer}
		\node [style=dot] (0) at (-0.75, 0) {};
		\node [style=dot] (1) at (0.75, 0) {};
		\node [style=white phase dot] (2) at (0, 0) {$x$};
	\end{pgfonlayer}
	\begin{pgfonlayer}{edgelayer}
		\draw (0) to (1);
	\end{pgfonlayer}
\end{tikzpicture}
}
\end{derivationchain}

Using
\[
    \interp{\scalebox{\scalefig}{\begin{tikzpicture}
	\begin{pgfonlayer}{nodelayer}
		\node [style=dot] (0) at (-0.75, 0) {};
		\node [style=dot] (1) at (0.75, 0) {};
		\node [style=white phase dot] (2) at (0, 0) {$x$};
	\end{pgfonlayer}
	\begin{pgfonlayer}{edgelayer}
		\draw (0) to (1);
	\end{pgfonlayer}
\end{tikzpicture}
}}
    =x,
\]
we obtain a diagram of non-zero interpretation. By functoriality of
$\interp{\cdot}$, this is impossible if $\interp{N}=0$. Thus
\[
    \interp{N}=0
    \quad\Longrightarrow\quad
    x=0.
\]
Therefore
\[
    \interp{N}=0
    \quad\Longleftrightarrow\quad
    x=0.
\]
The same argument applies to $N'$. Hence, if
\[
    \interp{N}=\interp{N'}=0,
\]
then
\[
    x=x'=0,
\]
and both diagrams rewrite to the distinguished zero normal form.

In the rest of the proof we may therefore assume
\[
    \interp{N}=\interp{N'}\neq0,
    \qquad
    x\neq0,
    \qquad
    x'\neq0.
\]

\medskip

\noindent
\(\boxed{\text{\itshape Step 2}:\ r=r'}\)

First, note that 
\[P M' = I_n PM' = \pmm{E_r PM'\\
                        \pmm{0 & I_{n-r}}PM'}
                        \]
Let us denote:
\[
    \begin{cases}
        C := E_r PM',\\
        D := \pmm{0 & I_{n-r}}PM'.
    \end{cases}
\]

Then, by plugging the same gadget 
\begin{derivationchain}
    \scalebox{\scalefig}{\input{./figures/aW-nf-uniqueness-proof/gadget_1.tikz}} 
\end{derivationchain}

to N' this time, we get the derivation:
\begin{derivationchain}
    \scalebox{\scalefig}{\input{./figures/aW-nf-uniqueness-proof/step2_1.tikz}} &
    $\eqann{\L3}{}$ &
    \scalebox{\scalefig}{\input{./figures/aW-nf-uniqueness-proof/step2_2.tikz}} &
    $\eqann{\D2}{PM'= \pmm{C\\D}}$ &
    \scalebox{\scalefig}{\input{./figures/aW-nf-uniqueness-proof/step2_3.tikz}} &
    $\eqann{\L4}{}$ &
    \scalebox{\scalefig}{\input{./figures/aW-nf-uniqueness-proof/step2_4.tikz}} \\[1em]

    \scalebox{\scalefig}{\input{./figures/aW-nf-uniqueness-proof/step2_4.tikz}} &
     $\eqann{\S\A6}{}$ &
    \scalebox{\scalefig}{\input{./figures/aW-nf-uniqueness-proof/step2_5.tikz}} &
     $\eqann{\L9}{}$ &
    \scalebox{\scalefig}{\input{./figures/aW-nf-uniqueness-proof/step2_6.tikz}} &
     $\eqann{\L24}{}$ &
    \scalebox{\scalefig}{\input{./figures/aW-nf-uniqueness-proof/step2_7.tikz}}
\end{derivationchain}

The interpretation of the diagram obtained from $N$ is $x\neq0$.
Since $\interp{N}=\interp{N'}$, the corresponding diagram obtained
from $N'$ must also have non-zero interpretation. As
$\interp{\scalebox{\scalefig}{\begin{tikzpicture}
	\begin{pgfonlayer}{nodelayer}
		\node [style=none] (0) at (0.75, 0) {};
		\node [style=rw] (1) at (-0.75, 0) {};
		\node [style=dot] (2) at (0.75, 0) {};
		\node [style=modenb] (3) at (0, 0.25) {$k$};
	\end{pgfonlayer}
	\begin{pgfonlayer}{edgelayer}
		\draw [style=thick] (1) to (0.center);
	\end{pgfonlayer}
\end{tikzpicture}
}} = 0$ when $k > 0$, 
the above reduction therefore forces
\[
    r=\mathrm{rank}(C)=r'.
\]
Hence
\[
    r=r'.
\]

\medskip

\noindent
\(\boxed{\text{\itshape Step 3}:\ 
(j_1,\ldots,j_r)=(j'_1,\ldots,j'_r)}\)

Suppose that the pivot rows differ, and let $k$ be the smallest index
such that
\[
    j_k\neq j'_k.
\]
Without loss of generality, suppose
\[
    j_k<j'_k.
\]

Note that $C=E_rP_MM'$ is the submatrix of M' composed of the rows
$(j_1,\ldots,j_r)$ in this order, and therefore reads:

\[C = 
\begin{bsmallmatrix}
        1           & 0           & \cdots & 0           & 0      & \cdots & 0 \\
        c_{2,1}     & 1           & \cdots & 0           & \vdots & \cdots & 0 \\
        \vdots      & \vdots      & \ddots & \vdots      & \vdots & \vdots & \vdots \\
        c_{k-1,1}   & c_{k-1,2}   & \cdots & 1           & 0      & \cdots & 0 \\
        c_{k,1}     & c_{k,2}     & \cdots & c_{k,k-1}   & 0      & \cdots & 0 \\
        \vdots      & \vdots      & \vdots & \vdots      & \ddots &        & \vdots \\
        c_{r,1}     & c_{r,2}     & \cdots & c_{r,k-1}   & \cdots & \cdots & 1
        \end{bsmallmatrix}.
\]
Notice that the $k$-th row is a linear combination of the first 
$k-1$ rows, which is not possible as $\mathrm{rank}(C) = r$, as shown in step 2.

Therefore
\[
    (j_1,\ldots,j_r)
    =
    (j'_1,\ldots,j'_r).
\]
It directly follows that 
\[
P = P'
\]

\medskip

\noindent
\(\boxed{\text{\itshape Step 4}:\ x=x'}\)
As $P = P'$ and $r=r'$, we have:
\[
    PM
    =
    \begin{bmatrix}
        I_r\\
        B
    \end{bmatrix},
    \qquad
    PM'
    =
    \begin{bmatrix}
        I_r\\
        B'
    \end{bmatrix},
\]
and therefore, plugging the same gadget
\begin{derivationchain}
    \scalebox{\scalefig}{\input{./figures/aW-nf-uniqueness-proof/gadget_1.tikz}} 
\end{derivationchain}
to $N'$ gives.
\begin{derivationchain}
    \scalebox{\scalefig}{\input{./figures/aW-nf-uniqueness-proof/step4_1.tikz}} &
    $\eqann{}{}$ &
    \scalebox{\scalefig}{\input{./figures/aW-nf-uniqueness-proof/step4_2.tikz}} &
    $\eqann{\L3, \D2}{P'M' = \pmm{I_r\\B'}}$ &
    \scalebox{\scalefig}{\input{./figures/aW-nf-uniqueness-proof/step4_3.tikz}} &
    $\eqann{\substack{\L4,\L8\\ \S\A6, \S\A5\\\L26, S\A11 }}{}$ &
    \scalebox{\scalefig}{\begin{tikzpicture}
	\begin{pgfonlayer}{nodelayer}
		\node [style=dot] (0) at (-0.75, 0) {};
		\node [style=dot] (1) at (0.75, 0) {};
		\node [style=white phase dot] (2) at (0, 0) {$x'$};
	\end{pgfonlayer}
	\begin{pgfonlayer}{edgelayer}
		\draw (0) to (1);
	\end{pgfonlayer}
\end{tikzpicture}
} \\[1em]
\end{derivationchain}

Recall that the interpretation of the diagram obtained from plugging
this gadget to $N$ in step 1 was $x$. This yields:
\[
    x=x'.
\]

\medskip

\noindent
\(\boxed{\text{\itshape Step 5}:\ M=M'}\)

We now know that
\[
    r=r',
    \qquad
    P=P'
    \qquad
    x=x'\neq0.
\]
And
\[
    PM
    =
    \begin{bmatrix}
        I_r\\
        B
    \end{bmatrix},
    \qquad
    PM'
    =
    \begin{bmatrix}
        I_r\\
        B'
    \end{bmatrix}.
\]
It remains only to prove
\[
    B=B'.
\]

Let
\[
    1\leq i\leq r,
    \qquad
    1\leq q\leq n-r,
\]
and let $S_{i,r+q}\in\C^{n\times n}$ be the permutation matrix
exchanging rows $i$ and $r+q$.

Define
\[
    R_{i,q}
    :=
    E_rS_{i,r+q}P
    \in\C^{r\times n}.
\]
such that
\[
    R_{i,q}M
    =
    E_rS_{i,r+q}
    \begin{bmatrix}
        I_r\\
        B
    \end{bmatrix}.
\]
This is the identity matrix in which the $i$-th row has been replaced
by the $q$-th row of $B$. Its determinant can be computed as:

\[
        \det\!\left(R_{i,q}M\right)
        =
        \left|
        \begin{smallarray}{c c c|c|c c c}
                &           &           & 0       &           &       & \\
                & I_{i-1}   &           & \svdots &           & (0)   & \\
                &           &           & 0       &           &       & \\
        \hline
        B_{q,1}   & \cdots    & B_{q,i-1} & B_{q,i} & B_{q,i+1}& \cdots& B_{q,r}\\
        \hline
                &           &           & 0       &           &       & \\
                & (0)       &           & \svdots &           & I_{r-i}& \\
                &           &           & 0       &           &       &
        \end{smallarray}
        \right|
        =(-1)^{i+i}B_{q,i} = B_{q,i}.
        \]
Similarly,
\[
    \det(R_{i,q}M')=(-1)^{i+i}B'_{q,i} = B'_{q,i}.
\]

Since
\[
    \interp{N}=\interp{N'}
    \qquad\text{and}\qquad
    x=x'\neq0,
\]
we may plug the following second gadget to both $N$ and $N'$
\begin{derivationchain}
    \scalebox{\scalefig}{\input{./figures/aW-nf-uniqueness-proof/gadget_2.tikz}} 
\end{derivationchain}
We get
\begin{derivationchain}
    \scalebox{\scalefig}{\input{./figures/aW-nf-uniqueness-proof/step5_1.tikz}} &
    $\eqann{\L3,\D2}{}$ &
    \scalebox{\scalefig}{\input{./figures/aW-nf-uniqueness-proof/step5_2.tikz}} &
    $\eqann{\L4, \S\A6}{}$ &
    \scalebox{\scalefig}{\input{./figures/aW-nf-uniqueness-proof/step5_3.tikz}} &
    $\eqann{\substack{\L24\\ \S\A6, \S\A5\\\L26, S\A11 }}{}$ &
    \scalebox{\scalefig}{\input{./figures/aW-nf-uniqueness-proof/step5_4.tikz}} \\[1em]
\end{derivationchain}

The two sides therefore reduce respectively
to
\[
    x\det(R_{i,q}M)
    \qquad\text{and}\qquad
    x\det(R_{i,q}M').
\]
Since $x\neq0$,
\[
    \det(R_{i,q}M)
    =
    \det(R_{i,q}M'),
\]
and therefore
\[
    B_{q,i}=B'_{q,i}.
\]
Since $i$ and $q$ were arbitrary,
\[
    B=B'.
\]
Hence
\[
    PM=PM',
\]
and since $P$ is invertible,
\[
    M=M'.
\]

Therefore
\[
    (r,x,M)=(r',x',M'),
\]
so the two normal forms are identical.

\subsection{Completeness of $\FpW$}

\textbf{To normal form.}
Let $D:0\to n$ be a diagram of $\S\FpW$.  
Due to \L31, \S\A32~ and \L36, we can consider without loss of generality that $D$
only contains generators $\scalebox{\scalefig}{\begin{tikzpicture}
	\begin{pgfonlayer}{nodelayer}
		\node [style=none] (0) at (0.5, 0) {};
		\node [style=none] (1) at (0, 0) {};
		\node [style=dot] (4) at (-0.5, 0) {};
	\end{pgfonlayer}
	\begin{pgfonlayer}{edgelayer}
		\draw (0.center) to (1.center);
		\draw (1.center) to (4);
	\end{pgfonlayer}
\end{tikzpicture}
}$ and $\scalebox{\scalefig}{\begin{tikzpicture}
	\begin{pgfonlayer}{nodelayer}
		\node [style=none] (0) at (-0.25, 0.5) {};
		\node [style=none] (1) at (0.25, 0.5) {};
		\node [style=none] (2) at (0.25, -0.5) {};
		\node [style=none] (3) at (-0.25, -0.5) {};
	\end{pgfonlayer}
	\begin{pgfonlayer}{edgelayer}
		\draw (3.center) to (2.center);
		\draw (0.center) to (1.center);
		\draw [bend right=270, looseness=2.50] (1.center) to (2.center);
	\end{pgfonlayer}
\end{tikzpicture}
}$
on top of \EaW generators.
Indeed, we can rewrite the other generators like so:
\[
\boxeq{A36}
    {FpW-rules/1w-2w_1}
    {FpW-rules/1w-2w_2}
\]

\[
\boxeq{L31}
    {SFpW-lemma/w-2-to-1_1}
    {SFpW-lemma/w-2-to-1_2}
\]

\[\boxeq{L36}
    {SFpW-lemma/dot-cup_1}
    {SFpW-lemma/dot-cup_2}\]

Using \L25, we can slide  $\scalebox{\scalefig}{\begin{tikzpicture}
	\begin{pgfonlayer}{nodelayer}
		\node [style=none] (0) at (0.5, 0) {};
		\node [style=none] (1) at (0, 0) {};
		\node [style=dot] (4) at (-0.5, 0) {};
	\end{pgfonlayer}
	\begin{pgfonlayer}{edgelayer}
		\draw (0.center) to (1.center);
		\draw (1.center) to (4);
	\end{pgfonlayer}
\end{tikzpicture}
}$ to the left side of the diagram, while
naturality of $\scalebox{\scalefig}{\begin{tikzpicture}
	\begin{pgfonlayer}{nodelayer}
		\node [style=none] (0) at (-0.25, 0.5) {};
		\node [style=none] (1) at (0.25, 0.5) {};
		\node [style=none] (2) at (0.25, -0.5) {};
		\node [style=none] (3) at (-0.25, -0.5) {};
	\end{pgfonlayer}
	\begin{pgfonlayer}{edgelayer}
		\draw (3.center) to (2.center);
		\draw (0.center) to (1.center);
		\draw [bend right=270, looseness=2.50] (1.center) to (2.center);
	\end{pgfonlayer}
\end{tikzpicture}
}$ allows us to slide them to the top-right corner of the diagram.
We thus start with $D$ written as:
\[
\scalebox{\scalefig}{\input{./figures/FpW-proof/pre-nf_1.tikz}}
\]
And then derive :

\begin{derivationchain}
    \scalebox{\scalefig}{\input{./figures/FpW-proof/pre-nf_1.tikz}} &
    $\eqann{D3}{M_1 = \pmm{M_2\\M_3\\M_4}}$ &
    \scalebox{\scalefig}{\input{./figures/FpW-proof/pre-nf_2.tikz}} &
    $\eqann{\L22, \L15, \S\A36}{Z_1 = M_3^TM_2 - M_2^TM_3}$ &
    \scalebox{\scalefig}{\input{./figures/FpW-proof/pre-nf_3.tikz}} &
    $\eq\limits_{\L29}$ &
    \scalebox{\scalefig}{\input{./figures/FpW-proof/pre-nf_4.tikz}} \\[1em]
    \scalebox{\scalefig}{\input{./figures/FpW-proof/pre-nf_4.tikz}} &
    $\eqann{\L14,\L22,\S\A6}{Z_2 = Alt_nZ_1Alt_n}$ &
    \scalebox{\scalefig}{\input{./figures/FpW-proof/pre-nf_5.tikz}} &
    $\eqann{\L23}{Z_2 = P_2\pmm{A_2 & 0 \\ 0&0}P_2^T}$ &
    \scalebox{\scalefig}{\input{./figures/FpW-proof/pre-nf_6.tikz}} &
    $\eqann{\L32}{x_1 = \det P_2}$ &
    \scalebox{\scalefig}{\input{./figures/FpW-proof/pre-nf_7.tikz}} \\[1em]
    \scalebox{\scalefig}{\input{./figures/FpW-proof/pre-nf_7.tikz}} &
    $\eqann{\L3}{M_5 = P_2^{-1}M_4}$ &
    \scalebox{\scalefig}{\input{./figures/FpW-proof/pre-nf_8.tikz}} &
    $\eqann{Dist,\L35}{\substack{M_5 = \pmm{M_6\\M_7}\\ x_2 = \Pf(A_2)x_1}}$ &
    \scalebox{\scalefig}{\input{./figures/FpW-proof/pre-nf_9.tikz}} &
    $\eqann{\L14,\S\A36}{Z_3 = M_6A_2^{-1}M_6^T}$ &
    \scalebox{\scalefig}{\input{./figures/FpW-proof/pre-nf_10.tikz}} \\[1em]
\end{derivationchain}

Similarly to what was done for \S\aW, 
if $\mathrm{rank}(M_7) \neq r_2$, we can apply \L9, \L24, \L28, and \S\A30~
to get to the normal form for the zero diagram of type $0 \to n$, which is essentially the same as for 
\S\aW due to \L20~ and \S\A6. 
Otherwise, we can assume that $\mathrm{rank}(M_7) = r_2 \leq n$. Then
Using \L24, $\S\A26$, and $\S\A9$, we can multiply $M_7$ on the right by a
sequence of elementary matrices so as to put it in Reduced Column Echelon
Form (RCEF), while multiplying $x_2$ by the determinant of these matrices.
This gives $M_8$ full rank and in RCEF, and $x_2$ so that $D$ reads:
\[
\scalebox{\scalefig}{\input{./figures/FpW-proof/nf_1.tikz}}
\]
As $M_8$ is in RCEF, \cref{lemma:dec-rcef} gives $P_3$ a permutation matrix such that:
\[M_8 = P_3\pmm{I_r\\B}\]
We then derive:

\begin{derivationchain}
    \scalebox{\scalefig}{\input{./figures/FpW-proof/nf_1.tikz}} &
    $\eqann{\L4,\D3,\L2}{M_8 = P_3\pmm{I_r\\B}}$ &
    \scalebox{\scalefig}{\input{./figures/FpW-proof/nf_2.tikz}} &
    $\eqann{\L13,\Dist}{P_3^{-1}Z_3P_3^{-T} = \pmm{Z_4 & c \\ -C^T & Z_5}}$ &
    \scalebox{\scalefig}{\input{./figures/FpW-proof/nf_3.tikz}} &
    $\eqann{\S\A1, \S\A3, \S\A36}{}$ &
    \scalebox{\scalefig}{\input{./figures/FpW-proof/nf_4.tikz}} \\[1em]
    \scalebox{\scalefig}{\input{./figures/FpW-proof/nf_4.tikz}} &
    $\eqann{\L34}{}$ &
    \scalebox{\scalefig}{\input{./figures/FpW-proof/nf_5.tikz}} &
    $\eqann{\S\A6,\L3}{M_9 = (-1)^{n+1}BAlt_n}$ &
    \scalebox{\scalefig}{\input{./figures/FpW-proof/nf_6.tikz}} &
    $\eqann{\S\A35,\L29}{}$ &
    \scalebox{\scalefig}{\input{./figures/FpW-proof/nf_8.tikz}} \\[1em]
    \scalebox{\scalefig}{\input{./figures/FpW-proof/nf_8.tikz}} &
    $\eqann{\S\A35,\S\A16,\L14}{Z_6 = Alt_nZ_4Alt_n^T}$ &
    \scalebox{\scalefig}{\input{./figures/FpW-proof/nf_9.tikz}} &
    $\eqann{\S\A35}{}$ &
    \scalebox{\scalefig}{\input{./figures/FpW-proof/nf_10.tikz}} &
    $\eqann{\L29}{}$ &
    \scalebox{\scalefig}{\input{./figures/FpW-proof/nf_11.tikz}} \\[1em]
    \scalebox{\scalefig}{\input{./figures/FpW-proof/nf_11.tikz}} &
    $\eqann{\S\A35}{}$ &
    \scalebox{\scalefig}{\input{./figures/FpW-proof/nf_12.tikz}} &
    $\eqann{\S\A16}{}$ &
    \scalebox{\scalefig}{\input{./figures/FpW-proof/nf_13.tikz}} &
    $\eqann{\L2,\L2}{}$ &
    \scalebox{\scalefig}{\input{./figures/FpW-proof/nf_14.tikz}} \\[1em]
    \scalebox{\scalefig}{\input{./figures/FpW-proof/nf_14.tikz}} &
    $\eqann{\S\A1}{Z_7 = M_9Z_6M_9^T}$ &
    \scalebox{\scalefig}{\input{./figures/FpW-proof/nf_15.tikz}} &
    $\eqann{\L14,\L3,\L15}{Z_8 = C^TAlt_nM_9^T - M_9Alt_nC}$ &
    \scalebox{\scalefig}{\input{./figures/FpW-proof/nf_16.tikz}} &
    $\eqann{\S\A1,\L10,\L3}{M_9 = (-1)^{n+1}BAlt_n}$ &
    \scalebox{\scalefig}{\input{./figures/FpW-proof/nf_17.tikz}} \\[1em]
    \scalebox{\scalefig}{\input{./figures/FpW-proof/nf_17.tikz}} &
    $\eqann{\S\A6}{Z_9 = Z_7+Z_5+Z_8}$ &
    \scalebox{\scalefig}{\input{./figures/FpW-proof/nf_18.tikz}} &
    $\eqann{\L29}{}$ &
    \scalebox{\scalefig}{\input{./figures/FpW-proof/nf_19.tikz}} &
    $\eqann{\S\A6, \S\A3}{}$ &
    \scalebox{\scalefig}{\input{./figures/FpW-proof/nf_20.tikz}} \\[1em]
    \scalebox{\scalefig}{\input{./figures/FpW-proof/nf_20.tikz}} &
    $\eqann{\L29}{}$ &
    \scalebox{\scalefig}{\input{./figures/FpW-proof/nf_21.tikz}} &
    $\eqann{\L29,\S\A36}{\substack{M_8 = P_3\pmm{I_r\\B}\\Z_{10} = P_3\pmm{0 & 0\\ Z_9 & 0}P_3^T}}$ &
    \scalebox{\scalefig}{\input{./figures/FpW-proof/nf_22.tikz}} &
    {} &
    {} \\[1em]
\end{derivationchain}

Therefore we can rewrite $D$ in normal form, with $M_8$ full rank and RCEF and $Z_10$ zero on the columns and rows
corresponding to $M_8$'s pivots.


\textbf{Uniqueness of normal form.}

Let
\[
    N=N(r,x,Z,M)
    \qquad\text{and}\qquad
    N'=N(r',x',Z',M')
\]
be two $\S\FpW$ normal forms of the same type $0\to n$, where
\(
    M\in\C^{n\times r},
    M'\in\C^{n\times r'}
\) 
are in RCEF and have ranks $r$ and $r'$ respectively. Let
\(
    1\leq j_1<\cdots<j_r\leq n,
    1\leq j'_1<\cdots<j'_{r'}\leq n
\) 
be their respective pivot rows.

Assume that
\[
    \interp{N}=\interp{N'},
\quad \text{i.e.} \quad
            \interp{\scalebox{\scalefig}{\input{./figures/FpW-nf-uniqueness-proof/step1_1.tikz}}}
            \eq
            \interp{\scalebox{\scalefig}{\input{./figures/FpW-nf-uniqueness-proof/step1_1-prime.tikz}}}.
\]

Due to \cref{lemma:dec-rcef}, there exists permutation matrices 
$P, P' \in \C^{n\times n}$ such that

\[
    P M
    =
    \begin{bmatrix}
        I_r\\
        B
    \end{bmatrix}
    \qquad 
    \text{and} 
    \qquad
    P'M'
    =
    \begin{bmatrix}
        I_{r'}\\
        B'
    \end{bmatrix}.
\]

For $s\leq n$, we write
\[
    E_s:=
    \begin{bmatrix}
        I_s & 0
    \end{bmatrix}
    \in\C^{s\times n}.
\]
Hence
\[
    E_r P M = I_r,
    \qquad
    E_{r'} P' M' = I_{r'}.
\]

We also know that $Z$ (resp $Z'$) is zero on the rows and columns corresponding to pivots, i.e.
\[
 P Z P^T
    =
    \begin{bmatrix}
        0 & 0\\
        0 & K 
    \end{bmatrix}
    \qquad 
    \text{and} 
    \qquad
    P' Z P'^T
    =
    \begin{bmatrix}
        0 & 0\\
        0 & K' 
    \end{bmatrix}
\]
for $K,K'$ antisymmetric matrices of size respectively $n-r$ and $n-r'$.

We prove equality of the normal-form data in six steps. The first
step separates the zero case; in the remaining steps we may therefore
assume that the common interpretation is non-zero.
\medskip

\noindent
\(\boxed{\text{\itshape Step 1}:\ 
\interp{N}=\interp{N'}=0
\ \Longrightarrow\ x=x'=0}\)

First, let us show that a normal form has zero interpretation if and
only if its global scalar vanishes.

If $x=0$, then \L28, \S\A30, \S\A26~ rewrites $N$ to
\scalebox{\scalefig}{\input{./figures/aW-proof/zero-nf.tikz}},
the distinguished normal form of the zero $0\to n$ state. Therefore
\[
    x=0
    \quad\Longrightarrow\quad
    \interp{N}=0.
\]

Conversely, suppose that $x\neq0$. Diagrammatically, we get

\begin{derivationchain}
    \scalebox{\scalefig}{\input{./figures/FpW-nf-uniqueness-proof/step1_1.tikz}} &
    $\eqann{\L3, \D2, \L14}{\substack{PM = \pmm{I_r\\B}\\ PZP^T = \pmm{0&0\\0&K}}}$ &
    \scalebox{\scalefig}{\input{./figures/FpW-nf-uniqueness-proof/step1_2.tikz}} &
    $\eqann{\L2,\Dist,\S\A11,\S\A6}{}$ &
    \scalebox{\scalefig}{\input{./figures/FpW-nf-uniqueness-proof/step1_3.tikz}} &
\end{derivationchain}

By plugging the gadget
\begin{derivationchain}
    \scalebox{\scalefig}{\input{./figures/FpW-nf-uniqueness-proof/gadget_1.tikz}} 
\end{derivationchain}
in N, we get:

\begin{derivationchain}
    \scalebox{\scalefig}{\input{./figures/FpW-nf-uniqueness-proof/step1_4.tikz}} &
    $\eqann{\L3, \L8, \Exp}{}$ &
    \scalebox{\scalefig}{\input{./figures/FpW-nf-uniqueness-proof/step1_5.tikz}} &
    $\eqann{\S\A4}{}$ &
    \scalebox{\scalefig}{\input{./figures/FpW-nf-uniqueness-proof/step1_6.tikz}} &
    $\eqann{\L4, \L8, \S\A6}{}$ &
    \scalebox{\scalefig}{\input{./figures/FpW-nf-uniqueness-proof/step1_7.tikz}} \\[1em]

    \scalebox{\scalefig}{\input{./figures/FpW-nf-uniqueness-proof/step1_7.tikz}} &
    $\eqann{\substack{\S\A6, \S\A5, \\\L26,\S\A9, \S\A11}}{}$ &
    \scalebox{\scalefig}{\begin{tikzpicture}
	\begin{pgfonlayer}{nodelayer}
		\node [style=dot] (0) at (-0.75, 0) {};
		\node [style=dot] (1) at (0.75, 0) {};
		\node [style=white phase dot] (2) at (0, 0) {$x$};
	\end{pgfonlayer}
	\begin{pgfonlayer}{edgelayer}
		\draw (0) to (1);
	\end{pgfonlayer}
\end{tikzpicture}
} &
\end{derivationchain}

Using
\[
    \interp{\scalebox{\scalefig}{\begin{tikzpicture}
	\begin{pgfonlayer}{nodelayer}
		\node [style=dot] (0) at (-0.75, 0) {};
		\node [style=dot] (1) at (0.75, 0) {};
		\node [style=white phase dot] (2) at (0, 0) {$x$};
	\end{pgfonlayer}
	\begin{pgfonlayer}{edgelayer}
		\draw (0) to (1);
	\end{pgfonlayer}
\end{tikzpicture}
}}
    =x,
\]
we obtain a diagram of non-zero interpretation. By functoriality of
$\interp{\cdot}$, this is impossible if $\interp{N}=0$. Thus
\[
    \interp{N}=0
    \quad\Longrightarrow\quad
    x=0.
\]
Therefore
\[
    \interp{N}=0
    \quad\Longleftrightarrow\quad
    x=0.
\]
The same argument applies to $N'$. Hence, if
\[
    \interp{N}=\interp{N'}=0,
\]
then
\[
    x=x'=0,
\]
and both diagrams rewrite to the distinguished zero normal form.

In the rest of the proof we may therefore assume
\[
    \interp{N}=\interp{N'}\neq0,
    \qquad
    x\neq0,
    \qquad
    x'\neq0.
\]

\medskip

\noindent
\(\boxed{\text{\itshape Step 2}:\ r=r'}\)

Without loss of generality, we can assume $r \leq r'$. Note that 
\[P M' = I_n PM' = \pmm{E_r PM'\\
                        \pmm{0 & I_{n-r}}PM'}
                        \]
Let us denote:
\[
    \begin{cases}
        C := E_r PM',\\
        D := \pmm{0 & I_{n-r}}PM'.
    \end{cases}
\]

Then, by plugging the same gadget 
\begin{derivationchain}
    \scalebox{\scalefig}{\input{./figures/aW-nf-uniqueness-proof/gadget_1.tikz}} 
\end{derivationchain}

to N' this time, we get the derivation:
\begin{derivationchain}
    \scalebox{\scalefig}{\input{./figures/FpW-nf-uniqueness-proof/step2_1.tikz}} &
    $\eqann{\L2}{}$ &
    \scalebox{\scalefig}{\input{./figures/FpW-nf-uniqueness-proof/step2_2.tikz}} &
    $\eqann{\Dist}{}$ &
    \scalebox{\scalefig}{\input{./figures/FpW-nf-uniqueness-proof/step2_3.tikz}} &
    $\eqann{\S\A4}{C = E_rPM'}$ &
    \scalebox{\scalefig}{\input{./figures/FpW-nf-uniqueness-proof/step2_4.tikz}} \\[1em]

    \scalebox{\scalefig}{\input{./figures/FpW-nf-uniqueness-proof/step2_4.tikz}} &
     $\eqann{\L9}{}$ &
    \scalebox{\scalefig}{\input{./figures/FpW-nf-uniqueness-proof/step2_5.tikz}} &
     $\eqann{\L24, \Dist}{}$ &
    \scalebox{\scalefig}{\input{./figures/FpW-nf-uniqueness-proof/step2_6.tikz}} &
\end{derivationchain}

The interpretation of the diagram obtained from $N$ is $x\neq0$.
Since $\interp{N}=\interp{N'}$, the corresponding diagram obtained
from $N'$ must also have non-zero interpretation. As
$\interp{\scalebox{\scalefig}{\begin{tikzpicture}
	\begin{pgfonlayer}{nodelayer}
		\node [style=none] (0) at (0.75, 0) {};
		\node [style=rw] (1) at (-0.75, 0) {};
		\node [style=dot] (2) at (0.75, 0) {};
		\node [style=modenb] (3) at (0, 0.25) {$k$};
	\end{pgfonlayer}
	\begin{pgfonlayer}{edgelayer}
		\draw [style=thick] (1) to (0.center);
	\end{pgfonlayer}
\end{tikzpicture}
}} = 0$ when $k > 0$, 
the above reduction therefore forces
\[
    r' \leq \mathrm{rank}(C)
\]
Besides, the assumption $r \leq r'$ directly gives $\mathrm{rank}(C) \leq r$, and therefore
\[
r' \leq \mathrm{rank}(C) \leq r \leq r'
\]
Hence,
\[
    r=r'.
\]

\medskip

\noindent
\(\boxed{\text{\itshape Step 3}:\ 
(j_1,\ldots,j_r)=(j'_1,\ldots,j'_r)}\)

Suppose that the pivot rows differ, and let $k$ be the smallest index
such that
\[
    j_k\neq j'_k.
\]
Without loss of generality, suppose
\[
    j_k<j'_k.
\]
Again, plug the gadget
\begin{derivationchain}
    \scalebox{\scalefig}{\input{./figures/FpW-nf-uniqueness-proof/gadget_1.tikz}} 
\end{derivationchain}
to $N'$ and derive
\begin{derivationchain}
    \scalebox{\scalefig}{\input{./figures/FpW-nf-uniqueness-proof/step3_1.tikz}} &
    $\eqann{\L2, \L3, \Dist, \L4}{\substack{PM' = \pmm{C\\D}\\ PZ'P^T = \pmm{W & *\\ * & *}}}$ &
    \scalebox{\scalefig}{\input{./figures/FpW-nf-uniqueness-proof/step3_2.tikz}} &
    $\eqann{\L27}{}$ &
    \scalebox{\scalefig}{\input{./figures/FpW-nf-uniqueness-proof/step3_3.tikz}} &
    $\eqann{\substack{\L11, \S\A6, \S\A5, \\\L26,\S\A9, \S\A11}}{}$ &
    \scalebox{\scalefig}{\input{./figures/FpW-nf-uniqueness-proof/step3_4.tikz}} \\[1em]
\end{derivationchain}
Which gives that $ x'detC = x\neq 0$.
Note that $C=E_rPM'$ is the submatrix of M' composed of the rows
$(j_1,\ldots,j_r)$ in this order, and therefore reads:

\[C = 
\begin{bsmallmatrix}
        1           & 0           & \cdots & 0           & 0      & \cdots & 0 \\
        c_{2,1}     & 1           & \cdots & 0           & \vdots & \cdots & 0 \\
        \vdots      & \vdots      & \ddots & \vdots      & \vdots & \vdots & \vdots \\
        c_{k-1,1}   & c_{k-1,2}   & \cdots & 1           & 0      & \cdots & 0 \\
        c_{k,1}     & c_{k,2}     & \cdots & c_{k,k-1}   & 0      & \cdots & 0 \\
        \vdots      & \vdots      & \vdots & \vdots      & \ddots &        & \vdots \\
        c_{r,1}     & c_{r,2}     & \cdots & c_{r,k-1}   & \cdots & \cdots & 1
        \end{bsmallmatrix}.
\]
Notice that the $k$-th row is a linear combination of the first 
$k-1$ rows, which is not possible as $detC \neq 0$.

Therefore
\[
    (j_1,\ldots,j_r)
    =
    (j'_1,\ldots,j'_r).
\]
It directly follows that 
\[
P = P'
\]

\medskip

\noindent
\(\boxed{\text{\itshape Step 4}:\ x=x'}\)
As $P = P'$ and $r=r'$, we have:
\[
    PM
    =
    \begin{bmatrix}
        I_r\\
        B
    \end{bmatrix},
    \qquad
    PM'
    =
    \begin{bmatrix}
        I_r\\
        B'
    \end{bmatrix},
\]
and therefore, plugging the same gadget
\begin{derivationchain}
    \scalebox{\scalefig}{\input{./figures/aW-nf-uniqueness-proof/gadget_1.tikz}} 
\end{derivationchain}
to $N'$ gives.
\begin{derivationchain}
    \scalebox{\scalefig}{\input{./figures/FpW-nf-uniqueness-proof/step4_1.tikz}} &
    $\eqann{}{P=P'}$ &
    \scalebox{\scalefig}{\input{./figures/FpW-nf-uniqueness-proof/step4_2.tikz}} &
    $\eqann{\L2, \L3}{\substack{P'M' = \pmm{I_r\\B'}\\ P'Z'P'^T = \pmm{0 & 0\\ 0 & K'}}}$ &
    \scalebox{\scalefig}{\input{./figures/FpW-nf-uniqueness-proof/step4_3.tikz}} &
    $\eqann{\S\A4}{}$ &
    \scalebox{\scalefig}{\input{./figures/FpW-nf-uniqueness-proof/step4_4.tikz}} \\[1em]

    \scalebox{\scalefig}{\input{./figures/FpW-nf-uniqueness-proof/step4_4.tikz}} &
    $\eqann{\substack{\L11, \S\A6, \S\A5, \\\L26,\S\A9, \S\A11}}{}$ &
    \scalebox{\scalefig}{\begin{tikzpicture}
	\begin{pgfonlayer}{nodelayer}
		\node [style=dot] (0) at (-0.75, 0) {};
		\node [style=dot] (1) at (0.75, 0) {};
		\node [style=white phase dot] (2) at (0, 0) {$x'$};
	\end{pgfonlayer}
	\begin{pgfonlayer}{edgelayer}
		\draw (0) to (1);
	\end{pgfonlayer}
\end{tikzpicture}
} &
\end{derivationchain}

Recall that the interpretation of the diagram obtained from plugging
this gadget to $N$ in step 1 was $x$. This yields:
\[
    x=x'.
\]

\medskip
\noindent\(\boxed{\text{\itshape Step 5}:\ Z=Z'}\)

We now know that 
\[
r = r' \qquad P = P'\qquad x = x'
\]
and that 
\[
P^T Z P
    =
    \begin{bmatrix}
        0 & 0\\
        0 & K 
    \end{bmatrix}
    \qquad 
    P^T Z' P
    =
    \begin{bmatrix}
        0 & 0\\
        0 & K' 
    \end{bmatrix}
\]

Remains to show that $K = K'$. Let $T_{k,l} \in \C^{(n-r)\times(n-r)}$ be the permutation matrix exchanging rows $k$
and $l$. Let $i < j \in \interN{1}{n-r}$.
Define
\[
L_{ij} := E_2 T_{1,i} T_{2,j} 
\]
such that 
\[
L_{i,j} K L_{i,j}^T = \begin{bmatrix}
        0 & K_{i,j}\\
        -K{i,j} & 0' 
    \end{bmatrix}
\]
Let us plug the following gadget to both $N$ and $N'$
\begin{derivationchain}
    \scalebox{\scalefig}{\input{./figures/FpW-nf-uniqueness-proof/gadget_2.tikz}} 
\end{derivationchain}
we can derive
\begin{derivationchain}
    \scalebox{\scalefig}{\input{./figures/FpW-nf-uniqueness-proof/step5_1.tikz}} &
    $\eqann{\L2, \L3}{\substack{PM = \pmm{I_r\\B}\\ PZP^T = \pmm{0 & 0\\ 0 & K}}}$ &
    \scalebox{\scalefig}{\input{./figures/FpW-nf-uniqueness-proof/step5_2.tikz}} &
    $\eqann{\L27}{}$ &
    \scalebox{\scalefig}{\input{./figures/FpW-nf-uniqueness-proof/step5_3.tikz}} &
    $\eqann{\substack{\L4, \S\A6, \S\A5, \\\L26,\S\A9, \S\A11}}{}$ &
    \scalebox{\scalefig}{\input{./figures/FpW-nf-uniqueness-proof/step5_4.tikz}} \\[1em]

    \scalebox{\scalefig}{\input{./figures/FpW-nf-uniqueness-proof/step5_4.tikz}} &
    $\eqann{\D2}{L_{ij} := E_2 T_{1,i} T_{2,j} }$ &
    \scalebox{\scalefig}{\input{./figures/FpW-nf-uniqueness-proof/step5_5.tikz}} &
    $\eqann{\S\A6, \L14}{L_{i,j} K L_{i,j}^T = \pmm{0 & K_{i,j}\\-K{i,j} & 0' }}$ &
    \scalebox{\scalefig}{\input{./figures/FpW-nf-uniqueness-proof/step5_6.tikz}} &
    $\eqann{\L13, \D1, \L20}{}$ &
    \scalebox{\scalefig}{\input{./figures/FpW-nf-uniqueness-proof/step5_7.tikz}} \\[1em]

    \scalebox{\scalefig}{\input{./figures/FpW-nf-uniqueness-proof/step5_7.tikz}} &
    $\eqann{\A6, \A37}{}$ &
    \scalebox{\scalefig}{\begin{tikzpicture}
	\begin{pgfonlayer}{nodelayer}
		\node [style=dot] (0) at (-1.5, 0) {};
		\node [style=dot] (1) at (1.5, 0) {};
		\node [style=white phase dot] (2) at (0, 0) {$xK_{i,j}$};
	\end{pgfonlayer}
	\begin{pgfonlayer}{edgelayer}
		\draw (0) to (1);
	\end{pgfonlayer}
\end{tikzpicture}
} &
\end{derivationchain}
And similarly,
\begin{derivationchain}
    \scalebox{\scalefig}{\input{./figures/FpW-nf-uniqueness-proof/step5_9.tikz}} &
    $\eqann{}{}$ &
    \scalebox{\scalefig}{\begin{tikzpicture}
	\begin{pgfonlayer}{nodelayer}
		\node [style=dot] (0) at (-1.5, 0) {};
		\node [style=dot] (1) at (1.5, 0) {};
		\node [style=white phase dot] (2) at (0, 0) {$x'K'_{i,j}$};
	\end{pgfonlayer}
	\begin{pgfonlayer}{edgelayer}
		\draw (0) to (1);
	\end{pgfonlayer}
\end{tikzpicture}
} &
\end{derivationchain}
Then $x= x'$ implies $K_{i,j} = K'_{i,j}$.
\medskip

\noindent
\(\boxed{\text{\itshape Step 6}:\ M=M'}\)

We now know that
\[
    r=r',
    \qquad
    P=P'
    \qquad
    x=x'\neq0
    \qquad
    Z=Z'
\]
And
\[
    PM
    =
    \begin{bmatrix}
        I_r\\
        B
    \end{bmatrix},
    \qquad
    PM'
    =
    \begin{bmatrix}
        I_r\\
        B'
    \end{bmatrix}.
\]
It remains only to prove
\[
    B=B'.
\]

Let
\[
    1\leq i\leq r,
    \qquad
    1\leq q\leq n-r,
\]
and let $S_{i,r+q}\in\C^{n\times n}$ be the permutation matrix
exchanging rows $i$ and $r+q$.

Define
\[
    R_{i,q}
    :=
    E_rS_{i,r+q}P
    \in\C^{r\times n}.
\]
such that
\[
    R_{i,q}M
    =
    E_rS_{i,r+q}
    \begin{bmatrix}
        I_r\\
        B
    \end{bmatrix}.
\]
This is the identity matrix in which the $i$-th row has been replaced
by the $q$-th row of $B$. Its determinant can be computed as:

\[
        \det\!\left(R_{i,q}M\right)
        =
        \left|
        \begin{smallarray}{c c c|c|c c c}
                &           &           & 0       &           &       & \\
                & I_{i-1}   &           & \svdots &           & (0)   & \\
                &           &           & 0       &           &       & \\
        \hline
        B_{q,1}   & \cdots    & B_{q,i-1} & B_{q,i} & B_{q,i+1}& \cdots& B_{q,r}\\
        \hline
                &           &           & 0       &           &       & \\
                & (0)       &           & \svdots &           & I_{r-i}& \\
                &           &           & 0       &           &       &
        \end{smallarray}
        \right|
        =(-1)^{i+i}B_{q,i} = B_{q,i}.
        \]
Similarly,
\[
    \det(R_{i,q}M')=(-1)^{i+i}B'_{q,i} = B'_{q,i}.
\]
Since
\[
    \interp{N}=\interp{N'}
    \qquad\text{and}\qquad
    x=x'\neq0,
\]
we may plug the following gadget to both $N$ and $N'$
\begin{derivationchain}
    \scalebox{\scalefig}{\input{./figures/FpW-nf-uniqueness-proof/gadget_3.tikz}} 
\end{derivationchain}
We get
\begin{derivationchain}
    \scalebox{\scalefig}{\input{./figures/FpW-nf-uniqueness-proof/step6_1.tikz}} &
    $\eqann{\S\A1}{}$ &
    \scalebox{\scalefig}{\input{./figures/FpW-nf-uniqueness-proof/step6_2.tikz}} &
    $\eqann{\L10, \L20,\S\A6}{}$ &
    \scalebox{\scalefig}{\input{./figures/FpW-nf-uniqueness-proof/step6_3.tikz}} &
    $\eqann{\D2,\L3, \L4}{S_{i,r+q}PM = \pmm{R_{i,q}M\\ *}}$ &
    \scalebox{\scalefig}{\input{./figures/FpW-nf-uniqueness-proof/step6_4.tikz}} \\[1em]

    \scalebox{\scalefig}{\input{./figures/FpW-nf-uniqueness-proof/step6_4.tikz}} &
    $\eqann{\S\A6}{}$ &
    \scalebox{\scalefig}{\input{./figures/FpW-nf-uniqueness-proof/step6_5.tikz}} &
    $\eqann{\substack{\L24\\ \S\A6, \S\A5, \\\L26,\S\A9, \S\A11}}{\det(R_{i,q}M) = B_{q,i}}$ &
    \scalebox{\scalefig}{\input{./figures/FpW-nf-uniqueness-proof/step6_6.tikz}}
\end{derivationchain}

The two sides respectively have interpretation
to
\[
    xB_{q,i}
    \qquad\text{and}\qquad
    xB'_{q,i}.
\]
Since $x\neq0$, $\interp{N}=\interp{N'}$ gives
\[
    B_{q,i}=B'_{q,i}.
\]
Since $i$ and $q$ were arbitrary,
\[
    B=B'.
\]
Hence
\[
    PM=PM',
\]
and since $P$ is invertible,
\[
    M=M'.
\]

Hence $(r,x,Z,M)=(r',x',Z',M')$, and the two normal forms are identical.

\subsection{Completeness of $\EpW$}

\textbf{To normal form.}
Let $D:0\to n$ be a diagram of $\S\EpW$.  Using \L17, we first slide every cup to the
top-left corner of the diagram and every cap to the top-right corner.  The
remainder is an $\EaW$-diagram and can therefore be rewritten as a matrix
arrow. We therefore start from $D$ written as:
\[
 \scalebox{\scalefig}{\input{./figures/EpW-proof/pre-nf_1.tikz}}
\]
Then, we derive:

\begin{derivationchain}
    \scalebox{\scalefig}{\input{./figures/EpW-proof/pre-nf_1.tikz}} &
    $\eqann{D2,\L15}{\substack{M_1 = \pmm{M_2 & M_3}\\ Z_1 = M_2M_3^T -M_3M_2^T}}$ &
    \scalebox{\scalefig}{\input{./figures/EpW-proof/pre-nf_2.tikz}} &
    $\eqann{\L23}{Z_1 = P_1\pmm{A_1 & 0 \\ 0 & 0}P_1^T}$ &
    \scalebox{\scalefig}{\input{./figures/EpW-proof/pre-nf_3.tikz}} &
    $\eqann{\L21}{}$ &
    \scalebox{\scalefig}{\input{./figures/EpW-proof/pre-nf_4.tikz}} \\[1em]
    \scalebox{\scalefig}{\input{./figures/EpW-proof/pre-nf_4.tikz}} &
    $\eqann{\L4,D3,D3,\L15}{\substack{P_1 = \pmm{M_6 & *\\M_5 & *\\M_4 & *}\\ Z_2 = M_5^TM_6 - M_6^TM_5}}$ &
    \scalebox{\scalefig}{\input{./figures/EpW-proof/pre-nf_5.tikz}} &
    $\eqann{\S\A16,\S\A5}{}$ &
    \scalebox{\scalefig}{\input{./figures/EpW-proof/pre-nf_6.tikz}} &
    $\eqann{\L22,\S\A1,\L10,\L22}{Z_3 = A_1^-1 + Z_2}$ &
    \scalebox{\scalefig}{\input{./figures/EpW-proof/pre-nf_7.tikz}} \\[1em]
    \scalebox{\scalefig}{\input{./figures/EpW-proof/pre-nf_7.tikz}} &
    $\eqann{\L23}{\substack{Z_3 = P_2\pmm{A_3 & 0 \\ 0 & 0}P_2^T}}$ &
    \scalebox{\scalefig}{\input{./figures/EpW-proof/pre-nf_8.tikz}} &
    $\eqann{\L2,\L3}{}$ &
    \scalebox{\scalefig}{\input{./figures/EpW-proof/pre-nf_9.tikz}} &
    $\eqann{\L18}{}$ &
    \scalebox{\scalefig}{\input{./figures/EpW-proof/pre-nf_10.tikz}} \\[1em]
    \scalebox{\scalefig}{\input{./figures/EpW-proof/pre-nf_10.tikz}} &
    $\eqann{\Dist}{x_2 = (x_1+1)\det P_2^T -1}$ &
    \scalebox{\scalefig}{\input{./figures/EpW-proof/pre-nf_11.tikz}} &
    $\eqann{\S\A6}{}$ &
    \scalebox{\scalefig}{\input{./figures/EpW-proof/pre-nf_12.tikz}} &
    $\eqann{\L17}{}$ &
    \scalebox{\scalefig}{\input{./figures/EpW-proof/pre-nf_13.tikz}} \\[1em]
    \scalebox{\scalefig}{\input{./figures/EpW-proof/pre-nf_13.tikz}} &
    $\eqann{\L21}{\substack{r_3 = r_1-r_2 \\ x_3 = (x_2+1)(-1)^{r_2(r_1-r_2)}}}$ &
    \scalebox{\scalefig}{\input{./figures/EpW-proof/pre-nf_14.tikz}} &
    $\eqann{D2}{\substack{M_4P_2^{-T} = \pmm{M_7 & M_8}\\ Z_4 = M_7A_3^{-1}M_7^T}}$ &
    \scalebox{\scalefig}{\input{./figures/EpW-proof/pre-nf_15.tikz}} &
    {} &
    {} \\[1em]
\end{derivationchain}

Similarly to what was done for \S\FpW, 
if $\mathrm{rank}(M_8) \neq 2r_3$, we can apply \L9, \L18, \S\A14, and \S\A13~
to get to one of the normal forms for the zero diagram of type $0 \to n$ (depending of the parity), 
which is essentially the same as for \S\FpW~ up to \L30. 
Otherwise, we can assume that $\mathrm{rank}(M_8) = 2r_3 \leq n$. Then
Using \L24, $\S\A26$, and $\S\A9$, we can multiply $M_8$ on the right by a
sequence of elementary matrices so as to put it in Reduced Column Echelon
Form (RCEF), while multiplying $x_3$ by the determinant of these matrices.
This gives $M_9$ full rank and in RCEF, and $x_4$ so that $D$ reads:
\[
\scalebox{\scalefig}{\input{./figures/EpW-proof/nf_1.tikz}} 
\]
As $M_9$ is in RCEF, \cref{lemma:dec-rcef} gives $P_3$ a permutation matrix such that:
\[M_9 = P_3\pmm{I_{2r_3}\\B}\]
We then derive:

\begin{derivationchain}
    \scalebox{\scalefig}{\input{./figures/EpW-proof/nf_1.tikz}} &
    $\eqann{\L8,\L2,\L3,D3}{M_9 = P_3\pmm{I_{2r_3}\\B}}$ &
    \scalebox{\scalefig}{\input{./figures/EpW-proof/nf_2.tikz}} &
    $\eqann{\L14}{Z_5 = P_3^{-1}Z_4P_3^{-T}}$ &
    \scalebox{\scalefig}{\input{./figures/EpW-proof/nf_3.tikz}} &
    $\eqann{\L13}{Z_5 = \pmm{Z_6 & C \\ -C^T & Z_7}}$ &
    \scalebox{\scalefig}{\input{./figures/EpW-proof/nf_4.tikz}} \\[1em]
    \scalebox{\scalefig}{\input{./figures/EpW-proof/nf_4.tikz}} &
    $\eqann{\S\A1}{}$ &
    \scalebox{\scalefig}{\input{./figures/EpW-proof/nf_5.tikz}} &
    $\eqann{\L19}{}$ &
    \scalebox{\scalefig}{\input{./figures/EpW-proof/nf_6.tikz}} &
    $\eqann{\L2,\L14,\S\A1,\L10}{Z_8 = BZ_6B^T +Z_7}$ &
    \scalebox{\scalefig}{\input{./figures/EpW-proof/nf_8.tikz}} \\[1em]
    \scalebox{\scalefig}{\input{./figures/EpW-proof/nf_8.tikz}} &
    $\eqann{\S\A3}{}$ &
    \scalebox{\scalefig}{\input{./figures/EpW-proof/nf_9.tikz}} &
    $\eqann{\S\A18}{}$ &
    \scalebox{\scalefig}{\input{./figures/EpW-proof/nf_10.tikz}} &
    $\eqann{\L12,\S\A13}{}$ &
    \scalebox{\scalefig}{\input{./figures/EpW-proof/nf_11.tikz}} \\[1em]
    \scalebox{\scalefig}{\input{./figures/EpW-proof/nf_11.tikz}} &
    $\eqann{\L19}{}$ &
    \scalebox{\scalefig}{\input{./figures/EpW-proof/nf_12.tikz}} &
    $\eqann{\L2,\L3}{}$ &
    \scalebox{\scalefig}{\input{./figures/EpW-proof/nf_13.tikz}} &
    $\eqann{\S\A1}{}$ &
    \scalebox{\scalefig}{\input{./figures/EpW-proof/nf_14.tikz}} \\[1em]
    \scalebox{\scalefig}{\input{./figures/EpW-proof/nf_14.tikz}} &
    $\eqann{\L15}{}$ &
    \scalebox{\scalefig}{\input{./figures/EpW-proof/nf_15.tikz}} &
    $\eqann{\S\A1,\L10}{Z_9 = Z_8 +C^TB^T - BC}$ &
    \scalebox{\scalefig}{\input{./figures/EpW-proof/nf_16.tikz}} &
    $\eqann{\S\A3,\S\A6}{}$ &
    \scalebox{\scalefig}{\input{./figures/EpW-proof/nf_17.tikz}} \\[1em]
    \scalebox{\scalefig}{\input{./figures/EpW-proof/nf_17.tikz}} &
    $\eqann{\L5}{}$ &
    \scalebox{\scalefig}{\input{./figures/EpW-proof/nf_18.tikz}} &
    $\eqann{\L13,\L2,\L3}{\substack{M_9 = P_3\pmm{I_{2r_3}\\B}\\ Z_{10} = P_3\pmm{Z_9 & 0 \\ 0 & 0}}}P_3^T$ &
    \scalebox{\scalefig}{\input{./figures/EpW-proof/nf_19.tikz}} &
    {} &
    {} \\[1em]
\end{derivationchain}


\textbf{Uniqueness of normal form.}

Let
\[
    N=N(r,x,Z,M)
    \qquad\text{and}\qquad
    N'=N(r',x',Z',M')
\]
be two $\S\EpW$ normal forms of the same type $0\to n$, where
\(
    M\in\C^{n\times 2r},
    M'\in\C^{n\times 2r'}
\) 
are in RCEF and have ranks $2r$ and $2r'$ respectively. Let
\(
    1\leq j_1<\cdots<j_{2r}\leq n,
    1\leq j'_1<\cdots<j'_{2r'}\leq n
\) 
be their respective pivot rows.

Assume that
\[
    \interp{N}=\interp{N'},
\quad \text{i.e.} \quad
            \interp{\scalebox{\scalefig}{\begin{tikzpicture}
	\begin{pgfonlayer}{nodelayer}
		\node [style=none] (4) at (1.75, -0.75) {};
		\node [style=rmat] (5) at (0.75, -0.75) {};
		\node [style=none] (6) at (0.75, 0) {$M$};
		\node [style=modenb] (7) at (1.5, -0.5) {$n$};
		\node [style=rw] (9) at (2.75, 0) {};
		\node [style=none] (11) at (3.75, 0) {};
		\node [style=rw] (12) at (1, 1) {};
		\node [style=none] (13) at (1.5, 1) {};
		\node [style=none] (14) at (1, 1.75) {$Z$};
		\node [style=none] (15) at (-1.5, -1.25) {};
		\node [style=none] (16) at (-1, -1) {};
		\node [style=none] (17) at (-1.25, -1.25) {};
		\node [style=fswap] (18) at (-1.75, -0.75) {};
		\node [style=lw] (19) at (-2.25, -1) {};
		\node [style=lw] (20) at (-2.25, -0.5) {};
		\node [style=none] (21) at (-1.5, -0.25) {};
		\node [style=none] (22) at (-1.75, -0.75) {};
		\node [style=none] (23) at (-1, -0.5) {};
		\node [style=none] (24) at (-1.25, -0.75) {};
		\node [style=none] (25) at (-1.25, -0.25) {};
		\node [style=none] (26) at (-0.5, -0.75) {};
		\node [style=modenb] (27) at (-1.5, 0) {$r$};
		\node [style=modenb] (28) at (-0.25, -0.5) {$2r$};
		\node [style=white phase dot] (162) at (-0.25, 1.5) {$x$};
		\node [style=none] (163) at (-0.75, 1.5) {};
		\node [style=none] (164) at (-0.75, 0.5) {};
		\node [style=none] (165) at (0.25, 1.5) {};
		\node [style=none] (166) at (0.25, 0.5) {};
	\end{pgfonlayer}
	\begin{pgfonlayer}{edgelayer}
		\draw [style=thick] (11.center) to (9);
		\draw [style=thick, in=0, out=-150, looseness=1.25] (9) to (4.center);
		\draw [style=thick, in=150, out=0] (13.center) to (9);
		\draw [style=thick] (13.center) to (12);
		\draw [style=thick, in=-120, out=0] (17.center) to (16.center);
		\draw [style=thick] (15.center) to (17.center);
		\draw [style=thick, in=180, out=30] (20) to (21.center);
		\draw [style=thick, in=120, out=0] (25.center) to (23.center);
		\draw [style=thick, in=-60, out=180, looseness=1.25] (26.center) to (23.center);
		\draw [style=thick] (21.center) to (25.center);
		\draw [style=thick, bend right=90, looseness=2.00] (20) to (19);
		\draw [style=thick, in=150, out=-30, looseness=1.25] (20) to (22.center);
		\draw [style=thick, in=-180, out=-30] (19) to (15.center);
		\draw [style=thick, in=-90, out=-30, looseness=1.25] (22.center) to (24.center);
		\draw [style=thick, bend left=90, looseness=2.00] (19) to (20);
		\draw [style=thick, in=-150, out=30, looseness=1.25] (19) to (22.center);
		\draw [style=thick, in=90, out=30, looseness=1.25] (22.center) to (24.center);
		\draw [style=thick, in=60, out=-180, looseness=1.25] (26.center) to (16.center);
		\draw [style=thick] (26.center) to (4.center);
		\draw [bend right=90, looseness=1.25] (163.center) to (164.center);
		\draw [bend left=90, looseness=1.25] (165.center) to (166.center);
		\draw (165.center) to (163.center);
		\draw (166.center) to (164.center);
	\end{pgfonlayer}
\end{tikzpicture}
}}
            \eq
            \interp{\scalebox{\scalefig}{\begin{tikzpicture}
	\begin{pgfonlayer}{nodelayer}
		\node [style=none] (4) at (1.75, -0.75) {};
		\node [style=rmat] (5) at (0.75, -0.75) {};
		\node [style=none] (6) at (0.75, 0) {$M'$};
		\node [style=modenb] (7) at (1.5, -0.5) {$n$};
		\node [style=rw] (9) at (2.75, 0) {};
		\node [style=none] (11) at (3.75, 0) {};
		\node [style=rw] (12) at (1, 1) {};
		\node [style=none] (13) at (1.5, 1) {};
		\node [style=none] (14) at (1, 1.75) {$Z'$};
		\node [style=none] (15) at (-1.5, -1.25) {};
		\node [style=none] (16) at (-1, -1) {};
		\node [style=none] (17) at (-1.25, -1.25) {};
		\node [style=fswap] (18) at (-1.75, -0.75) {};
		\node [style=lw] (19) at (-2.25, -1) {};
		\node [style=lw] (20) at (-2.25, -0.5) {};
		\node [style=none] (21) at (-1.5, -0.25) {};
		\node [style=none] (22) at (-1.75, -0.75) {};
		\node [style=none] (23) at (-1, -0.5) {};
		\node [style=none] (24) at (-1.25, -0.75) {};
		\node [style=none] (25) at (-1.25, -0.25) {};
		\node [style=none] (26) at (-0.5, -0.75) {};
		\node [style=modenb] (27) at (-1.5, 0) {$r'$};
		\node [style=modenb] (28) at (-0.25, -0.5) {$2r'$};
		\node [style=white phase dot] (162) at (-0.25, 1.5) {$x'$};
		\node [style=none] (163) at (-0.75, 1.5) {};
		\node [style=none] (164) at (-0.75, 0.5) {};
		\node [style=none] (165) at (0.25, 1.5) {};
		\node [style=none] (166) at (0.25, 0.5) {};
	\end{pgfonlayer}
	\begin{pgfonlayer}{edgelayer}
		\draw [style=thick] (11.center) to (9);
		\draw [style=thick, in=0, out=-150, looseness=1.25] (9) to (4.center);
		\draw [style=thick, in=150, out=0] (13.center) to (9);
		\draw [style=thick] (13.center) to (12);
		\draw [style=thick, in=-120, out=0] (17.center) to (16.center);
		\draw [style=thick] (15.center) to (17.center);
		\draw [style=thick, in=180, out=30] (20) to (21.center);
		\draw [style=thick, in=120, out=0] (25.center) to (23.center);
		\draw [style=thick, in=-60, out=180, looseness=1.25] (26.center) to (23.center);
		\draw [style=thick] (21.center) to (25.center);
		\draw [style=thick, bend right=90, looseness=2.00] (20) to (19);
		\draw [style=thick, in=150, out=-30, looseness=1.25] (20) to (22.center);
		\draw [style=thick, in=-180, out=-30] (19) to (15.center);
		\draw [style=thick, in=-90, out=-30, looseness=1.25] (22.center) to (24.center);
		\draw [style=thick, bend left=90, looseness=2.00] (19) to (20);
		\draw [style=thick, in=-150, out=30, looseness=1.25] (19) to (22.center);
		\draw [style=thick, in=90, out=30, looseness=1.25] (22.center) to (24.center);
		\draw [style=thick, in=60, out=-180, looseness=1.25] (26.center) to (16.center);
		\draw [style=thick] (26.center) to (4.center);
		\draw [bend right=90, looseness=1.25] (163.center) to (164.center);
		\draw [bend left=90, looseness=1.25] (165.center) to (166.center);
		\draw (165.center) to (163.center);
		\draw (166.center) to (164.center);
	\end{pgfonlayer}
\end{tikzpicture}
}}.
\]

Notice that using \L30, we get back a normal form for $\S\FpW$ which we already know is unique.
\[
N_{\EpW}(r,x,Z,M) \qquad \text{becomes} \qquad N_{\FpW}(2r,(-1)^{r(r-1)/2}(x+1), Z, M)
\]

\[
N_{\EpW}(r',x',Z',M') \qquad \text{becomes} \qquad N_{\FpW}(2r',(-1)^{r'(r'-1)/2}(x'+1), Z', M')
\]

We can directly deduce that 
\[
(r,x,Z,M) = (r',x',Z',M').
\]

\end{document}